\documentclass[fleqn,11pt,a4paper]{report}
\usepackage{front,front_report}
\usepackage{makeidx}
\makeindex

\usepackage{tikz}
\usetikzlibrary{decorations.pathmorphing,chains,matrix,positioning,scopes}
\pgfrealjobname{semi}
\newlength{\unit}
\usepackage{pifont}
\usepackage[geometry]{ifsym}
\newcommand{\bulletc}{\tiny\FilledSmallCircle}
\newcommand{\bullets}{\tiny\FilledSmallSquare}
\newcommand{\bulletd}{\scriptsize\FilledSmallDiamondshape}

\usepackage{amsmath}
\allowdisplaybreaks
\usepackage{mathtools,atmath}
\usepackage{amssymb}

\usepackage[amsmath,thmmarks]{ntheorem}
\usepackage{ntheorem_extbig}

\newcommand{\Half}[1][1]{\frac{#1}{2}}
\newcommand{\tHalf}[1][1]{\tfrac{#1}{2}}

\newcommand{\rG}{\mathrm G}
\newcommand{\rH}{\mathrm H}
\newcommand{\rP}{\mathrm P}
\newcommand{\cG}{\mathcal G}
\newcommand{\cH}{\mathcal H}

\newcommand{\hi}{\hat\imath}
\newcommand{\hj}{\hat\jmath}
\newcommand{\hk}{\hat k}

\newcommand{\hl}{\hat l}
\newcommand{\Hs}{\hat s}
\newcommand{\Ht}{\hat t}
\newcommand{\Hu}{\hat u}

\newcommand{\Ta}{\tilde a}
\newcommand{\Tb}{\tilde b}

\newcommand{\Tm}{\tilde m}
\newcommand{\Tn}{\tilde n}
\newcommand{\Tp}{\tilde p}
\newcommand{\Tt}{\tilde t}

\newcommand{\Id}{\mathrm I}

\newcommand{\Fx}{\rule{2pt}{0pt}}
\newcommand{\Dx}{\rule{2pt}{2pt}}
\newcommand{\Ux}{\rule[2pt]{2pt}{2pt}}

\newcommand{\tfbox}[1]{{\setlength\fboxsep{0pt}\fbox{\rule{0pt}{4pt}#1}}}

\newcommand{\ssub}{\tfbox{\Dx\Fx}}
\newcommand{\presuf}{\tfbox{\Fx\Dx}}
\newcommand{\sufpre}{\tfbox{\Ux\Fx}}
\newcommand{\subs}{\tfbox{\Fx\Ux}}

\newcommand{\ssubX}{\tfbox{\rule{2pt}{4pt}\rule[0pt]{2pt}{2pt}}}
\newcommand{\subsX}{\tfbox{\rule[2pt]{2pt}{2pt}\rule{2pt}{4pt}}}

\newcommand{\why}[1]{\tag*{\small (#1)}}

\newcommand{\charguard}{\textsf{\$}}
\newcommand{\charwild}{\textrm ?}
\usepackage{subfig}
\usepackage{verbatim,url}
\usepackage{ifthen,atprog}
\usepackage{multirow}

\newcommand{\extra}[1]{}
\begin{document}

\newcommand{\mytitle}{%
Semi-local string comparison:\\ Algorithmic techniques and applications}

\newcommand{\mytitleshort}{%
Semi-local string comparison: Algorithmic techniques and applications}

\newcommand{\myabstract}{%
A classical measure of string comparison 
is given by the longest common subsequence (LCS) problem on a pair of strings.
We consider its generalisation, called the semi-local LCS problem,
which arises naturally in many string-related problems.
The semi-local LCS problem asks for the LCS scores 
for each of the input strings
against every substring of the other input string,
and for every prefix of each input string
against every suffix of the other input string.
Such a comparison pattern provides 
a much more detailed picture of string similarity
than a single LCS score;
it also arises naturally in many string-related problems.
In fact, the semi-local LCS problem 
turns out to be fundamental for string comparison,
providing a powerful and flexible alternative to classical dynamic programming.
It is especially useful when the input 
to a string comparison problem may not be available all at once:
for example, comparison of dynamically changing strings;
comparison of compressed strings;
parallel string comparison.
The same approach can also be applied to permutation strings,
providing efficient solutions for local versions 
of the longest increasing subsequence (LIS) problem,
and for the problem of computing a maximum clique in a circle graph.
Furthermore, the semi-local LCS problem 
turns out to have surprising connections in a few seemingly unrelated fields, 
such as computational geometry and algebra of semigroups.
This work is devoted to exploring the structure of the semi-local LCS problem,
its efficient solutions, 
and its applications in string comparison and other related areas, 
including computational molecular biology.}

\renewcommand{\mythanks}{%
Research supported by 
the Centre for Discrete Mathematics and Its Applications (DIMAP), 
University of Warwick, 
and by the Royal Society Leverhulme Trust Senior Research Fellowship.}

\myfront

\tableofcontents

%%===========================================================================%%
\mychapter{Introduction}
\label{c-intro}

A classical measure of string comparison 
is given by the longest common subsequence (LCS) problem.
Given two strings $a$, $b$ of lengths $m$, $n$ respectively,
the LCS problem asks for the length of the longest possible string
that is a subsequence of both $a$ and $b$.
This length is called the strings' LCS score.
The LCS problem has numerous applications 
both within and outside computer science.
We refer the reader to monographs \cite{Crochemore_Rytter:94,Gusfield:97}
for the background and further references.

For a more detailed approach to string comparison,
let us consider the following generalisation of the LCS problem.
Given two strings $a$, $b$ as before,
the semi-local LCS problem asks for the LCS score 
of each string against all substrings of the other string,
and of all prefixes of each string 
against all suffixes of the other string.
Such a comparison pattern provides 
a much more detailed picture of string similarity
than a single LCS score;
it also arises naturally in many string-related problems.
Upon closer look, the semi-local LCS problem 
turns out to be a powerful and flexible alternative to classical dynamic programming
in situations where the input to a string comparison problem 
may not be available all at once:
for example, comparison of dynamically changing strings;
comparison of compressed strings;
parallel string comparison.
Furthermore, the semi-local LCS problem 
turns out to have surprising connections in a few seemingly unrelated fields, 
such as computational geometry, algebra of semigroups, and graph theory.

This work is devoted to exploring the structure of the semi-local LCS problem,
its efficient solutions, 
and its applications in string comparison and other related areas, 
including computational molecular biology.

This work is organised as follows.
In \chapref{c-prelim}, we give the necessary preliminaries.
In \chapref{c-mmult}, we investigate the algebraic structure
underlying the semi-local LCS problem. 
This is done in two alternative forms:
as matrix distance multiplication on simple unit-Monge matrices,
and as a formal monoid of seaweed braids.
In \chapref{c-semi}, we establish rigorously the relationship 
between this structure and the semi-local LCS problem.
In \chapref{c-seaweed}, we use our structural results
to obtain a simple algorithm for the semi-local LCS problem.
We also show a number of this algorithm's applications.
In \chapref{c-weighted}, we generalise our techniques from LCS scores 
to arbitrary rational-weighted alignment scores and edit distances.
In \chaprefs{c-periodic}--\ref{c-compressed}, we apply our techniques
to several particular classes of string comparison problems:
comparison of a periodic string against a plain string;
comparison of permutation strings;
comparison of compressed strings.
In \chapref{c-network}, 
we explore the connection between semi-local string comparison 
and a subclass of comparison networks, known as transposition networks.
Using this connection, we develop algorithms 
for several important variants of the LCS problem:
parameterised, dynamic, bit-parallel and subword-parallel.
In \chapref{c-beyond}, we discuss ways of extending our techniques 
beyond semi-local string comparison,
towards the ultimate goal of detailed and efficient 
fully-local string comparison.
We also discuss an implementation of our method, which has been used 
to solve several problems in computational molecular biology.

\begin{comment}
\begin{figure}
\centering

\begin{tikzpicture}
%
\matrix [matrix of nodes,column sep=5mm,row sep=5mm,
  every node/.style={draw,anchor=center}]{
%
|(overview)| \ref{s-overview} & 
|(points)| \ref{s-points} &
|(distribution)| \ref{s-distribution} &
|(umonge)| \ref{s-umonge} & 
|(z) [coordinate]| \\
%
|(a) [coordinate]| & & & |(b) [coordinate]|\\
%
|(mmult)| \ref{s-mmult} & 
|(braids)| \ref{s-braids} & & & &
|(mmult-alg)| \ref{s-mmult-alg} & 
|(bruhat)| \ref{s-bruhat} \\
%
|(llcs)| \ref{s-llcs} & & & &
|(score-matrices)| \ref{s-score-matrices} &
|(score-mmult)| \ref{s-score-mmult} \\
%
|(d) [coordinate]| & |(e) [coordinate]| \\
%
|(seaweed)| \ref{s-seaweed} &
|(micro)| \ref{s-micro} &
|(incremental)| \ref{s-incremental}--\ref{s-lrs} &
|(h) [coordinate]| \\
%
|(weighted)| \ref{s-weighted} &
|(amatch)| \ref{s-amatch} \\
};
{ [start chain,every join/.style=->]
%
\chainin (points);
\chainin (distribution) [join];
\chainin (umonge) [join];
\chainin (z) [join=by -];
%
\chainin (b) [join=with umonge by -];
\chainin (a) [join=by -];
%
\chainin (mmult) [join];
\chainin (braids) [join];
\chainin (mmult-alg) [join];
\chainin (bruhat) [join];
%
\chainin (llcs);
\chainin (score-matrices) [join, join=with z];
\chainin (score-mmult) [join, join=with mmult-alg];
%
\chainin (e) [join=with braids by -];
\chainin (d) [join=by -];
%
\chainin (seaweed) [join];
\chainin (micro) [join];
\chainin (incremental) [join];
%
\chainin (weighted);
}
\end{tikzpicture}

\end{figure}
\end{comment}

Many results presented in this work 
appeared incrementally in the author's publications 
\cite{Tiskin:06_CSR,Tiskin:06_CPM,Tiskin:08_MCS,%
Tiskin:08_JDA,Tiskin:09_JMS,Krusche_Tiskin:09,Tiskin:09_CPM,Tiskin:11_CSR,%
Tiskin:13_EDBT,Tiskin:Algorithmica}.
The aim of this work is to consolidate these results,
unifying the terminology and notation.
However, a number of results are original to this work.

%%===========================================================================%%
%%===========================================================================%%
\mychapter{Preliminaries}
\label{c-prelim}

In this chapter, we give the necessary preliminaries for the rest of the work.
It is organised as follows.
In \secref{s-points}, we establish the terminology and notation,
borrowing main concepts from planar Euclidean geometry and matrix algebra.
In \secref{s-distribution}, we introduce 
some basic combinatorial operations on matrices.
In \secref{s-umonge}, we describe our main algorithmic tool:
a special class of integer matrices, called simple unit-Monge matrices.
These matrices are intimately related 
to the combinatorial concept of a permutation,
and the dominance counting problem arising in computational geometry.

%%=-=-=-=-=-=-=-=-=-=-=-=-=-=-=-=-=-=-=-=-=-=-=-=-=-=-=-=-=-=-=-=-=-=-=-=-=-=%%
\mysection{Points and matrices}
\label{s-points}

For indices, we will use either integers, or half-integers%
\footnote{The intuition behind using both integers and half-integers
is that we are dealing with planar grid-like graphs and, 
implicitly, with their dual graphs.
In this setting, it is natural to index the nodes of a primal graph
by pairs of integers, and the nodes of its dual graph
(corresponding to the faces of the primal graph)
by pairs of half-integers.}:
\begin{gather*}
\brc{\ldots, -2, -1, 0, 1, 2, \ldots} \\
\bigbrc{\ldots, -\tHalf[5], -\tHalf[3], -\tHalf[1], 
 \tHalf[1], \tHalf[3], \tHalf[5], \ldots}
\end{gather*}
For ease of reading, half-integer variables will be indicated by hats
(e.g.\ $\hi$, $\hj$).
Ordinary variable names 
(e.g.\ $i$, $j$, with possible subscripts or superscripts), 
will normally denote integer variables,
but can sometimes denote a variable 
that may be either integer, or half-integer.

\index{$i^-$, $i^+$: decrement/increment by $\tHalf$}%
It will be convenient to denote
\begin{gather*}
i^- = i-\tHalf \qquad i^+ = i+\tHalf
\end{gather*}
for any integer or half-integer $i$.
The set of all half-integers can now be written as
\begin{gather*}
\bigbrc{\ldots, (-3)^+, (-2)^+, (-1)^+, 0^+, 1^+, 2^+, \ldots}
\end{gather*}

\index{interval notation}%
We denote integer and half-integer \emph{intervals} by
\begin{gather*}
\index{$\protect\bra{i:j}$: integer interval}%
\index{$\protect\ang{i:j}$: half-integer interval}%   %% bug in AMSLaTeX?
\bra{i:j} = \brc{i, i+1, \ldots, j-1, j} \\
\ang{i:j} = \bigbrc{i^+, i+\tHalf[3], \ldots, j-\tHalf[3], j^-}
\end{gather*}
In both cases, the interval is defined by its integer endpoints.
For finite intervals $\bra{i:j}$ and $\ang{i:j}$,
we call the difference $j-i$ interval \emph{length}.
Note that an integer (respectively, half-integer) 
interval of length $n$ consists of $n+1$ (respectively, $n$) elements.

To denote infinite intervals of integers and half-integers,
we will use $-\infty$ and $+\infty$ where appropriate.
In particular, $\bra{-\infty:+\infty}$ denotes the set of all integers,
and $\ang{-\infty:+\infty}$ the set of all half-integers.

When dealing with pairs of numbers,
we will often use geometric language and call them \emph{points}.
\index{$\ll$, $\gtrless$: dominance orders}%
\index{order!$\ll$-dominance}%
\index{order!$\gtrless$-dominance}%
We define two natural strict partial orders on points,
called \emph{$\ll$-} and \emph{$\gtrless$-dominance}:
\begin{alignat*}{2}
{}
&(i_0,j_0) \ll (i_1,j_1)      &\quad &\text{if $i_0 < i_1$ and $j_0 < j_1$}\\
&(i_0,j_0) \gtrless (i_1,j_1) &\quad &\text{if $i_0 > i_1$ and $j_0 < j_1$}
\end{alignat*}
When visualising points, 
we will deviate from the standard Cartesian convention 
on the direction of the coordinate axes.
We will use instead the matrix indexing convention:
the first coordinate in a pair increases downwards,
and the second coordinate rightwards.
Hence, \emph{$\ll$-} and $\gtrless$-dominance correspond respectively
to the ``above-left'' and ``below-left'' partial orders.
The latter order corresponds visually 
to the standard definition of dominance in computational geometry.

\index{order!lexicographic}%
We also define the natural \emph{lexicographic order} on points:
point $(i_0,j_0)$ precedes point $(i_1,j_1)$ in this order,
if either $i_0 < i_1$, or $i_0 = i_1$ and $j_0 < j_1$.
The lexicographic order is a strict total order,
compatible with the partial $\ll$-dominance order.

We use standard terminology for special elements and subsets in partial orders.
\index{chain}%
In particular, a set of elements form a \emph{chain},
if they are pairwise comparable,
\index{antichain}%
and an \emph{antichain},
if they pairwise incomparable.
Note that a $\ll$-chain is a $\gtrless$-antichain, and vice versa.
An element in a partially ordered set is \emph{minimal} 
(respectively, \emph{maximal}), if, in terms of the partial order, 
it does not dominate (respectively, is not dominated by)
any other element in the set.
All minimal (respectively, maximal) elements 
in a partially ordered set form an antichain.

A function of an integer argument 
\index{function!unit-monotone}%
will be called \emph{unit-monotone increasing}
(respectively, \emph{decreasing}),
if for every successive pair of values,
the difference between the successor and the predecessor
is either $0$ or $1$ (respectively, $0$ or $-1$).

We will make extensive use of vectors and matrices
with integer (occasionally, also rational or real) elements, 
and with integer or half-integer indices%
\footnote{When integers and half-integers are used as matrix indices,
it is convenient to imagine that the matrices are written on squared paper.
The entries of an integer-indexed matrix 
are at integer points of line intersections;
the entries of a half-integer-indexed matrix 
are at half-integer points within the squares.}.
We regard a vector or matrix as a one- (respectively, two-) argument function,
so we can speak e.g.\ about unit-monotone increasing matrices.

% We will sometimes consider matrices where one or both index ranges
% are non-consecutive sets of integers or half-integers;
% however, index ranges will always be assumed to be linearly ordered.
\index{$(I \mid J)$: Cartesian product}%
Given two index ranges $I$, $J$,
it will be convenient to denote their Cartesian product by $(I \mid J)$.
\index{$\bra{i_0:i_1 \mid j_0:j_1}$: interval Cartesian product}%
\index{$\ang{i_0:i_1 \mid j_0:j_1}$: interval Cartesian product}%
We extend this notation to Cartesian products of intervals:
\begin{alignat*}{2}
{}
&\bra{i_0:i_1 \mid j_0:j_1} &&= (\bra{i_0:i_1} \mid \bra{j_0:j_1})\\
&\ang{i_0:i_1 \mid j_0:j_1} &&= (\ang{i_0:i_1} \mid \ang{j_0:j_1})
\end{alignat*}
Given index ranges $I$, $J$,
a \emph{vector over $I$} is indexed by $i \in I$,
and a \emph{matrix over $(I \mid J)$} is indexed by $i \in I$, $j \in J$.
A vector or matrix is \emph{nonnegative}, if all its elements are nonnegative.

\index{matrix!implicit}%
The matrices we consider can be \emph{implicit},
i.e.\ represented by a compact data structure
that supports access to every matrix element 
in a specified (typically small, but not necessarily constant) time.
If the query time is not given, it is assumed to be constant by default.

%%% REVISE

% When considering matrices over non-consecutive index ranges,
% we will occasionally perform operations on such matrices 
% as if they were over consecutive intervals.
% This will have the following meaning:
% we remap the ranges to consecutive intervals 
% preserving the linear order within each range,
% then we perform a matrix operation,
% and finally we remap the intervals back to the original ranges.

We will use the parenthesis notation for indexing matrices, e.g.\ $A(i,j)$.
We will also use a straightforward notation for selecting subvectors and submatrices:
for example, given a matrix $A$ over $\bra{0:n \mid 0:n}$, 
we denote by $A\bra{i_0:i_1 \mid j_0:j_1}$
the submatrix defined by the given sub-intervals.
\index{$*$: implicit range}%
A star $*$ will indicate that for a particular index, 
its whole range is being used,
e.g.\ $A\bra{* \mid j_0:j_1} = A\bra{0:n \mid j_0:j_1}$.
In particular, $A(*,j)$ and $A(i,*)$ will denote 
a full matrix column and row, respectively.

\index{$A^T$: matrix transpose}%
\index{$A^R$: matrix rotation}%
We will denote by $A^T$ the transpose of matrix $A$,
and by $A^R$ the matrix obtained from $A$
by counterclockwise 90-degree rotation.
Given a matrix $A$ over $\bra{0:n \mid 0:n}$ or $\ang{0:n \mid 0:n}$,
we have 
\begin{gather*}
A^T(i,j)=A(j,i)\qquad A^R(i,j)=A(j,n-i)
\end{gather*}
for all $i,j$.

%%=-=-=-=-=-=-=-=-=-=-=-=-=-=-=-=-=-=-=-=-=-=-=-=-=-=-=-=-=-=-=-=-=-=-=-=-=-=%%
\mysection{Distribution, density and Monge matrices}
\label{s-distribution}

We now introduce two fundamental combinatorial operations on matrices.
The first operation obtains 
an integer-indexed matrix from a half-integer-indexed matrix
by summing up, for each of the integer points,
all matrix elements that are $\gtrless$-dominated by the given point.
\begin{definition}
\label{def-distribution}
Let $D$ be a matrix over $\ang{i_0:i_1 \mid j_0:j_1}$.
\index{matrix!distribution}%
\index{$D^\Sigma$: distribution matrix}%
Its \emph{distribution matrix} $D^\Sigma$
over $\bra{i_0:i_1 \mid j_0:j_1}$ is defined by
\begin{gather*}
D^\Sigma(i,j) = 
\sum_{\hi \in \ang{i:i_1},\hj \in \ang{j_0:j}} D(\hi,\hj)
\end{gather*}
for all $i \in \bra{i_0:i_1}$, $j \in \bra{j_0:j_1}$.
\end{definition}

The second operation obtains a half-integer-indexed matrix 
from an integer-indexed matrix,
by taking a four-point difference around each given point.
\begin{definition}
\label{def-density}
Let $A$ be a matrix over $\bra{i_0:i_1 \mid j_0:j_1}$.
\index{matrix!density}%
\index{$A^\square$: density matrix}%
Its \emph{density matrix} $A^\square$
over $\ang{i_0:i_1 \mid j_0:j_1}$ is defined by
\begin{gather*}
A^\square(\hi,\hj) =
A\pa{\hi^+,\hj^-} - A\pa{\hi^-,\hj^-} -
A\pa{\hi^+,\hj^+} + A\pa{\hi^-,\hj^+}
\end{gather*}
for all $\hi \in \ang{i_0:i_1}$, $\hj \in \ang{j_0:j_1}$.
\end{definition}
\begin{example}
\label{ex-distr}
We have
\begin{gather*}
\bmat{0 & 1 & 0 \\ 1 & 0 & 0 \\ 0 & 0 & 1}^\Sigma = 
\bmat{0 & 1 & 2 & 3 \\ 0 & 1 & 1 & 2 \\ 0 & 0 & 0 & 1 \\ 0 & 0 & 0 & 0}
\qquad
\bmat{0 & 1 & 2 & 3 \\ 0 & 1 & 1 & 2 \\ 0 & 0 & 0 & 1 \\ 0 & 0 & 0 & 0}^\square =
\bmat{0 & 1 & 0 \\ 1 & 0 & 0 \\ 0 & 0 & 1}
\end{gather*}
\end{example}

The definitions of distribution and density matrices 
extend naturally to matrices over an infinite index range,
as long as the sum in \defref{def-distribution} is defined.

The operations of taking the distribution and the density matrix 
are close to be mutually inverse.
For any finite matrices $D$, $A$ as above, and for all $i$, $j$, we have
\begin{gather*}
D^{\Sigma\square} = D\\
A^{\square\Sigma}(i,j) = A(i,j) - A(i_1,j) - A(i,j_0) + A(i_1,j_0)
\end{gather*}
When matrix $A$ is restricted to have 
all zeros on its bottom-left boundary
(i.e.\ in the leftmost column and the bottom row),
the two operations become truly mutually inverse.
We introduce special terminology for such matrices.
\begin{definition}
\label{def-simple}
\index{matrix!simple}% 
Matrix $A$ over $\bra{i_0:i_1 \mid j_0:j_1}$ will be called \emph{simple}, 
if $A(i_1,j) = A(i,j_0) = 0$ for all $i$, $j$. 
Equivalently, $A$ is simple if $A^{\square\Sigma}=A$.
\end{definition}

The following classes of matrices
play an important role in optimisation theory
(see Burkard et al.\ \cite{Burkard+:96} and Burkard \cite{Burkard:07} 
for an extensive survey),
and also arise in graph and string algorithms.
\begin{definition}
\label{def-totally-monotone}
\index{matrix!totally monotone}%
Matrix $A$ is called \emph{totally monotone}, if
\begin{gather*}
A(i,j) > A(i,j') \implies A(i',j) > A(i',j')
\end{gather*}
for all $i \leq i'$, $j \leq j'$.
\end{definition}
\begin{definition}
\label{def-Monge}
\index{matrix!Monge}%
Matrix $A$ is called \emph{a Monge matrix}, if
\begin{gather*}
A(i,j) + A(i',j') \leq A(i,j') + A(i',j)
\end{gather*}
for all $i \leq i'$, $j \leq j'$.
Equivalently, matrix $A$ is a Monge matrix, 
if $A^\square$ is nonnegative.
\index{matrix!anti-Monge}%
Matrix $A$ is called \emph{an anti-Monge matrix}, if $-A$ is Monge.
\end{definition}
\noindent
It is easy to see that Monge matrices
form a subclass of totally monotone matrices.
The characterisation of Monge matrices via their density matrices 
given by \defref{def-Monge} 
is equivalent to the canonical structure theorem
for Monge matrices by Burdyuk and Trofimov \cite{Burdyuk_Trofimov:76}
and Bein and Pathak \cite{Bein_Pathak:96}
(see also \cite{Burkard+:96,Burkard:07}).

%%=-=-=-=-=-=-=-=-=-=-=-=-=-=-=-=-=-=-=-=-=-=-=-=-=-=-=-=-=-=-=-=-=-=-=-=-=-=%%
\mysection{Permutation and unit-Monge matrices}
\label{s-umonge}

Our techniques will rely on structures 
that have permutations as their basic building blocks.
We will be dealing with permutations in matrix form,
exploiting the symmetry between indices and elements of a permutation.
\begin{definition}
\label{def-permutation}
\index{matrix!permutation}%
\index{matrix!subpermutation}% 
A \emph{permutation}% 
\footnote{Strictly speaking, such a matrix corresponds 
to a bijection, rather than a permutation,
since the ranges of column and row indices may be different
(although they must be of the same cardinality).
Therefore, there is a slight abuse of terminology
in calling it a permutation matrix, rather than a ``bijection matrix''.} %
(respectively, \emph{subpermutation}) \emph{matrix}
is a zero-one matrix 
containing exactly one (respectively, at most one) nonzero 
in every row and every column.
\end{definition}
\begin{example}
The $3 \times 3$ matrix in \exref{ex-distr} is a permutation matrix.
\end{example}

Typically, permutation and subpermutation matrices
will be indexed by half-integers.

\index{matrix!permutation!identity}%
\index{$\Id$: identity matrix}%
\index{matrix!permutation!offset identity}%
\index{$\Id_h$: offset identity matrix}%
An \emph{identity matrix} is a permutation matrix $\Id$ 
over an interval range $\ang{i_0:i_1 \mid i_0:i_1}$,
such that $\Id(\hi,\hj)=1$, iff $\hi=\hj$.
More generally, an \emph{offset identity matrix} 
is a permutation matrix $\Id_h$ 
over an interval range $\ang{i_0:i_1 \mid j_0:j_1}$,
where $j_0-i_0 = j_1-i_1 = h$,
such that $\Id_h(\hi,\hj)=1$, iff $\hj-\hi=h$.
Note that $\Id_0=\Id$.
Clearly, an identity or offset identity matrix
can be represented implicitly in constant space
and with constant query time.
When dealing with identity and offset identity matrices,
we will often omit their index ranges,
as long as they are clear from the context.

When dealing with (sub)permutation matrices,
we will write ``nonzeros'' for ``index pairs corresponding to nonzeros'',
as long as this does not lead to confusion.

Due to the extreme sparsity of (sub)permutation matrices, 
it would obviously be wasteful and inefficient to store them explicitly.
Instead, we will normally assume that a permutation matrix $P$ of size $n$
is given implicitly by the underlying permutation and its inverse,
i.e.\ by a pair of arrays $\pi$, $\pi^{-1}$,
such that $P\bigpa{\hi,\pi(\hi)} = 1$ for all $\hi$,
and $P\bigpa{\pi^{-1}(\hj),\hj} = 1$ for all $\hj$.
This compact representation has size $O(n)$,
and allows constant-time access to each nonzero of $P$
by its row index, as well as by its column index.
The implicit representation for subpermutation matrices is analogous.

% \index{$\cdot$: induced subrange}%
% Given a permutation matrix $P$ over $(I \mid J)$, and a set $I' \subseteq I$, 
% we will denote by $P(I' \mid \cdot)$
% the permutation submatrix \emph{row-induced by $I'$},
% i.e.\ the permutation submatrix obtained by deleting from $P$ 
% all columns in $I \setminus I'$,
% and then deleting from the remaining submatrix all zero rows.
% A column-induced permutation submatrix $P(\cdot \mid J')$ 
% is defined analogously.
% Both these operations can be implemented in linear time
% by a sweep of the nonzeros of matrix $P$.

The following subclasses of Monge matrices 
will play a crucial role in this work.
\begin{definition}
\label{def-unit-Monge}
\index{matrix!unit-Monge}%
\index{matrix!subunit-Monge}%
\index{matrix!unit-anti-Monge}%
\index{matrix!subunit-anti-Monge}%
Matrix $A$ is called \emph{a unit-Monge} 
(respectively, \emph{subunit-Monge}) \emph{matrix}, 
if $A^\square$ is a permutation (respectively, subpermutation) matrix.
Matrix $A$ is called \emph{a unit-anti-Monge} 
(respectively, \emph{subunit-anti-Monge}) \emph{matrix}, 
if $-A$ is unit-Monge (respectively, subunit-Monge).
\end{definition}
\noindent
By \defrefs{def-Monge}, \ref{def-unit-Monge}, 
any unit-Monge matrix is subunit-Monge, 
and any subunit-Monge matrix is Monge
(since the corresponding density matrix $A^\square$ 
is a (sub)permutation matrix, and hence nonnegative).
Similar inclusions hold for (sub)unit-anti-Monge matrices.
\begin{example}
The $4 \times 4$ matrix in \exref{ex-distr} is unit-Monge.
It is also simple.
\end{example}

We will use the following straightforward criterion 
for a simple integer Monge matrix to be unit-Monge.
\begin{lemma}
\label{lm-unit-Monge}
Let $A$ be a simple square integer Monge matrix over $\bra{i_0:i_1 \mid j_0:j_1}$.
Matrix $A$ is unit-Monge, if and only if
\begin{gather*}
A(\hi^-, j_1) - A(\hi^+, j_1) = 1\qquad
A(i_0, \hj^+) - A(i_0, \hj^-) = 1
\end{gather*}
for all $\hi \in \ang{i_0:i_1}$, $\hj \in \ang{j_0:j_1}$.
\end{lemma}
\begin{proof}
Note that a nonnegative integer matrix $P$ over $\ang{i_0:i_1 \mid j_0:j_1}$ 
is a permutation matrix, if and only if
\begin{gather*}
\sum_{\hj' \in \ang{j_0:j_1}} P(\hi,\hj')=1 \qquad 
\sum_{\hi' \in \ang{i_0:i_1}} P(\hi',\hj)=1 
\end{gather*}
for all $\hi \in \ang{i_0:i_1}$, $\hj \in \ang{j_0:j_1}$.
Applying this observation to matrix $A^\square$, we have
\begin{gather*}
\sum_{\hj' \in \ang{j_0:j_1}} A^\square(\hi,\hj') ={} 
\why{definition of $\square$}\\
\sum_{\hj' \in \ang{j_0:j_1}} \bigl(
A(\hi^+,\hj'^-) - A(\hi^-,\hj'^-) -
A(\hi^+,\hj'^+) + A(\hi^-,\hj'^+)\bigr) ={}
\why{telescoping sum}\\
A(\hi^+, j_0) - A(\hi^-, j_0) -
A(\hi^+, j_1) + A(\hi^-, j_1) ={}
\why{$A$ simple}\\
0 - 0 - A(\hi^+, j_1) + A(\hi^-, j_1) = 1
\end{gather*}
for all $\hi \in \ang{i_0:i_1}$.
Analogously,
\begin{gather*}
\sum_{\hi' \in \ang{i_0:i_1}} A^\square(\hi',\hj) = 1 
\end{gather*}
for all $\hj \in \ang{j_0:j_1}$.
Therefore, $A^\square$ is a permutation matrix, so $A$ is unit-Monge.
\end{proof}

The algorithms presented in this work are based 
on dealing with implicitly represented matrices.
While we will occasionally use the term ``implicit matrix'' in the general sense,
it will have the following specific meaning when applied to a (sub)unit-Monge matrix.
\begin{definition}
\label{def-implicit}
Let $A$ be a (sub)unit-Monge matrix over $\bra{0:n_1 \mid 0:n_2}$.
The \emph{implicit representation} for matrix $A$ 
is given by the (sub)permutation matrix $P=A^\square$
and vectors $b=A(n_1,*)$, $c=A(*,0)$:
\begin{gather*}
A(i,j) = P^\Sigma(i,j) + b(j) + c(i) - b(0)
\end{gather*}
for all $i$, $j$.
\end{definition}

Our particular focus will be on matrices that are both simple and unit-Monge. 
Our particular focus will be on square matrices $A$ 
that are both simple and unit-Monge.
By \defrefs{def-simple}, \ref{def-unit-Monge}, this holds
if and only if $A=P^\Sigma$, where $P$ is a permutation matrix.
\defref{def-implicit} can be specialised to such matrices as follows.
\begin{definition}
\label{def-implicit-simple}
Let $A$ be a simple unit-Monge matrix over $\bra{0:n \mid 0:n}$.
The \emph{implicit representation} for matrix $A$ 
is given by the permutation matrix $P=A^\square$:
\begin{gather*}
A = P^\Sigma
\end{gather*}
\end{definition}
\begin{example}
The $4 \times 4$ matrix in \exref{ex-distr} is simple unit-Monge.
\end{example}

\index{dominance counting}%
Thinking of elements of $P$ and $P^\Sigma$ 
as respectively half-integer and integer points in the plane,
the value $P^\Sigma(i,j)$ represents the count of nonzeros in $P$, 
that are $\gtrless$-dominated by the point $(i,j)$.
This type of query to an (implicit) set of points 
is known as \emph{dominance counting}.
An individual element $P^\Sigma(i,j)$ can be queried in time $O(n)$ 
by a linear sweep of the nonzeros of $P$,
counting those that are $\gtrless$-dominated by $(i,j)$.
Using a classical data structure, matrix $P$
can be preprocessed to allow element queries on $P^\Sigma$
much more efficiently.

\begin{figure}[tb]
\centering

\includegraphics{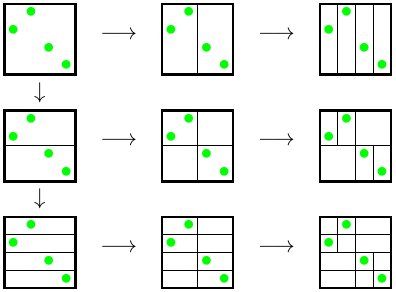}

\caption{\label{f-range-tree}%
A permutation matrix and the corresponding range tree}
\end{figure}
\begin{theorem}
\label{th-query}
Given a (sub)permutation matrix $P$ of size $n$, 
there exists a data structure that
\begin{itemize}
\item has size $O\bigpa{n \log n}$;
\item can be built in time $O\bigpa{n \log n}$;
\index{matrix!implicit!element query}%
\item allows to query an individual element
of the simple (sub)unit-Monge matrix $P^\Sigma$
in time $O\bigpa{\log^2 n}$;
\end{itemize}
\end{theorem}
\begin{proof}
The required structure is a two-dimensional range tree 
\cite{Bentley:80_CACM} (see also \cite{Preparata_Shamos:85}),
built on the set of nonzeros in $P$.
There are at most $n$ nonzeros,
hence the total number of nodes in the tree is $O\bigpa{n \log n}$.
A dominance counting query on the set of nonzeros
can be answered by accessing $O\bigpa{\log^2 n}$ of the tree nodes.
\end{proof}

\begin{example}
\figref{f-range-tree} shows a $4 \times 4$ permutation matrix,
with nonzeros indicated by green%
\footnote{For colour illustrations, the reader is referred 
to the online version of this work.
If the colour version is not available, 
all references to colour can be ignored.}
bullets, and the corresponding range tree.
\end{example}

The bounds given by \thref{th-query} can be improved 
by employing more advanced data structures.
Successive improvements to the efficiency of orthogonal range counting
(which includes dominance counting as a special case)
were obtained by Chazelle \cite{Chazelle:88},
J\'{a}J\'{a} et al.\ \cite{JaJa+:04},
Chan and P\v{a}tra\c{s}cu \cite{Chan_Patrascu:10}.
The currently most efficient data structure of \cite{Chan_Patrascu:10}
has size $O(n)$, can be built in time $O\bigpa{n(\log n)^{1/2}}$,
and answers a dominance counting query 
in time $O\bigpa{\frac{\log n}{\log\log n}}$.
However, the standard range tree data structure employed by \thref{th-query} 
is simpler, requires a less powerful computation model,
and is more likely to be practical.
Therefore, we will be using \thref{th-query} as our main technique
for implicit representation of simple (sub)unit-Monge matrices.

\index{matrix!implicit!incremental query}%
In addition to ordinary element queries described by \thref{th-query},
we will also access matrix elements via \emph{incremental queries}.
Given an element of an implicit simple (sub)unit-Monge matrix,
such a query returns the value of a specified adjacent element.
Incremental queries can be answered directly from the (sub)permutation matrix,
without any non-trivial data structures or preprocessing.
\begin{theorem}
\label{th-query-inc}
Given a (sub)permutation matrix $P$ of size $n$,
and the value $P^\Sigma(i,j)$, $i,j \in \bra{0:n}$,
the values $P^\Sigma(i \pm 1,j)$, $P^\Sigma(i,j \pm 1)$,
where they exist, can be queried in time $O(1)$.
\end{theorem}
\begin{proof}
Let $P$ be a permutation matrix; 
a generalisation to subpermutation matrices is straightforward.
Consider a query of the type $P^\Sigma(i+1,j)$;
the proof for other query types is analogous.
Let $\hj \in \ang{0:n}$ be such that $P(i+\tHalf,\hj) = 1$;
value $\hj$ can be obtained from the 
permutation representation of $P$ in time $O(1)$.
We have 
\begin{gather*}
P^\Sigma(i+1,j) = P^\Sigma(i,j)-
\begin{cases}
1 &\text{if $\hj < j$}\\
0 &\text{otherwise}
\end{cases}
\end{gather*}
\end{proof}

\index{matrix!implicit!incremental query!columnwise}%
\index{matrix!implicit!incremental query!rowwise}%
\index{matrix!implicit!batch query}%
We will call the incremental queries 
of type $P^\Sigma(i \pm 1,j)$ \emph{columnwise}, 
and of type $P^\Sigma(i,j \pm 1)$ \emph{rowwise}.
Incremental queries described by \thref{th-query-inc}
can be used to answer \emph{batch queries},
returning a set of elements in a row, column or diagonal 
of an implicit simple (sub)unit-Monge matrix.
In particular, all elements 
in a given row, column or diagonal of matrix $P^\Sigma$
can be obtained by a sequence of incremental queries in time $O(n)$,
and a subset of $r$ consecutive elements in time $O\bigpa{r + \log^2 n}$.

%%===========================================================================%%

%%===========================================================================%%
\mychapter{Matrix distance multiplication}
\label{c-mmult}

In this chapter, we lay the mathematical foundation for the rest of this work.
Our main mathematical structure is presented in two alternative forms:
first as distance multiplication of simple unit-Monge matrices,
and then via an algebraic formalism of seaweed braids.
The reader interested primarily in the algorithmic applications of our method
may wish to skip this chapter at first reading,
and then return to it as necessary for details of specific definitions,
theorems and proofs.

This chapter is organised as follows. 
In \secref{s-mmult}, we introduce matrix distance multiplication,
and study its algebraic properties 
in the classes of Monge and simple unit-Monge matrices.
In \secref{s-mmult-monge}, we describe efficient algorithms for Monge matrices:
in particular, row/column minima searching, 
matrix-vector and matrix-matrix multiplication.
In \secrefs{s-mvmult-umonge} and \ref{s-mmult-umonge}, 
we extend these algorithmic results
to matrix-vector and matrix-matrix multiplication
of simple unit-Monge matrices.
In \secref{s-braids}, we define the seaweed braid monoid,
and establish its isomorphism 
with the distance multiplication monoid of simple unit-Monge matrices.
In \secref{s-bruhat}, we describe the first application of our method,
obtaining an efficient algorithm for deciding 
Bruhat comparability of permutations.

%%=-=-=-=-=-=-=-=-=-=-=-=-=-=-=-=-=-=-=-=-=-=-=-=-=-=-=-=-=-=-=-=-=-=-=-=-=-=%%
\mysection{Distance multiplication monoids}
\label{s-mmult}

\index{$(\min,+)$-semiring}%
\index{distance (tropical) algebra}%
\index{$\oplus$: distance addition ($\min$)}%
\index{$\odot$: distance multiplication ($+$)}%
The $(\min,+)$-semiring of integers 
is one of the fundamental structures in algorithm design.
In this semiring, the operators $\min$ and $+$, denoted by $\oplus$ and $\odot$,
play the role of addition and multiplication, respectively.
The $(\min,+)$-semiring is often called 
\emph{distance} (or \emph{tropical}) algebra.
For a detailed introduction into this and related topics, 
see e.g.\ Rote \cite{Rote:90}, Gondran and Minoux \cite{Gondran_Minoux:08},
Butkovi{\'c} \cite{Butkovic:10}.
An application of the distance algebra to string comparison
has been previously suggested by Comet \cite{Comet:03}.

Throughout this chapter, 
vectors and matrices will be indexed by integers beginning from $0$, 
or half-integers beginning from $0^+ = \tHalf$.
All our definitions and statements can easily be generalised
to indexing over arbitrary integer or half-integer intervals.

Multiplication in the $(\min,+)$-semiring of integers
can be naturally extended to integer matrices and vectors.
\begin{definition}
\label{def-mvmult}
\index{distance multiplication!matrix-vector}%
\index{$\odot$: matrix distance multiplication}%
Let $A$ be a matrix over $\bra{0:n_1 \mid 0:n_2}$.
Let $b$, $c$ be vectors over $\bra{0:n_2}$ and $\bra{0:n_1}$ respectively.
The \emph{matrix-vector distance product} $A \odot b = c$ is defined by
\begin{gather*}
c(i) = 
\bigoplus_{j \in \bra{0:n_2}} \bigpa{A(i,j) \odot b(j)} =
\min_{j \in \bra{0:n_2}} \bigpa{A(i,j) + b(j)}
\end{gather*}
for all $i \in \bra{0:n_1}$.
\end{definition}
\begin{definition}
\label{def-mmult}
\index{distance multiplication!matrix-matrix}%
\index{$\odot$: matrix distance multiplication}
Let $A$, $B$, $C$ be matrices over 
$\bra{0:n_1 \mid 0:n_2}$, 
$\bra{0:n_2 \mid 0:n_3}$, 
$\bra{0:n_1 \mid 0:n_3}$
respectively.
The \emph{matrix distance product} $A \odot B = C$ is defined by
\begin{gather*}
C(i,k) = 
\bigoplus_{j \in \bra{0:n_2}} \bigpa{A(i,j) \odot B(j,k)} =
\min_{j \in \bra{0:n_2}} \bigpa{A(i,j) + B(j,k)}
\end{gather*}
for all $i \in \bra{0:n_1}$, $k \in \bra{0:n_3}$.
\end{definition}

We now consider three different monoids of integer matrices 
with respect to matrix distance multiplication.

\paragraph{Monoid of all nonnegative matrices.}
\index{matrix!general!$\odot$-monoid}%
\index{matrix!general!$\odot$-identity}%
\index{$\Id_{\odot}$: $\odot$-identity}%
\index{matrix!general!$\odot$-zero}%
\index{$O_{\odot}$: $\odot$-zero}%
Consider the set of all square matrices with elements in $\bra{0:+\infty}$
over a fixed index range. 
This set forms a monoid with zero with respect to distance multiplication.
The identity and the zero element in this monoid are respectively the matrices
\begin{gather*}
\Id_{\odot}(i,j)=
\begin{cases}
0 & \text{if $i=j$}\\ +\infty & \text{otherwise}
\end{cases}\qquad
O_{\odot}(i,j)=+\infty
\end{gather*}
for all $i$, $j$.
For any matrix $A$, we have 
\begin{gather*}
A \odot \Id_{\odot} = \Id_{\odot} \odot A = A\qquad
A \odot O_{\odot} = O_{\odot} \odot A = O_{\odot}
\end{gather*}

\paragraph{Monge monoid.}
It is well-known (see e.g.\ \cite{Atallah+:89})
that the set of all Monge matrices is closed under distance multiplication.
\begin{theorem}
\label{th-monoid-monge}
Let $A$, $B$, $C$ be matrices, such that $A \odot B = C$.
If $A$, $B$ are Monge, then $C$ is also Monge.
\end{theorem}
\begin{proof}
Let $A$, $B$ be 
over $\bra{0:n_1 \mid 0:n_2}$, $\bra{0:n_2 \mid 0:n_3}$, respectively.
Let $i',i'' \in \bra{0:n_1}$, $i' \leq i''$, 
and $k',k'' \in \bra{0:n_3}$, $k' \leq k''$.
By definition of matrix distance multiplication, we have
\begin{gather*}
C(i',k'') = \min_j \bigpa{A(i',j) + B(j,k'')}\\
C(i'',k') = \min_j \bigpa{A(i'',j) + B(j,k')}
\end{gather*}
Let $j'$, $j''$ respectively be the values of $j$ 
on which these minima are attained.
Suppose $j' \leq j''$.
We have
\begin{gather*}
C(i',k') + C(i'',k'') = {}
\why{definition of $\odot$}\\
\min_j \bigpa{A(i',j) + B(j,k')} + 
\min_j \bigpa{A(i'',j) + B(j,k'')} \leq {}
\why{minimisation over $j$}\\
\bigpa{A(i',j') + B(j',k')} + \bigpa{A(i'',j'') + B(j'',k'')} = {}
\why{term rearrangement}\\
\bigpa{A(i',j') + A(i'',j'')} + \bigpa{B(j',k') + B(j'',k'')} \leq {}
\why{$A$ is Monge}\\
\bigpa{A(i',j'') + A(i'',j')} + \bigpa{B(j',k') + B(j'',k'')} = {}
\why{term rearrangement}\\
\bigpa{A(i',j'') + B(j'',k'')} + \bigpa{A(i'',j') + B(j',k')} = {}
\why{definition of $j'$, $j''$}\\
C(i',k'') + C(i'',k')
\end{gather*}
The case $j' \geq j''$ is treated symmetrically, 
making use of the Monge property of $B$.
Hence, matrix $C$ is Monge.
\end{proof}

\index{matrix!Monge!$\odot$-monoid}%
\index{matrix!Monge!$\odot$-identity}%
\index{matrix!Monge!$\odot$-zero}%
\thref{th-monoid-monge} implies that
the set of all square nonnegative Monge matrices over a fixed index range
forms a submonoid (the \emph{Monge monoid}) 
in the distance multiplication monoid of all nonnegative matrices
(where the range of elements has to be formally extended by $+\infty$).

The ambient monoid's identity $\Id_{\odot}$ and zero $O_{\odot}$
are inherited by the Monge monoid.
Indeed, in the expansion of their density matrices 
$\Id_{\odot}^\square$ and $O_{\odot}^\square$ by \defref{def-density},
all indeterminate expressions of the form $+\infty-\infty$ 
can be formally considered to be nonnegative.
Therefore, matrices $\Id_{\odot}$ and $O_{\odot}$
can be formally considered to be Monge matrices.

\paragraph{Unit-Monge monoid.}
It is somewhat surprising, 
but crucial for the development of our techniques,
that the set of all simple (sub)unit-Monge matrices
is also closed under distance multiplication.
\begin{theorem}
\label{th-monoid-umonge}
Let $A$, $B$, $C$ be matrices, such that $A \odot B = C$.
If $A$, $B$ are simple unit-Monge (respectively, simple subunit-Monge),
then $C$ is also simple unit-Monge (respectively, simple subunit-Monge).
\end{theorem}
\begin{proof}
Let $A$, $B$ be simple unit-Monge matrices over $\bra{0:n \mid 0:n}$.
We have $A = P_A^\Sigma$, $B = P_B^\Sigma$,
where $P_A$, $P_B$ are permutation matrices.
It is easy to check that matrix $C$ is simple,
therefore $C = P_C^\Sigma$ for some matrix $P_C$.

We now have $P_A^\Sigma \odot P_B^\Sigma = P_C^\Sigma$,
and we need to show that $P_C$ is a permutation matrix.
Clearly, matrices $C$ and $P_C$ are both integer.
Furthermore, matrix $C$ is Monge by \thref{th-monoid-monge}, 
and therefore matrix $C^\square = P_C$ is nonnegative.

Since $P_B$ is a permutation matrix, we have 
\begin{gather*}
P_B^\Sigma(j,0) = 0 \qquad P_B^\Sigma(j,n) = n-j
\end{gather*}
for all $j \in \bra{0:n}$.
Hence 
\begin{gather*}
C(i,0) =
\min_j \bigpa{P_A^\Sigma(i,j) + P_B^\Sigma(j,0)} =
\min_j \bigpa{P_A^\Sigma(i,j) + 0} = 0\\
C(i,n) =
\min_j \bigpa{P_A^\Sigma(i,j) + P_B^\Sigma(j,n)} =
\min_j \bigpa{P_A^\Sigma(i,j) + n - j} = n-i
\end{gather*}
for all $i \in \bra{0:n}$,
since the minimum is attained respectively at $j=0$ and $j=n$.
Therefore, we have
\begin{gather*}
\sum_{\hk} P_C(\hi,\hk) = {}
\why{definition of $\Sigma$ and $\square$}\\
\sum_{\hk}
\bigpa{
C(\hi^+,\hk^-) - C(\hi^-,\hk^-) -
C(\hi^+,\hk^+) + C(\hi^-,\hk^+)} = {}
\why{term cancellation}\\
C(\hi^+,0) - C(\hi^-,0) - 
C(\hi^+,n) + C(\hi^-,n) = {}\\
0 - 0 - (n-\hi^+) + (n-\hi^-) = 1
\end{gather*}
for all $\hi \in \ang{0:n}$.
Symmetrically, we have
\begin{gather*}
\sum_{\hi} P_C(\hi,\hk) = 1
\end{gather*}
for all $\hk \in \ang{0:n}$.
Taken together, the above properties imply
that matrix $P_C$ is a permutation matrix.
Therefore, $C$ is a simple unit-Monge matrix.

Finally, let $A$, $B$ be simple subunit-Monge matrices
over $\bra{0:n_1 \mid 0:n_2}$, $\bra{0:n_2 \mid 0:n_3}$, respectively.
We have $A = P_A^\Sigma$, $B = P_B^\Sigma$,
where $P_A$, $P_B$ are subpermutation matrices.
As before, let $C=P_C^\Sigma$, for some matrix $P_C$; 
we have to show that $P_C$ is a subpermutation matrix.
Suppose that for some $\hi$, row $P_A(\hi,*)$ contains only zeros.
Then, it is easy to check that 
the corresponding row $P_C(\hi,*)$ also contains only zeros,
and that upon deleting rows $P_A(\hi,*)$ and $P_C(\hi,*)$ 
from the respective matrices,
the equality $P_A^\Sigma \odot P_B^\Sigma = P_C^\Sigma$ still holds.
Symmetrically, a zero column $P_B(*,\hk)$ 
results in a zero column $P_C(*,\hk)$,
and upon deleting both these columns from the respective matrices,
the equality $P_A^\Sigma \odot P_B^\Sigma = P_C^\Sigma$ still holds.
Therefore, we may assume without loss generality that
$n_1 \leq n_2$, $n_2 \geq n_3$, 
and that subpermutation matrix $P_A$ (respectively, $P_B$)
does not have any zero rows (respectively, zero columns).

Let us now extend matrix $P_A$ 
to a square $n_2 \times n_2$ matrix $\bmat{* \\ P_A}$,
where the top $n_2 - n_1$ rows are filled by zeros and ones
so that the resulting matrix is a permutation matrix.
Likewise, let us extend matrix $P_B$ 
to an $n_2 \times n_2$ permutation matrix $\bmat{P_B & *}$.
We now have 
\begin{gather*}
\bmat{* \\ P_A}^\Sigma \odot \bmat{P_B & *}^\Sigma = \bmat{* & * \\ P_C & *}^\Sigma
\end{gather*}
where $\bmat{* & * \\ P_C & *}$ 
is an $n_2 \times n_2$ permutation matrix,
with matrix $P_C$ occupying its lower-left corner.
Hence, matrix $P_C$ is a subpermutation matrix,
and the original matrix $C$ is a simple subunit-Monge matrix.
\end{proof}
\index{matrix!unit-Monge!$\odot$-monoid}%
\index{matrix!unit-Monge!$\odot$-identity}%
\index{matrix!unit-Monge!$\odot$-zero}%
\thref{th-monoid-umonge} implies that
the set of all simple unit-Monge matrices over a fixed index range
forms a submonoid (the \emph{unit-Monge monoid}) in the Monge monoid.

Without loss of generality, 
let the matrices be over $\bra{0:n \mid 0:n}$.
The Monge monoid's identity $\Id_{\odot}$ and zero $O_{\odot}$
are neither simple nor unit-Monge matrices,
and therefore are not inherited by the unit-Monge monoid.
Instead, its identity and zero elements are given respectively by the matrices
\begin{gather*}
\Id^\Sigma(i,j) = \max(j-i,0)\qquad
\Id^{R \Sigma}(i,j) = \min(n-i,j)
\end{gather*}
(recall that $\Id^R$ is the matrix obtained by 90-degree rotation 
of the identity permutation matrix $\Id$).
For any permutation matrix $P$, we have
\begin{gather*}
P^\Sigma \odot \Id^\Sigma = \Id^\Sigma \odot P^\Sigma = P^\Sigma\qquad
P^\Sigma \odot \Id^{R \Sigma} = \Id^{R \Sigma} \odot P^\Sigma = \Id^{R \Sigma}
\end{gather*}

\newcommand{\domain}[3]{%
\path[overlay] (p) ++(#3,0) -- ++(-#2,0) coordinate (q);
\draw (p) ++(-0.5,0) -- (p) -- (q) -- ++(0.5,0);
\path (p) coordinate(pp);
\foreach \count in {#2,...,#3}{%
  \draw (pp) node[circle](#1\count){};
  \path[overlay] (pp) ++(1,0) coordinate(pp);}}
\newcommand{\snake}[2]{%
\draw[green] (#1) to[out=-90,in=90,looseness=1.2] (#2);}

\begin{figure}[tb]
\centering

\setlength{\unit}{0.3cm}

\subfloat[\label{f-mmult-matrix}As permutation matrices]{%
\includegraphics{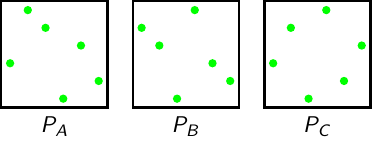}}

\subfloat[\label{f-mmult-seaweed}As seaweed braids]{%
\includegraphics{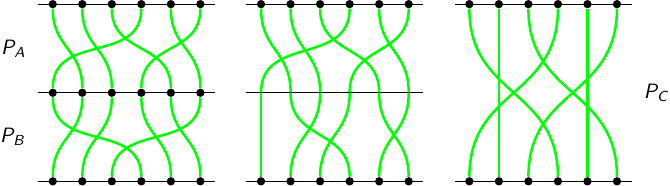}}
\caption{\label{f-mmult-example}%
Implicit matrix distance product $P_A \boxdot P_B = P_C$}
\end{figure}
\thref{th-monoid-umonge} gives us the basis for performing
distance multiplication of simple (sub)unit-Monge matrices implicitly,
by taking the density (sub)permutation matrices as input,
and producing a density (sub)permutation matrix as output.
It will be convenient to introduce special notation 
for such implicit distance matrix (and also matrix-vector) multiplication.
\begin{definition}
\label{def-mvmult-imp}
\index{distance multiplication!matrix-vector!implicit}%
\index{distance multiplication!vector-matrix!implicit}%
\index{$\boxdot$: implicit distance multiplication}
Let $P$ be a (sub)permutation matrix. 
Let $b$, $c$ be vectors.
The \emph{implicit matrix-vector distance product} $P \boxdot b = c$ 
is defined by $P^\Sigma \odot b = c$.
The \emph{implicit matrix-vector distance product} $b \boxdot P = c$
is defined analogously.
\end{definition}
\begin{definition}
\label{def-mmult-imp}
\index{distance multiplication!matrix-matrix!implicit}%
\index{$\boxdot$: implicit distance multiplication}
Let $P_A$, $P_B$, $P_C$ be (sub)permutation matrices. 
The \emph{implicit matrix distance product} $P_A \boxdot P_B = P_C$ 
is defined by $P_A^\Sigma \odot P_B^\Sigma = P_C^\Sigma$.
\end{definition}
The set of all permutation matrices over $\ang{0:n \mid 0:n}$
is therefore a monoid 
with respect to implicit distance multiplication $\boxdot$.
This monoid has identity element $\Id$ and zero element $\Id^R$,
and is isomorphic to the unit-Monge monoid.
Note that, although defined on the set of all permutation matrices of size $n$,
this monoid is substantially different from the symmetric group $\mathcal S_n$,
defined by standard permutation composition
(equivalently, by standard multiplication of permutation matrices).
In particular, 
the implicit distance multiplication monoid has a zero element $\Id^R$, 
whereas $\mathcal S_n$, being a group, cannot have a zero.
More generally, the implicit distance multiplication monoid
has plenty of idempotent elements (defined by involutive permutations),
whereas $\mathcal S_n$ has the only trivial idempotent $\Id$.
However, both the implicit distance multiplication monoid
and the symmetric group $\mathcal S_n$ still share the same identity element $\Id$.
\begin{example}
In \figref{f-mmult-example},
\sfigref{f-mmult-matrix} shows a triple of $6 \times 6$ permutation matrices
$P_A$, $P_B$, $P_C$, such that $P_A \boxdot P_B = P_C$.
Nonzeros are indicated by green circles.
\end{example}

%%=-=-=-=-=-=-=-=-=-=-=-=-=-=-=-=-=-=-=-=-=-=-=-=-=-=-=-=-=-=-=-=-=-=-=-=-=-=%%
\mysection{Monge distance multiplication}
\label{s-mmult-monge}

In this section, we study algorithms 
for distance multiplication of Monge matrices.
For simplicity, we only consider square matrices,
although the results generalise to rectangular ones.

We begin with matrix-vector distance multiplication.
For generic, explicitly presented matrices, the only reasonable method 
for matrix-vector distance multiplication of size $n$
is by direct application of \defref{def-mvmult} in time $O(n^2)$.

For implicit Monge matrices 
(including the case where the matrix is stored in random-access memory,
so that only subset of its elements may need to be queried,
and the rest can be ignored),
the running time can be substantially reduced.
This is achieved by an application 
of a classical row minima searching algorithm
by Aggarwal et al.\ \cite{Aggarwal+:87} (see also \cite{Giancarlo:97}),
often nicknamed the ``SMAWK algorithm''.
\begin{lemma}[\cite{Aggarwal+:87}]
\label{lm-rowmin-tmon}
Let $A$ be an $n_1 \times n_2$ implicit totally monotone matrix,
where each element can be queried in time $q$.
The problem of finding the (say, leftmost) minimum element 
in every row of $A$ can be solved in time $O(qn)$,
where $n = \max(n_1,n_2)$.
\end{lemma}
\begin{proof}
We give a sketch of the proof;
for details, see \cite{Aggarwal+:87,Giancarlo:97}.

Without loss of generality, let $A$ be over $\bra{0:n \mid 0:n}$.
Let $B$ be an implicit $\tHalf[n] \times n$ matrix
over $\bigbra{0:\tHalf[n] \mid 0:n}$,
obtained by taking every other row of $A$.
Clearly, at most $\tHalf[n]$ columns of $B$
contain a leftmost row minimum.
The key idea of the algorithm is to eliminate 
$\tHalf[n]$ of the remaining columns
in an efficient process, based on the total monotonicity property.

We call a matrix element \emph{marked} (for elimination), 
if its column has not (yet) been eliminated,
but the element is already known not to be a leftmost row minimum.
A column gets eliminated when all its elements become marked.

Initially, both the set of eliminated columns 
and the set of marked elements are empty.
In the process of column elimination,
marked elements may only be contained 
in the $i$ leftmost uneliminated columns;
the value of $i$ is initially equal to $1$,
and gets either incremented or decremented
in every step of the algorithm.
The marked elements form a \emph{staircase}:
that is, the marked elements 
in the first, second, \ldots, $i$-th uneliminated column,
are respectively the zero, one, \ldots, $i-1$ topmost elements.
In every iteration of the algorithm, two outcomes are possible:
either the staircase gets extended 
to the right to the $i+1$-st uneliminated column,
or the whole $i$-th uneliminated column 
gets eliminated from matrix $B$, and therefore also from the staircase.

Let $j$, $j'$ denote respectively the indices 
of the $i$-th and $i+1$-st uneliminated column
in the original matrix (across both uneliminated and eliminated columns).
The outcome of the current iteration depends on the comparison 
of element $B(i,j)$, 
which is the topmost unmarked element in the $i$-th uneliminated column,
against element $B(i,j')$,
which is the next uneliminated (and unmarked) element immediately to its right.
The outcomes of this comparison and the rest of the elimination procedure 
are given in \tabref{t-elim}.
\begin{table}[tb]
\newcommand{\Q}{\hspace*{8pt}}
\begin{quote}
$i \becomes 0$; $j \becomes 0$; $j' \becomes 1$\\
while $j' \leq n$:\\
\Q case $B(i,j) \leq B(i,j')$:\\
\Q\Q case $i < \tHalf[n]$:\Q $i \becomes i+1$; $j \becomes j'$\\
\Q\Q case $i = \tHalf[n]$:\Q eliminate column $j'$\\
\Q\Q $j'\becomes j'+1$\\
\Q case $B(i,j) > B(i,j')$:\\
\Q\Q eliminate column $j$\\
\Q\Q case $i = 0$:\Q $j\becomes j'$; $j'\becomes j'+1$\\
\Q\Q case $i > 0$:\Q
  $i \becomes i-1$;
  $j \becomes \max\{ k: \text{$k$ uneliminated and $<j$} \}$
\end{quote}
\caption{\label{t-elim} Elimination procedure 
of \lmrefs{lm-rowmin-tmon} and \ref{lm-rowmin-loglog}.}
\end{table}
By storing indices of uneliminated columns 
in an appropriate dynamic data structure, such as a doubly-linked list, 
a single iteration of this procedure can be implemented to run in time $O(q)$.
The whole procedure runs in time $O(qn)$, and eliminates $\tHalf[n]$ columns.

Let $A'$ be the $\tHalf[n] \times \tHalf[n]$ matrix
obtained from $B$ by deleting the $\tHalf[n]$ eliminated columns.
We call the algorithm recursively on $A'$.
This recursive call returns the leftmost row minima of $A'$,
and therefore also of $B$.
It is now straightforward to fill in  
the leftmost minima in the remaining rows of $A$ in time $O(qn)$.
Thus, the top level of recursion runs in time $O(qn)$.
The amount of work gets halved with every recursion level,
therefore the overall running time is $O(qn)$.
\end{proof}
\begin{figure}[tb]
\centering
%
%\subfloat[\label{f-elim-1}Case $B(i,j) \leq B(i,j')$, $i < \tHalf[n]$]{% % does not compile
\subfloat[\label{f-elim-1}Case $B(i,j) \leq B(i,j')$, $i < \frac{n}{2}$]{%
\beginpgfgraphicnamed{f-elim-1}%
\begin{tikzpicture}[x=0.5cm, y=-0.5cm]
\fill[gray!50] (1,0) -- ++(0,1) -- ++(1,0) -- ++(0,1) -- ++(1,0)
  -- ++(0,1) -- ++(1,0) -- ++(0,1) -- ++(1,0) -- ++(0,1) -- ++(1,0) |- (1,0);
\fill[gray!10] (6,0) -- ++(0,6) -- ++(1,0) |- (6,0);
\foreach \i in {1,...,6}
  \draw[dotted] (\i,0) -- (\i,8);
\foreach \i/\x in {5.5/0,6.5/1}
  \draw (\i,5.5) node[inner sep=1pt](x\x){$\circ$};
\foreach \i/\y in {6.5/0,7.5/1}
  \draw (\i,6.5) node[inner sep=1pt](y\y){$\bullet$};
\draw[->] (x1) -- (y0); 
\draw[->] (y0) -- (y1);
\draw (0,5.5) node[left]{\small $i$}
  (5.5,0) node[above]{\small $j$} (6.5,0) node[above]{\small $j'$};
\draw[thick] (10,0) -- (0,0) -- (0,8) -- (10,8);
\draw[thick,decorate,decoration=random steps] (10,0) -- (10,8);
\end{tikzpicture}
\endpgfgraphicnamed}
\qquad
\subfloat[\label{f-elim-2}Case $B(i,j) > B(i,j')$, $i > 0$]{%
\beginpgfgraphicnamed{f-elim-2}%
\begin{tikzpicture}[x=0.5cm, y=-0.5cm]
\fill[gray!50] (1,0) -- ++(0,1) -- ++(1,0) -- ++(0,1) -- ++(1,0)
  -- ++(0,1) -- ++(1,0) -- ++(0,1) -- ++(1,0) -- ++(0,1) -- ++(1,0) |- (1,0);
\fill[gray!10] (5,5) -- ++(0,3) -- ++(1,0) |- (5,5);
\foreach \i in {1,...,6}
  \draw[dotted] (\i,0) -- (\i,8);
\foreach \i/\x in {5.5/0,6.5/1}
  \draw (\i,5.5) node[inner sep=1pt](x\x){$\circ$};
\foreach \i/\y in {4.5/0,6.5/1}
  \draw (\i,4.5) node[inner sep=1pt](y\y){$\bullet$};
\draw[->] (x1) -- (y1);
\draw (0,5.5) node[left]{\small $i$}
  (5.5,0) node[above]{\small $j$} (6.5,0) node[above]{\small $j'$};
\draw[thick] (10,0) -- (0,0) -- (0,8) -- (10,8);
\draw[thick,decorate,decoration=random steps] (10,0) -- (10,8);
\end{tikzpicture}
\endpgfgraphicnamed}
\caption{\label{f-elim}
A snapshot of the elimination procedure 
in \lmrefs{lm-rowmin-tmon} and \ref{lm-rowmin-loglog}}
\end{figure}
\begin{example}
\figref{f-elim} gives a snapshot of the two non-boundary cases
of the elimination algorithm described in the proof of \lmref{lm-rowmin-tmon}.
Each vertical dotted line represents an arbitrary number  
of consecutive eliminated columns.
Dark-shaded cells represent the staircase of marked elements.
The current elements $B(i,j)$, $B(i,j')$ are shown by white circles.

\sfigref{f-elim-1} shows the case $B(i,j) \leq B(i,j')$, $i < \tHalf[n]$, and
\sfigref{f-elim-2} the case $B(i,j) > B(i,j')$, $i > 0$.
In both cases, the light-shaded cells represent the newly marked elements.
In \sfigref{f-elim-1},
these new elements extend the staircase by one column to the right.
In \sfigref{f-elim-2},
the marking of new elements results in the elimination of the whole column $j$,
reducing the staircase by its rightmost column.
In both cases, the elements $B(i,j)$, $B(i,j')$ for the next iteration
are shown by black circles.
\end{example}

It is straightforward to apply 
the ``SMAWK algorithm'' of \lmref{lm-rowmin-tmon}
to the distance multiplication of an implicit Monge matrix by a vector.

\begin{theorem}
\label{th-mvmult-monge}
Let $A$ be an $n_1 \times n_2$ implicit Monge matrix,
where each element can be queried in time $q$.
Let $b$ be an $n_1$-vector, and $c$ an $n_2$-vector, such that $A \odot b = c$.
Given vector $b$, vector $c$ can be computed in time $O(qn)$,
where $n = \max(n_1,n_2)$.
\end{theorem}
\begin{proof}
Let $\tilde A(i,j) = A(i,j) + b(j)$ for all $i$, $j$.
Matrix $\tilde A$ is an implicit Monge matrix,
where each element can be queried in time $q+O(1)$.
The problem of computing the product $A \odot b = c$
is equivalent to searching for row minima in matrix $\tilde A$,
which can be solved in time 
%(and therefore also memory) 
$O(qn)$ by \lmref{lm-rowmin-tmon}.
\end{proof}

The simplest case of application of \thref{th-mvmult-monge}
is when matrix $A$ is represented explicitly in random-access memory.
In such case, we have $q=1$, and Monge matrix-vector multiplication
can be performed in time $O(n)$, 
without even reading most elements of the matrix.

We now consider matrix-matrix distance multiplication.
For generic, explicitly presented matrices, 
direct application of \defref{def-mmult} gives an algorithm 
for matrix distance multiplication of size $n$, running in time $O(n^3)$.
Slightly subcubic algorithms for this problem have also been obtained.
The fastest currently known algorithm is by Chan \cite{Chan:07},
running in time $O\bigpa{\frac{n^3 (\log\log n)^3}{\log^2 n}}$.

For Monge matrices, distance multiplication can easily be performed
in quadratic time (see also \cite{Atallah+:89}).
For simplicity, we restrict ourselves to square Monge matrices.
\begin{theorem}
\label{th-mmult-monge}
Let $A$, $B$, $C$ be $n \times n$ matrices, 
such that $A$ is Monge, and $A \odot B = C$.
Given matrices $A$, $B$, matrix $C$ can be computed 
in time and memory $O(n^2)$.
\end{theorem}
\begin{proof}
The problem of computing the product $A \odot B = C$
is equivalent to $n$ instances of the matrix-vector product $A \odot b = c$,
where $b$ (respectively, $c$) is a column of $B$ (respectively, $C$).
Every one of these instances can be solved in time $O(n)$
by \thref{th-mvmult-monge},
so the overall running time 
% (and therefore also memory) 
is $n \cdot O(n) = O(n^2)$.

Alternatively, an algorithm with the same asymptotic running time
can be obtained directly by the divide-and-conquer technique
(see e.g.\ \cite{Apostolico+:90}).
\end{proof}

%%=-=-=-=-=-=-=-=-=-=-=-=-=-=-=-=-=-=-=-=-=-=-=-=-=-=-=-=-=-=-=-=-=-=-=-=-=-=%%
\mysection{Unit-Monge matrix-vector distance multiplication}
\label{s-mvmult-umonge}

We will now discuss algorithms for multiplication
of implicit simple unit-Monge matrices.
We begin with matrix-vector multiplication,
which turns out to be already a non-trivial problem.

By \thref{th-query}, an element of an implicit simple unit-Monge matrix,
represented by an appropriate data structure,
can be queried in time $q=O(\log^2 n)$.
By plugging this query time into \thref{th-mvmult-monge},
we obtain immediately an algorithm 
for implicit matrix-vector distance multiplication,
running in time $O(n \log^2 n)$.

A more careful analysis of the elimination procedure of \lmref{lm-rowmin-tmon}
shows that the required matrix elements can be obtained, 
instead of the standalone element queries of \thref{th-query},
by more efficient incremental queries of \thref{th-query-inc}.
At the top level of recursion, the query time is $q=O(1)$.
However, the query time per matrix element grows with each recursion level,
as the queried elements become more and more distant from each other
with respect to the original top-level matrix.
The resulting combined query time is $O(n)$ in every recursion level,
so the overall running time becomes $O(n \log n)$.

We now show that it is possible to speed up 
the implicit matrix-vector distance multiplication algorithm still further.
We describe two solutions:
first, a relatively straightforward extension 
of the elimination procedure of \lmref{lm-rowmin-tmon},
running in time $O(n \log\log n)$;
second, an algorithm based on sophisticated data structures 
for the union-find problem in the unit-cost RAM model,
running in optimal time $O(n)$.

The first solution, using incremental queries
and a new ``coarse-grain'' recursive fill-in procedure,
is as follows.
\begin{lemma}
\label{lm-rowmin-loglog}
Let $A$ be an implicit (sub)unit-Monge matrix 
over $\bra{0:n_1 \mid 0:n_2}$,
represented as in \defref{def-implicit}
by the (sub)permutation matrix $P=A^\square$
and vectors $b=A(n_1,*)$, $c=A(*,0)$.
The problem of finding the (say, leftmost) minimum element 
in every row of $A$ can be solved in time $O(n \log\log n)$,
where $n=\max(n_1,n_2)$.
\end{lemma}
\begin{proof}
First, observe that vector $c$ has no effect 
on the positions (as opposed to the values) of any row minima.
Therefore, we assume without loss of generality that $c(i)=0$ for all $i$
(and, in particular, $b(0)=c(n_1)=0$).
Further, suppose that some column $P(*,\hj)$ is identically zero;
then, depending on whether 
$b(\hj^-) \leq b(\hj^+)$ or $b(\hj^-) > b(\hj^+)$,
we may delete respectively column $A(*,\hj^+)$ or $A(*,\hj^-)$
as it does not contain any leftmost row minima.
Also, suppose that some row $P(\hi,*)$ is identically zero;
then the minimum value in row $A(\hi^-,*)$ 
lies in the same column as the minimum value in row $A(\hi^-,*)$,
hence we can delete one of these rows.
Therefore, we assume without loss of generality
that $A$ is an implicit unit-Monge matrix 
over $\bra{0:n \mid 0:n}$,
and hence $P$ is a permutation matrix.

To find the leftmost row minima, we adopt 
the column elimination procedure of \lmref{lm-rowmin-tmon} 
(see \tabref{t-elim}, \figref{f-elim}),
with some modifications outlined below.

Let $B$ be an implicit $n^{1/2} \times n$ matrix,
obtained by taking a subset of $n^{1/2}$ rows of $A$
at regular intervals of $n^{1/2}$.
Clearly, at most $n^{1/2}$ columns of $B$
contain a leftmost row minimum.
We need to eliminate $n-n^{1/2}$ of the remaining columns.

Let $B$ be over $\bigbra{0:\tHalf[n] \mid 0:n}$.
Throughout the elimination procedure,
we maintain a vector $d(i)$, $i \in \bra{0:n^{1/2}-1}$,
initialised by zero values.
In every iteration, given a current value of the index $j'$, 
each value $d(i)$ gives the count of nonzeros $P(s,t)=1$ 
within the rectangle 
$s \in \bigang{n^{1/2}i : n^{1/2}(i+1)}$, $t \in \ang{0:j'}$.

Consider an iteration of 
the column elimination procedure of \lmref{lm-rowmin-tmon}
with given values $i$, $j$, $j'$,
operating on matrix elements $B(i,j)$, $B(i,j')$.
For the iteration that follows the current one, 
the following matrix elements may be required:
\begin{itemize}
\item $B(i-1,j')$, $B(i+1,j')$. 
These values can be obtained respectively as 
$B(i,j) + d(i-1)$ and $B(i,j) - d(i)$.
\item $B(i,j'+1)$, $B(i+1,j'+1)$.
These values can be obtained respectively from $B(i,j')$, $B(i+1,j')$
by a rowwise incremental query of matrix $P^\Sigma$ via \thref{th-query-inc},
plus a single access to vector $b$.
\item $B(i-1,\brc{k : \text{$k$ uneliminated and $<j$}})$.
This element was already queried in the iteration at which 
its column was first added to the staircase.
There is at most one such element per column, 
therefore each of them can be stored and subsequently queried in constant time.
\end{itemize}

At the end of the current iteration, index $j'$ may be incremented
(i.e.\ the staircase may grow by one column).
In this case, we also need to update vector $d$ for the next iteration.
Let $s \in \ang{0:n}$ be such that $P(s,j'-\tHalf) = 1$.
Let $i = \bigfloor{s/n^{1/2}}$; 
we have $s \in \bigang{n^{1/2}i : n^{1/2}(i+1)}$.
The update consists in incrementing the vector element $d(i)$ by $1$.

The total number of iterations in the elimination procedure is at most $2n$.
This is because in total, at most $n$ columns are added to the staircase,
and at most $n$ (in fact, exactly $n-n^{1/2}$) columns are eliminated.
Therefore, the elimination procedure runs in time $O(n)$.

Let $A'$ be the $n^{1/2} \times n^{1/2}$ matrix
obtained from $B$ by deleting the $n-n^{1/2}$ eliminated columns.
Using incremental queries to matrix $P$,
it is straightforward to obtain matrix $A'$ explicitly
in random-access memory in time $O(n)$.
We now call the algorithm of \lmref{lm-rowmin-tmon}
to compute the row minima of $A'$,
and therefore also of $B$, in time $O(n)$.

We now need to fill in the remaining row minima of matrix $A$. 
The row minima of matrix $A'$ define a chain of $n^{1/2}$ submatrices in $A$
at which these remaining row minima may be located.
More specifically, given two successive row minima of $A'$,
all the $n^{1/2}$ row minima that are located 
between the two corresponding rows in $A$
must also be located between the two corresponding columns.
Each of the resulting submatrices has $n^{1/2}$ rows;
the number of columns may vary from submatrix to submatrix.
It is straightforward to eliminate from each submatrix 
all columns not containing any nonzero of matrix $P$;
therefore, without loss of generality, 
we may assume that every submatrix
is of size $n^{1/2} \times n^{1/2}$.

We now call the algorithm recursively on each submatrix
to fill in the remaining leftmost row minima.
The amount of work remains $O(n)$ in every recursion level.
There are $\log\log n$ recursion levels,
therefore the overall running time of the algorithm is $O(n \log\log n)$.
\end{proof}
\begin{example}
In \figref{f-elim}, the incremental queries made by the elimination algorithm 
in the proof of \lmref{lm-rowmin-loglog} are shown by arrows.
Note that no incremental query can cross a vertical dotted line,
since every such line represents an arbitrary number of eliminated columns.
\end{example}

A faster, time-optimal solution was suggested by Gawrychowski \cite{Gawrychowski:12}.
\begin{lemma}
\label{lm-rowmin}
Under the conditions of \lmref{lm-rowmin-loglog},
the running time can be reduced to $O(n)$.
\end{lemma}
\begin{proof}
As before, observe that vector $c$ has no effect 
on the positions of any row minima.
Therefore, we assume without loss of generality 
that $c(i)=A(i,0)=0$ for all $i$,
so $A(i,j) = P^\Sigma(i,j) + b(j)$ for all $i$, $j$.
Also note that we can perturb the elements of vector $b$ slightly,
so that each leftmost row minimum becomes the only minimum in its row.
Therefore, from now on we will omit the adjective ``leftmost''.

Consider vector $b$, which coincides with the bottom row of matrix $A$: 
$b=A(n,*)$.
Suppose that for some $j$, $j'$, $j \leq j'$, 
we have $A(n,j) \leq A(n,j')$.
Then, the element $A(n,j')$ cannot be the minimum in row $n$.
Furhermore, by the Monge property of matrix $A$,
we have $A(i,j) \leq A(i,j')$ for all $i$,
therefore an element $A(i,j')$ cannot be the minimum in any row $i$,
and hence column $j'$ can be safely excluded from the search for row minima.
After excluding all such columns,
the remaining elements in row $n$ 
form a decreasing subsequence of \emph{record minimal values,}
which we call for short the \emph{record subsequence}.
Here, an element $b(j')$ is called a record minimal value,
if we have $b(j) > b(j')$ for all $j \leq j'$. 
The record subsequence can be found trivially 
in a single pass of the input vector $b$ in time $O(n)$.
The final element in the record subsequence is the row minimum in row $n$.

Let $j_0 < j_1 < \ldots < j_r$ be the indices of the record minimal values
in row $n$, so the initial record subsequence is
\begin{gather*}
b(j_0) > b(j_1) > \ldots > b(j_r)
\end{gather*}
Our goal now is to compute the record subsequence for every row in matrix $A$.
We will represent the record subsequences implicitly 
by storing the differences between successive pairs of elements.
In particular, the initial record subsequence
is represented by the sequence of (all negative) values 
\begin{gather*}
d_{\hk} = b(j_{\hk^+}) - b(j_{\hk^-})\qquad \hk \in \ang{0:r}
\end{gather*}

We now move through rows of matrix $A$
from the bottom row $n$ towards the top row $0$,
updating the implicit record subsequence incrementally for each row.
We describe the procedure for updating this subsequence
from row $n$ to row $n-1$; the other updates are analogous.

Let $P(n^-,\hj)=1$ be the nonzero of matrix $P$ in row $n^-$.
Let 
\begin{gather*}
k_0 = \max \brc{k : j_k < \hj}\\
k_1 = \min \brc{k : j_k > \hj \text{ and } b(j_k) + 1 < b(j_{k_0})}
\end{gather*}

Recall that $A(i,j) = P^\Sigma(i,j) + b(j)$ for all $i$, $j$.
Assuming that both $k_0$ and $k_1$ above are well-defined
(i.e.\ the set under respectively the $\max$ and the $\min$ operator is non-empty),
it is easy to see that the record subsequence in row $n-1$ is
\begin{multline*}
b(j_0) > b(j_1) > \ldots > b(j_{k_0}) > {}\\
b(j_{k_1}) + 1 > b(j_{k_1+1}) + 1 > \ldots > b(j_r) + 1
\end{multline*}
In other words, we take all the elements of the record subsequence 
from $b(j_0)$ to $b(j_{k_0})$ inclusive,
we delete all the elements strictly between $b(j_{k_0})$ and $b(j_{k_1})$,
and then we take all the elements from $b(j_{k_1})$ to $b(j_r)$,
incrementing them by $1$.

The described updated record subsequence 
is represented implicitly by the updated difference sequence
\begin{gather*}
d_{0^+}, d_{1^+}, \ldots, d_{j_0^-}, 
b(j_{k_1}) - b(j_{k_0}) + 1,
d_{j_1^+}, d_{j_1^++1}, \ldots, d_{r^-} 
\end{gather*}
In other words, we take all the elements 
of the original difference subsequence 
from $d_{0^+}$ to $d_{j_0^-}$ inclusive,
we delete all the elements from $d_{j_0^+}$ to $d_{j_1^-}$ inclusive,
we create a new element $b(j_{k_1}) - b(j_{k_0}) + 1$,
and then we take all the elements of the original difference subsequence 
from $d_{j_1^+}$ to $d_{r^-}$ inclusive.
Assuming indices $k_0$ and $k_1$ are known,
such an update can be performed in time $O(1)$.

In case $k_0$ is undefined (this happens whenever $j_0 > \hj$), 
the updated record subsequence becomes
\begin{gather*}
b(j_0)+1 > b(j_1)+1 > \ldots > b(j_r)+1
\end{gather*}
hence the corresponding difference sequence 
remains the same as for the original record subsequence, 
and does not need to be updated.
In case $k_1$ is undefined (this happens whenever $b(j_r) + 1 > b(j_{k_0})$), 
the record subsequence becomes
\begin{gather*}
b(j_0) > b(j_1) > \ldots > b(j_{k_0})
\end{gather*}
hence the corresponding difference sequence 
is obtained by taking the original record subsequence
from $d_{0^+}$ to $d_{k_0^-}$ inclusive.
In both above cases, the update can still be performed in time $O(1)$.

Now, assume that only index $k_0$ is known before the start of the update.
Then, index $k_1$ can be found by linear search through the difference sequence.
The size of this linear search is equal to the number of elements
deleted from the sequence by the subsequent update.
Hence, the amortized running time 
of the linear search across all the updates is $O(n)$.

It remains to show how to find the index $k_0$ efficiently.
Consider the partitioning of interval $\ang{0:n}$ 
into a disjoint union of sub-intervals
\begin{gather*}
\ang{0:n} = \ang{0:j_0} \uplus \ang{j_0:j_1} \uplus\cdots\uplus \ang{j_{r-1}:j_r}
\end{gather*}
The problem of finding $k_0$ 
is equivalent to finding the interval $\ang{j_{k_0}:j_{k_0+1}}$ 
containing the index $\hj$ of the nonzero $P(n^-,\hj)$.
The same problem has to be solved repeatedly for each subsequent row,
where we need to find the interval between elements 
of the current record subsequence,
containing the current nonzero of matrix $P$.
As elements get deleted from the record subsequence by the update,
pairs of adjacent intervals also have to be merged into one interval.

The described problem fits in the classical setup 
of the \emph{union-find problem},
in particular its special case called the \emph{interval union-find problem} 
(see e.g.\ Italiano and Raman \cite{Italiano_Raman:10}).
This is a highly non-trivial problem that,
in the most general setting, 
has a marginally superlinear lower bound on the running time.
However, in the unit-cost RAM model of computation
this problem can be solved 
by an algorithm of Gabow and Tarjan \cite{Gabow_Tarjan:85}
(see also \cite{Italiano_Raman:10}) in time $O(n)$.

The overall running time of the algorithm
(assuming, as usual, the unit-cost RAM model of computation) is $O(n)$.
\end{proof}

\lmref{lm-rowmin} can now be applied to obtain an optimal algorithm
for distance multiplication of a simple (sub)unit-Monge matrix by a vector.
\begin{theorem}
\label{th-mvmult}
Let $P$ be an $n_1 \times n_2$ (sub)permutation matrix.
Let $b$ be an $n_1$-vector, and $c$ an $n_2$-vector, such that $P \boxdot b = c$.
Given the nonzeros of $P$ and the full vector $b$,
vector $c$ can be computed in time $O(n)$,
where $n = \max(n_1,n_2)$.
\end{theorem}
\begin{proof}
Analogous to \thref{th-mvmult-monge}, 
but using \lmref{lm-rowmin} for finding row minima.
\end{proof}

%%=-=-=-=-=-=-=-=-=-=-=-=-=-=-=-=-=-=-=-=-=-=-=-=-=-=-=-=-=-=-=-=-=-=-=-=-=-=%%
\mysection{Unit-Monge matrix-matrix distance multiplication}
\label{s-mmult-umonge}

\newcommand{\PA}[1]{P_{A,\mathit{#1}}}
\newcommand{\PB}[1]{P_{B,\mathit{#1}}}
\newcommand{\PC}[1]{P_{C,\mathit{#1}}}
\newcommand{\Mx}[1]{M_{\mathit{#1}}}
\newcommand{\Lo}{\mathit{lo}}
\newcommand{\Hi}{\mathit{hi}}

We now consider matrix-matrix distance multiplication.
While the quadratic running time of \thref{th-mmult-monge} 
is trivially optimal for explicit matrices,
it is possible to break through this time barrier 
in the case of implicitly represented matrices.

For simplicity, we restrict ourselves once again to square matrices
(which is trivially the case for unit-Monge matrices,
but not so for subunit-Monge matrices).
Subquadratic distance multiplication algorithms 
for implicit simple (sub)unit-Monge matrices 
were given in \cite{Tiskin:06_CSR,Tiskin:08_MCS}, 
and culminated with the following result in \cite{Tiskin:10_SODA}.
\begin{theorem}
\label{th-mmult}
Let $P_A$, $P_B$, $P_C$ be $n \times n$ (sub)permutation matrices, 
such that $P_A \boxdot P_B = P_C$.
Given the nonzeros of $P_A$, $P_B$, the nonzeros of $P_C$
can be computed in time $O(n \log n)$.
\end{theorem}
\begin{proof}
Let $P_A$, $P_B$, $P_C$ be permutation matrices over $\ang{0:n \mid 0:n}$.
The algorithm follows a divide-and-conquer approach,
in the form of recursion on $n$.
\begin{trivlist}
\setlabelit
\item[Recursion base: $n=1$.] The computation is trivial.

\item[Recursive step: $n>1$.]
Assume without loss of generality that $n$ is even.
Informally, the idea is to split the range of index $j$ 
in the definition of matrix distance product (\defref{def-mmult})
into two sub-intervals of size $\tHalf[n]$.
For each of these half-sized sub-intervals of $j$, we use the sparsity 
of the input permutation matrix $P_A$ (respectively, $P_B$)
to reduce the range of index $i$ (respectively, $k$)
to a (not necessarily contiguous) subset of size $\tHalf[n]$;
this completes the \emph{divide phase}.
We then call the algorithm recursively 
on the two resulting half-sized subproblems.
Using the subproblem solutions, 
we reconstruct the output permutation matrix $P_C$;
this is the \emph{conquer phase}.

We now describe each phase of the recursive step in more detail.

\item[Divide phase.]
By \defref{def-mmult-imp}, we have
\begin{gather*}
P_A^\Sigma \odot P_B^\Sigma = P_C^\Sigma 
\end{gather*}
Consider the partitioning of matrices $P_A$, $P_B$ into subpermutation matrices
\begin{gather*}
P_A = \bmat{\PA{lo} & \PA{hi}} \qquad
P_B = \bmat{\PB{lo} \\ \PB{hi}}
\end{gather*}
where $\PA{lo}$, $\PA{hi}$, $\PB{lo}$, $\PB{hi}$
are over $\ang{0:n \mid 0:\Half[n]}$, $\ang{0:n \mid \Half[n]:n}$,
$\ang{0:\Half[n] \mid 0:n}$, $\ang{\Half[n]:n \mid 0:n}$, respectively;
in each of these matrices, we maintain the indexing 
of the original matrices $P_A$, $P_B$.
We now have two implicit matrix multiplication subproblems
\begin{gather*}
\PA{lo}^\Sigma \odot \PB{lo}^\Sigma = \PC{lo}^\Sigma\qquad
\PA{hi}^\Sigma \odot \PB{hi}^\Sigma = \PC{hi}^\Sigma
\end{gather*}
where $\PC{lo}$, $\PC{hi}$ are of size $n \times n$.
Each of the subpermutation matrices 
$\PA{lo}$, $\PA{hi}$, $\PB{lo}$, $\PB{hi}$, $\PC{lo}$, $\PC{hi}$
has exactly $\tHalf[n]$ nonzeros.

Recall from the proof of \thref{th-monoid-umonge}
that a zero row in $\PA{lo}$ (respectively, a zero column in $\PB{lo}$)
corresponds to a zero row (respectively, column) 
in their implicit distance product $\PC{lo}$.
Therefore, we can delete all zero rows and columns 
from $\PA{lo}$, $\PB{lo}$, $\PC{lo}$,
obtaining, after appropriate index remapping,
three $\tHalf[n] \times \tHalf[n]$ permutation matrices.
Consequently, the first subproblem can be solved
by first performing a linear-time index remapping 
(corresponding to the deletion 
of zero rows and columns from $\PA{lo}$, $\PB{lo}$),
then making a recursive call on the resulting half-sized problem,
and then performing an inverse index remapping
(corresponding to the reinsertion 
of the zero rows and columns into $\PC{lo}$).
The second subproblem can be solved analogously.

\item[Conquer phase.]
We now need to combine the solutions for the two subproblems
to a solution for the original problem.
Note that we cannot simply put together 
the nonzeros of the subproblem solutions.
The original problem depends on the subproblems in a more subtle way:
some elements of $P_A^\Sigma$ 
depend on elements of both $\PA{lo}$ and $\PA{hi}$,
and therefore would not be accounted for directly 
by the solution to either subproblem on its own.
A similar observation holds for elements of $P_B^\Sigma$.
However, note that the nonzeros in the two subproblems 
have disjoint index ranges, and therefore
the direct combination of subproblem solutions $\PC{lo}+\PC{hi}$,
although not a solution to the original problem,
is still a permutation matrix.

In order to combine correctly the solutions of the two subproblems,
let us consider the relationship between these subproblems in more detail.
First, we split the range of index $j$ 
in the definition of matrix distance product (\defref{def-mmult})
into a ``low'' and a ``high'' sub-interval, each of size $\tHalf[n]$.
\begin{gather}
P_C^\Sigma(i,k) = 
\min_{j \in \bra{0:n}}
  \bigpa{P_A^\Sigma(i,j) + P_B^\Sigma(j,k)} = {} \notag \\
\min\Bigpa{%
  \min_{j \in \bigbra{0:\tHalf[n]}}
    \bigpa{P_A^\Sigma(i,j) + P_B^\Sigma(j,k)},
  \min_{j \in \bigbra{\tHalf[n]:n}}
    \bigpa{P_A^\Sigma(i,j) + P_B^\Sigma(j,k)}}
\label{eq-split}
\end{gather}
for all $i,k \in \bra{0:n}$.
Let us denote the two arguments in \eqref{eq-split}
by $\Mx{lo}(i,k)$ and $\Mx{hi}(i,k)$, respectively:
\begin{gather}
\label{eq-mlohi}
P_C^\Sigma(i,k) = \min\bigpa{\Mx{lo}(i,k), \Mx{hi}(i,k)}
\end{gather}
for all $i,k \in \bra{0:n}$.
The first argument in \eqref{eq-split}, \eqref{eq-mlohi} 
can be expressed via the solutions of the two subproblems as follows:
\begin{gather}
\Mx{lo}(i,k) =
\min_{j \in \bigbra{0:\tHalf[n]}}
  \bigpa{P_A^\Sigma(i,j) + P_B^\Sigma(j,k)} = {}
\why{definition of $\Sigma$}\\
\min_{j \in \bigbra{0:\tHalf[n]}} 
  \bigpa{\PA{lo}^\Sigma(i,j) + \PB{lo}^\Sigma(j,k) + 
    \PB{hi}^\Sigma(\tHalf[n],k)} = {}
\why{term rearrangement}\\
\min_{j \in \bigbra{0:\tHalf[n]}}
  \bigpa{\PA{lo}^\Sigma(i,j) + \PB{lo}^\Sigma(j,k)} +
  \PB{hi}^\Sigma(\tHalf[n],k) = {}
\why{definition of $\odot$}\\
\PC{lo}^\Sigma(i,k) + \PC{hi}^\Sigma(0,k)
\label{eq-mlo}
\end{gather}
Here, the final equality is due to
\begin{gather*}
\PC{hi}^\Sigma(0,k) = \min_{j \in \bigbra{\tHalf[n]:n}} 
  \bigpa{\PA{hi}^\Sigma(0,j) + \PB{hi}^\Sigma(j,k)} = {}\\
\min_{j \in \bigbra{\tHalf[n]:n}} 
  \bigpa{j - \tHalf[n] + \PB{hi}^\Sigma(j,k)} = 
\PB{hi}^\Sigma(\tHalf[n],k)
\end{gather*}
since the minimum is attained at $j = \tHalf[n]$.
The second argument in \eqref{eq-split}, \eqref{eq-mlohi} 
can be expressed analogously as
\begin{gather}
\Mx{hi}(i,k) =
\min_{j \in \bigbra{\tHalf[n]:n}}
  \bigpa{P_A^\Sigma(i,j) + P_B^\Sigma(j,k)} = {} \notag\\
\PC{hi}^\Sigma(i,k) + \PC{lo}^\Sigma(i,n)
\label{eq-mhi}
\end{gather}
The minimisation operator in \eqref{eq-split}, \eqref{eq-mlohi} 
is equivalent to evaluating the sign of the difference of its two arguments:
\begin{gather*}
\delta(i,k) = \Mx{lo}(i,k) - \Mx{hi}(i,k) = {}
\why{by \eqref{eq-mlo}, \eqref{eq-mhi}}\\
\bigpa{\PC{lo}^\Sigma(i,k) + \PC{hi}^\Sigma(0,k)} - 
\bigpa{\PC{hi}^\Sigma(i,k) + \PC{lo}^\Sigma(i,n)} = {}
\why{term rearrangement}\\
\bigpa{\PC{hi}^\Sigma(0,k) - \PC{hi}^\Sigma(i,k)} -
\bigpa{\PC{lo}^\Sigma(i,n) - \PC{lo}^\Sigma(i,k)} = {}
\why{definition of $\Sigma$}\\
\sum_{\hi \in \ang{0:i},\hk \in \ang{0:k}} \PC{hi}(\hi,\hk) - 
\sum_{\hi \in \ang{i:n},\hk \in \ang{k:n}} \PC{lo}(\hi,\hk) = {}
\why{definition of $\Sigma$, $R$}\\
\PC{hi}^{R\Sigma}(n-k,i) - \PC{lo}^{RRR\Sigma}(k,n-i)
\end{gather*}
Since $\PC{lo}$, $\PC{hi}$ are subpermutation matrices,
and $\PC{lo}+\PC{hi}$ a permutation matrix,
it follows that function $\delta$ 
is unit-monotone increasing in each of its arguments.

The sign of function $\delta$ determines 
the positions of nonzeros in $P_C$ as follows.
Let us fix some half-integer point $\hi,\hk \in \ang{0:n}$ in $P_C$,
and consider the signs of the four values $\delta(\hi^\pm,\hk^\pm)$ 
at neighbouring integer points.
Due to the unit-monotonicity of $\delta$, only three cases are possible.
\begin{trivlist}
\item{\itshape Case $\delta(\hi^\pm,\hk^\pm) \leq 0$ 
for all four sign combinations.}
We have
\begin{gather*}
\Mx{lo}(\hi^\pm,\hk^\pm) \leq \Mx{hi}(\hi^\pm,\hk^\pm)
\end{gather*}
for each sign combination taken consistently on both sides of the inequality,
and, by \eqref{eq-mlohi},
\begin{gather*}
P_C^\Sigma(\hi^\pm,\hk^\pm) = \Mx{lo}(\hi^\pm,\hk^\pm)
\end{gather*}
Hence, we have
\begin{gather*}
P_C(\hi,\hk) = P_C^{\Sigma\square}(\hi,\hk) = \Mx{lo}^{\square}(\hi,\hk) = \PC{lo}(\hi,\hk)
\why{definition of $\Sigma$, $\square$, \eqref{eq-mlo}, \eqref{eq-mhi}}
\end{gather*}
Thus, in this case $P_C(\hi,\hk) = 1$ is equivalent to $\PC{hi}(\hi,\hk) = 1$.
Note that this also implies $\delta(\hi^-,\hk^-) < 0$,
since otherwise we would have $\delta(\hi^\pm,\hk^\pm) = 0$ for all four sign combinations,
and hence, by symmetry, also $\PC{hi}(\hi,\hk) = 1$.
However, that would imply $\PC{lo}(\hi,\hk)+\PC{hi}(\hi,\hk) = 1+1 = 2$,
which is a contradiction to $\PC{lo}+\PC{hi}$ being a permutation matrix.
\item{\itshape Case $\delta(\hi^\pm,\hk^\pm) \geq 0$ 
for all four sign combinations.}
Symmetrically to the previous case, we have
\begin{gather*}
P_C(\hi,\hk) = \PC{hi}(\hi,\hk)
\end{gather*}
Thus, in this case $P_C(\hi,\hk) = 1$ is equivalent to $\PC{lo}(\hi,\hk) = 1$,
and implies $\delta(\hi^+,\hk^+) > 0$.
\item{\itshape Case $\delta(\hi^-,\hk^-) < 0$, 
$\delta(\hi^-,\hk^+) = \delta(\hi^+,\hk^-) = 0$, $\delta(\hi^+,\hk^+) > 0$.}
By \eqref{eq-mlohi}, we have
\begin{gather*}
P_C^\Sigma(\hi^-,\hk^-) = \Mx{lo}(\hi^-,\hk^-) \\
P_C^\Sigma(\hi^+,\hk^-) = \Mx{lo}(\hi^+,\hk^-) \\
P_C^\Sigma(\hi^-,\hk^+) = \Mx{lo}(\hi^-,\hk^+) \\
P_C^\Sigma(\hi^+,\hk^+) = \Mx{hi}(\hi^+,\hk^+) < \Mx{lo}(\hi^+,\hk^+)
\end{gather*}
Hence,
\begin{gather*}
P_C(\hi,\hk) = P_C^{\Sigma\square}(\hi,\hk) >
\Mx{lo}^{\square}(\hi,\hk) = \PC{lo}(\hi,\hk)
\why{definition of $\Sigma$, $\square$, \eqref{eq-mlo}, \eqref{eq-mhi}}
\end{gather*}
Since both $P_C$ and $\PC{lo}$ are zero-one matrices,
the strict inequality implies that $P_C(\hi,\hk)=1$ and $\PC{lo}(\hi,\hk)=0$.
Symmetrically, also $\PC{hi}(\hi,\hk)=0$.
\end{trivlist}

Summarising the above three cases, we have $P_C(\hi,\hk)=1$, 
if and only if one of the following conditions holds:
\begin{gather}
\text{$\delta(\hi^-,\hk^-) < 0$ and $\PC{lo}(\hi,\hk)=1$} \label{eq-lo} \\
\text{$\delta(\hi^+,\hk^+) > 0$ and $\PC{hi}(\hi,\hk)=1$} \label{eq-hi} \\
\text{$\delta(\hi^-,\hk^-) < 0$ and $\delta(\hi^+,\hk^+) > 0$} \label{eq-lohi}
\end{gather}
By the discussion above, these three conditions are mutually exclusive.

In order to check the conditions \eqref{eq-lo}--\eqref{eq-lohi},
we need an efficient procedure for determining the sign of function $\delta$
in points of the integer square $\bra{0:n \mid 0:n}$.
Informally, low (respectively, high) values of both $i$ and $k$ 
correspond to negative (respectively, positive) values of $\delta(i,k)$.
By unit-monotonicity of $\delta$, 
there must exist a pair of monotone rectilinear paths 
from the bottom-left to the top-right corner 
of the half-integer square $\ang{-1:n+1 \mid -1:n+1}$,
that separate strictly negative and nonnegative 
(respectively, strictly positive and nonpositive) values of $\delta$.

We now give a simple efficient procedure for finding such a pair of separating paths.
By symmetry we only need to the consider the \emph{lower separating path}.
For all integer points $(i,k)$ above-left (respectively, below-right) of this path,
we have $\delta(i,k) < 0$ (respectively, $\delta(i,k) \geq 0$).

We start at the bottom-left corner of the square,
with $(\hi,\hk) = (n^+,0^-)$ as the initial point on the lower separating path.
We have $\delta(\hi^-,\hk^+) = \delta(n,0) = 0$.

Let $(\hi,\hk)$ now denote any current point on the lower separating path,
and suppose that we have evaluated $\delta(\hi^-,\hk^+)$.
The sign of this value determines the next point on the path:
\begin{alignat*}{2}
{}
&(\hi,\hk+1) &\qquad &\text{if $\delta(\hi^-,\hk^+) < 0$}\\
&(\hi-1,\hk) &       &\text{if $\delta(\hi^-,\hk^+) \geq 0$}
\end{alignat*}
Following this choice, we then evaluate 
either $\delta\bigpa{\hi^-,(\hk+1)^+}$, or $\delta\bigpa{(\hi-1)^-,\hk^+}$
from $\delta(\hi^-,\hk^+)$
by an incremental query of \thref{th-query-inc} in time $O(1)$.
The computation is now repeated with the new current point.

The described path-finding procedure runs until 
either $\hi = 0^-$, or $\hk = n^+$.
We then complete the path by moving 
in a straight horizontal (respectively, vertical) line
to the final destination $(\hi,\hk) = (0^-,n^+)$.
The whole procedure of finding the lower separating path runs in time $O(n)$.
A symmetric procedure procedure with the same running time
can be used to find the \emph{upper separating path,}
for which we have $\delta(i,k) \leq 0$ on the above-left, 
and $\delta(i,k) > 0$ on the below-right.

Given a value $d \in \bra{-n+1:n-1}$, 
let us now consider the set of points $(\hi,\hk)$ with $\hk-\hi=d$;
such a set forms a diagonal in the half-integer square.
Let $\bigpa{\hi_\Lo,\hk_\Lo}$, where $\hk_\Lo - \hi_\Lo = d$, 
be the unique intersection point of the given diagonal with the lower separating path.
Let $r_\Lo(d) = \hi_\Lo + \hk_\Lo$.
Define $r_\Hi(d)$ analogously, using the upper separating path.
Conditions \eqref{eq-lo}--\eqref{eq-lohi} can now be expressed 
in terms of arrays $r_\Lo$, $r_\Hi$ as follows:
\begin{gather}
\text{$\hi+\hk \leq r_\Lo(\hk-\hi)$ and $\PC{lo}(\hi,\hk)=1$} \label{eq-r-lo} \\
\text{$\hi+\hk \geq r_\Hi(\hk-\hi)$ and $\PC{hi}(\hi,\hk)=1$} \label{eq-r-hi} \\
\text{$\hi+\hk = r_\Lo(\hk-\hi) = r_\Hi(\hk-\hi)$} \label{eq-r-lohi}
\end{gather}
Here, we make use of the fact that $\hi+\hk \leq r_\Lo(\hk-\hi)$
is equivalent to $\hi^- + \hk^- < r_\Lo(\hk-\hi)$,
and $\hi+\hk \geq r_\Hi(\hk-\hi)$ to $\hi^+ + \hk^+ > r_\Hi(\hk-\hi)$.

The nonzeros of $P_C$ satisfying either of the conditions 
\eqref{eq-r-lo}, \eqref{eq-r-hi} can be found in time $O(n)$
by checking directly each of the nonzeros in matrices $\PC{lo}$ and $\PC{hi}$.
The nonzeros of $P_C$ satisfying condition 
\eqref{eq-r-lohi} can be found in time $O(n)$
by a linear sweep of the points $(\hi,\hk)$ on the two separating paths.
We have now obtained all the nonzeros of matrix $P_C$.

\item[(End of recursive step)]

\end{trivlist}

The generalisation to subpermutation matrices is as in \thref{th-monoid-umonge}.

\begin{trivlist}
\setlabelit
\item[Time analysis.]
The recursion tree is a balanced binary tree of height $\log n$.
In the root node, the computation runs in time $O(n)$.
In each subsequent level, the number of nodes doubles, 
and the running time per node decreases by a factor of $2$.
Therefore, the overall running time is $O(n \log n)$.
\end{trivlist}
                                              % empty line generates qedsymbol
\end{proof}

\begin{figure}[p]
\centering

\subfloat[\label{f-mmult-pab}Input matrices $P_A$, $P_B$]{%
\includegraphics{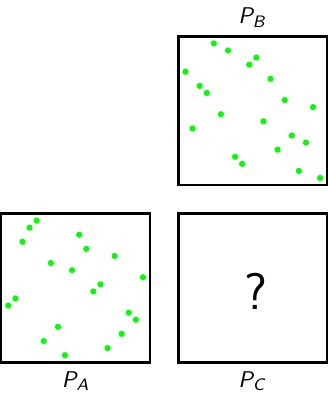}}
\qquad
\subfloat[\label{f-mmult-pabx}%
Subproblems 
$\PA{lo} \boxdot \PB{lo} = \PC{lo}$ and
$\PA{hi} \boxdot \PB{hi} = \PC{hi}$]{%
\includegraphics{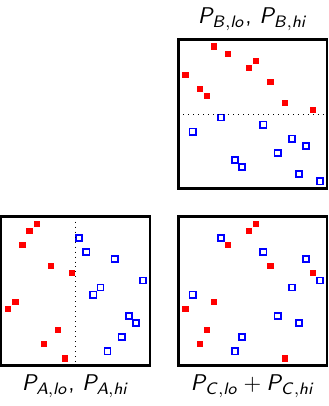}}

\subfloat[\label{f-mmult-pabxc}%
Conversion of $\PC{lo}+\PC{hi}$ into $P_C$]{%
\includegraphics{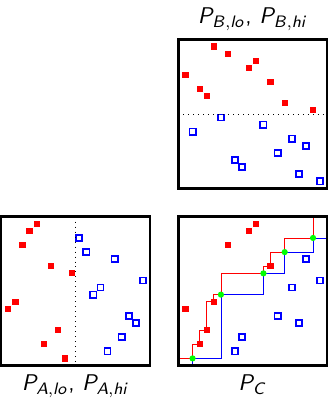}}
\qquad
\subfloat[\label{f-mmult-pabc}Output matrix $P_C$]{%
\includegraphics{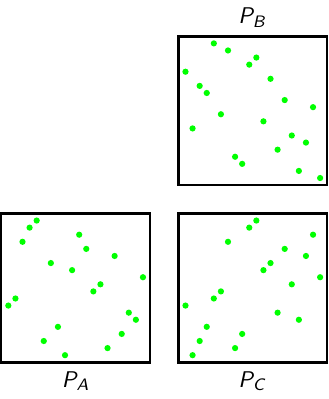}}

\caption{\label{f-mmult}%
Proof of \thref{th-mmult}: $P_A \boxdot P_B = P_C$}
\end{figure}

\begin{example}
\figref{f-mmult} illustrates the proof of \thref{th-mmult}
on a problem instance
with a solution generated by the Wolfram Mathematica software.
\sfigref{f-mmult-pab} shows a pair of input 
$20 \times 20$ permutation matrices $P_A$, $P_B$,
with nonzeros indicated by green circles.
\sfigref{f-mmult-pabx} shows the partitioning 
of the implicit $20 \times 20$ matrix distance multiplication problem
into two $10 \times 10$ subproblems.
The nonzeros in the two subproblems are shown respectively 
by filled red squares and hollow blue squares.
\sfigref{f-mmult-pabxc} shows a recursive step.
The lower and the upper separating paths 
are shown respectively in red and in blue
(note that the lower path is visually above the upper one;
the lower/upper terminology refers to the relative values of $\delta$,
rather than the visual position of the paths).
The nonzeros in the output matrix $P_C$ satisfying 
\eqref{eq-r-lo}, \eqref{eq-r-hi}, \eqref{eq-r-lohi} are shown respectively 
by filled red squares, hollow blue squares, and green circles;
note that overall, there are 20 such nonzeros,
and that they define a permutation matrix.
\sfigref{f-mmult-pabc} shows the output matrix $P_C$.
\end{example}

%%=-=-=-=-=-=-=-=-=-=-=-=-=-=-=-=-=-=-=-=-=-=-=-=-=-=-=-=-=-=-=-=-=-=-=-=-=-=%%
\mysection{Seaweed braids}
\label{s-braids}

Further understanding of the unit-Monge monoid
(and, by isomorphism, 
of the implicit distance multiplication monoid of permutation matrices)
can be gained via an algebraic formalism closely related to braid theory.
We refer the reader to \cite{Kassel_Turaev:08} 
for the background on classical braid theory.

\index{seaweed}%
\index{seaweed!braid}%
\index{seaweed!braid!width}%
Consider two sets of $n$ nodes each,
drawn on two parallel horizontal lines in the Euclidean plane.
We put the two node sets into one-to-one correspondence
by connecting them pairwise, in some order, 
with continuous monotone curves. 
(Here, a curve is called monotone, 
if its vertical projection is always directed downwards.)
These curves will be called \emph{seaweeds}%
\footnote{A tongue-in-cheek justification for this term is that
seaweed braids are like ordinary braids, except that they are sticky:
a pair of seaweeds, once they have crossed, cannot be fully untangled.}.
We call the resulting configuration 
a \emph{seaweed braid} of width $n$.

\begin{example}
In \figref{f-mmult-example}, \sfigref{f-mmult-seaweed} 
shows three different seaweed braids.
\end{example}

There is remarkable similarity between seaweed braids and classical braids.
However, there is also a crucial difference:
all crossings between seaweeds are ``level crossings'',
i.e.\ a pair of crossing seaweeds 
are not assumed to pass under/over one another as in classical braids.
We will also assume that all crossings are between exactly two seaweeds,
hence three or more seaweeds can never meet at a single point.

\index{seaweed!braid!reduced}%
In a seaweed braid, a given pair of seaweeds 
may cross an arbitrary number of times.
We call a seaweed braid \emph{reduced}, 
if every pair of its seaweeds cross at most once 
(i.e.\ either once, or not at all).

\index{seaweed!braid!multiplication}%
Similarly to classical braids, 
two seaweed braids of the same width can be \emph{multiplied}.
The product braid is obtained as follows.
First, we draw one braid above the other,
identifying the bottom nodes of the top braid
with the top nodes of the bottom braid.
Then, we join up each pair of seaweeds 
that became incident in the previous step.
Note that, even if both original seaweed braids were reduced,
their product may in general not be reduced.
\begin{example}
In \sfigref{f-mmult-seaweed}, 
the left-hand side is a product of two reduced seaweed braids.
In this product braid, some pairs of seaweeds cross twice,
hence it is not reduced.
\end{example}

Seaweed braids can be transformed
(and, in particular, unreduced braids can be reduced) 
according to a specific set of algebraic rules.
These rules are incorporated into the following formal definition. 
\begin{definition}
\label{def-monoid}
\index{seaweed!monoid}%
The \emph{seaweed monoid} $\mathcal T_n$
is a finitely presented monoid on $n$ generators:
$\mathit{id}$ (the identity element), $g_1$, $g_2$, \ldots, $g_{n-1}$.
The presentation of monoid $\mathcal T_n$ 
consists of the \emph{idempotence relations}
\begin{alignat}{2}
{}
&g_t^2 = g_t &\qquad &t \in [1:n-1] \label{rel-idem}\\
\intertext{the \emph{far commutativity relations}}
&g_t g_u = g_u g_t &&t,u \in [1:n-1], u-t \geq 2 \label{rel-farcomm}\\
\intertext{and the \emph{braid relations}}
&g_t g_u g_t = g_u g_t g_u &&t,u \in [1:n-1], u-t = 1 \label{rel-braid}
\end{alignat}
\end{definition}
Traditionally, this structure is also known as the 
\emph{$0$-Hecke monoid of the symmetric group} $H_0(\mathcal S_n)$,
or the \emph{Richardson--Springer monoid}
(for details, see e.g.\ Denton et al.\ \cite{Denton+:11}, 
Mazorchuk and Steinberg \cite{Mazorchuk_Steinberg:12},
Deng et al.\ \cite{Deng+:08}).

\begin{figure}[tb]
\centering
\subfloat[\label{f-relations-idem}Idempotence relations \eqref{rel-idem}]{%
\makebox[\textwidth]{\includegraphics{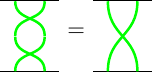}}}

\subfloat[\label{f-relations-farcomm}Far commutativity relations \eqref{rel-farcomm}]{%
\includegraphics{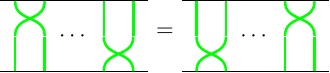}}

\subfloat[\label{f-relations-braid}Braid relations \eqref{rel-braid}]{%
\includegraphics{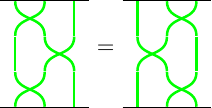}}
\caption{\label{f-relations} 
Defining relations of the seaweed monoid}
\end{figure}
The correspondence between elements of the seaweed monoid
and seaweed braids is as follows.
The monoid multiplication
(i.e.\ concatenation of words in the generators)
corresponds to the multiplication of seaweed braids.
The identity element $\mathit{id}$
corresponds to a seaweed braid where the top nodes
are connected to the bottom nodes in the left-to-right order, 
without any crossings.
Each of the remaining generators $g_t$ corresponds
to an \emph{elementary crossing},
i.e.\ to a seaweed braid where the only crossing
is between a pair of neighbouring seaweeds 
in half-integer positions $t^-$ and $t^+$.
\figref{f-relations} shows the defining relations
of the seaweed monoid \eqref{rel-idem}--\eqref{rel-braid}
in terms of seaweed braids.

\begin{example}
In \sfigref{f-mmult-seaweed}, 
the left-hand side is an unreduced product of two seaweed braids.
We now \emph{comb} the seaweeds by running through all their crossings,
respecting the top-to-bottom partial order of the crossings.
For each crossing, we check whether the two crossing seaweeds
have previously crossed above the current point.
If this is the case, 
then we undo the current crossing by removing it from the braid
and replacing it by two non-crossing seaweed pieces.
The correctness of this combing procedure is easy to prove
by the seaweed monoid relations \eqref{rel-idem}--\eqref{rel-braid}.
After all the crossings have been combed,
we obtain a reduced seaweed braid shown in the middle of \sfigref{f-mmult-seaweed}.
Another equivalent reduced seaweed braid in shown in the right-hand side.
\end{example}

A permutation matrix $P$ over $\ang{0:n \mid 0:n}$
can be represented by a seaweed braid as follows.
The row and column indices correspond respectively 
to the top and the bottom nodes, ordered from left to right.
A nonzero $P(\hi,\hj)=1$ corresponds to a seaweed
connecting top node $\hi$ and bottom node $\hj$.
For a given permutation, it is always possible to draw the seaweeds 
so that the resulting seaweed braid is reduced.
In general, this reduced braid will not be unique;
however, it turns out that all the reduced braids 
corresponding to the same permutation are equivalent.
We formalise this observation by the following lemma.
\begin{lemma}
\label{lm-nfact}
The seaweed monoid $\mathcal T_n$ consists of at most $n!$ distinct elements.
\end{lemma}
\begin{proof}
It is straightforward to see that any seaweed braid 
can be transformed into a reduced one,
using relations \eqref{rel-idem}--\eqref{rel-braid}.
Then, any two reduced seaweed braids corresponding to the same permutation
can be transformed into one another,
using far commutativity \eqref{rel-farcomm}
and the braid relations \eqref{rel-braid}.
Therefore, each permutation corresponds to a single element of $\mathcal T_n$.
This mapping is surjective,
therefore the number of elements in $\mathcal T_n$ 
is at most the total number of permutations $n!$.
\end{proof}

We now establish a direct connection 
between elements of the seaweed monoid and permutation matrices.
The identity generator $\mathit{id}$
corresponds to the identity matrix $\Id$.
Each of the remaining generators $g_t$ corresponds 
to an \emph{elementary transposition matrix} $G_t$, defined as
\begin{gather*}
G_t(\hi,\hj) =
\begin{cases}
1-\Id(\hi,\hj) & \text{if $\hi,\hj \in \brc{t^-,t^+}$}\\
\Id(\hi,\hj) & \text{otherwise}
\end{cases}
\end{gather*}

\begin{lemma}
\label{lm-gen}
The set $\bigbrc{G_t^\Sigma}$, $t \in \bra{1:n-1}$, generates
the full distance multiplication monoid of simple unit-Monge matrices.
\end{lemma}
\begin{proof}
Let $P$ be a permutation matrix.
Consider an arbitrary reduced seaweed braid corresponding to $P$,
and let $t$ be the position of its first elementary crossing.
Consider the truncated seaweed braid, 
obtained by removing this seaweed crossing.
This braid is still reduced, and such that the pair of seaweeds 
originating in $t^-$, $t^+$ do not cross.
Let this pair of seaweeds terminate 
at indices $\hk_0$, $\hk_1$, where $\hk_0 < \hk_1$.
Let $Q$ be the permutation matrix corresponding 
to the truncated seaweed braid.
We have
\begin{gather*}
P(t^-,\hk_1) = P(t^+,\hk_0) = 1\\
Q(t^-,\hk_0) = Q(t^+,\hk_1) = 1
\end{gather*}

We will now show that $P = G_t \boxdot Q$ or,
equivalently $P^\Sigma = G_t^\Sigma \odot Q^\Sigma$.
The lemma statement then follows by induction.

Note that $G_t^\Sigma(i,j) = \Id^\Sigma(i,j)$
and $Q^\Sigma(i,j) = P^\Sigma(i,j)$
for all $i \in \bra{0:n}$, $i \neq t$,
and for all $j \in \bra{0:n}$.
Therefore, we have 
\begin{gather*}
\bigpa{G_t^\Sigma \odot Q^\Sigma} (i,k) = 
\bigpa{\Id^\Sigma \odot Q^\Sigma} (i,k) = 
Q^\Sigma (i,k) =
P^\Sigma (i,k)
\end{gather*}
for all $i \in \bra{0:n}$, $i \neq t$,
and for all $k \in \bra{0:n}$.

It remains to consider the case $i=t$.
Note that $G_t^\Sigma(t,j) = \Id^\Sigma(t,j)$
for all $j \in \bra{0:n}$, $j \neq t$.
Let $k \in \bra{0:n}$.
We have
\begin{gather}
\label{eq-gen-min}
\bigpa{G_t^\Sigma \odot Q^\Sigma} (t,k) =
\min_{j \in \bra{0:n}} \bigpa{G_t^\Sigma(t,j) + Q^\Sigma(j,k)}
\end{gather}
By definition of the distribution matrix (\defref{def-distribution}), we have
\begin{gather*}
G_t^\Sigma(t,t-1) = 0\\
G_t^\Sigma(t,t) = G_t^\Sigma(t,t+1) = 1\\
0 \leq Q^\Sigma(t-1,k) - Q^\Sigma(t,k) \leq 1\\
0 \leq Q^\Sigma(t,k) - Q^\Sigma(t+1,k) \leq 1
\end{gather*}
Hence, we have
\begin{multline*}
G_t^\Sigma(t,t) + Q^\Sigma(t,k) = 1 + Q^\Sigma(t,k) \geq {} \\
0 + Q^\Sigma(t-1,k) = G_t^\Sigma(t,t-1) + Q^\Sigma(t-1,k)
\end{multline*}
and, analogously,
\begin{gather*}
G_t^\Sigma(t,t) + Q^\Sigma(t,k) \geq
G_t^\Sigma(t,t+1) + Q^\Sigma(t+1,k)
\end{gather*}
We have established that the value under the minimum 
operator in \eqref{eq-gen-min} for $j=t$
is always no less than the values for both $j=t-1$ and $j=t+1$.
Therefore, the minimum is never attained solely at $j=t$, 
so we may assume $j \neq t$.
We now consider two cases: either $j \in \bra{0:t-1}$, or $j \in \bra{t+1:n}$.

For $j \in \bra{0:t-1}$, we have $G_t^\Sigma(t,j) = 0$.
Therefore,
\begin{gather*}
\min_{j \in \bra{0:t-1}} \bigpa{G_t^\Sigma(t,j) + Q^\Sigma(j,k)} = {}\\
\min_{j \in \bra{0:t-1}} \bigpa{0 + Q^\Sigma(j,k)} = {} 
\why{attained at $j=t-1$}\\
Q^\Sigma(t-1,k) = P^\Sigma(t-1,k)
\end{gather*}
Similarly, for $j \in \bra{t+1:n}$, we have $G_t^\Sigma(t,j) = j-t$.
Therefore, 
\begin{gather*}
\min_{j \in \bra{t+1:n}} \bigpa{G_t^\Sigma(t,j) + Q^\Sigma(j,k)} = {}\\
\min_{j \in \bra{t+1:n}} \bigpa{j-t + Q^\Sigma(j,k)} = {} 
\why{attained at $j=t+1$}\\
1 + Q^\Sigma(t+1,k) = 1 + P^\Sigma(t+1,k)
\end{gather*}
Substituting into \eqref{eq-gen-min}, we now have
\begin{gather*}
\bigpa{G_t^\Sigma \odot Q^\Sigma} (t,k) =
\min \bigpa{P^\Sigma(t-1,k), 1 + P^\Sigma(t+1,k)}
\end{gather*}
Recall that $P(t^-,\hk_1) = P(t^+,\hk_0) = 1$.
We have
\begin{alignat*}{2}
{}
&P^\Sigma(t-1,k) = P^\Sigma(t,k) = P^\Sigma(t+1,k)          &\quad 
&\text{for $k < \hk_0$}\\
&P^\Sigma(t-1,k) = P^\Sigma(t,k) = 1 + P^\Sigma(t+1,k)      &\quad 
&\text{for $\hk_0 < k < \hk_1$}\\
&P^\Sigma(t-1,k) - 1 = P^\Sigma(t,k) = 1 + P^\Sigma(t+1,k)  &\quad 
&\text{for $\hk_1 < k$}
\end{alignat*}
In all three above cases, we have 
\begin{gather*}
\min \bigpa{P^\Sigma(t-1,k), 1 + P^\Sigma(t+1,k)} = P^\Sigma(t,k)
\end{gather*}
which completes the proof.
\end{proof}

We are now able to establish a formal connection
between the unit-Monge monoid and the seaweed monoid.
\begin{theorem}
\label{th-iso}
The distance multiplication monoid of $n \times n$ simple unit-Monge matrices
is isomorphic to the seaweed monoid $\mathcal T_n$.
\end{theorem}
\begin{proof}
We have already established a bijection between the generators of both monoids:
a generator simple unit-Monge matrix $G_t^\Sigma$ corresponds to
a generator $g_t$ of the seaweed monoid $\mathcal T_n$.
It is straightforward to check that relations \eqref{rel-idem}--\eqref{rel-braid}
are verified by matrices $G_t^\Sigma$,
therefore the bijection on the generators defines a homomorphism
from the seaweed monoid to the unit-Monge matrix monoid.
By \lmref{lm-gen}, this homomorphism is surjective,
hence the cardinality of $\mathcal T_n$ 
is at least the number 
of all simple unit-Monge matrices of size $n$, equal to $n!$.
However, by \lmref{lm-nfact}, 
the cardinality of $\mathcal T_n$ is at most $n!$.
Thus, the cardinality of $\mathcal T_n$ is exactly $n!$,
and the two monoids are isomorphic.
\end{proof}

\begin{example}
In \figref{f-mmult-example}, 
the seaweed braids shown in \sfigref{f-mmult-seaweed} 
correspond to the implicit matrix distance product $P_A \boxdot P_B = P_C$
in \sfigref{f-mmult-matrix}.
\end{example}

The seaweed monoid is closely related 
to some other well-known algebraic structures:
\begin{itemize}
\item by replacing the idempotence relations \eqref{rel-idem}
with involution relations $g_t^2 = \mathit{id}$,
we obtain the \emph{Coxeter presentation} of the symmetric group;
\item by removing the idempotence relations \eqref{rel-idem},
and keeping far commutativity \eqref{rel-farcomm}
and braid relations \eqref{rel-braid}, we obtain
the classical \emph{positive braid monoid} 
(see e.g.\ \cite[Section 6.5]{Kassel_Turaev:08});
\item by removing the braid relations \eqref{rel-braid}, 
and keeping idempotence \eqref{rel-idem} 
and far commutativity \eqref{rel-farcomm},
we obtain the \emph{locally free idempotent monoid} 
\cite{Vershik+:00} (see also \cite{Esyp+:05});
\item by introducing the generators' inverses $g_t^{-1}$,
and replacing the idempotence relations \eqref{rel-idem}
with cancellation relations $g_t g_t^{-1} = \mathit{id}$,
we obtain the classical \emph{braid group}.
\end{itemize}

A generalisation of the seaweed monoid is given by 
$0$-Hecke monoids of general Coxeter groups,
also known as \emph{Coxeter monoids}.
These monoids arise naturally as subgroup monoids in groups.
The theory of Coxeter monoids 
can be traced back to Bourbaki \cite{Bourbaki:68},
and was developed in 
\cite{Tsaranov:90,Richardson_Springer:90,Fomin_Greene:98,Buch+:08}.
A further generalisation to \emph{$\mathcal J$-trivial monoids} 
has been studied by Denton et al.\ \cite{Denton+:11}.
The contents of this chapter can be regarded as a first step
in the algorithmic study of such general classes of monoids.

%%=-=-=-=-=-=-=-=-=-=-=-=-=-=-=-=-=-=-=-=-=-=-=-=-=-=-=-=-=-=-=-=-=-=-=-=-=-=%%
\mysection{Bruhat order}
\label{s-bruhat}

Given a permutation, it is natural to ask how well-sorted it is.
In particular, a permutation may be 
either fully sorted (the identity permutation),
or fully anti-sorted (the reverse identity permutation),
or anything in between.
More generally, given two permutations, it is natural to ask 
whether, in some sense, one is ``more sorted'' than the other.

Let $P_A$, $P_B$ be permutation matrices over $\ang{0:n \mid 0:n}$.
A classical ``degree-of-sortedness'' comparison
is given by the following partial order
(see e.g.\ B\'ona \cite{Bona:04}, Hammett and Pittel \cite{Hammett_Pittel:08}, 
and references therein).
\begin{definition}
\label{def-bruhat}
\index{matrix!permutation!Bruhat order}%
\index{$\preceq$: Bruhat order}%
Matrix $P_A$ is lower than matrix $P_B$ in the \emph{Bruhat order},
$P_A \preceq P_B$,
if $P_A$ can be transformed to $P_B$ by a sequence of \emph{anti-sorting} steps.
Each such step substitutes a (not necessarily contiguous)
submatrix of the form $\bigbra{\smat{1&0\\0&1}}$
by a submatrix of the form $\bigbra{\smat{0&1\\1&0}}$.
\end{definition}
Informally, $P_A \preceq P_B$,
if $P_A$ defines a ``more sorted'' permutation than $P_B$.
More precisely, $P_A \preceq P_B$, 
if the permutation defined by $P_A$ can be transformed 
into the one defined by $P_B$ by successive pairwise anti-sorting
between arbitrary pairs of elements.
Symmetrically, the permutation defined by $P_B$ can be transformed 
into the one defined by $P_A$ by successive pairwise sorting
(or, equivalently, by an application of a comparison network;
see e.g.\ Knuth \cite{Knuth:98_3}).

Bruhat order is an important group-theoretic concept,
which can be generalised to arbitrary Coxeter groups
(see Bj\"orner and Brenti \cite{Bjorner_Brenti:05},
Denton et al.\ \cite{Denton+:11}
for more details and further references).

Many equivalent definitions of the Bruhat order on permutations are known;
see e.g.\ Bj\"orner and Brenti \cite{Bjorner_Brenti:05},
Drake et al.\ \cite{Drake+:04}, 
Johnson and Nasserasr \cite{Johnson_Nasserasr:10}.
Probably the simplest one, 
known as \emph{Ehresmann's tableau criterion} \cite{Hammett_Pittel:08}
or \emph{dot criterion} \cite{Bjorner_Brenti:05}, 
is as follows.
\begin{theorem}
\label{th-bruhat}
We have $P_A \preceq P_B$,
if and only if $P_A^\Sigma \leq P_B^\Sigma$ elementwise.
\end{theorem}
\begin{proof}
Straightforward from the definitions; see \cite{Bjorner_Brenti:05}.
\end{proof}
\begin{example}
We have
\begin{gather*}
\bmat{1 & 0 & 0 \\ 0 & 0 & 1 \\ 0 & 1 & 0} = 
\bmat{0 & 1 & 2 & 3 \\ 0 & 0 & 1 & 2 \\ 0 & 0 & 1 & 1 \\ 0 & 0 & 0 & 0}^\square \preceq
\bmat{0 & 1 & 2 & 3 \\ 0 & 1 & 2 & 2 \\ 0 & 0 & 1 & 1 \\ 0 & 0 & 0 & 0}^\square = 
\bmat{0 & 0 & 1 \\ 1 & 0 & 0 \\ 0 & 1 & 0}
\end{gather*}
Note that the permutation matrix on the right can be obtained from the one on the left
by anti-sorting the $2 \times 2$ submatrix at the intersection of the top two rows
with the leftmost and rightmost columns.

We also have
\begin{gather*}
\bmat{1 & 0 & 0 \\ 0 & 0 & 1 \\ 0 & 1 & 0} = 
\bmat{0 & 1 & 2 & 3 \\ 0 & 0 & 1 & 2 \\ 0 & 0 & 1 & 1 \\ 0 & 0 & 0 & 0}^\square \qquad
\bmat{0 & 1 & 0 \\ 1 & 0 & 0 \\ 0 & 0 & 1} = 
\bmat{0 & 1 & 2 & 3 \\ 0 & 1 & 1 & 2 \\ 0 & 0 & 0 & 1 \\ 0 & 0 & 0 & 0}^\square
\end{gather*}
The above two permutation matrices are incomparable in the Bruhat order.
\end{example}
\thref{th-bruhat} immediately gives one an algorithm 
for deciding whether two permutations 
are Bruhat-comparable in time $O(n^2)$.
To the author's knowledge, no asymptotically faster algorithm
for deciding Bruhat comparability has been known so far.

To demonstrate an application of our techniques,
we now give a new characterisation of the Bruhat order 
in terms of the unit-Monge monoid (or, equivalently, the seaweed monoid).
This characterisation will give us a substantially faster algorithm 
for deciding Bruhat comparability.

Intuitively, the connection between the Bruhat order 
and seaweeds is as follows.
Consider matrix $P_A$ and the rotated matrix $P_A^R$.
The matrix rotation induces a one-to-one correspondence
between the nonzeros of $P_A^R$ and $P_A$, and therefore also between 
individual seaweeds in their reduced seaweed braids.
A pair of seaweeds cross in a reduced braid of $P_A^R$, if and only if 
the corresponding pair of seaweeds do not cross in a reduced braid of $P_A$.
Now consider the product braid $P_A^R \boxdot P_A$,
where each seaweed is made up of two mutually corresponding seaweeds
from $P_A^R$ and $P_A$, respectively.
Every pair of seaweeds in braid $P_A^R \boxdot P_A$
either cross in the top sub-braid $P_A^R$,
or in the bottom sub-braid $P_A$, but not in both.
Therefore, the product braid is a reduced seaweed braid, 
in which every pair of seaweeds cross exactly once.
Thus, we have $P_A^R \boxdot P_A = \Id^R$.

Now suppose $P_A \preceq P_B$.
By \thref{th-bruhat}, we have $P_A^\Sigma \leq P_B^\Sigma$ elementwise.
Therefore, by \defref{def-mmult}, 
$P_A^{R\Sigma} \odot P_A^\Sigma \leq P_A^{R\Sigma} \odot P_B^\Sigma$ elementwise,
hence by \thref{th-bruhat}, we have $P_A^R \boxdot P_A \preceq P_A^R \boxdot P_B$.
However, as argued above, $P_A^R \boxdot P_A = \Id^R$,
which is the highest possible permutation matrix in the Bruhat order,
corresponding to the reverse identity permutation.
Therefore, $P_A^R \boxdot P_B = \Id^R$.
We thus have a necessary condition for $P_A \preceq P_B$.
It turns out that this condition is also sufficient,
giving us a new, computationally efficient criterion for Bruhat comparability.
\begin{theorem}
\label{th-bruhat-fast}
We have $P_A \preceq P_B$, if and only if $P_A^R \boxdot P_B = \Id^R$.
\end{theorem}
\begin{proof}
Let $i,j \in \bra{0:n}$.
We have
\begin{gather}
P_A^{R\Sigma}(i,j) + P_B^\Sigma(j,n-i) ={} 
\why{definition of $R$} \\
\bigpa{n-i - P_A^\Sigma(j,n-i)} + P_B^\Sigma(j,n-i) ={} 
\why{term rearrangement} \\
\bigpa{P_B^\Sigma(j,n-i) - P_A^\Sigma(j,n-i)} + n-i 
\label{eq-x}
\end{gather}
We now prove the implication separately in each direction.
\begin{trivlist}
\setlabelit
\item[Necessity.]
Let $P_A \preceq P_B$.
By \eqref{eq-x} and \thref{th-bruhat}, we have
\begin{gather*}
P_A^{R\Sigma}(i,j) + P_B^\Sigma(j,n-i) \geq n-i 
\end{gather*}
This lower bound is attained at $j=0$ (and, symmetrically, $j=n$): 
we have $P_A^{R\Sigma}(i,0) + P_B^\Sigma(0,n-i) = 0+(n-i) = n-i$.
Therefore,
\begin{gather*}
(P_A^R \boxdot P_B)^\Sigma(i,n-i) ={} 
\why{definition of $\boxdot$} \\
\min_j \bigpa{P_A^{R\Sigma}(i,j) + P_B^\Sigma(j,n-i)} = n-i
\why{attained at $j=0$}
\end{gather*}
It is now easy to prove (e.g.\ by induction on $n$) 
that $P_A^R \boxdot P_B = \Id^R$ is the only permutation matrix 
satisfying the above equation for all $i$.
\item[Sufficiency.]
Let $P_A^R \boxdot P_B = \Id^R$.
By \defref{def-mmult}, we have
\begin{gather*}
\min_j \bigpa{P_A^{R\Sigma}(i,j) + P_B^\Sigma(j,n-i)} = 
\Id^{R\Sigma}(i,n-i) = n-i
\end{gather*}
for all $i$.
Therefore, for all $i$, $j$,
$P_A^{R\Sigma}(i,j) + P_B^\Sigma(j,n-i) \geq n-i$.
By \eqref{eq-x}, this is equivalent to 
$P_B^\Sigma(j,n-i) - P_A^\Sigma(j,n-i) \geq 0$,
therefore $P_A^\Sigma(j,n-i) \leq P_B^\Sigma(j,n-i)$,
hence by \thref{th-bruhat}, we have $P_A \preceq P_B$.
\qed
\end{trivlist}
\renewcommand{\qed}{}
\end{proof}

The combination of \threfs{th-mmult} and \ref{th-bruhat-fast}
gives us a fast algorithm for deciding Bruhat comparability of permutations.
\begin{theorem}
Given permutation matrices $P_A$, $P_B$, 
it is possible to determine whether $P_A \preceq P_B$
in time $O(n \log n)$.
\end{theorem}
\begin{proof}
Immediately from \threfs{th-mmult} and \ref{th-bruhat-fast}.
\end{proof}

%%===========================================================================%%

%%===========================================================================%%
\mychapter{Semi-local string comparison}
\label{c-semi}

In this chapter, we introduce semi-local string comparison,
and establish its connection with the mathematical concepts
of the previous chapter.

This chapter is organised as follows. 
In \secrefs{s-lcs}, \ref{s-slcs},
we formally define the LCS and the semi-local LCS problems.
In \secref{s-adag}, we describe a representation 
of the semi-local LCS problem by an alignment dag,
and introduce the associated score matrix and seaweed matrix.
In \secref{s-matrix-notation},
we introduce some further notation relevant to the semi-local LCS problem.
In \secref{s-adag-composition}, we discuss the fundamental operation
of seaweed matrix composition.
We then obtain an efficient algorithm for this operation,
based on the algebraic framework developed in the previous chapter.

%%=-=-=-=-=-=-=-=-=-=-=-=-=-=-=-=-=-=-=-=-=-=-=-=-=-=-=-=-=-=-=-=-=-=-=-=-=-=%%
\mysection{The LCS problem}
\label{s-lcs}

We will consider strings of characters taken from an alphabet.
No a priori assumptions are made on the size of the alphabet
and on the set of primitive character operations;
we will make specific assumptions in different contexts
(e.g.\ a fixed finite alphabet with only equality comparisons,
or an alphabet of integers up to a given $n$ 
with standard arithmetic operations, etc.)
Two alphabet characters $\alpha$, $\beta$ \emph{match}, 
if $\alpha=\beta$, and \emph{mismatch} otherwise.
\index{\charguard: guard character}%
\index{\charwild: wildcard character}%
In addition to alphabet characters, we introduce two special extra characters:
the \emph{guard character} `\charguard',
which only matches itself and no other characters,
and the \emph{wildcard character} `\charwild',
which matches itself and all other characters.

\index{string!indexing}%
\index{string!concatenation!right}%
\index{string!concatenation!left}%
\index{$ab$, $\underline{a}b$: string right concatenation}%
\index{$a \underline{b}$: string left concatenation}%
It will be convenient to index strings 
by half-integer, rather than integer indices, e.g.\ %
string $a = \alpha_{0^+} \alpha_{1^+} \ldots \alpha_{m^-}$.
We will index strings as vectors, writing e.g.\ %
$a(\hi) = \alpha_{\hi}$, $a\ang{i:j} = \alpha_{i^+} \ldots \alpha_{j^-}$.
Given strings $a$ over $\ang{i:j}$ and $b$ over $\ang{i':j'}$,
we will distinguish between string \emph{right concatenation} $\underline{a}b$,
which is over $\ang{i:j+j'-i'}$ and preserves the indexing within $a$,
and \emph{left concatenation} $a \underline{b}$,
which is over $\ang{i'-j+i:j'}$ and preserves the indexing within $b$.
We extend this notation to concatenation of more than two strings,
e.g.\ $a \underline{b} c$ is a concatenation of three strings,
where the indexing of the second string is preserved.
If no string is marked in the concatenation, 
then right concatenation is assumed by default.

\index{substring}%
\index{subsequence}%
Given a string, we distinguish between its contiguous \emph{substrings},
and not necessarily contiguous \emph{subsequences}.
\index{prefix}
\index{suffix}
Special cases of a substring 
are \emph{a prefix} and \emph{a suffix} of a string.
Unless indicated otherwise, an algorithm's input is
a string $a$ of length $m$, and a string $b$ of length $n$.

A classical approach to string comparison is based 
on the following numerical measure of string similarity.
\begin{definition}
\label{def-lcs}
\index{problem!longest common subsequence (LCS)}%
\index{alignment score!LCS}
\index{$\lcs(a,b)$: LCS score}
Given strings $a$, $b$, the \emph{longest common subsequence (LCS) problem}
asks for the length of the longest string 
that is a subsequence of both $a$ and $b$.
We will call this length the \emph{LCS score} of strings $a$, $b$,
and denote it by $\lcs(a,b)$.
\end{definition}

\begin{example}
Let 
\begin{gather*}
a = \textsf{``BAABCBCA''}\\
b = \textsf{``BAABCABCABACA''}
\end{gather*}
This example, borrowed from Alves et al.\ \cite{Alves+:08},
will serve as a running example for this chapter.
String $b$ of length $13$ contains 
the whole string $a$ of length $8$ as a subsequence,
therefore we have
\begin{gather*}
\lcs(a,b) = 8
\end{gather*}
The LCS score of string $a$ 
against substring $b\ang{4:11} = \textsf{``CABCABA''}$ of length $7$
is realised by a common subsequence \textsf{``ABCBA''} of length $5$,
therefore we have
\begin{gather*}
\lcs(a,b\ang{4:11}) = 5
\end{gather*}
\end{example}

The classical dynamic programming algorithm for the LCS problem 
\cite{Needleman_Wunsch:70,Wagner_Fischer:74} runs in time $O(mn)$.
The best known algorithms improve on this running time
by some (model-dependent) polylogarithmic factors
\cite{Masek_Paterson:80,Crochemore+:03_SIAM,Wu+:96,Bille_Farach:08}.
We will recall the necessary background on LCS algorithms
in \chaprefs{c-seaweed}, \ref{c-compressed}.

\index{problem!subsequence recognition!global}%
\index{problem!subsequence matching}%
A simple special case of the LCS problem
is the \emph{(global) subsequence recognition problem}
(also known as the ``subsequence matching problem'').
Given a \emph{text string} $t$ of length $n$ 
and a \emph{pattern string} $p$ of length $m \leq n$,
the problem asks whether the text $t$ 
contains the whole pattern $p$ as a subsequence.
This is equivalent to asking 
whether the LCS score of $t$ against $p$ is exactly $m$. 
The global subsequence recognition problem has been considered 
e.g.\ by Aho et al.\ \cite[Section 9.3]{Aho+:76},
who describe a straightforward algorithm running in time $O(n)$.
Various extensions of this problem have been explored
by Crochemore et al.\ \cite{Crochemore+:03_JDA}.

A more detailed measure of string similarity
can be obtained by comparing strings locally by their substrings.
Such an approach is particularly useful in biological applications.
We will consider local string comparison in \chapref{c-beyond}.

%%=-=-=-=-=-=-=-=-=-=-=-=-=-=-=-=-=-=-=-=-=-=-=-=-=-=-=-=-=-=-=-=-=-=-=-=-=-=%%
\mysection{Semi-local LCS}
\label{s-slcs}

Although global comparison (full string against full string)
and local comparison (all substrings against all substrings)
are the two most common approaches to comparing strings,
in between of them there is another important type of string comparison.
\begin{definition}
\label{def-semi-local}
\index{problem!semi-local LCS}%
\index{problem!string-substring LCS}%
\index{problem!prefix-suffix LCS}%
\index{problem!suffix-prefix LCS}%
\index{problem!substring-string LCS}%
Given strings $a$, $b$, the \emph{semi-local LCS problem}
asks for the LCS scores as follows:
\begin{itemize}
\item the whole $a$ against every substring of $b$
(\emph{string-substring LCS});
\item every prefix of $a$ against every suffix of $b$
(\emph{prefix-suffix LCS});
\item every suffix of $a$ against every prefix of $b$
(\emph{suffix-prefix LCS});
\item every substring of $a$ against the whole $b$
(\emph{substring-string LCS}).
\end{itemize}
\index{problem!string-substring LCS!extended}%
\index{problem!substring-string LCS!extended}%
The first three (respectively, the last three) components, taken together,
will also be called the \emph{extended string-substring}
(respectively, \emph{substring-string}) \emph{LCS problem}.
These versions of the problem will be useful 
whenever string $a$ (respectively $b$) is too long
for considering all its substrings.
\end{definition}
Semi-local string comparison will be the main focus of this work.

Some alternative terms for semi-local comparison,
used especially in biological texts, 
are ``end-free comparison" \cite[Subsection 11.6.4]{Gusfield:97}
or ``semi-global alignment'' 
\cite[Problem 6.24]{Jones_Pevzner:04}, \cite[Section 8.4]{Gogol+:10}.
The string-substring (and its symmetric substring-string) component
of semi-local string comparison is also called ``fitting alignment'' 
\cite[Problem 6.23]{Jones_Pevzner:04}.
String-substring LCS is an important problem in its own right,
closely related to approximate pattern matching,
where a short fixed pattern string 
is compared to various substrings of a long text string.
We will consider approximate pattern matching in \chapref{c-weighted}.
The prefix-suffix (and the symmetric suffix-prefix) LCS problem,
sometimes called ``overlap alignment'' 
\cite[Problem 6.22]{Jones_Pevzner:04}, \cite[Section 8.4]{Gogol+:10},
also occurs independently in some applications.

Many string comparison algorithms
output either a single optimal comparison score 
across all local comparisons,
or a number of local comparison scores 
that are ``sufficiently close'' to the globally optimal.
In contrast with this approach, \defref{def-semi-local}
asks for all the locally optimal comparison scores.
This approach is more flexible, and will be useful 
for various algorithmic applications described later in this work.

It turns out that, although more general than the LCS problem,
the semi-local LCS problem can still be solved in time $O(mn)$;
similarly to the classical LCS problem.
It is also possible to obtain (model-dependent) 
polylogarithmic speedups on this running time.
We will consider semi-local LCS algorithms 
on plain strings in \chapref{c-seaweed},
and on compressed strings in \chapref{c-compressed}.

\index{problem!subsequence recognition!local}%
A special case of the semi-local LCS problem
is the \emph{local subsequence recognition problem},
which, given a text $t$ and a pattern $p$, 
asks for the substrings in $t$ containing $p$ as a subsequence.
This problem can also be regarded as a basic form
of approximate pattern matching.
We will consider algorithms for local subsequence recognition
and other types of approximate pattern matching
on plain strings in \chapref{c-weighted},
and on compressed strings in \chapref{c-compressed}.

%%=-=-=-=-=-=-=-=-=-=-=-=-=-=-=-=-=-=-=-=-=-=-=-=-=-=-=-=-=-=-=-=-=-=-=-=-=-=%%
\mysection{Alignment dags and seaweed matrices}
\label{s-adag}

A standard method for the LCS problem represents
a problem instance by a \emph{dag} (directed acyclic graph)
on a rectangular grid of nodes,
where every edge is assigned a score of either $0$ or $1$.
\begin{definition}
\label{def-gd-dag}
\index{dag!grid-diagonal}%
A \emph{grid-diagonal dag} is a weighted dag, 
defined on the set of nodes $v_{l,i}$, 
$l \in \bra{0:m}$, $i \in \bra{0:n}$.
The edge and path weights are called \emph{scores}.
For all $l \in \bra{0:m}$, $\hl \in \ang{0:m}$,
$i \in \bra{0:n}$, $\hi \in \ang{0:n}$,
the grid-diagonal dag contains:
\begin{itemize}
\item the horizontal edge $v_{l,\hi^-} \to v_{l,\hi^+}$
and the vertical edge $v_{\hl^-,i} \to v_{\hl^+,i}$,
both with score $0$;
\item the diagonal edge $v_{\hl^-,\hi^-} \to v_{\hl^+,\hi^+}$ 
with score either $0$ or $1$.
\end{itemize}
\end{definition}
A grid-diagonal dag can be viewed as an $m \times n$ grid of \emph{cells}.

\begin{definition}
\label{def-alignment-dag}
\index{dag!alignment}%
\index{$\rG_{a,b}$: alignment dag}%
An instance of the semi-local LCS problem on strings $a$, $b$
corresponds to an $m \times n$ grid-diagonal dag $\rG_{a,b}$,
called the \emph{alignment dag} of $a$ and $b$.
A cell indexed by $\hl \in \ang{0:m}$, $\hi \in \ang{0:n}$
is called a \emph{match cell}, if $a(\hl)$ matches $b(\hi)$,
and a \emph{mismatch cell} otherwise
(recall that the strings may contain wildcard characters).
The diagonal edges in match cells have score $1$,
and in mismatch cells score $0$.
\end{definition}
\begin{figure}[tb]
\centering
\includegraphics{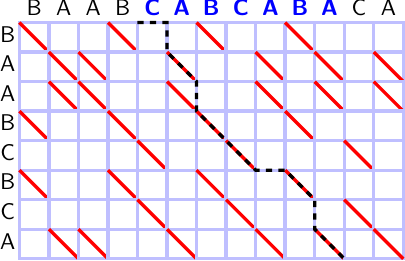}
\caption{\label{f-align} Alignment dag $\rG_{a,b}$ and a highest-scoring path}
\end{figure}

\begin{example}
\figref{f-align} shows the alignment dag for strings 
$a = \textsf{``BAABCBCA''}$, $b = \textsf{``BAABCABCABACA''}$.
All edges are directed left-to-right and top-to-bottom.
The diagonal edges of score $0$ are not shown.
The colour of the remaining edges indicates their scores:
blue (respectively, red) corresponds to edge score $0$ (respectively, $1$).
\end{example}

\index{dag!alignment!full-mismatch}%
\index{dag!alignment!full-match}%
A particular special case of an alignment dag 
is the \emph{full-mismatch dag}, 
which consists entirely of mismatch cells.
This dag can be obtained as the alignment dag 
of a pair strings that have no characters in common.
Another special case is the \emph{full-match dag},
which consists entirely of match cells.
This dag can be obtained as the alignment dag 
of a pair of strings over an alphabet of a single character,
or, alternatively, 
a pair of strings, one of which consists entirely of wildcard characters.

Given a pair of strings $a$, $b$, 
their semi-local common subsequences correspond
to boundary-to-boundary paths in the alignment dag $\rG_{a,b}$.
In particular, a common string-substring, 
suffix-prefix, prefix-suffix, or substring-string subsequence
corresponds, respectively,
to a top-to-bottom, left-to-bottom, top-to-right, and left-to-right path.
The length of a subsequence is equal 
to the total score of the corresponding path.
The semi-local LCS problem is therefore equivalent 
to finding the maximum path scores of the following four types:
\begin{alignat}{2}
{}
&\lcs\bigpa{a, b\ang{i:i'}} &&{}= 
\max \mathit{score}(v_{0,i} \pathto v_{m,i'}) \label{eq-lcs1}\\
&\lcs\bigpa{a\ang{l:m}, b\ang{0:i'}} &&{}= 
\max \mathit{score}(v_{l,0} \pathto v_{m,i'}) \label{eq-lcs2}\\
&\lcs\bigpa{a\ang{0:l'}, b\ang{i:n}} &&{}= 
\max \mathit{score}(v_{0,i} \pathto v_{l',n}) \label{eq-lcs3}\\
&\lcs\bigpa{a\ang{l:l'}, b} &&{}= 
\max \mathit{score}(v_{l,0} \pathto v_{l',n}) \label{eq-lcs4}
\end{alignat}
where $l,l' \in [0:m]$, $i,i' \in [0:n]$,
and the score maxima are taken across all paths between the given endpoints.
The diagonal edges with score $0$ in mismatch cells 
do not affect maximum node-to-node scores, and can therefore be ignored.
\begin{example}
In \figref{f-align}, the highlighted top-to-bottom path
corresponds to the string-substring LCS score
$\lcs\bigpa{a,b\ang{4:11}} = \lcs(a,\textsf{``CABCABA''}) = 5$.
\end{example}

Finding the maximum boundary-to-boundary path scores in an alignment dag 
is equivalent to finding the corresponding distances 
in an undirected graph, obtained from the alignment dag
by assigning length $1$ to vertical and horizontal edges,
assigning lengths $0$ and $2$ to diagonal edges 
in match and mismatch cells respectively, 
and ignoring edge directions.
The problem thus becomes a special case of
the problem to find distances between boundary nodes
and all nodes in a planar graph.
In the case a grid-diagonal dag with arbitrary real weights,
this problem has been studied by Schmidt \cite{Schmidt:98}.
In the still more general case of 
an arbitrary real-weighted undirected planar graph,
the problem has been studied by Klein \cite{Klein:05} 
and Cabello and Chambers \cite{Cabello_Chambers:07}.
These more general algorithms are slower than the ones presented in this work,
where we exploit both the grid-diagonal structure of the dag,
and the discreteness of edge lengths.

The analysis of the four different path types \eqref{eq-lcs1}--\eqref{eq-lcs4} 
can be simplified by padding string $b$
with wildcard characters on both sides,
and then considering paths in the alignment dag
for string $a$ over $\ang{0:m}$ 
against the padded string $\charwild^m \underline{b} \charwild^m$ 
over $\ang{-m:m+n}$.
\begin{definition}
\label{def-matrix}
\index{matrix!score!semi-local}%
\index{$\rH_{a,b}$: semi-local score matrix}%
Given strings $a$, $b$,
the corresponding \emph{semi-local score matrix%
\footnote{These matrices are called ``DIST matrices'' 
e.g.\ in \cite{Schmidt:98,Crochemore+:04}, 
and ``score matrices'' in \cite{Tiskin:06_CPM}.
Our current terminology is chosen to reflect 
the semi-local score-maximising nature of the matrix elements, 
while avoiding confusion with pairwise substitution score matrices
used in comparative genomics (see e.g.\ \cite{Jones_Pevzner:04}).}}
is a matrix over $\bra{-m:n \mid 0:m+n}$, defined by
\begin{gather*}
\rH_{a,b}(i,j) = \max \mathit{score}(v_{0,i} \pathto v_{m,j})
\end{gather*}
where $i \in \bra{-m:n}$, $j \in \bra{0:m+n}$,
and the maximum is taken across all paths 
between the given endpoints $v_{0,i}$, $v_{m,j}$
in the $m \times (2m+n)$ padded alignment dag 
$\rG_{a,\charwild^m \underline{b} \charwild^m}$.
If $i=j$, we have $\rH_{a,b}(i,j)=0$. 
By convention, if $j < i$, then we let $\rH_{a,b}(i,j)=j-i<0$.
\end{definition}
\begin{figure}[tb]
\centering
\includegraphics{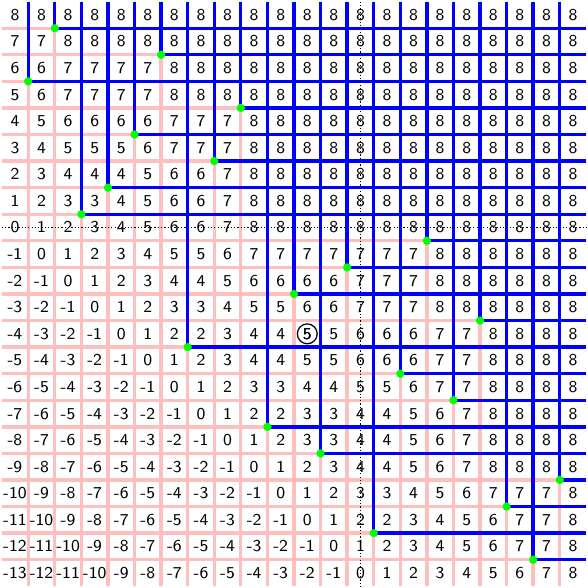}
\caption{\label{f-score} 
Matrices $\rH_{a,b}$ and $\rP_{a,b}$}
\end{figure}
\begin{example}
\figref{f-score} shows the matrix $\rH_{a,b}$,
giving all semi-local LCS scores 
for strings $a$, $b$ as in the previous examples.
The entry $\rH_{a,b}(4,11)=5$ is circled.
\end{example}

The solution for each of the four components
of the semi-local LCS problem in \defref{def-semi-local}
can now be obtained from \eqref{eq-lcs1}--\eqref{eq-lcs4} as follows:
\begin{alignat*}{2}
{}
&\lcs\bigpa{a, b\ang{i:i'}} &&{}= \rH_{a,b}(i,i')\\
&\lcs\bigpa{a\ang{l:m}, b\ang{0:i'}} &&{}= \rH_{a,b}(-l,i')-l\\
&\lcs\bigpa{a\ang{0:l'}, b\ang{i:n}} &&{}= \rH_{a,b}(i,m+n-l')-m+l'\\
&\lcs\bigpa{a\ang{l:l'}, b} &&{}= \rH_{a,b}(-l,m+n-l')-m-l+l'
\end{alignat*}

where $l,l' \in [0:m]$, $i,i' \in [0:n]$.

The key property of semi-local score matrices 
is captured by the following theorem.
\begin{theorem}
\label{th-ps}
Given strings $a$, $b$,
the corresponding semi-local score matrix $\rH_{a,b}$ is unit-anti-Monge.
More precisely, we have
\begin{gather*}
\rH_{a,b}(i,j) = j - i - \rP_{a,b}^\Sigma (i,j) = m - \rP_{a,b}^{T \Sigma T} (i,j)
\end{gather*}
where $\rP_{a,b}$ is a permutation matrix over $\ang{-m:n \mid 0:m+n}$.
In particular, string $a$ is a subsequence of substring $b\ang{i:j}$ 
for some $i,j \in \bra{0:n}$, if and only if $\rP_{a,b}^{T \Sigma T} (i,j) = 0$.
\end{theorem}
\begin{proof}
For any crossing pair of highest-scoring paths,
$v_{0,\hi^+} \pathto v_{m,\hj^-}$ and $v_{0,\hi^-} \pathto v_{m,\hj^+}$,
where $\hi \in \ang{-m:n}$, $\hj \in \ang{0:m+n}$,
there exists a non-crossing pair of paths 
$v_{0,\hi^-} \pathto v_{m,\hj^-}$ and $v_{0,\hi^+} \pathto v_{m,\hj^+}$
of at least the same total score,
which can be obtained by rearranging the edges in the paths.
Therefore, we have 
\begin{gather*}
\rH_{a,b}(\hi^+,\hj^-) + \rH_{a,b}(\hi^-,\hj^+) \leq
\rH_{a,b}(\hi^-,\hj^-) + \rH_{a,b}(\hi^+,\hj^+)
\end{gather*}
for all $\hi \in \ang{-m:n}$, $\hj \in \ang{0:m+n}$,
hence $\rH_{a,b}^\square (\hi,\hj) \leq 0$,
and matrix $\rH_{a,b}$ is anti-Monge.

Let $N(i,j) = j - i - \rH_{a,b}(i,j)$.
From the above, matrix $N$ is Monge.
We have
\begin{gather*}
N(n,j) = j - n - \rH_{a,b}(n,j) = j-n - (j-n) = 0\\
N(i,0) = 0 - i - \rH_{a,b}(i,0) = -i - (-i) = 0
\end{gather*}
for all $i \in \bra{-m:n}$, $j \in \bra{0:m+n}$,
hence $P$ is simple.
Furthermore, we have
\begin{gather*}
N(\hi^-, m+n) - N(\hi^+, m+n) ={}
\why{definitions of $N$ and $\rH_{a,b}$} \\
\bigpa{m+n - \hi^- - (m+n)} - \bigpa{m+n - \hi^+ - (m+n)} ={} 
\why{cancellation} \\ 
\tHalf+\tHalf = 1
\end{gather*}
for all $\hi \in \ang{-m:n}$.
Analogously,
\begin{gather*}
N(-m, \hj^+) - N(-m, \hj^-) = 1
\end{gather*}
for all $\hj \in \ang{0:m+n}$.
By \lmref{lm-unit-Monge}, matrix $P$ is simple unit-Monge.
Therefore, $N = \rP_{a,b}^\Sigma$, where $\rP_{a,b}$ is a permutation matrix.
The second equality in the theorem statement 
is now straightforward from the definitions.
\end{proof}
The intuition behind \thref{th-ps} is as follows.
Let $a'$ be a substring of $a$, and consider its LCS score against string $b$.
If substring $a'$ is extended at either end by one character,
the LCS score is either unchanged, or increased by $1$,
depending on whether or not the new character of $a'$ 
can be usefully matched to a character of $b$.
The unit-anti-Monge property of matrix $\rH_{a,b}$
describes the fact that, as substring $a'$ grows in size,
so does the number of obstacles to useful matching of characters.
Each obstacle, in other words 
a ``critical point'' at which a useful match become useless,
corresponds to a nonzero in the permutation matrix $\rP_{a,b}$. 
However, this is only very general intuition:
in fact, the notion of a ``useful match'' is not absolute,
but depends on the choice of a particular highest-scoring path
through the alignment dag.

\thref{th-ps} and its proof hold more generally 
for any grid-diagonal dag, not necessarily obtainable 
as an alignment dag of any particular pair of strings.
However, the given form of the theorem 
will be sufficient for the rest of this work.

The key idea of our approach is to view \thref{th-ps}
as a description of an implicit solution to the semi-local LCS problem:
the semi-local score matrix $\rH_{a,b}$
is represented implicitly by (the nonzeros of) the permutation matrix $\rP_{a,b}$.
\begin{definition}
\label{def-matrix-seaweed}
\index{matrix!seaweed!semi-local}%
\index{$\rP_{a,b}$: semi-local seaweed matrix}%
Given strings $a$, $b$, the \emph{semi-local seaweed matrix}
is a permutation matrix $\rP_{a,b}$ over $\ang{-m:n \mid 0:m+n}$,
defined by \thref{th-ps}.
\end{definition}
\defref{def-matrix-seaweed} leads to 
the following interpretation of \thref{th-ps}.
The LCS score for string $a$ against substring $b\ang{i:j}$
is determined by the length $j-i$ of $b\ang{i:j}$,
and by the number $\rP_{a,b}^\Sigma(i,j)$ of nonzeros in $\rP_{a,b}$ 
that are $\gtrless$-dominated by the point $(i,j)$.
Alternatively, the same LCS score is determined
by the length $m$ of string $a$,
and by the number $\rP_{a,b}^{T \Sigma T}(i,j)$ of nonzeros in $\rP_{a,b}$
that $\gtrless$-dominate the point $(i,j)$.

\begin{example}
\figref{f-score} shows the unit-anti-Monge property of matrix $\rH_{a,b}$ 
by the coloured grid pattern: red (respectively, blue) lines 
separate matrix elements that differ by $1$ (respectively, by $0$).
The nonzeros of the semi-local seaweed matrix $\rP_{a,b}$ 
over $\ang{-8:13 \mid 0:8+13}$ are shown by green bullets.

Nonzeros of $\rP_{a,b}$ that are $\gtrless$-dominated by the point $(4,11)$
correspond to the green bullets 
lying below-left of the circled entry.
There are exactly two such nonzeros, 
therefore $\rP_{a,b}^\Sigma(4,11) = 2$,
and we have $\rH_{a,b}(4,11) = 11-4-\rP_{a,b}^\Sigma(4,11) = 11-4-2 = 5$.

Nonzeros of $\rP_{a,b}$ that $\gtrless$-dominate the point $(4,11)$
correspond to the green bullets
lying above-right of the circled entry.
There are exactly three such nonzeros, 
therefore $\rP_{a,b}^{T \Sigma T}(4,11) = 3$,
and we have $\rH_{a,b}(4,11) = 8-\rP_{a,b}^{T \Sigma T}(4,11) = 8-3 = 5$.
\end{example}

\begin{figure}[tb]
\centering
\includegraphics{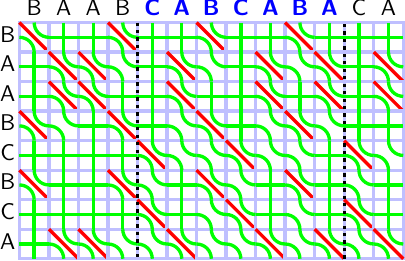}
\caption{\label{f-align-seaweeds} 
Alignment dag $\rG_{a,b}$ and nonzeros of $\rP_{a,b}$ as seaweeds}
\end{figure}

A semi-local seaweed matrix can be naturally identified
with a reduced seaweed braid of width $m+n$
(or, more precisely, with a family of equivalent reduced seaweed braids).
As we show in the following section,
divide-and-conquer solutions to the semi-local LCS problem 
generally correspond to implicit distance multiplication of seaweed matrices,
and therefore also to the multiplication of the corresponding seaweed braids.

\begin{example}
\figref{f-align-seaweeds} shows matrix $\rP_{a,b}$ 
as a reduced seaweed braid of width $8+13=21$,
laid out directly on the alignment dag $\rG_{a,b}$.
The nonzeros correspond to seaweeds, laid out as paths in the dual graph.
We say that a seaweed goes \emph{from $\hi$ to $\hj$},
if it originates between the nodes $v_{0,\hi^-}$ and $v_{0,\hi^+}$,
and terminates between the nodes $v_{8,\hj^-}$ and $v_{8,\hj^+}$.
In particular, every nonzero $\rP_{a,b}(\hi,\hj)=1$, 
where $\hi,\hj \in \ang{0:13}$, 
is represented by a seaweed going from $\hi$ to $\hj$.
The remaining seaweeds,
originating or terminating at the sides of the dag,
correspond to nonzeros $\rP_{a,b}(\hi,\hj)=1$,
where either $\hi \in \ang{-8:0}$, or $\hj \in \ang{13:8+13}$, or both.
For the purposes of this example,
the specific layout of the seaweeds between their endpoints is not important.
However, this layout will become meaningful 
in the context of the algorithms described in the next chapter.

The full set of $8+13=21$ nonzeros %in \figref{f-score}
corresponds to the full set of $21$ seaweeds in \figref{f-align-seaweeds}.
The two nonzeros that are $\gtrless$-dominated by the point $(4,11)$
correspond to the two seaweeds,
going from $4.5$ to $6.5$ and from $7.5$ to $9.5$.
These two seaweeds fit completely 
between the two dashed vertical lines $i=4$ and $j=11$.
The three nonzeros that $\gtrless$-dominate the point $(4,11)$
correspond to the three seaweeds
going from $0.5$ to the right boundary, from $1.5$ to $13.5$,
and from $3.5$ to the right boundary.
These three seaweeds pierce both dashed vertical lines.
\end{example}

When talking about semi-local score and seaweed matrices,
we will sometimes omit the qualifier ``semi-local'',
as long as it is clear from the context.

The definitions of score and seaweed matrices are not symmetric
with respect to the order of the input strings.
The precise relationship between score matrices $\rH_{a,b}$, $\rH_{b,a}$,
and between seaweed matrices $\rP_{a,b}$, $\rP_{b,a}$,
is given by the following lemma.
\begin{lemma}
\label{lm-switch}
Given input strings $a$, $b$, we have
\begin{gather*}
\rH_{b,a}(i,j) = \rH_{a,b}(-i,m+n-j) - i + j - n
\end{gather*}
for all $i \in \bra{0:n}$, $j \in \bra{0:m}$, and
\begin{gather*}
\rP_{b,a}(\hi,\hj) = \rP_{a,b}(-\hi,m+n-\hj)
\end{gather*}
for all $\hi \in \ang{0:n}$, $\hj \in \ang{0:m}$. 
\end{lemma}
\begin{proof}
Straightforward by definitions.
\end{proof}

%%=-=-=-=-=-=-=-=-=-=-=-=-=-=-=-=-=-=-=-=-=-=-=-=-=-=-=-=-=-=-=-=-=-=-=-=-=-=%%
\mysection{Seaweed submatrix notation}
\label{s-matrix-notation}

The four individual components of the semi-local LCS problem
correspond to a partitioning of both the score matrix $\rH_{a,b}$ 
and the seaweed matrix $\rP_{a,b}$ into submatrices.
It will be convenient to introduce a special notation
for these submatrices as follows:
\begin{center}
\begin{tabular}{|c|cc|}
\hline
& $\bra{0:n}$ & $\bra{n:m+n}$ \\ \hline
$\bra{-m:0}$ & $\rH_{a,b}^\sufpre$ & $\rH_{a,b}^\subs$ \\
$\bra{0:n}$  & $\rH_{a,b}^\ssub$   & $\rH_{a,b}^\presuf$ \\ \hline
\end{tabular}
\qquad
\begin{tabular}{|c|cc|}
\hline
& $\ang{0:n}$ & $\ang{n:m+n}$ \\ \hline
$\ang{-m:0}$ & $\rP_{a,b}^\sufpre$ & $\rP_{a,b}^\subs$ \\
$\ang{0:n}$  & $\rP_{a,b}^\ssub$   & $\rP_{a,b}^\presuf$ \\ \hline
\end{tabular}
\end{center}
For example, we have $\rH^\ssub_{a,b} = \rH_{a,b}\bra{0:n \mid 0:n}$,
which is the matrix of all string-substring LCS scores for strings $a$, $b$.
Note that the four resulting submatrices of $\rH_{a,b}$
overlap by one row/column at the boundaries
(in particular, the global LCS score $\rH_{a,b}(0,n)$
belongs to all four score submatrices),
while the corresponding submatrices of $\rP_{a,b}$ are disjoint:
\begin{gather*}
\rP_{a,b} = \bmat{\rP^\sufpre_{a,b} & \rP_{a,b}^\subs \\ 
                  \rP^\ssub_{a,b}   & \rP^\presuf_{a,b}}
\end{gather*}
\begin{definition}
\label{def-matrix-split}
\index{matrix!score!string-substring}%
\index{matrix!score!prefix-suffix}%
\index{matrix!score!suffix-prefix}%
\index{matrix!score!substring-string}%
\index{$\rH^\ssub_{a,b}$: string-substring score matrix}%
\index{$\rH^\presuf_{a,b}$: prefix-suffix score matrix}%
\index{$\rH^\sufpre_{a,b}$: suffix-prefix score matrix}%
\index{$\rH^\subs_{a,b}$: substring-string score matrix}%
\index{matrix!seaweed!string-substring}%
\index{matrix!seaweed!prefix-suffix}%
\index{matrix!seaweed!suffix-prefix}%
\index{matrix!seaweed!substring-string}%
\index{$\rP^\ssub_{a,b}$: string-substring seaweed matrix}%
\index{$\rP^\presuf_{a,b}$: prefix-suffix seaweed matrix}%
\index{$\rP^\sufpre_{a,b}$: suffix-prefix seaweed matrix}%
\index{$\rP^\subs_{a,b}$: substring-string seaweed matrix}%
Given strings $a$, $b$, 
the corresponding \emph{suffix-prefix, substring-string, 
string-substring and prefix-suffix score} 
(respectively, \emph{seaweed}) \emph{matrices}
are the submatrices $\rH^\sufpre_{a,b}$, $\rH^\subs_{a,b}$, $\rH^\ssub_{a,b}$, $\rH^\presuf_{a,b}$
(respectively, $\rP^\sufpre_{a,b}$, $\rP^\subs_{a,b}$, $\rP^\ssub_{a,b}$, $\rP^\presuf_{a,b}$).
\end{definition}
Similarly to the full semi-local seaweed matrix, 
its component seaweed submatrices can be processed 
into an efficient data structure of \thref{th-query}
for answering individual semi-local score queries.
\begin{example}
\figref{f-score} shows by thin dotted lines
the partitioning of $\rH_{a,b}$ and $\rP_{a,b}$
into submatrices of \defref{def-matrix-split}.
The suffix-prefix, substring-string, 
string-substring and prefix-suffix submatrices are respectively 
on the top-left, top-right, bottom-left and bottom-right.
The elements of $\rH_{a,b}$ lying directly on the dotted lines
are shared by the bordering score submatrices.
Note that the substring-string score submatrix $\rH^\subs_{a,b}$ 
and seaweed submatrix $\rP^\subs_{a,b}$ 
are both filled with a constant value ($8$ and $0$, respectively);
this is due to the fact that the whole string $a$ is a subsequence of $b$.
\end{example}

The nonzeros of each seaweed submatrix introduced in \defref{def-matrix-split}
can be regarded as an implicit solution 
to the corresponding component of the semi-local LCS problem.
Similarly, by considering only three out of the four submatrices,
we can define an implicit solution 
to the extended string-substring (respectively, substring-string) LCS problem.
\begin{definition}
\label{def-matrix-ext}
%
\begin{comment}
\index{matrix!score!string-substring!extended}%
\index{matrix!score!substring-string!extended}%
\index{$\rH^\ssubX_{a,b}$: extended string-substring score matrix}%
\index{$\rH^\subsX_{a,b}$: extended substring-string score matrix}%
Given strings $a$, $b$, we define the \emph{extended string-substring} 
(respectively, \emph{substring-string}) \emph{score matrix} 
over $\bra{-m:n \mid 0:m+n}$ as
%
\begin{gather*}
%
\rH^\ssubX_{a,b} = \bmat{\rH^\sufpre_{a,b} & \cdot \\ \rH^\ssub_{a,b} & \rH^\presuf_{a,b}}
\qquad
%
\rH^\subsX_{a,b} = \bmat{\rH^\sufpre_{a,b} & \rH^\subs_{a,b} \\ \cdot & \rH^\presuf_{a,b}}
%
\draw[dotted] (+0.2,-2) |- (-2,+0.2);
\draw[dashed] (-0.2,-2) |- (+2,+0.2);
\draw[dashdotted] (-0.2,+2) |- (+2,-0.2);
%
\end{tikzpicture}\right]
%
\end{gather*}
\end{comment}
%
\index{matrix!seaweed!string-substring!extended}%
\index{matrix!seaweed!substring-string!extended}%
\index{$\rH^\ssubX_{a,b}$: extended string-substring seaweed matrix}%
\index{$\rH^\subsX_{a,b}$: extended substring-string seaweed matrix}%
Given strings $a$, $b$, we define the \emph{extended string-substring} 
(respectively, \emph{substring-string}) \emph{seaweed matrix} 
over $\ang{-m:n \mid 0:m+n}$ as
\begin{gather*}
\rP^\ssubX_{a,b} = \bmat{\rP^\sufpre_{a,b} & \cdot \\ \rP^\ssub_{a,b} & \rP^\presuf_{a,b}}
\qquad
\rP^\subsX_{a,b} = \bmat{\rP^\sufpre_{a,b} & \rP^\subs_{a,b} \\ \cdot & \rP^\presuf_{a,b}}
\end{gather*}
\end{definition}
The extended string-substring seaweed matrix $\rP^\ssubX_{a,b}$
contains at least $n$ and at most $\min(m+n,2n)$ nonzeros.
Note that for $m \geq n$, the number of nonzeros in $\rP^\ssubX_{a,b}$ 
is at most $2n$, which is convenient when $m$ is large.
Analogously, for $m \leq n$, the number of nonzeros 
in the extended substring-string matrix $\rP^\subsX_{a,b}$ 
is at most $2m$, which is convenient when $n$ is large.

\index{substring!cross-}%
\index{prefix!cross-}%
\index{suffix!cross-}%
Let string $a$ of length $m$ be a concatenation of two fixed strings:
$a=a'a''$, where $a'$, $a''$ are nonempty strings 
of length $m'$, $m''$ respectively, and $m=m'+m''$.
A substring of the form $a\ang{i':i''}$ 
with $i' \in \bra{0:m'}$, $i'' \in \bra{m':m}$
will be called a \emph{cross-substring}.
A cross-substring that is either a suffix of $a'$, or a prefix of $a''$
(i.e.\ with either $i'=m'$, or $i''=m'$), will be called \emph{degenerate}.
In other words, a cross-substring of $a$ consists 
of a suffix of $a'$ and a prefix of $a''$;
a cross-substring is non-degenerate, if and only if both of these are non-empty.
A cross-substring that is a prefix or a suffix of $a$
will be called a \emph{cross-prefix} and a \emph{cross-suffix}, respectively.
Given string $b$ of length $n$ that is a concatenation of two fixed strings,
$b=b'b''$, cross-substrings of $b$ are defined analogously.
\begin{definition}
\label{def-matrix-seaweed-cross}
\index{matrix!score!cross-semi-local}%
\index{matrix!seaweed!cross-semi-local}%
\index{$\rH_{a',a'';b}$, $\rH_{a;b',b''}$: cross-semi-local score matrix}%
\index{$\rP_{a',a'';b}$, $\rP_{a;b',b''}$: cross-semi-local seaweed matrix}%
Given strings $a=a'a''$ and $b$, 
the corresponding \emph{cross-semi-local score matrix} is the submatrix 
\begin{gather*}
\rH_{a',a'';b} = \rH_{a,b} \bra{-m':n \mid 0:m''+n}
\end{gather*}
Symmetrically, given strings $a$ and $b=b'b''$, 
the corresponding \emph{cross-semi-local score matrix} is the submatrix 
\begin{gather*}
\rH_{a;b',b''} = \rH_{a,b}\bra{-m:n' \mid n':m+n}
\end{gather*}
The \emph{cross-semi-local seaweed matrices} are defined analogously:
\begin{gather*}
\rP_{a',a'';b} = \rP_{a,b} \ang{-m':n \mid 0:m''+n}\\
\rP_{a;b',b''} = \rP_{a,b}\ang{-m:n' \mid n':m+n}
\end{gather*}
\end{definition}

A cross-semi-local score matrix represents the solution 
of a restricted version of the semi-local LCS problem.
In this version, instead of all substrings (prefixes, suffixes) 
of string $a$ (respectively, $b$),
we only consider cross-substrings (cross-prefixes, cross-suffixes).
At the submatrix boundaries $\rH_{a',a'';b}(*,m''+n)$ and $\rH_{a',a'';b}(-m',*)$,
cross-substrings of string $a$ degenerate to suffixes of $a'$ and prefixes of $a''$;
in particular, cross-prefixes and cross-suffixes of $a$
degenerate respectively to the whole $a'$ and $a''$.
The submatrix boundaries $\rH_{a;b',b''}(*,n')$ and $\rH_{a;b',b''}(n',*)$
correspond to similarly degenerate cross-substrings of string $b$.
Submatrix elements on these boundaries will also be called \emph{degenerate}.

Occasionally, we will use cross-semi-local score and seaweed matrices
in combination with the superscripted submatrix notation,
introduced earlier in this section.
In such cases, the range of the resulting matrix
will be determined by the intersection of the ranges
implied by the superscript and the subscript.
For example, matrix $\rH_{a;b',b''}^\ssub = \rH_{a,b}\bra{0:n' \mid n':n}$
is the \emph{string-cross-substring score matrix},
i.e.\ the matrix of all LCS scores
between string $a$ and all cross-substrings of string $b=b'b''$.

In all the above definitions, the seaweed matrices can be used
to give an implicit representation for the corresponding score matrices,
just as the full seaweed matrix $\rP_{a,b}$ 
does for the full score matrix $\rH_{a,b}$.

%%=-=-=-=-=-=-=-=-=-=-=-=-=-=-=-=-=-=-=-=-=-=-=-=-=-=-=-=-=-=-=-=-=-=-=-=-=-=%%
\mysection{Alignment dag composition}
\label{s-adag-composition}

We now combine all the concepts and techniques introduced in previous sections
to provide a mechanism for string comparison in a divide-and-conquer framework.

% Let $a'$, $a''$, $b'$, $b''$ be nonempty strings 
% of length $m'$, $m''$, $n'$, $n''$ respectively.
% We will consider the comparison 
% of a concatenation string $a=a'a''$ of length $m=m'+m''$ 
% against a fixed string $b$ of length $n$;
% symmetrically, the definitions and results will also apply to the comparison 
% of a fixed string $a$ of length $m$ against 
% a concatenation string $b=b'b''$ of length $n=n'+n''$.

Let $a'$, $a''$ be nonempty strings of length $m'$, $m''$ respectively.
We consider the semi-local LCS problem 
on a concatenation string $a=a'a''$ of length $m=m'+m''$ 
against a fixed string $b$ of length $n$.

\index{dag!alignment!composite}%
The alignment dag $\rG_{a,b}$ consists of 
alignment subdags $\rG_{a',b}$, $\rG_{a'',b}$,
sharing a horizontal row of $n$ nodes and $n-1$ edges,
which is simultaneously 
the bottom row of $\rG_{a',b}$ and the top row of $\rG_{a'',b}$.
We will say that dag $\rG_{a,b}$ is the \emph{composite}
of dags $\rG_{a',b}$ and $\rG_{a'',b}$.

\newcommand{\showdown}[1]{\genfrac{}{}{0pt}{}{#1}{\blacktriangledown}}
\newcommand{\Showdown}[1]{\genfrac{}{}{0pt}{}{#1}{\triangledown}}
\newcommand{\showup}[1]{\genfrac{}{}{0pt}{}{\blacktriangle}{#1}}
\newcommand{\Showup}[1]{\genfrac{}{}{0pt}{}{\vartriangle}{#1}}

\begin{figure}[tb]
\centering
\subfloat[\label{f-comp-pab}Seaweed braids for matrices $\rP_{a',b}$, $\rP_{a'',b}$ 
  and block decomposition]{%
\beginpgfgraphicnamed{f-comp-pab}%
\begin{tikzpicture}[x=0.5cm,y=-0.5cm]
\tikzstyle{every circle node}=[inner sep=0pt,minimum size=0pt,fill]
\path (0,0) coordinate(p);

\draw (p) ++(-0.5,0)
  +( 0,0) node[above,fill=white]{$\showdown{-m}$}
  +( 3,0) node[above,fill=white]{$\showdown{-m'}$}
  +( 7,0) node[above,fill=white]{$\showdown{0}$} 
  +(16,0) node[above,fill=white]{$\showdown{n}$};

\domain{I}{0}{15}
\path (p) ++(4,4) coordinate(p);
\domain{J}{0}{15}
\path (p) ++(3,3) coordinate(p);
\domain{K}{0}{15}

\draw (p) ++(-0.5,0)
  +( 0,0) node[below]{$\showup{0}$}
  +( 9,0) node[below]{$\showup{n}$}
  +(12,0) node[below]{$\showup{m''+n}$}
  +(16,0) node[below]{$\showup{m+n}$};

\draw[very thick,dotted,green]
  (I0) -- (J0)
  (I1) -- (J1)
  (I2) -- (J2)
  (J12) -- (K12)
  (J13) -- (K13)
  (J14) -- (K14)
  (J15) -- (K15);
  
\draw[very thick,rounded corners,green]
  (I3) -- ++(3.5,3.5) ++(0,0) -- ++(3,0) -- (J6)
  (I4) -- ++(2.5,2.5) ++(0,0) -- (J4)
  (I5) -- ++(1.5,1.5) ++(0,0) -- ++(2,0) -- ++(2,2) -- ++(5,0) ++(0,0) -- (J12)
  (I6) -- ++(0.5,0.5) ++(0,0) -- ++(0.5,0.5) -- (J3)
  (I7) -- ++(1,1) -- ++(0,2) -- (J5)
  (I8) -- ++(3,3) -- (J7)
  (I9) -- ++(2.5,2.5) -- ++(4,0) ++(0,0) -- (J13)
  (I10) -- ++(2,2) -- (J8)
  (I11) -- ++(2,2) -- (J9)
  (I12) -- ++(3,3) -- (J11)
  (I13) -- ++(1,1) -- (J10)
  (I14) -- ++(1.5,1.5) ++(0,0) -- (J14)
  (I15) -- ++(0.5,0.5) ++(0,0) -- (J15)
  (J0) -- ++(2.5,2.5) ++(0,0) -- ++(1,0) -- (K1)
  (J1) -- ++(1.5,1.5) ++(0,0) -- ++(0.5,0.5) -- (K0)
  (J2) -- ++(0.5,0.5) ++(0,0) -- ++(2,2) -- ++(4,0) -- (K6)
  (J3) -- ++(2,2) -- (K2)
  (J4) -- ++(1.5,1.5) -- ++(1,0) -- (K5)
  (J5) -- ++(1,1) -- (K3)
  (J6) -- ++(1,1) -- (K4)
  (J7) -- ++(0.5,0.5) -- ++(4,0) ++(0,0) -- (K11)
  (J8) -- ++(0,1) -- (K7)
  (J9) -- ++(0,1) -- (K8)
  (J10) -- ++(1.5,1.5) ++(0,0) -- (K10)
  (J11) -- ++(0,2) -- ++(0.5,0.5) ++(0,0) -- (K9);
\draw[very thick] 
  (6.5,0) rectangle +(9,4)
  (6.5,4) rectangle +(9,3);
\draw[dotted]
  (6.5,4) ++(3,0) -- +(0,3) ++(3,0) -- +(0,3);
\draw[draw opacity=0]
  (-0.5,0) -- ++(4,4) node[midway,auto=right]{\small $\rP_{a',b}$} --
    ++(3,3) node[midway,auto=right]{\small $\rP_{a'',b}$};
\end{tikzpicture}
\endpgfgraphicnamed}

\subfloat[\label{f-comp-pc}Seaweed braid for output matrix $\rP_{a,b}$]{%
\beginpgfgraphicnamed{f-comp-pc}%
\begin{tikzpicture}[x=0.5cm,y=-0.5cm]
\tikzstyle{every circle node}=[inner sep=0pt,minimum size=0pt,fill]
\path (0,0) coordinate(p);

\draw (p) ++(-0.5,0)
  +( 0,0) node[above,fill=white]{$\showdown{-m}$}
  +( 3,0) node[above,fill=white]{$\showdown{-m'}$}
  +( 7,0) node[above,fill=white]{$\showdown{0}$} 
  +(16,0) node[above,fill=white]{$\showdown{n}$};

\domain{I}{0}{15}
\path (p) ++(4,4) coordinate(p);
\domain{J}{0}{15}
\path (p) ++(3,3) coordinate(p);
\domain{K}{0}{15}

\draw (p) ++(-0.5,0)
  +( 0,0) node[below]{$\showup{0}$}
  +( 9,0) node[below]{$\showup{n}$}
  +(12,0) node[below]{$\showup{m''+n}$}
  +(16,0) node[below]{$\showup{m+n}$};

\draw[very thick,green]
  (I0) -- (J0)
  (I1) -- (J1)
  (I2) -- (J2)
  (J12) -- (K12)
  (J13) -- (K13)
  (J14) -- (K14)
  (J15) -- (K15);
  
\draw[very thick,rounded corners,green]
  (I3) -- ++(3.5,3.5) ++(0,0) -- ++(3,0) -- (J6)
  (I4) -- ++(2.5,2.5) ++(0,0) -- (J4)
  (I5) -- ++(1.5,1.5) ++(0,0) -- ++(2,0) -- ++(2,2) -- ++(5,0) ++(0,0) -- (J12)
  (I6) -- ++(0.5,0.5) ++(0,0) -- ++(0.5,0.5) -- (J3)
  (I7) -- ++(1,1) -- ++(0,2) -- (J5)
  (I8) -- ++(3,3) -- (J7)
  (I9) -- ++(2.5,2.5) -- ++(4,0) ++(0,0) -- (J13)
  (I10) -- ++(2,2) -- (J8)
  (I11) -- ++(2,2) -- (J9)
  (I12) -- ++(3,3) -- (J11)
  (I13) -- ++(1,1) -- (J10)
  (I14) -- ++(1.5,1.5) ++(0,0) -- (J14)
  (I15) -- ++(0.5,0.5) ++(0,0) -- (J15)
  (J0) -- ++(2.5,2.5) ++(0,0) -- ++(1,0) -- (K1)
  (J1) -- ++(1.5,1.5) ++(0,0) -- ++(0.5,0.5) -- (K0)
  (J2) -- ++(0.5,0.5) ++(0,0) -- ++(2,2) -- ++(4,0) -- (K6)
  (J3) -- ++(2,2) -- (K2)
  (J4) -- (K4)
  (J5) -- ++(1,1) -- (K3)
  (J6) -- ++(2,2) -- (K5)
  (J7) -- ++(0.5,0.5) -- ++(4,0) ++(0,0) -- (K11)
  (J8) -- ++(0,1) -- (K7)
  (J9) -- ++(0,1) -- (K8)
  (J10) -- ++(0,1) -- ++(1.5,1.5) ++(0,0) -- (K9)
  (J11) -- ++(0,1) -- ++(0.5,0.5) ++(0,0) -- (K10);
\draw[very thick]
  (6.5,0) rectangle +(9,7);
\draw[dotted]
  (6.5,4) ++(3,0) -- +(0,3) ++(3,0) -- +(0,3);
\draw[draw opacity=0]
  (-0.5,0) -- ++(7,7) node[midway,auto=right]{\small $\rP_{a,b}$};
\end{tikzpicture}
\endpgfgraphicnamed}
\caption{\label{f-comp}Semi-local seaweed matrix multiplication:
$\rP_{a',b} \boxdot \rP_{a'',b} = \rP_{a,b}$}
\end{figure}

\index{matrix!seaweed!composite}%
Our goal is, given the seaweed matrices $\rP_{a',b}$, $\rP_{a'',b}$, 
to compute the \emph{composite} seaweed matrix $\rP_{a,b}$.
It turns out that this problem can be solved efficiently
by the techniques of \chapref{c-mmult}.

\begin{theorem}
\label{th-comp-mmult}
The composite string-substring (respectively, semi-local) seaweed matrix 
can be obtained from the respective original seaweed matrices as
\begin{gather}
\label{eq-comp-ssub}
\rP^\ssub_{a,b} = \rP^\ssub_{a',b} \boxdot \rP^\ssub_{a'',b}\\
\label{eq-comp-semi}
\rP_{a,b} = 
\bmat{\Id_{m'} & \cdot \\ \cdot & \rP_{a',b}}
\boxdot
\bmat{\rP_{a'',b} & \cdot \\ \cdot & \Id_{m''}}
\end{gather}
Here, the dimensions of the offset identity submatrices
are chosen to conform to the dimensions of the product.
\end{theorem}
\begin{proof}
By \thref{th-ps}, we have
\begin{gather*}
\rH_{a,b}(i,k) = k-i - \rP_{a,b}^\Sigma(i,k)
\end{gather*}
for all $i \in \bra{-m:n}$, $k \in \bra{0:m+n}$.
Three cases are possible, 
based on the partitioning of the index ranges.

\begin{trivlist}

\setlabelit
\item[Case $i \in \bra{-m':n}$, $k \in \bra{0:m''+n}$.]
By \defref{def-matrix} and \thref{th-ps}, we have
\begin{gather*}
\rH_{a,b}(i,k) =
\max_{j \in \bra{0:n}} \bigpa{\rH_{a',b}(i,j) + \rH_{a'',b}(j,k)} = {}\\
\qquad\max_{j \in \bra{0:n}} 
\bigpa{j-i-\rP_{a',b}^\Sigma(i,j) + k-j-\rP_{a'',b}^\Sigma(j,k)} = {}\\
\qquad k-i - \min_{j \in \bra{0:n}} 
\bigpa{\rP_{a',b}^\Sigma(i,j) + \rP_{a'',b}^\Sigma(j,k)}
\end{gather*}
Therefore, 
\begin{gather*}
\rP_{a,b}^\Sigma(i,k) =
\min_{j \in \bra{0:n}} \bigpa{\rP_{a',b}^\Sigma(i,j) + \rP_{a'',b}^\Sigma(j,k)} =
\bigpa{\rP_{a',b}^\Sigma \odot \rP_{a'',b}^\Sigma}(i,k)
\end{gather*}
In particular, this holds for $i,k \in \bra{0:n}$.
Hence, we have \eqref{eq-comp-ssub}:
\begin{gather*}
\rP^\ssub_{a,b} = \rP^\ssub_{a',b} \boxdot \rP^\ssub_{a'',b}
\end{gather*}
\item[Case $i \in \bra{-m:-m'}$, $k \in \bra{0:m+n}$.]
We have
\begin{gather*}
\rH_{a,b}(i,k) = m' + \rH_{a'',b}(i+m',k) ={}\\
\qquad m' + k - (i+m') - \rP_{a'',b}^\Sigma(i+m',k) ={}\\
\qquad k - i - \rP_{a'',b}^\Sigma(i+m',k)
\end{gather*}
Therefore, 
\begin{gather*}
\rP_{a,b}^\Sigma(i,k) = \rP_{a'',b}^\Sigma(i+m',k) = 
(\Id_{m'} \boxdot \rP_{a'',b})^\Sigma(i,k)
\end{gather*}
\item[Case $i \in \bra{-m:n}$, $k \in \bra{m''+n:m+n}$.]
Symmetrically to the previous case, we have
\begin{gather*}
\rP_{a,b}^\Sigma(i,k) = 
(\rP_{a',b} \boxdot \Id_{m''})^\Sigma(i,k)
\end{gather*}
\end{trivlist}
Summarising the above three cases, we have the proof of \eqref{eq-comp-semi}.
\end{proof}

\begin{example}
\figref{f-comp} shows an instance of seaweed matrix composition,
as obtained by \thref{th-comp-mmult}
(using an arbitrary layout for individual seaweeds).
\sfigref{f-comp-pab} shows the input matrices $\rP_{a',b}$, $\rP_{a'',b}$.
Additionally, it shows the auxiliary matrices $\Id_{m'}$, $\Id_{m''}$
by dotted seaweeds.
\sfigref{f-comp-pc} shows the output matrix $\rP_{a,b}$.
\end{example}

\begin{theorem}
\label{th-comp-mmult-comp}
Given a pair of seaweed matrices,
the corresponding composite seaweed matrix
can be computed in the following time:
\begin{center}
\begin{tabular}{|l|l|l|l|}
\hline
Type & Input & Output & Time \\ \hline
string-substring & $\rP^\ssub_{a',b}$, $\rP^\ssub_{a'',b}$ & 
$\rP^\ssub_{a,b}$ & $O\bigpa{n \log N}$\\ \hline
extended string-substring & $\rP^\ssubX_{a',b}$, $\rP^\ssubX_{a'',b}$ &
$\rP^\ssubX_{a,b}$ & $O\bigpa{n \log N}$\\
cross-semi-local & & $\rP_{a',a'';b}$ & \\ \hline
semi-local & $\rP_{a',b}$, $\rP_{a'',b}$ & 
$\rP_{a,b}$ & $O\bigpa{m + n \log N}$ \\ \hline
\end{tabular}
\end{center}
Here, $N = \min(m',m'',n)$.
\end{theorem}
\begin{proof}
We prove the statement in the final row of the table; 
the proof for other rows is analogous.

We will represent the product \eqref{eq-comp-semi}
by a product of seaweed braids.
To achieve such a representation, we need to extend 
the notion of seaweed braid multiplication to pairs of braids 
whose index ranges are not identical, but have a partial overlap.
Consider two seaweed braids of size $n'$, $n''$ respectively,
whose index ranges overlap over a range of width $n \leq \min(n',n'')$.
The product of such seaweed braids can be defined naturally as follows.
We join up the $n$ seaweeds across the overlap,
and extend each of the remaining $n'+n''-2n$ seaweeds
either upwards or downwards by a straight vertical line,
forming the resulting product braid of width $n'+n''-2n$.
By \thref{th-mmult}, such a staggered seaweed product 
can be computed in time $O(n'+n''+n \log n)$.

In the setting of the current theorem,
the semi-local seaweed matrices $\rP_{a',b}$, $\rP_{a'',b}$
correspond each to a seaweed braid of width $m'+n$, $m''+n$ respectively.
The two braids overlap along the common boundary of the two dags;
the width of the overlap is $n$,
so the width of the (generally non-reduced) product seaweed braid is $m+n$.

Suppose $N=n$. 
Then, as discussed above, 
the seaweed product of width $n$ can be computed directly
in time $O(m'+m''+n \log n) = O(m'+m''+n \log n^*)$.

We may now assume without loss of generality that $N=m''$. 
Assume, also without loss of generality, that $\frac{n}{m''} \geq 1$ is an integer.
We partition the dag $\rG_{a'',b}$
into $\frac{n}{m''}$ square blocks of size $m'' \times m''$.
Consider the seaweed braid of width $m''+n$, corresponding to matrix $\rP_{a'',b}$.
It is now straightforward to decompose this seaweed braid
into a staggered product of $\frac{n}{m''}$ seaweed braids of width $2m''$,
so that each pair of successive braids in this product
overlap over an interval of width $m''$.
Given the matrix $\rP_{a'',b}$, 
the seaweed matrices for all the braids in the product
can easily be obtained in time $\frac{n}{m''} \cdot O(m'') = O(n)$.

We now compute the product \eqref{eq-comp-semi} by first performing 
$\frac{n}{m''}$ successive multiplications
of the seaweed braid corresponding to $\rP_{a',b}$
by each of the braids in the decomposition of $\rP_{a'',b}$.
The resulting running time is
$O\bigpa{m + \frac{n}{m''} \cdot m'' \log m''} = O(m + n \log m'')$.
\end{proof}

\begin{example}
\figref{f-comp} illustrates the proof 
of \thref{th-comp-mmult-comp} as follows.
Decomposition of the seaweed braid corresponding to matrix $\rP_{a'',b}$ 
into blocks is shown in \sfigref{f-comp-pab} by thin dotted lines.
Clearly, the whole seaweed braid for $\rP_{a,b}$
can be obtained from $\rP_{a',b}$ by successive multiplication 
with staggered block subbraids of $\rP_{a'',b}$, block-by-block.
\sfigref{f-comp-pc} shows the resulting seaweed braid 
for the product matrix $\rP_{a,b}$.
\end{example}

\begin{comment}
Observe that if only a single row 
of the highest-score composition is required,
this can be easily computed by matrix-vector distance multiplication.
An equivalent procedure is given
(using different terminology and notation)
in \cite{Landau_Ziv-Ukelson:01,Crochemore+:04,Kent+:06},
based on techniques from \cite{Kannan_Myers:96,Benson:95}.
%
\begin{theorem}
\label{th-score-mvmult}
%
Given the core elements of row $i$ in matrix $\rH_{a',b}$, 
and the core nonzeros of matrix $\rP_{a'',b}$,
it is possible to compute the core elements of row $i$ in matrix $\rH_{a,b}$
in time $O(n \log n)$ and memory $O(n)$.
%
\end{theorem}
%
\begin{proof}
By \thref{th-mvmult}.
%
\end{proof}
\end{comment}

%%===========================================================================%%

%%===========================================================================%%
\mychapter{The seaweed method}
\label{c-seaweed}

In this chapter, we develop a simple and efficient algorithm
for the semi-local LCS problem, called the seaweed combing algorithm.
We then give some of the algorithm's modifications and applications.

This chapter is organised as follows. 
In \secref{s-seaweed}, we describe the main version 
of the seaweed combing algorithm,
and in \secref{s-micro}, its micro-block speedup.
In \secrefs{s-incremental} and \ref{s-blockwise}, 
we apply the seaweed combing algorithm to solving several incremental 
and blockwise versions of the semi-local LCS problem.
In \secref{s-cyclic} we give algorithms 
for the window LCS and cyclic LCS problems,
and in \secref{s-lrs} for the longest repeating subsequence problem.
All our algorithms match, improve on, and/or generalise 
the existing algorithms.

%%=-=-=-=-=-=-=-=-=-=-=-=-=-=-=-=-=-=-=-=-=-=-=-=-=-=-=-=-=-=-=-=-=-=-=-=-=-=%%
\mysection{Seaweed combing}
\label{s-seaweed}

\index{algorithm!dynamic programming}%
\index{algorithm!dynamic programming!traceback}%
\index{problem!prefix-prefix LCS}%
A classical solution to the global LCS problem on input strings $a$ and $b$
is given by the dynamic programming algorithm,
discovered independently by Needleman and Wunsch
\cite{Needleman_Wunsch:70} (without an explicit analysis), 
and by Wagner and Fischer \cite{Wagner_Fischer:74}.
This algorithm solves in fact the more general \emph{prefix-prefix LCS problem,}
which consists in computing the LCS scores
for all prefixes of string $a$ against all prefixes of string $b$.
The algorithm assumes a character comparison model that only allows 
comparison outcomes ``equal'' and ``unequal'',
and the unit-cost RAM computation model;
it runs in time $O(mn)$. 
The solution to the prefix-prefix LCS problem 
can be used to \emph{trace back} (i.e.\ to obtain character by character) 
the actual LCS of strings $a$ and $b$
in time proportional to the size of the output
(i.e.\ the length of the output subsequence).
A memory-saving recomputation technique by Hirschberg \cite{Hirschberg:75}
can be applied to achieve LCS traceback
in the same asymptotic time, but in a linear amount of memory.

\index{problem!prefix-substring LCS}%
The semi-local LCS problem can be solved naively
by computing the LCS score 
for each semi-local substring pair (string-substring, etc.) 
independently by the classical dynamic programming algorithm,
in overall time $O\bigpa{(m+n)^4}$.
By making use of all the prefix-prefix LCS scores
produced by each run of the classical algorithm,
this can be improved to $O\bigpa{(m+n)^3}$.
Based on the ideas of Schmidt \cite{Schmidt:98}, 
Alves et al.\ \cite{Alves+:08} gave an algorithm 
for the string-substring LCS problem.
Their algorithm solves in fact the more general \emph{prefix-substring LCS problem,}
which consists in computing an implicit representation of the LCS scores
for all prefixes of string $a$ against all substrings of string $b$.
The algorithm runs in time $O(mn)$.
It therefore extends the functionality 
of the standard dynamic programming LCS algorithm
(prefix-substring instead of just prefix-prefix),
while matching its asymptotic running time.

We now give a simple and efficient algorithm for the semi-local LCS problem,
which extends still further the functionality of the above algorithms,
while matching their model assumptions and asymptotic running time.
We call it the \emph{seaweed combing algorithm},
since it can be interpreted in terms of combing 
(i.e.\ removing double crossings from) seaweed braids.

\index{algorithm!seaweed combing}%
\begin{algorithm}
\textbf{(Semi-local LCS: Seaweed combing)}%
\label{alg-seaweed}
\setlabelitbf
\nobreakitem[Input:]
strings $a$, $b$ of length $m$, $n$, respectively.
\item[Output:]
nonzeros of semi-local seaweed matrix $\rP_{a,b}$.
\item[Description.]
We obtain the output matrix $\rP_{a,b}$ by constructing  
its corresponding reduced seaweed braid of width $m+n$,
laid out within the cells of the alignment dag $\rG_{a,b}$.
Individual seaweeds within the braid 
will be identified by their starting index 
at the top boundary of the padded alignment dag 
$\rG_{a,\charwild^m \underline{b} \charwild^m}$.

We initialise the seaweed braid 
as a composition of $mn$ elementary seaweed braids, 
each laid over a single cell of $\rG_{a,b}$.
An elementary braid in a match (respectively, mismatch) cell
consists of two non-crossing (respectively, crossing) seaweeds.
The composition of all elementary braids is 
a generally unreduced seaweed braid of width $m+n$,
laid over the full dag $\rG_{a,b}$.
By \thref{th-comp-mmult}, an equivalent reduced seaweed braid
corresponds to the semi-local seaweed matrix $\rP_{a,b}$.

We now need to comb the initial unreduced seaweed braid, 
in order to obtain an equivalent reduced one.
The combing is performed by sweeping the cells of the alignment dag $\rG_{a,b}$
from left to right and from top to bottom,
either in the lexicographic order, or in any other total order
compatible with the $\ll$-dominance partial order of the cells.

For each cell, we consider the two seaweeds passing through it.
The two seaweeds initially cross within the cell,
if and only if it is a mismatch cell.
The new layout for the two seaweeds within the current cell 
is decided as follows.
In a match cell, we always leave the seaweeds uncrossed.
In a mismatch cell, we first determine whether the two crossing seaweeds
have ever crossed previously (in terms of the cells' $\ll$-dominance order) 
in the alignment dag.
This check can be performed efficiently 
by comparing the starting indices of the two current seaweeds.
If the current pair of seaweeds have never crossed previously,
we keep their crossing in the current cell.
However, if the two seaweeds have crossed previously,
we undo (``comb away'') their crossing in the current cell.
We then move on to the next cell in the sweeping order.

From the definition of the seaweed monoid (\defref{def-monoid}),
it is easy to see that each step 
is an equivalence transformation on the current seaweed braid.
Hence, the resulting seaweed braid is equivalent to the initial one.
Furthermore, since our procedure never allows 
a given pair of seaweeds to cross twice,
the resulting seaweed braid is reduced,
and gives us the nonzeros of the output matrix $\rP_{a,b}$.
\item[Cost analysis.]
Within a cell, both the crossing check and the update run in time $O(1)$.
Therefore, the overall running time of the algorithm is $O(mn)$.

Note that we never need to store the full seaweed braid explicitly.
For each seaweed, we only need to store its starting index,
and its index on the current frontier of processed cells.
Therefore, the overall memory cost of the algorithm is $O(m+n)$.
\end{algorithm}

\begin{figure}[tbp]
\centering
\subfloat[\label{f-seaweed-init}Initial state: unreduced seaweed braid]{%
\includegraphics{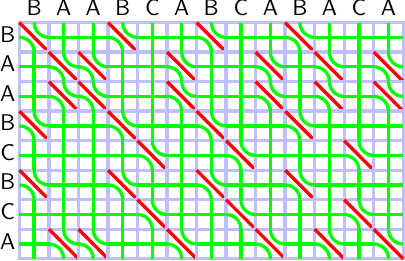}}

\subfloat[\label{f-seaweed-snap}A snapshot of the execution]{%
\includegraphics{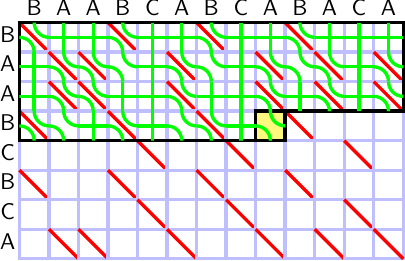}}

\subfloat[\label{f-seaweed-final}Final state: reduced seaweed braid]{%
\includegraphics{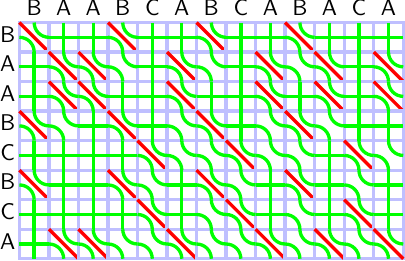}}
\caption{\label{f-seaweed} 
Execution of \algref{alg-seaweed} (seaweed combing)}
\end{figure}
\begin{example}
\figref{f-seaweed} shows the execution of \algref{alg-seaweed}.

\sfigref{f-seaweed-init} shows the initial state of the seaweed braid,
obtained by direct composition of elementary braids 
in all the individual cells. 
The dag $\rG_{a,b}$ is then swept
in the top-to-bottom, left-to-right lexicographic cell order,
which is compatible with $\ll$-dominance of the cells.

\sfigref{f-seaweed-snap} shows a snapshot 
of some intermediate state of the seaweed braid.
The dag area that has already been processed is shown by the dark border;
the current cell is shaded in yellow.
Since the two seaweeds passing through the current cell 
have previously crossed,
their crossing has been undone in the current cell.

\sfigref{f-seaweed-final} shows the final state of the seaweed braid,
identical to the one shown in \figref{f-align-seaweeds}.
It is a reduced seaweed braid corresponding to the output matrix $\rP_{a,b}$.
\end{example}

%%% EXPAND

\index{problem!prefix semi-local LCS}%
Recall that the dag cells in \algref{alg-seaweed}
can be swept in any order compatible with $\ll$-dominance of the cells.
Note that, irrespective of the sweeping order,
the algorithm ends up computing a reduced seaweed braid
(and therefore, implicitly, the semi-local seaweed matrix)
for each prefix of string $a$ against whole string $b$,
and for the whole string $a$ against each prefix of string $b$.
By building an appropriate data structure for querying these implicit matrices,
we obtain a solution to the more general \emph{prefix semi-local LCS problem,}
which incorporates the (ordinary) semi-local LCS problem,
and consists in computing an implicit representation of the LCS scores
for all prefixes of string $a$ against all substrings of string $b$,
and for all substrings of string $a$ against all prefixes of string $b$.

\index{algorithm!seaweed combing!traceback}%
As in the standard dynamic programming LCS algorithm,
the solution to the prefix semi-local LCS problem 
can be used to trace back the actual LCS corresponding to 
any prefix semi-local (i.e.\ prefix-substring or substring-prefix) LCS query,
in time proportional to the size of the output.
A technique similar to the one by Hirschberg \cite{Hirschberg:75}
can also be applied to achieve prefix semi-local LCS traceback
in the same asymptotic time, but in a linear amount of memory.

%Alternatively, cells can be processed 
%in the recursively defined order of \cite{Frigo_Strumpen:05},
%resulting in a cache-efficient version of the algorithm.

Assuming the ``equal/unequal'' character comparison model,
Aho et al.\ \cite{Aho+:76_JACM} gave a lower bound of $\Omega(mn)$
on the solution running time of the (global) LCS problem
(see also a survey by Bergroth et al.\ \cite{Bergroth+:00}).
Both the standard dynamic programming LCS algorithm,
and the seaweed combing algorithm (\algref{alg-seaweed}) match this lower bound,
and are therefore asymptotically optimal in this model.

%%=-=-=-=-=-=-=-=-=-=-=-=-=-=-=-=-=-=-=-=-=-=-=-=-=-=-=-=-=-=-=-=-=-=-=-=-=-=%%
\mysection{The micro-block speedup}
\label{s-micro}

In the previous section, we developed a semi-local LCS algorithm
running in time $O(mn)$,
assuming the character comparison model 
that only allows comparison outcomes ``equal'' and ``unequal''.
We now aim for an asymptotically faster algorithm.
In such a situation, we cannot afford to perform 
all the $mn$ pairwise comparisons of characters from each string.
Therefore, it is natural to switch to a more powerful comparison model,
where the alphabet is a totally ordered set,
and comparison outcomes are ``less than'', ``equal'' and ``greater than''.
The ``missing'' comparisons can be obtained by transitivity.
In this model, we no longer need to process every dag cell individually,
so algorithms with running time $o(mn)$ become possible%
\footnote{This holds true even if 
the computation model assumption is weakened,
so that character comparisons and arithmetic operations 
are charged using the log-cost RAM model.
However, for uniformity we will stick to our original assumption 
of the unit-cost RAM model.}.

A classical LCS computation speedup 
originates from a matrix multiplication method
by Arlazarov et al.\ \cite{Arlazarov+:70},
often nicknamed the ``four Russians method''.
\index{algorithm!dynamic programming!micro-block}%
In this work, we call it the \emph{micro-block speedup,}
adopting the terminology of Bille and G{\o}rtz \cite{Bille_Goertz:09}.
The main idea of this method is to sweep the alignment dag 
in regular micro-blocks of a small, suitably chosen size,
such that running time can be saved  
by precomputing all possible micro-block updates in advance.
Without loss of generality, we assume that $m \leq n$.
By applying the micro-block speedup 
to the classical dynamic programming algorithm,
Masek and Paterson \cite{Masek_Paterson:80}
gave an algorithm for the (global) LCS problem
running in time $O\bigpa{\frac{mn}{\log^2 n} + n}$
for a constant-size alphabet%
\footnote{The original algorithm
by Masek and Paterson \cite{Masek_Paterson:80}
runs in time $O\bigpa{\frac{mn}{\log n} + n}$
for a constant-size alphabet in the log-cost RAM model.
The unit-cost RAM version of the algorithm 
was given in \cite{Wu+:96,Bille_Farach:08}.}.
An alternative approach to subquadratic LCS computation
was developed by Crochemore et al.\ \cite{Crochemore+:03_SIAM}.

An extension of the micro-block subquadratic LCS algorithm 
to an alphabet of unbounded size,
running in time $O\bigpa{\frac{mn (\log\log n)^2}{\log^2 n} + n}$,
was suggested by Paterson and Dan\v c\'\i k \cite{Paterson_Dancik:94},
and fully developed by Bille and Farach-Colton \cite{Bille_Farach:08}.
In this extension, a second, coarser level 
of alignment dag partitioning is introduced.
The blocks of this second level, called \emph{macro-blocks},
are used for reducing the effective alphabet size,
maximising the number of input string characters 
that fit into a machine word for each micro-block update. 

We now give an algorithm for semi-local LCS running in subquadratic time,
which makes a slight improvement on \algref{alg-seaweed}.
The algorithm uses the two-level micro-block method,
similar to the global LCS algorithm of \cite{Bille_Farach:08},
and matches it in running time.
At the same time, our algorithm provides 
a substantially more detailed string comparison.
However, in contrast to the global LCS algorithms,
our algorithm does not allow any further speedup
in the case of a constant alphabet size.

\index{algorithm!seaweed combing!micro-block}%
\begin{algorithm}
\textbf{(Semi-local LCS: Seaweed combing with micro-block speedup)}%
\label{alg-seaweed-micro}%
\setlabelitbf
\nobreakitem[Input:]
strings $a$, $b$ of length $m$, $n$, respectively; we assume $m \leq n$.
\item[Output:]
nonzeros of semi-local seaweed matrix $\rP_{a,b}$.
\item[Description.]
Without loss of generality, 
we may assume that the alphabet size is at most $2n$,
and that the characters are encoded 
by half-integers in the range $\ang{-n:n}$.
We call two strings of equal length \emph{isomorphic},
if one can obtained from the other by a permutation of the alphabet.

We process the alignment dag in square \emph{micro-blocks} of size
\begin{gather*}
\textstyle t=\min\bigpa{\frac{\log n}{16 \cdot \log\log n},m}
\end{gather*}
where the logarithms are base 2.
Similarly to \algref{alg-seaweed},
we start with an initial seaweed braid of width $m+n$,
obtained as a composite of elementary seaweed braids in individual cells.

We then transform the seaweed braid by incremental micro-block combing 
into an equivalent reduced seaweed braid, 
corresponding to the output matrix $\rP_{a,b}$.
The algorithm sweeps the alignment dag $\rG_{a,b}$ in an arbitrary order
compatible with the $\ll$-dominance order of the micro-blocks.
For each micro-block, we perform an update on the current seaweed braid.

Consider a micro-block defined by the input substrings 
$a\ang{l:l+t}$, $b\ang{i:i+t}$,
where $l \in \bra{0:m-h}$, $i \in \bra{0:n-h}$.
Let $i^*=i+m-l$.
A total of $2t$ seaweeds pass through the current micro-block.
A micro-block can be regarded as a function,
parameterised by the current input substrings,
and performing an update on the $2t \times 2t$ permutation matrix $P$,
corresponding to a subbraid of width $2t$ defined by the micro-block.
The states of matrix $P$ before and after the update
will be called the micro-block's 
\emph{input matrix} and \emph{output matrix}, respectively.
Note that both of these are permutation matrices,
and therefore can be represented implicitly by their nonzeros.

The alignment dag can be swept in an arbitrary order 
compatible with the $\ll$-dominance partial order of the micro-blocks.
Combined with precomputation, 
this is already sufficient to obtain a subquadratic algorithm.
However, in order to achieve higher speedup,
we introduce a second level of \emph{macro-blocks} of size 
\begin{gather*}
\textstyle s=\min\bigpa{\frac{\log^2 n}{2},m}
\end{gather*}
We define a macro-block's input and output matrices 
similarly to a micro-block's ones.

The characters of a macro-block's defining input substrings 
are encoded by values in the range $\ang{-n:n}$.
The macro-block's input and output matrices are represented  
by the row and column indices of the nonzeros;
the natural range of these indices is also $\ang{-n:n}$.
In order to perform the computation efficiently,
we remap each of these ranges to a smaller range $\ang{-s:s}$
before passing the values to the micro-block level.
The range remapping preserves the linear order of the values
(both characters and matrix indices).

We process each macro-block as follows.
First, we obtain its defining substrings and the input matrix;
overall, we have $O(s)$ characters and index values for the macro-block.
We then remap both the characters and the index values 
from $\ang{-n:n}$ to $\ang{-s:s}$ 
by removing $2n-2s$ unused values from the range,
while preserving the relative order of the remaining $2s$ values.
This remapping requires sorting of the $O(s)$ values,
and can be performed in time $O(s \log s)$.
We then sweep the current macro-block by micro-blocks,
in an arbitrary order compatible with 
the $\ll$-dominance partial order of the micro-blocks.
For each micro-block, we obtain 
its defining input substrings and the input matrix;
overall, we have $O(t)$ values for the micro-block of size $t$.
We then apply a precomputed update
to the micro-block's $2t \times 2t$ input matrix $P$.

The micro-block's defining substrings and the input matrix
consist each of $2t$ values, ranging over $\ang{-s:s}$.
For each of the at most $(2s)^{2t+2t}=(2s)^{4t}$ possible combinations 
of the input character and index values, 
the output index values resulting from the update
are precomputed in advance, using \algref{alg-seaweed}.

The algorithm maintains the same invariant as \algref{alg-seaweed}:
the current seaweed braid is equivalent to the initial one,
and corresponds to the semi-local LCS problem on the dag $\rG_{a,b}$.
Therefore, at the end of the sweep,
we obtain a reduced seaweed braid corresponding to the output matrix $\rP_{a,b}$.
\item[Cost analysis.]
In the precomputation stage, there are at most $(2s)^{4t}$ problem instances,
each of which runs in time $O(t^2)$.
Therefore, the running time of the precomputation is
\begin{gather*}
O\bigpa{(2s)^{4t} \cdot t^2} =
O\bigpa{(\log^2 n)^{\frac{\log n}{4 \cdot \log\log n}} \cdot t^2} =
O\bigpa{2^{\log(\log^2 n) \cdot \frac{\log n}{4 \log \log n}} \cdot t^2} = {}\\
O\bigpa{2^{2 \log\log n \cdot \frac{\log n}{4 \log \log n}} \cdot t^2} =
O\bigpa{2^{\frac{\log n}{2}} \cdot t^2} = o(n)
\end{gather*}
which is negligible, compared to the subsequent main computation stage.

In the main computation stage,
there are $\frac{mn}{t^2}$ micro-block update steps.
The micro-block's defining substrings and input-output matrices
are each represented by $2t$ values in the range $\ang{-s:s}$.
The full micro-block data are of size
\begin{gather*}
\textstyle
O\bigpa{2t \cdot \log (2s)} = 
O\bigpa{\frac{\log n}{\log\log n} \cdot \log(\log^2 n)} = 
O(\log n)
\end{gather*}
and hence fit into a constant number of machine words.
Therefore, the total running time of the algorithm is
\begin{gather*}
\textstyle
\frac{mn}{t^2} \cdot O(1) = 
O\bigpa{\frac{mn (\log\log n)^2}{\log^2 n} + n}
\end{gather*}
\end{algorithm}

\begin{figure}[tb]
\centering
\includegraphics{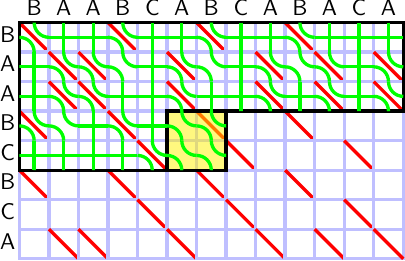}
\caption{\label{f-seaweed-micro} 
A snapshot of \algref{alg-seaweed-micro} (micro-block seaweed combing)}
\end{figure}
\begin{example}
\figref{f-seaweed-micro} shows a snapshot of \algref{alg-seaweed-micro},
using the same conventions as \figref{f-seaweed}.
For simplicity, the macro-blocks are not shown,
and the micro-blocks are assumed to be of size $2$.
As in \algref{alg-seaweed}, the final layout of the seaweed braid 
is identical to the one given in \figref{f-align-seaweeds}.
\end{example}

%%=-=-=-=-=-=-=-=-=-=-=-=-=-=-=-=-=-=-=-=-=-=-=-=-=-=-=-=-=-=-=-=-=-=-=-=-=-=%%
\mysection{Incremental LCS}
\label{s-incremental}

We now begin to explore the applications of the seaweed method. 
First, we consider on-line string comparison.
In this framework, rather than dealing with a pair of fixed input strings,
we consider either one or both strings to be variable.
These variable string(s) are subject to updates;
several different models of string updating may be considered.
The goal of an on-line LCS algorithm is to maintain a data structure
that will store the LCS score for the input strings,
and will allow efficient on-line updates of this score.
We denote by $a$, $b$ denote the current state of each input string,
and by $m$, $n$ their respective current size.

\paragraph{Appending characters.}
The simplest possible model of string updating
is by \emph{appending} a character at the end of an input string.
The classical dynamic programming LCS algorithm 
\cite{Needleman_Wunsch:70,Wagner_Fischer:74}
maintains an array of all prefix-prefix LCS scores.
Such an array can be efficiently updated 
whenever a character is appended (but not prepended)
to either of the input strings.
The running time of this on-line version 
of the classical dynamic programming algorithm 
is $O(n)$ (respectively, $O(m)$)
per update of string $a$ (respectively, $b$).

\index{problem!incremental LCS}%
\paragraph{Appending/prepending characters.}
Another possible model for string updating 
is to allow both appending a character at the end,
and \emph{prepending} a character at the beginning of a string.
This updating model is assumed
by the \emph{incremental LCS problem},
introduced by Landau et al.\ \cite{Landau+:98} 
and by Kim and Park \cite{Kim_Park:04}.
The problem asks to maintain the LCS score
for a fixed string $a$ against a variable string $b$, 
updated on-line by appending and/or prepending characters.
Works \cite{Landau+:98,Kim_Park:04} 
gave algorithms for the incremental LCS problem 
with worst-case time $O(m)$ per update of string $b$.

The more general \emph{fully-incremental LCS problem},
was introduced by Ishida et al.\ \cite{Ishida+:05}.
Here, both strings can be updated on-line 
by appending and/or prepending characters at either end.
Work \cite{Ishida+:05}
gave an algorithm for the fully-incremental LCS problem,
with worst-case time $O(n)$ (respectively, $O(m)$)
per update of string $a$ (respectively, $b$).

A new algorithm for the fully-incremental LCS problem,
matching the above algorithms in running time,
can be obtained by a straightforward generalisation 
of \algref{alg-seaweed} (seaweed combing) as follows.
Our dynamic LCS data structure will consist of the nonzeros
of semi-local seaweed matrix $\rP_{a,b}$.
Whenever a new character is prepended or appended 
to string $a$ (respectively, $b$),
this seaweed matrix is updated by processing 
a new row of cells along the top or bottom boundary
(respectively, a new column of cells along the left or right boundary) 
of the alignment dag $\rG_{a,b}$.
This amounts to the multiplication of seaweed braids
corresponding to alignment dags of size $m \times n$ and $1 \times n$
(respectively, $m \times n$ and $m \times 1$).
A product of such seaweed braids
can be computed in time $O(n)$ (respectively, $O(m)$)
by \thref{th-comp-mmult-comp}.
Having computed the product, it is also straightforward 
to update the value of the LCS in the same asymptotic running time.
Therefore, the overall running time per update 
is $O(n)$ (respectively, $O(m)$)
per update of string $a$ (respectively, $b$).

\begin{comment}
In \chapref{c-network}, we will consider an even more general model
of dynamic string comparison.
In addition to character appending/prepending,
this model allows character deletions at either end of the strings.
\end{comment}

%%=-=-=-=-=-=-=-=-=-=-=-=-=-=-=-=-=-=-=-=-=-=-=-=-=-=-=-=-=-=-=-=-=-=-=-=-=-=%%
\mysection{Blockwise LCS}
\label{s-blockwise}

To demonstrate another natural application of the seaweed method,
let us consider the following generalisation of string comparison.
Suppose that instead of individual characters, 
the input string $a$ is composed of character \emph{blocks},
taken from a pre-specified set of \emph{admissible blocks}
(for example, a list of frequently occurring words in a text).
The set of admissible blocks is known in advance,
and off-line preprocessing of the blocks 
against input string $b$ is allowed.

Block decomposition of strings is very common in string compression.
In this section, we are only dealing with the most general case
of such decomposition;
we will consider specific string compression models in \chapref{c-compressed}.

We denote by $\bar m$ the number of blocks in string $a$,
while keeping the notation $m$ for its total number of characters.
A block may consist of a single character,
therefore such \emph{blockwise string comparison}
generalises ordinary string comparison.
On the other hand, a block may be long;
in such a case, our goal is to process the block quickly,
ideally in the same time as it takes
to process a single character in string $a$.

\index{problem!blockwise LCS}%
\index{problem!blockwise semi-local LCS}%
Given a string of blocks $a$, and a string of ordinary characters $b$,
the \emph{blockwise LCS problem} asks for their LCS score.
Likewise, the \emph{blockwise semi-local LCS problem} 
asks for their implicit semi-local LCS scores,
represented by the seaweed matrix $\rP_{a,b}$.
We can solve both problems naively, by running \algref{alg-seaweed} 
while ignoring the block structure of string $a$, in time $O(mn)$.

\index{problem!common-substring LCS}%
A special case of the blockwise LCS problem
was introduced by Landau and Ziv-Ukelson \cite{Landau_Ziv-Ukelson:01}
as the \emph{common-substring LCS problem} (see also \cite{Crochemore+:04}).
In this version of the problem,
there is only a single non-trivial admissible block of length $l$,
called \emph{common substring}.
The algorithm of \cite{Landau_Ziv-Ukelson:01}
preprocesses the common substring against string $b$ in time $O(nl)$,
and then solves the common-substring LCS problem in time $O(\bar m n)$.

The seaweed method allows us to solve the blockwise LCS and 
the blockwise semi-local LCS problems efficiently as follows.

\paragraph{Preprocessing.}
We preprocess every admissible block $c$ against string $b$
by computing the semi-local seaweed matrix $\rP_{c,b}$.
Let $l$ be the length of the admissible block.
The preprocessing is by \algref{alg-seaweed}, 
and runs in time $O(nl)$ per block;
if $l$ is sufficiently large,
then \algref{alg-seaweed-micro} (the micro-block speedup) can be used.

\paragraph{Blockwise LCS.}
Let us define the score vector $h_{a,b}$ over $\bra{0:n}$ as
$h_{a,b} = \rH^\ssub_{a,b}(0,*)$.
We compute this vector incrementally as follows.
Suppose we have computed the vector $h_{a',b}$,
and let $a = a'c$, where $c$ is an admissible block.
Then we have $h_{a,b} = h_{a',b} \boxdot \rP_{c,b}$,
where matrix $\rP_{c,b}$ has been obtained in the preprocessing phase.
This implicit vector-matrix distance product can be computed
by \thref{th-mvmult} in time $O(n)$.
The overall running time is $O(\bar m n)$.

\paragraph{Blockwise semi-local LCS.}
We compute the seaweed matrix $\rP_{a,b}$ incrementally,
starting from the precomputed matrix $\rP_{d,b}$ 
for an arbitrary block $d$ of string $a$.
Suppose we have computed the matrix $\rP_{a',b}$,
and let $a = a'c$, where $c$ is an admissible block
(the case $a = ca'$ is dealt with analogously).
Then we have $\rP_{a,b} = \rP_{a',b} \boxdot \rP_{c,b}$,
where, as before, matrix $\rP_{c,b}$ has been obtained in the preprocessing phase.
This implicit matrix distance product can be computed
by \thref{th-comp-mmult-comp} in time $O(n \log l + l)$,
where $l$ is the length of block $c$.
The overall running time is $O(\bar m n \log L + m)$,
where $L$ is the maximum size of an admissible block.

\paragraph{Symmetric blockwise (semi-local) LCS.}
Let us consider the symmetric version of the above problems,
where input strings $a$ and $b$ are both composed of admissible blocks,
on which preprocessing is allowed.
In this setting, the preprocessing runs in time $O(l_1 l_2)$
for each pair of blocks of lengths $l_1$, $l_2$.
The blockwise LCS problem can then be solved in time $O(\bar m n + m \bar n)$,
where $\bar m$ and $\bar n$ is the number of blocks 
in string $a$ and $b$, respectively.
The blockwise semi-local LCS problem 
can be solved in time $O\bigpa{(\bar m n + m \bar n)\cdot \log L}$,
where $L$ is the maximum size of an admissible block.

\paragraph{Incremental blockwise (semi-local) LCS.}
We can also consider the incremental versions 
of the blockwise LCS and semi-local LCS problems,
generalising the incremental (semi-local) LCS problem
considered in \secref{s-incremental}.
Both input strings are considered to be variable;
either string can be updated by appending a block on the right,
or prepending a block on the left.
The \emph{block-incremental LCS problem} asks to maintain a data structure
that will store the LCS score for the input strings,
and will allow efficient on-line updates of this score.

The incremental blockwise (semi-local) LCS problem
can be solved by the above algorithms as follows.
The algorithm for the blockwise LCS problem
takes time $O(n+L)$ (respectively, $O(m+L)$)
to process a block appended to string $a$ (respectively, $b$);
here, blocks can only be appended to either string on the right.
The algorithm for the blockwise semi-local LCS problem
takes time $O(n \log L + \bar n L)$ (respectively, $O(m \log L + \bar m L)$) 
to process a block either appended 
or prepended to string $a$ (respectively, $b$);
blocks can be appended (prepended) both on the right or on the left.

%% the \emph{mosaic LCS problem} \cite{Huang+:07}: not applicable

%%=-=-=-=-=-=-=-=-=-=-=-=-=-=-=-=-=-=-=-=-=-=-=-=-=-=-=-=-=-=-=-=-=-=-=-=-=-=%%
\mysection{Window LCS and cyclic LCS}
\label{s-cyclic}

\index{substring!window}%
\index{problem!window LCS}%
Given a fixed parameter $w$, we call a substring of length $w$ 
a \emph{$w$-window} in the corresponding string.
Given strings $a$, $b$, and a window length $w$, the \emph{window LCS problem}
asks for the LCS score of the whole string $a$ 
against every $w$-window in string $b$.

Using the seaweed method, we are now able to give 
a simple efficient algorithm for the window LCS problem.
First, we run \algref{alg-seaweed-micro} (micro-block seaweed combing)
on strings $a$, $b$
obtaining the semi-local seaweed matrix $\rP_{a,b}$.
Then, we perform $m-w+1$ string-substring LCS score queries
for whole $a$ against every $w$-window of $b$.
This can be done efficiently 
as a diagonal batch query via \thref{th-query-inc}.
The overall running time is dominated 
by the call to \algref{alg-seaweed-micro},
which runs in time $O\bigpa{\frac{mn (\log\log n)^2}{\log^2 n} + n}$.

\index{problem!cyclic LCS}%
The following problem can also be solved 
as a special case of the window LCS problem.
Given strings $a$, $b$ of length $m$, $n$ respectively, 
the \emph{cyclic LCS problem} asks for the highest LCS score 
of $a$ against all cyclic shifts of $b$
(equivalently, all cyclic shifts of $a$ against $b$,
or all cyclic shifts of both strings against each other).
Cyclic string comparison has been considered
by Maes \cite{Maes:90}, Bunke and B\"uhler \cite{Bunke_Buehler:93},
Landau et al.\ \cite{Landau+:98}, Schmidt \cite{Schmidt:98},
Marzal and Barrachina \cite{Marzal_Barrachina:00}.
The algorithms of \cite{Landau+:98,Schmidt:98}
solve the cyclic LCS problem in worst-case time $O(mn)$.

We now give a simple algorithm for the cyclic LCS problem,
improving on the existing algorithms by the micro-block speedup.
First, we solve the window LCS problem
on string $a$ against string $bb$ (a concatenation of string $b$ with itself),
which is of length $2n$, with window size $w=n$.
We then take the maximum LCS score across all windows.
The overall running time is still dominated 
by the call to \algref{alg-seaweed-micro},
which runs in time $O\bigpa{\frac{mn (\log\log n)^2}{\log^2 n} + n}$.

%%=-=-=-=-=-=-=-=-=-=-=-=-=-=-=-=-=-=-=-=-=-=-=-=-=-=-=-=-=-=-=-=-=-=-=-=-=-=%%
\mysection{Longest repeating subsequence}
\label{s-lrs}

\index{problem!longest repeating subsequence}%
Given a string $a$ of length $n$,
the \emph{longest repeating subsequence problem}
asks for the length of the longest subsequence of $a$
that is a \emph{square}, i.e.\ a concatenation of two identical strings.

This problem has been considered
under the name ``longest tandem scattered subsequence problem''
by Kosowski \cite{Kosowski:04},
who gave an algorithm running in time $O(n^2)$.

Using the seaweed method, we are now able to give 
a new algorithm for the longest repeating subsequence problem,
improving on the existing algorithm by the micro-block speedup.
First, we run \algref{alg-seaweed-micro} (micro-block seaweed combing)
on string $a$ against itself,
obtaining the semi-local seaweed matrix $\rP_{a,a}$
in time $O\bigpa{\frac{n^2 (\log\log n)^2}{\log^2 n}}$.
Then, we perform $n-1$ prefix-suffix LCS score queries
for every possible non-trivial prefix-suffix decomposition of $a$;
this can be done in time $O(1)$ per query by \thref{th-query-inc}.
Finally, we take the maximum score among all the queries.
The overall running time is $O\bigpa{\frac{n^2 (\log\log n)^2}{\log^2 n}}$.

%%===========================================================================%%

%%===========================================================================%%
\mychapter{Weighted string comparison}
\label{c-weighted}

In this chapter, we generalise our techniques 
developed for the unweighted semi-local LCS problem 
to rational-weighted semi-local string comparison.

This chapter is organised as follows. 
In \secref{s-weighted}, we introduce weighted string alignment,
and describe the blow-up technique for rational-weighted alignment.
In \secref{s-amatch}, we use the framework of weighted alignment 
to obtain algorithms for several versions 
of the approximate pattern matching problem.

%%=-=-=-=-=-=-=-=-=-=-=-=-=-=-=-=-=-=-=-=-=-=-=-=-=-=-=-=-=-=-=-=-=-=-=-=-=-=%%
\mysection{Weighted scores and edit distances}
\label{s-weighted}

\newcommand{\wmatch}{w_{\scriptscriptstyle\mathrm M}}
\newcommand{\wmismatch}{w_{\scriptscriptstyle\mathrm X}}
\newcommand{\wgap}{w_{\scriptscriptstyle\mathrm G}}
\newcommand{\win}{w_{\scriptscriptstyle\mathrm I}}
\newcommand{\wdel}{w_{\scriptscriptstyle\mathrm D}}
\newcommand{\wsub}{w_{\scriptscriptstyle\mathrm S}}

\index{alignment score!weighted}
\index{$\wmatch$: match weight}
\index{$\wmismatch$: mismatch weight}
\index{$\wgap$: gap weight}
The concept of LCS score is generalised 
by that of \emph{(weighted) alignment score} (see e.g.\ \cite{Jackson_Aluru:06}).
An \emph{alignment} of strings $a$, $b$ is obtained 
by putting a subsequence of $a$ into one-to-one correspondence 
with a (not necessarily identical) subsequence of $b$,
character by character and respecting the index order.
The corresponding pair of characters, one from $a$ and the other from $b$,
are said to be \emph{aligned}.
A character that is not aligned against a character of another string 
is said to be aligned against a \emph{gap} in that string.
Each of the resulting character alignments is given a real \emph{weight}:
\begin{itemize}
\item a pair of aligned matching characters has weight $\wmatch \geq 0$;
\item a pair of aligned mismatching characters has weight $\wmismatch < \wmatch$;
\item a gap-character or character-gap pair
has weight $\wgap \leq \tHalf\wmismatch$;
it is normally assumed that $\wgap \leq 0$
(i.e.\ this weight is in fact a penalty).
\end{itemize}
The intuition behind the weight inequalities is as follows:
aligning a matching pair of characters is always better
than aligning a mismatching pair of characters,
which in its turn is never worse than leaving both characters unaligned
(aligned against a gap).
\begin{definition}
\label{def-score}
The \emph{(weighted) alignment score} for strings $a$, $b$ 
is the maximum total weight across all possible alignments of $a$ against $b$.
\end{definition}
\begin{example}
\label{ex-scorew}
\index{alignment score!LCS}%
The LCS alignment score is given by
\begin{gather*}
\wmatch=1 \qquad \wmismatch=\wgap=0
\end{gather*}
\index{alignment score!weighted!for DNA}%
A slightly more sophisticated alignment score, 
intended to penalise gaps in DNA sequence alignment, is given by 
\begin{gather*}
\wmatch=1 \qquad \wmismatch=0 \qquad \wgap=-0.5
\end{gather*}
Another alignment score used for DNA sequence comparison
\cite[Section 1.3]{Chao_Zhang:09} is given by
\begin{gather*}
\wmatch=2 \qquad \wmismatch=-1\qquad \wgap=-1.5
\end{gather*}
\end{example}
\begin{definition}
\index{problem!semi-local alignment score}%
We define the \emph{semi-local (weighted) alignment score problem}
and its component subproblems (string-substring alignment score, etc.)
by straightforward extension of \defref{def-semi-local},
replacing the LCS score by (weighted) alignment scores.
\end{definition}

The concepts of alignment dag and score matrix,
which we introduced in \chapref{c-semi} for unweighted alignments,
can also be naturally extended to the weighted case.
To distinguish between the weighted and unweighted cases,
we will use the script font (e.g.\ $\cG$, $\cH$) in the weighted case,
keeping the ordinary font ($\rG$, $\rH$) for the unweighted case.

\begin{definition}
\index{alignment dag!weighted}%
\index{$\cG_{a,b}$: weighted alignment dag}%
We define the \emph{weighted alignment dag} $\cG_{a,b}$
by straightforward extension of \defref{def-alignment-dag},
where diagonal match edges, diagonal mismatch edges, 
and horizontal/vertical edges
have weight $\wmatch$, $\wmismatch$, $\wgap$, respectively.
\end{definition}

\begin{definition}
\index{matrix!score!weighted semi-local}%
\index{$\cH_{a,b}$: semi-local weighted score matrix}%
We define the \emph{semi-local (weighted) score matrix} $\cH_{a,b}$
by straightforward extension of \defref{def-matrix}.
A semi-local alignment score corresponds to 
a boundary-to-boundary highest-scoring path 
in the $m \times (2m+n)$ padded weighted alignment dag
$\cG_{a,\charwild^m \underline{b} \charwild^m}$.
\end{definition}
Matrix $\cH_{a,b}$ is anti-Monge.
However, in contrast with the unweighted case,
it is not necessarily unit-anti-Monge.

\index{alignment score!weighted!normalised}
Given an arbitrary set of alignment weights,
it is often convenient to normalise them 
so that $0 = \wgap \leq \wmismatch < \wmatch = 1$.
To obtain such a normalisation, first observe that, 
given a pair of strings $a$, $b$, 
and arbitrary weights $\wmatch \geq 0$, 
$\wmismatch < \wmatch$, $\wgap \leq \tHalf\wmismatch$,
we can replace the weights respectively by 
$\wmatch + 2x$, $\wmismatch + 2x$, $\wgap + x$, for any real $x$.
This weight transformation increases the score of every 
global alignment (top-left to bottom-right path in $\rG_{a,b}$) by $(m+n)x$.
Therefore, the relative scores 
of different global alignment paths do not change.
In particular, the maximum global alignment score 
is attained by the same path for all values of $x$.
By taking $x=-\wgap$, 
and dividing the resulting weights by $\wmatch-2\wgap > 0$,
we achieve the desired normalisation.
(A similar method is used by Rice et al.\ \cite{Rice+:97};
see also \cite{Gusfield+:94,Jones_Pevzner:04}.)
\begin{definition}
\label{def-normalised}
\index{alignment score!weighted!normalised}%
Given original weights $\wmatch$, $\wmismatch$, $\wgap$,
the corresponding \emph{normalised weights} are
\begin{gather*}
\wmatch^* = 1\qquad
\wmismatch^* = \frac{\wmismatch-2\wgap}{\wmatch-2\wgap}\qquad
\wgap^* = 0
\end{gather*}
The resulting \emph{normalised score} is
\begin{gather*}
h^* = \frac{h - (m+n)\wgap}{\wmatch - 2\wgap}
\end{gather*}
where $h$ is the original alignment score.
This original score can be restored from the normalised one 
by reversing the normalisation:
$h = h^* \cdot (\wmatch-2\wgap) + (m+n) \cdot \wgap$.
\end{definition}
\begin{example}
In \exref{ex-scorew}, the LCS score is already normalised.
The other two scores correspond to normalised scores 
with weights $\wmatch^* = 1$,
$\wmismatch^* = \frac{0-2\cdot(-0.5)}{1-2\cdot(-0.5)} = 0.5$ 
(respectively, $\wmismatch^* = \frac{-1-2\cdot(-1.5)}{2-2\cdot(-1.5)} = 0.4$),
and $\wgap^* = 0$.
\end{example}

As discussed above, 
maximising the normalised global alignment score $h^*$
is equivalent to maximising the original score $h$,
for fixed string lengths $m$ and $n$.
However, more care is needed when it is required 
to maximise the alignment score across substrings of different lengths,
which is typical for various approximate string matching problems.
In such cases, explicit conversion from normalised weights
to original weights will be necessary prior to the maximisation.

In this work, we will mostly restrict ourselves to character alignment weights
that satisfy the following rationality condition.
\begin{definition}
\label{def-edist-rational}
\index{alignment score!weighted!rational}%
A set of character alignment weights will be called \emph{rational},
if all the weights are rational numbers.
\end{definition}
For a rational set of alignment weights, 
the corresponding normalised weights 
(in particular, the only non-trivial normalised weight $\wmismatch$) 
are also rational.
Given such a rational set of normalised weights,
the semi-local alignment score problem on strings $a$, $b$
can be reduced to the semi-local LCS problem 
by the following \emph{blow-up} procedure.
Let $\wmismatch = \tfrac{\mu}{\nu} < 1$,
where $\mu$, $\nu$ are positive natural numbers.
We transform input strings $a$, $b$ of lengths $m$, $n$
into new \emph{blown-up} strings $\Ta$, $\Tb$ 
of lengths $\Tm = \nu m$, $\Tn = \nu n$.
The transformation consists in replacing
every character $\gamma$ in each of the strings
by a substring $\charguard^\mu \gamma^{\nu-\mu}$ of length $\nu$
(recall that $\charguard$ is a special guard character,
not present in the original strings).
We have
\begin{gather*}
\cH_{a,b}(i,j) = \tfrac{1}{\nu} \cdot \rH_{\Ta,\Tb}(\nu i,\nu j)
\end{gather*}
for all $i \in \bra{-m:n}$, $j \in \bra{0:m+n}$,
where the matrix $\cH_{a,b}$ is defined 
by the normalised scores on the original strings $a$, $b$,
and the matrix $\rH_{\Ta,\Tb}$ 
by the LCS scores on the blown-up strings $\Ta$, $\Tb$.
Using this reduction, we are able to apply
the techniques of semi-local LCS 
to the more general semi-local alignment score problem,
assuming that all weights are rational and that $\nu$ is constant.

\begin{figure}[tbp]
\centering
\subfloat[\label{f-align-weighted}Weighted alignment dag $\cG_{a,b}$]{%
\includegraphics{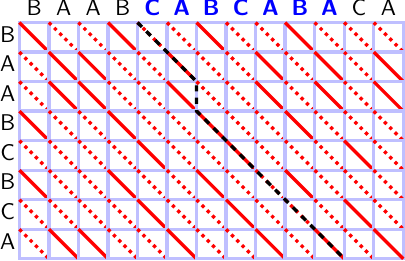}}

\subfloat[\label{f-align-blowup}Alignment dag $\rG_{\Ta,\Tb}$ for the blown-up strings]{%
\includegraphics{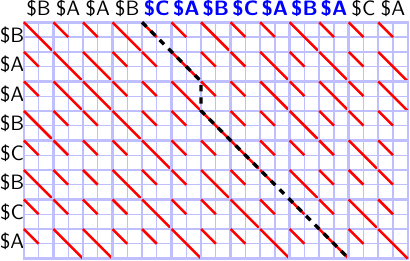}}

\subfloat[\label{f-align-blowup-seaweeds}Alignment dag $\rG_{\Ta,\Tb}$
and the nonzeros of $\rP_{\Ta,\Tb}$ as seaweeds]{%
\makebox[\textwidth]{\includegraphics{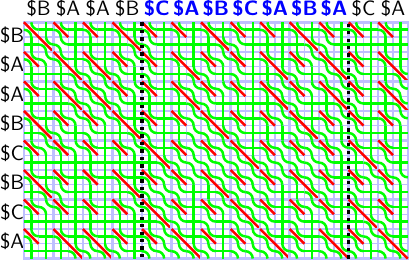}}}
\caption{\label{f-align-w}Semi-local weighted alignment}
\end{figure}

\begin{example}
\figref{f-align-w} shows semi-local weighted alignment of strings $a$, $b$,
using normalised alignment weights $\wmatch=1$, $\wmismatch=0.5$, $\wgap = 0$.

\sfigref{f-align-weighted} shows the alignment dag $\cG_{a,b}$.
Match edges of weight $\wmatch=1$ and mismatch edges of weight $\wmatch=0.5$
are shown respectively by solid and dotted red lines.
The highlighted path of score $5.5$ corresponds
to the string-substring weighted alignment score for string $a$
against substring $b\ang{4:11} = \textsf{``CABCABA''}$.

\sfigref{f-align-blowup} shows the alignment dag $\rG_{\Ta,\Tb}$
for the blown-up strings $\Ta$, $\Tb$. 
We have $\mu=1$, $\nu=2$,
\begin{gather*}
\Ta = \textsf{``\$B\$A\$A\$B\$C\$B\$C\$A''}\\
\Tb = \textsf{``\$B\$A\$A\$B\$C\$A\$B\$C\$A\$B\$A\$C\$A''}
\end{gather*}
The highlighted path has the the same meaning as in \sfigref{f-align-weighted}.

\sfigref{f-align-blowup-seaweeds} shows matrix $\rP_{\Ta,\Tb}$ 
as a reduced seaweed braid of width $(8+13) \cdot 2 = 42$.
Compared to the seaweed braid 
for the original strings $a$, $b$ (\figref{f-align-seaweeds}),
the complexity of the seaweed braid for the blown-up strings
has increased by a factor of $\nu^2 = 4$.
\end{example}

\index{problem!edit distance}%
An important special case of weighted string alignment
is the \emph{edit distance problem}.
Here, the characters are assumed to match ``by default'': $\wmatch = 0$.
The mismatches and gaps are penalised: $2\wgap \leq \wmismatch < 0$.
The resulting score is always nonpositive.
Equivalently, we can regard string $a$ as being transformed into string $b$
by a sequence of \emph{character edits,} each with an associated cost:
\begin{itemize}
\item character insertion or deletion (\emph{indel}) has cost $-\wgap > 0$;
\item character \emph{substitution} has cost $-\wmismatch > 0$.
\end{itemize}
\begin{definition}
\label{def-edist}
The \emph{(weighted) edit distance} between strings $a$, $b$ 
is the minimum total cost of a sequence of character edits
transforming $a$ into $b$.
Equivalently, it is the (nonnegative) absolute value 
of the corresponding (nonpositive) alignment score.
\end{definition}
The edit distance is a \emph{metric:}
it is nonnegative (zero on equal strings and positive otherwise),
symmetric, and satisfies the triangle inequality.
\begin{example}
\index{problem!edit distance!indel (LCS, simple)}%
\index{alignment score!weighted!indel}%
The \emph{indel distance} (also called the \emph{LCS distance} or \emph{simple distance})
\cite{Needleman_Wunsch:70,Apostolico_Guerra:87_Algorithmica,Bergroth+:00} 
has indel cost $1$ and substitution cost $2$,
making any substitution equivalent to an insertion-deletion pair, and thus redundant.
The corresponding \emph{indel alignment score} is given by 
\begin{gather*}
\wmatch=0 \qquad \wmismatch=-2 \qquad \wgap=-1 
\end{gather*}
\index{problem!edit distance!indelsub (Levenshtein)}%
\index{alignment score!weighted!indelsub}%
The \emph{indelsub distance} (also called the \emph{Levenshtein distance})
\cite{Levenshtein:65} 
has both indel cost and substitution cost equal to $1$.
The corresponding \emph{indelsub alignment score} is given by 
\begin{gather*}
\wmatch=0 \qquad \wmismatch=\wgap=-1 
\end{gather*}
\end{example}

\index{problem!edit distance!episode}%
The definition of edit distance can be generalised
by allowing insertions and deletions to have two separate, distinct costs.
For example, the asymmetric \emph{episode distance} \cite{Das+:97} 
corresponds to insertion cost $0$,
and strictly positive deletion and substitution costs.

In the rest of this work, the edit distance problem 
will be treated as a special case of the weighted alignment problem.
In particular, all the techniques of the previous sections
apply to the semi-local edit distance problem,
as long as the character edit costs are rational.

%%=-=-=-=-=-=-=-=-=-=-=-=-=-=-=-=-=-=-=-=-=-=-=-=-=-=-=-=-=-=-=-=-=-=-=-=-=-=%%
\mysection{Approximate pattern matching}
\label{s-amatch}

\index{problem!approximate matching}%
\index{substring!matching}%
Approximate pattern matching is a natural generalisation 
of classical (exact) pattern matching,
allowing for some character differences between the pattern
and a matching substring of the text.
It also fits well in our framework
of semi-local alignment score and edit distance problems.
Given a text string $t$ of length $n$ 
and a pattern string $p$ of length $m \leq n$, 
the \emph{approximate pattern matching problem}
asks for all the substrings of the text 
that are close to the pattern,
i.e.\ that have sufficiently high alignment score
(or, equivalently, sufficiently low edit distance) against the pattern.
Such substrings of the text will be called \emph{matching substrings}.
The precise definition of ``sufficiently high alignment score'' 
(or ``sufficiently low edit distance''),
and therefore of a matching substring,
may vary in different versions of the problem.
Typically, matching substrings will correspond 
to a certain set of maxima (global, local, row, column etc.) 
in the string-substring score matrix $\cH^\ssub_{p,t}$.

The most general form of approximate pattern matching is as follows.
\begin{definition}
\label{def-cmatch}
\index{problem!approximate matching!complete}%
The \emph{complete approximate matching problem}
assumes an alignment score with arbitrary weights.
For every suffix of text $t$, 
the problem asks for a prefix of this suffix 
that has the highest alignment score against pattern $p$.
This corresponds to the set of all row maxima 
in the matrix $\cH^\ssub_{p,t}$.
\end{definition}

\index{algorithm!dynamic programming}%
The complete approximate pattern matching problem can be solved 
by a classical dynamic programming algorithm 
due to Sellers \cite{Sellers:80}.
This algorithm solves in fact the more general problem
of complete approximate matching between every prefix of pattern $p$
against the whole text $t$.
The algorithm runs in time $O(mn)$.
By applying the micro-block speedup to Sellers' algorithm,
the running time can be improved to $O\bigpa{\frac{mn}{\log^2 n} + n}$ 
for a constant-size alphabet,
and to $O\bigpa{\frac{mn (\log\log n)^2}{\log^2 n} + n}$ 
for an unbounded-size alphabet,
assuming a rational set of alignment weights.
Various extensions of the problem have been considered 
by Cormode and Muthukrishnan \cite{Cormode_Muthukrishnan:07} and many others 
(see e.g.\ a survey by Navarro \cite{Navarro:01} and references therein).

\index{$\tau_h(A)$: threshold subset}%
The complete approximate pattern matching problem
asks for the best approximate pattern match at every position of the text.
However, in most cases we are only interested in matches
that are above a certain fixed similarity threshold.
Given a matrix $A$ and a threshold $h$, 
it will be convenient to denote 
the subset of entries above the threshold by
\begin{gather*}
\tau_h(A) = \bigbrc{\text{$(i,j)$, such that $A(i,j) \geq h$}}
\end{gather*}

\begin{definition}
\label{def-tmatch}
\index{problem!approximate matching!threshold}%
The \emph{threshold approximate matching problem}
(often called simply ``approximate matching'')
assumes an alignment score with arbitrary weights.
Given a threshold score $h$,
the problem asks for all substrings of text $t$
that have alignment score at least $h$ against pattern $p$.
This corresponds to all points in the set $\tau_h(\cH^\ssub_{p,t})$.
\end{definition}

\index{problem!approximate matching!output filtering}%
The introduction of a threshold
concentrates the search for matching substrings 
on the areas of high similarity between the pattern and the text.
However, the resulting set of matching substrings may be highly redundant.
In particular, the substrings asked for 
by \defref{def-tmatch} (threshold approximate matching) 
will typically include some highly overlapping ones,
with starting and/or ending positions only differing by a few characters.
The usual convention for eliminating such redundancy 
is to $\emph{filter}$ the output,
by retaining only a specially constrained subset of all the matching substrings.
The filtered output may retain, for instance:
\begin{itemize}
\item only inclusion-minimal matching substrings;
\item only the matching substrings of a fixed length $w$;
\item only the unique starting positions 
(or, symmetrically, the unique ending positions) of matching substrings.
\end{itemize}
In the results that we cite, the filtering is often left implicit,
and can usually be assumed to be of one of the above types.

The most basic version of threshold approximate matching
can be obtained by choosing the alignment score to be the unweighted LCS score,
and setting the threshold to $h=m$ (the pattern length).
A few different filtering methods can be applied 
to cut down on the matching substring redundancy.
\begin{definition}
\label{def-lsubrec}
\index{problem!subsequence recognition!local}%
\index{problem!episode matching}%
\index{problem!subsequence recognition!minimal-window}%
\index{problem!subsequence recognition!fixed-window}%
The \emph{local subsequence recognition problem}
(also known as the \emph{episode matching problem})
asks for all substrings in text $t$ 
containing pattern $p$ as a subsequence.
This corresponds to all points in the set $\tau_m(\rH^\ssub_{p,t})$.
If this set is nonempty 
(i.e.\ if $p$ is contained in the whole $t$ as a subsequence), 
it is the set of all global maxima in the matrix $\rH^\ssub_{p,t}$.
The \emph{minimal-window subsequence recognition problem}
asks for all inclusion-minimal substrings in the text
containing the pattern as a subsequence.
This corresponds to all $\gtrless$-minimal points 
in the set $\tau_m(\rH^\ssub_{p,t})$.
The \emph{fixed-window subsequence recognition problem,}
given a window length $w$,
asks for all substrings of length $w$ of the text
containing the pattern as a subsequence.
This corresponds to all points lying 
on the intersection of the diagonal $j-i=w$ 
with the set $\tau_m(\rH^\ssub_{p,t})$.
\end{definition}

The local subsequence recognition problem
has been considered by Das et al.\ \cite{Das+:97}.
For both the minimal-window and fixed-window versions, 
they give an algorithm running in time $O\bigpa{\frac{mn}{\log^2 n} + n}$
for a constant-size alphabet, which can be modified to 
an algorithm running in time $O\bigpa{\frac{mn (\log\log n)^2}{\log^2 n} + n}$ 
for an unbounded-size alphabet.
A multi-pattern version of the problems has been considered 
by C\'egielski et al.\ \cite{Cegielski+:06_IPL}.

The most well-studied version of threshold approximate pattern matching 
is under the Levenshtein edit distance (equivalently, the indelsub alignment score).
\begin{definition}
\label{def-edmatch}
\index{problem!edit distance matching}%
The problem of \emph{edit distance matching}
is a special case of the approximate matching problem with threshold $h < 0$,
where the alignment score is given by weights 
$\wmatch=0$, $\wmismatch$, $\wgap$, where $2\wgap \leq \wmismatch < 0$.
This problem is usually defined
in terms of the corresponding edit distance threshold $k= -h > 0$,
and only the unique starting positions of matching substrings are required.
\end{definition}
\index{problem!matching with $k$ differences}%
The most well-studied case of the edit distance matching problem 
is Levenshtein distance matching, 
also known as \emph{matching with $k$ differences},
where $\wmismatch = \wgap = -1$.
When the threshold $k$ is low,
the best known algorithm for matching with $k$ differences
is by Cole and Hariharan \cite{Cole_Hariharan:02},
running in time $O\bigpa{\frac{nk^4}{m}+n}$.
For higher values of $k$, the best known algorithm 
is by Landau and Vishkin \cite{Landau_Vishkin:89},
running in time $O(nk)$.

Seaweed combing provides us with a unified algorithm 
for approximate pattern matching,
applicable to any of the problem's versions
described by \defrefs{def-cmatch}--\ref{def-edmatch},
in the case of rational weights.
Our algorithm matches the micro-block version of Sellers' algorithm
in running time for an unbounded-size alphabet.

The algorithm is as follows.
First, we apply the normalisation and blow-up technique of \secref{s-weighted}
to transform the semi-local alignment score problem on strings $p$, $t$ 
into the semi-local LCS problem on blown-up strings $\Tp$, $\Tt$.
Then, we call \algref{alg-seaweed-micro} 
(seaweed combing with micro-block speedup) on strings $\Tp$, $\Tt$
obtaining the semi-local seaweed matrix $\rP_{\Tp,\Tt}$.
By \thref{th-query}, we then build a data structure that allows to query 
any element of the semi-local LCS score matrix $\rH_{\Tp,\Tt}$ in polylogarithmic time.
This data structure provides us with an implicit representation
of the semi-local alignment score matrix $\cH_{p,t}$.
Since both this matrix 
and its string-substring submatrix $\cH^\ssub_{p,t}$ are anti-Monge, 
all the row maxima can now be found efficiently 
by the algorithm of \lmref{lm-rowmin-tmon}.
This solves the complete approximate matching problem.
Both the minimal-window and the fixed-window versions 
of the local subsequence recognition problem,
as well as the approximate matching problem with $k$ differences,
can now be solved by selecting the rows where the maxima
satisfy the additional filtering conditions.

The algorithm's running is independent of the threshold parameter $k$.
In all the described versions of the algorithm, 
the overall running time is dominated 
by the call to \algref{alg-seaweed-micro},
which runs in time $O\bigpa{\frac{mn (\log\log n)^2}{\log^2 n} + n}$.

%%===========================================================================%%

%%===========================================================================%%
\mychapter{Periodic string comparison}
\label{c-periodic}

In this chapter, we introduce the periodic string-substring LCS problem,
and solve it by the seaweed method.

This chapter is organised as follows. 
In \secref{s-periodic}, we define the periodic string-substring LCS problem,
and develop an algorithm for its solution,
called the wraparound seaweed combing algorithm.
In \secref{s-tandem}, we apply this algorithm
to the tandem LCS problem and the tandem cyclic alignment problem,
improving on existing algorithms in running time.

%%=-=-=-=-=-=-=-=-=-=-=-=-=-=-=-=-=-=-=-=-=-=-=-=-=-=-=-=-=-=-=-=-=-=-=-=-=-=%%
\mysection{Wraparound seaweed combing}
\label{s-periodic}

Strings with periodic (or approximately periodic) structure
play an important role 
in both the theory and the applications of string algorithms.
In particular, a variant of periodic string comparison is a key subroutine 
in the fastest known low-threshold approximate pattern matching algorithm
by Cole and Hariharan \cite{Cole_Hariharan:02}.
In computational molecular biology, approximately periodic substrings
of a genome are known as \emph{tandem repeats},
and are crucial for efficient genome analysis
(see e.g.\ Schmidt \cite{Schmidt:98} and references therein).

In this chapter, we give a method of string comparison
that is designed to exploit a periodic structure in the input.
This is made possible by adapting the seaweed method (\chapref{c-seaweed})
to account for input string periodicity.

\index{string!periodic}%
As before, we denote by $a$ a finite input string of length $m$.
Let the input string $b$ be infinite in both directions and \emph{periodic}:
$b = u^{\pm\infty} = \ldots uuuu \ldots$
The \emph{period string} $u$ is finite of length $p$.

\begin{definition}
\label{def-semi-local-periodic}
\index{problem!periodic string-substring LCS}%
Given strings $a$, $u$, the \emph{periodic string-substring LCS problem}
asks for the LCS score of string $a$
against every finite substring of $b = u^{\pm\infty}$.
\end{definition}

Without loss of generality, we assume 
that every character of $a$ occurs in $u$ at least once.
Then, the length of the substring of $b$ in \defref{def-semi-local-periodic}
can be restricted to be at most $mp$
(for any longer substring of $b$, every character of $a$ 
can be matched to a different copy of the period $u$ within the substring,
and therefore the string-substring LCS score will be equal to $m$).

The definition of the alignment dag $\rG_{a,b}$ (\defref{def-alignment-dag})
extends naturally to the periodic string-substring LCS problem.
The alignment dag for this problem is itself infinite and \emph{periodic}:
all vertical and horizontal edges have score $0$,
and each pair of diagonal edges
\begin{gather*}
v_{\hl^-,\hi^-}   \to v_{\hl^+,\hi^+}\qquad
v_{\hl^-,\hi^-+p} \to v_{\hl^+,\hi^++p}
\end{gather*}
have equal scores for all $\hl \in \ang{0:m}$, $\hi \in \ang{-\infty:+\infty}$.
Such an infinite alignment dag can also be regarded as a horizontal composition
of an infinite sequence of \emph{period subdags},
each isomorphic to the $m \times p$ alignment dag $\rG_{a,u}$.

Since string $b$ is infinite, and therefore has no finite prefixes of suffixes, 
the semi-local score and seaweed matrices 
can be understood as just the respective infinite string-substring matrices:
matrix $\rH_{a,b} = \rH^\ssub_{a,b}$ over $\bra{-\infty:+\infty}$,
and matrix $\rP_{a,b} = \rP^\ssub_{a,b}$ over $\ang{-\infty:+\infty}$.
\index{matrix!periodic}%
Furthermore, matrices $\rH_{a,b}$, $\rP_{a,b}$ are again \emph{periodic}:
we have 
\begin{alignat*}{2}
{}
&\rH_{a,b}(i,j) &&= \rH_{a,b}(i+p,j+p)\\
&\rP_{a,b}(\hi,\hj) &&= \rP_{a,b}(\hi+p,\hj+p)
\end{alignat*}
for all $i,j \in \bra{-\infty:\infty}$, $\hi,\hj \in \ang{-\infty:\infty}$.
\index{matrix!periodic!period submatrix}%
To represent such matrices, it is sufficient to store the $p$ nonzeros
of an infinite submatrix of $\rP_{a,b}$: 
either the $p \times \infty$ \emph{row-period submatrix} 
$\rP_{a,b}\ang{0:p \mid *}$,
or, symmetrically, the $\infty \times p$ \emph{column-period submatrix} 
$\rP_{a,b}\ang{* \mid 0:p}$.
The nonzero sets of the two period submatrices 
can be obtained from one another in time $O(p)$.
When working with an infinite periodic seaweed matrix,
we will assume such a representation by default.
For example, accessing a column $\rP_{a,b}(*, \hi)$
will correspond to accessing all columns $\rP_{a,b}(*, \hi + kp)$,
where $k \in \bra{-\infty:+\infty}$.

For the periodic semi-local LCS problem,
the seaweeds only need to be traced within a single period subdag, 
with appropriate wraparound.
Therefore, the problem can be solved 
by the following variant of seaweed combing.

\index{algorithm!seaweed combing!wraparound}%
\begin{algorithm}
\textbf{(Periodic string-substring LCS: Wraparound seaweed combing)}
\label{alg-seaweed-periodic}
\setlabelitbf
\nobreakitem[Input:]
strings $a$, $u$ of length $m$, $p$, respectively; here, $b=u^{\pm\infty}$.
\item[Output:]
nonzeros of semi-local seaweed matrix $\rP_{a,b}$, represented by nonzeros 
of (say) row-period submatrix $\rP_{a,b}\ang{0:p \mid *}$.
\item[Description.]
Similarly to the ordinary seaweed combing algorithm (\algref{alg-seaweed}),
we start with a generally unreduced seaweed braid on the alignment dag $\rG_{a,b}$,
obtained by composition of elementary seaweed braids for individual cells.
We now need to comb this braid in order to obtain
an equivalent reduced seaweed braid.
Note that we only need to maintain $p$ seaweeds,
corresponding to (say) the row-period submatrix $\rP_{a,b}\ang{0:p \mid *}$;
every such seaweed will represent an infinite periodic family of seaweeds.

In contrast with \algref{alg-seaweed},
we are no longer able to sweep the cells of the alignment dag $\rG_{a,b}$
in an order compatible with $\ll$-dominance,
since there are no $\ll$-minimal cells to begin from.
Instead, we sweep the dag in the following special order.
In the outer loop, we run through the rows of cells top-to-bottom.
For each current row $\hat l \in \ang{0:m}$, 
we start the inner loop at an arbitrary match cell,
i.e.\ at an index $\hi_0 \in \ang{0:p}$, 
such that $a(\hat l) = u(\hi)$.
Such an index $\hi_0$ is guaranteed to exist by the assumption 
that every character of $a$ occurs in $u$ at least once.
Then, starting from $\hi=\hi_0$, 
we sweep the cells of the current row left-to-right,
wrapping around from $\hi=p^-$ to $\hi=0^+$,
and continuing the sweep left-to-right up to $\hi=\hi_0$.

For each cell, we consider the two seaweeds passing through it.
The new layout for the two seaweeds within the current cell 
is decided similarly to \algref{alg-seaweed}.
In a match cell, we always leave the seaweeds uncrossed.
In a mismatch cell, we leave the seaweeds crossed,
if and only if this pair of seaweeds have never crossed previously
(in terms of the $\ll$-dominance order) in the alignment dag.
However, if the two seaweeds have crossed previously,
we undo (``comb away'') their crossing in the current cell.
We then move on to the next cell in the sweeping order.

In order to perform the crossing check efficiently, we maintain,
similarly to \algref{alg-seaweed}, 
the starting index of every seaweed considered during the sweep.
This starting index will always be 
at the (infinite) top boundary of the dag $\rG_{a,b}$.
However, since the seaweeds enter and leave 
the current period subdag during the sweep,
a seaweed's starting index may not necessarily belong to the current subdag.
The starting indices of all the considered seaweeds
can still be maintained efficiently as follows.
As the row sweep wraps around, 
a seaweed leaves the current period subdag 
at the right boundary of the cell with $\hi=p^-$.
This seaweed is replaced by a seaweed from the same periodic family
entering the current period subdag 
at the left boundary of the cell with $\hi=0^+$.
The starting index of the new seaweed can be obtained
by subtracting the period length $p$ 
from the starting index of the leaving seaweed.
The crossing check for a pair of seaweeds is performed
by comparing their starting indices, similarly to \algref{alg-seaweed}.

The correctness of the combing procedure is implied,
similarly to \algref{alg-seaweed},
by \thref{th-comp-mmult} and \defref{def-monoid}.
The resulting seaweed braid is reduced, and gives us the nonzeros 
of the column-period submatrix $\rP_{a,b}\ang{* \mid 0:p}$.
By translating the indices of every nonzero by an appropriate multiple of $p$,
it is straightforward to convert this matrix 
to the output row-period submatrix $\rP_{a,b}\ang{0:p \mid *}$.
\item[Cost analysis.]
Similarly to \algref{alg-seaweed},
a seaweed crossing check and a cell update both run in time $O(1)$.
The conversion from column-period to row-period submatrix
runs in time $O(p)$.
Therefore, the overall running time of the algorithm is $O(mp)$.

Like in \algref{alg-seaweed},
for each seaweed, we only need to store its starting index,
and its index on the current frontier of processed cells.
Only the $p$ current seaweeds are represented explicitly,
therefore the overall memory cost of the algorithm is $O(p)$.
\end{algorithm}
\begin{figure}[tb]
\centering
\includegraphics{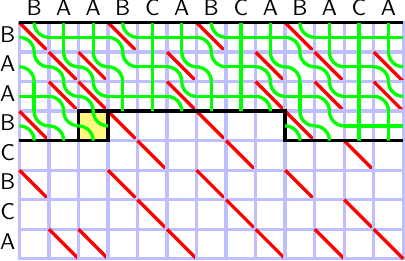}
\caption{\label{f-seaweed-periodic} 
A snapshot of \algref{alg-seaweed-periodic} (wraparound seaweed combing)}
\end{figure}
\begin{example}
\figref{f-seaweed-periodic} shows a snapshot of \algref{alg-seaweed-periodic},
using the same conventions as \figref{f-seaweed}.
The seaweed braid is laid out on a period subdag $\rG_{a,u}$;
each seaweed leaving the subdag on the right is replaced
by a seaweed from the same periodic family entering the subdag
at the corresponding point on the left.
Note that, although the two seaweeds meeting in the highlighted current cell
do not cross in the current period subdag,
they have both arrived from the previous period subdag,
where they did cross.
Therefore, their crossing has been undone in the current cell.
\end{example}

In contrast with \algref{alg-seaweed},
the extra data dependencies caused by the wraparound
and the resulting restrictions of the cell sweeping order
seem to rule out the possibility 
of a micro-block version for \algref{alg-seaweed-periodic}.

%%=-=-=-=-=-=-=-=-=-=-=-=-=-=-=-=-=-=-=-=-=-=-=-=-=-=-=-=-=-=-=-=-=-=-=-=-=-=%%
\begin{comment}

First phase.
Compute semi-local seaweed matrix $\rP_{a,u}$ in $O(mp)$, 
micro-block speedup $O(\frac{mp(\log\log)^2}{\log^2})$.

Submatrix $\rP^\ssub_{a,u}$: subset of $\rP_{a,b}$.

Second phase.
Consider $P\presuf_{a,u}$, $P\sufpre_{a,u}$, $P\subs_{a,u}$.
Compute $P\presuf_{a,u^2}$, $P\sufpre_{a,u^2}$, $P\subs_{a,u^2}$
in $O(m \log m)$.
Now have $\rP_{a,u^2}$.
Continue $\log m$ steps. Time $\log m \cdot O(m \log m) = O(m \log^2 m)$.

For reasonably close $m$, $p$, first phase dominates. Subquadratic.

\end{comment}

%%=-=-=-=-=-=-=-=-=-=-=-=-=-=-=-=-=-=-=-=-=-=-=-=-=-=-=-=-=-=-=-=-=-=-=-=-=-=%%
\mysection{Tandem alignment}
\label{s-tandem}

The periodic LCS problem has many variations that can be solved 
by an application of wraparound seaweed combing 
(\algref{alg-seaweed-periodic}).

\paragraph{Tandem LCS.}
\index{problem!tandem LCS}%
The first such variation is the \emph{tandem LCS problem}.
The problem asks for the LCS score of a string $a$ of length $m$
against a tandem $k$-repeat string $b=u^k$ of length $n=kp$.
As before, we assume that every character of $a$ occurs in $u$ at least once;
we may also assume that $k \leq m$.

The tandem LCS problem can be solved naively 
by considering the LCS problem directly on strings $a$ and $b$, 
ignoring the periodic structure of string $b$.
The standard dynamic programming LCS algorithm 
\cite{Needleman_Wunsch:70,Wagner_Fischer:74}
solves the problem in time $O(mn) = O(mkp)$.
This running time can be slightly improved 
by the micro-block speedup.

The tandem LCS problem can also be regarded 
as a special case of the common-substring LCS problem 
\cite{Landau_Ziv-Ukelson:01,Crochemore+:04} (see \secref{s-blockwise}).
Using this technique, the problem can be solved in time $O\bigpa{m(k+p)}$.
The techniques of Landau et al.\ \cite{Crochemore+:04,Landau+:07}
give an algorithm for the tandem LCS problem,
parameterised by the LCS score of the input strings;
however, the worst-case running time of this algorithm 
is still $O\bigpa{m(k+p)}$.
Landau \cite{Landau:06} asked if the running time for the tandem LCS problem
can be improved to $O\bigpa{m(\log k+p)}$.

We now give an algorithm that improves on the current algorithms
in time and functionality,
and even exceeds Landau's expectation.
First, we call \algref{alg-seaweed-periodic}
on strings $a$ and $u$.
Then, we count the number of nonzeros $\gtrless$-dominated by point $(0,n)$,
i.e.\ nonzeros in the submatrix $\rP_{a,b}\ang{0:+\infty \mid -\infty:n}$.
Given the (say) row-period submatrix $\rP_{a,b}\ang{0:p \mid *}$,
this can be done by a sweep of its $p$ nonzeros,
counting every nonzero with appropriate multiplicity.
More precisely, every nonzero $\rP_{a,b}(\hi,\hj)=1$,
$\hi \in \ang{0:p}$, $\hj \in \ang{-\infty:\infty}$,
is counted with multiplicity $k-\floor{\hj/p}$, if $\hj \in \ang{0:n}$,
and is skipped (counted with multiplicity $0$) otherwise.
The solution to the tandem LCS problem is then obtained by \thref{th-ps}.
The overall running time is dominated
by the call of \algref{alg-seaweed-periodic},
which runs in time $O(mp)$.

\paragraph{Tandem alignment.}
\index{problem!tandem alignment}%
Another set of variations on the periodic LCS problem
was introduced by Benson \cite{Benson:05} as the \emph{tandem alignment problem}.
Instead of asking for all string-substring LCS scores
of $a$ against $b=u^{\pm\infty}$,
the tandem alignment problem asks for a substring of $b$
that is closest to $a$ in terms of alignment score (or edit distance),
under different restrictions on the substring.
In particular:
\begin{itemize}
\item the \emph{pattern global, text global (PGTG) tandem alignment problem}
restricts the substring of $b$
to consist of a whole number of copies of $u$,
i.e.\ to be of the form $u^k = uu \ldots u$ for an arbitrary integer $k$;
\item the \emph{tandem cyclic alignment problem}
restricts the substring of $b$
to be of length $kp$ for an arbitrary integer $k$
(but it may not consist of a whole number of copies of $u$);
\item the \emph{pattern local, text global (PLTG) tandem alignment problem}
leaves the substring of $b$ unrestricted.
\end{itemize}

All these three versions of the tandem alignment problem 
can be regarded as special cases 
of the approximate pattern matching problem (see \secref{s-amatch})
on strings $a$ of length $m$ and $b' = u^m$ of length $n=mp$
(but with the roles of the text and the pattern reversed
with respect to Benson's terminology).
Therefore, the tandem LCS problem can be solved naively 
by considering the approximate pattern matching problem 
directly on strings $a$ and $b'$, 
ignoring the periodic structure of string $b'$.
Given an arbitrary (real) set of alignment weights,
the standard dynamic programming algorithm \cite{Sellers:80}
solves the problem in time $O(mn) = O(m^2 p)$.
For a rational set of weights, the running time can be slightly improved
by the micro-block speedup (see \secref{s-amatch}).

\index{algorithm!dynamic programming!wraparound}%
The PGTG and PLTG tandem alignment problems 
can be solved more efficiently
by the technique of \emph{wraparound dynamic programming}
\cite{Myers_Miller:89,Fischetti+:93} (see also \cite{Benson:05})
in time $O(mp)$.
For the tandem cyclic alignment problem, 
Benson \cite{Benson:05} modified this technique 
to give an algorithm running in time $O(mp \log p)$. % and memory $O(mp)$.

We now give a new algorithm for the tandem cyclic alignment problem,
which improves on the existing algorithm in running time,
assuming a rational set of alignment weights.
The running time of the new algorithm matches the current algorithms 
for the PGTG and PLTG tandem alignment problems.

Given input strings $a$, $u$,
we first solve the periodic string-substring problem
by calling \algref{alg-seaweed-periodic}.
This gives us a period submatrix of matrix $\rP_{a,b}$, where $b=u^{\pm\infty}$.
Then, for each $k$, $0 < k < m$, 
we perform independently the following procedure.
We solve the tandem LCS problem for strings $a$ and $u^k$
by the method described earlier in this section,
counting every nonzero in the period submatrix $\rP_{a,b}$
with an appropriate multiplicity.
This gives us the LCS score for $a$ against $u^k$ for every $k$.
We then update this score incrementally, 
obtaining the LCS score for string $a$ 
against a window of length $p$ in $b$,
sliding through $p$ successive positions.
This is equivalent to querying $p$ successive elements
in a diagonal of matrix $\rP_{a,b}$,
which can be achieved by $2p$ incremental dominance counting queries.
By \thref{th-query-inc},
every one of these queries can be performed in time $O(1)$.

The call to \algref{alg-seaweed-periodic} runs in time $O(mp)$;
its output is shared by the tandem LCS computation for all $k$.
For each $k$, the running time 
of the remaining computation is $O(p)$.
Therefore, the combined running time 
for all values of $k$ is $m \cdot O(p) = O(mp)$.
Overall, the algorithm runs in time $O(mp)$.

\begin{comment}

Periodic substring-substring approximate matching.

In \cite{Wu+:95}: approximate matching 
a regular expression of length $p$ in a text of size $n$.
Time $O(np/\log_{d+2} n)$ on unit-cost RAM,
where $d$ is Levenshtein distance.

Can improve to subquadratic in log-cost PRAM without parameterisation?

\end{comment}

%%===========================================================================%%

%%===========================================================================%%
\mychapter{Permutation string comparison}
\label{c-permutation}

In this chapter, we consider the semi-local comparison of permutation strings.

This chapter is organised as follows. 
In \secref{s-permutation}, 
we introduce the semi-local LCS problem on permutation strings,
and develop an algorithm for its solution.
By direct application of this algorithm, in \secref{s-permutation-cyclic} 
we obtain an improved algorithm for the cyclic LCS problem on permutations.
Further applications include improved algorithms 
for the longest pattern-avoiding subsequence problem,
given in \secref{s-permutation-lxs},
and for the longest $k$-increasing and $k$-modal subsequence problems,
given in \secref{s-permutation-piecewise}.
In \secref{s-circle}, we consider 
the maximum clique problem in a circle graph 
represented by an interval model.
By application our semi-local LCS algorithm on permutations,
we obtain new algorithms for this problem,
both for general and sparse circle graphs,
achieving a substantial improvement on existing algorithms in running time.
In \secref{s-linear}, we describe an application of these algorithms
to the problem of finding exact and approximate 
commonly structured patterns in linear graphs.

\mysection{Semi-local LCS between permutations}
\label{s-permutation}

\index{string!permutation}%
\index{string!permutation!identity}%
\index{$\id$: identity permutation string}%
An important special case of string comparison 
is where each of the input strings $a$, $b$ is a \emph{permutation string},
i.e.\ a string that consists of all distinct characters.
Without loss of generality, we may assume that $m=n$,
and that both strings are permutations 
of a given totally ordered alphabet of size $n$.
For consistency with the notation in previous chapters,
we will assume that a permutation string of length $n$
is indexed by half-integers $\ang{0:n}$,
and is over the alphabet $\ang{0:n}$, unless indicated otherwise.
The \emph{identity permutation} string of length $n$ is the string 
$\id=\pa{0^+,1^+,\ldots,n^-}$.

\index{string!reverse}%
\index{$\bar a$: string reverse}%
Given a string $a$, we denote its reverse string by $\bar a$.
In particular, the \emph{reverse identity permutation} string 
is $\overline{\id}=\pa{n^-,n^--1,\ldots,0^+}$.
\index{$\Sigma(a)$: character set}%
We denote by $\Sigma(a)$ the set of characters 
appearing in $a$ at least once.

The LCS problem on permutation strings is closely related 
to the following classical problem.
\begin{definition}
\label{def-lis}
\index{problem!longest increasing subsequence (LIS)}%
\index{LIS score}%
Given a string $a$, the \emph{longest increasing subsequence (LIS) problem}
asks for the length of the longest string 
that is an increasing subsequence of $a$.
For consistency with our previous terminology,
we will call this length the \emph{LIS score} of string $a$.
\end{definition}
Indeed, the LIS problem is equivalent to the LCS problem
on string $a$ against the identity string $\id$,
and the LCS problem on a pair of permutation strings
can be reduced to the LIS problem by sorting one of the input strings.

The LIS problem has a long history, 
going back to Erd{\"o}s and Szekeres \cite{Erdos_Szekeres:35}
and Robinson \cite{Robinson:38}.
Later, Knuth \cite{Knuth:70}, Fredman \cite{Fredman:75} 
and Dijkstra \cite{Dijkstra:80}
gave algorithms running in time $O(n \log n)$.
The problem was studied further by Chang and Wang \cite{Chang_Wang:92},
Bespamyatnikh and Segal \cite{Bespamyatnikh_Segal:00}.
Crochemore and Porat \cite{Crochemore_Porat:10}
gave an LIS algorithm running in time $O(n \log\log \lambda)$,
where $\lambda$ is the output LIS score.

\index{problem!local LIS}%
We consider the semi-local LCS problem on permutation strings.
Note that its string-substring component
is equivalent to the \emph{local LIS problem}, 
which asks for the LIS score in every substring of a given permutation string.
We now give an efficient algorithm 
for the semi-local LCS problem on permutation strings.

\index{algorithm!seaweed doubling}%
\begin{algorithm}
\textbf{(Semi-local LCS between permutation strings: Seaweed doubling)}%
\label{alg-perm}
\setlabelitbf
\nobreakitem[Input:]
permutation strings $a$, $b$ of length $n$ over an alphabet of size $n$.
\item[Output:]
nonzeros of the semi-local seaweed matrix $\rP_{a,b}$.
\item[Description.]
Recursion on $n$.
\setlabelnormal
\item[Recursion base: $n=1$.] The computation is trivial.
\item[Recursive step: $n>1$.]
Assume without loss of generality that $n$ is even.
We partition the input string $a$
into a concatenation $a=a' a''$ of two strings of length $\Half[n]$.
Each of the strings $a'$, $a''$ 
is a permutation string over an alphabet of size $\Half[n]$.

The semi-local seaweed matrices $\rP_{a',b}$, $\rP_{a'',b}$,
are each over $\bigang{-\Half[n]:n \mid 0:\Half[3n]}$, 
and each contain $\Half[3n]$ nonzeros.
The semi-local seaweed matrix $\rP_{a,b}$
is over $\ang{-n:n \mid 0:2n}$, 
and contains $2n$ nonzeros.

Note that whenever $b(\hi) \not\in \Sigma(a')$
(respectively $b(\hi) \not\in \Sigma(a'')$),
we have $\rP_{a',b}(\hi,\hi)=1$ (respectively, $\rP_{a'',b}(\hi,\hi)=1$),
so the corresponding set of indices $\hi$
defines within $\rP_{a',b}(\hi,\hi)$ (respectively, $\rP_{a'',b}(\hi,\hi)$)
an $\tHalf[n] \times \tHalf[n]$ non-contiguous submatrix 
equal to the identity matrix $\Id$.

Let $b'$ be a $\tHalf[n]$-subsequence of $b$,
obtained by deleting all characters not belonging to $\Sigma(a')$.
We call the algorithm recursively on strings $a'$, $b'$
to compute the semi-local seaweed matrix $\rP_{a',b'}$.
Then, matrix $\rP_{a',b}$ is obtained
by inserting a set of $\tHalf[n]$ rows 
and $\tHalf[n]$ columns into matrix $\rP_{a',b'}$,
which corresponds to reinserting 
the $\tHalf[n]$ deleted characters of string $b$ into string $b'$.
The newly inserted rows and columns 
are filled with (implicit) zeros and ones,
so that they form a $\tHalf[n] \times \tHalf[n]$ 
non-contiguous submatrix equal to the identity matrix $\Id$.

Symmetrically, matrix $\rP_{a'',b}$ is obtained
by deleting from string $b$ all characters not belonging to $\Sigma(a'')$,
calling the algorithm recursively, and then inserting 
a set of $\tHalf[n]$ rows and $\tHalf[n]$ columns into resulting matrix
to form a $\tHalf[n] \times \tHalf[n]$ non-contiguous identity submatrix.

Given matrices $\rP_{a',b}$, $\rP_{a'',b}$, 
the output matrix $\rP_{a,b}$ is computed 
by the algorithm of \thref{th-comp-mmult-comp},
which calls the algorithm of \thref{th-mmult} as a subroutine.
Note that we now have two nested recursions:
the outer recursion of the current algorithm, 
and the inner recursion of \thref{th-mmult}.
\item[(End of recursive step)]
\setlabelitbf
\item[Cost analysis.]
The recursion tree is a balanced binary tree of height $\log n$.
In the root node, the running time 
is dominated by the call to the algorithm of \thref{th-comp-mmult-comp}, 
and is therefore $O(n \log n)$.
In each subsequent level of the recursion tree, 
the number of nodes doubles, 
and the running time per node is reduced by at least a factor of $2$.
Therefore, the running time per level is $O(n \log n)$.
The overall running time is $\log n \cdot O(n \log n) = O(n \log^2 n)$.
\end{algorithm}
\begin{figure}[tb]
\centering
\includegraphics[scale=0.8]{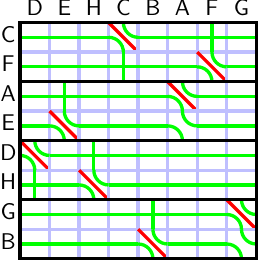}\qquad
\includegraphics[scale=0.8]{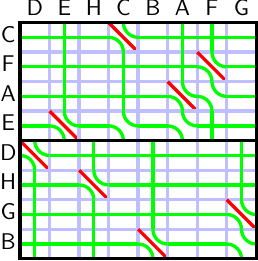}\qquad
\includegraphics[scale=0.8]{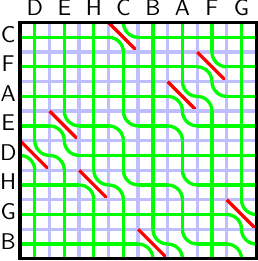}
\caption{\label{f-seaweed-perm} Snapshots of \algref{alg-perm}
(Local LIS)}
\end{figure}
\begin{example}
\figref{f-seaweed-perm} shows a series of snapshots 
of an execution of \algref{alg-perm} on permutation strings
$a = \textsf{``CFAEDHGB''}$, $b = \textsf{``DEHCBAFG''}$.
\end{example}

\index{algorithm!seaweed doubling!traceback}%
By keeping the algorithm's intermediate data, we obtain a data structure
that allows efficient traceback of any semi-local LCS query
on a pair of permutations, in time proportional to the size of the output
(i.e.\ the length of the output subsequence).

%% Sparse in $O(r \log^2 n)$

%%=-=-=-=-=-=-=-=-=-=-=-=-=-=-=-=-=-=-=-=-=-=-=-=-=-=-=-=-=-=-=-=-=-=-=-=-=-=%%
\mysection{Window and cyclic LIS}
\label{s-permutation-cyclic}

\index{substring!window}%
\index{problem!window LIS}%
\index{problem!window LIS!between permutations}%
Recall from \secref{s-cyclic} that, 
given a fixed parameter $w$, we call a substring of length $w$ 
a \emph{$w$-window} in the corresponding string.
The window LCS problem between permutation strings
is equivalent to the \emph{window LIS problem}, 
which asks for the LIS score in every $w$-window of a given permutation string.

% A related problem of tracing back LIS 
% in every substring of a \emph{fixed size} in a permutation

The window LIS problem has been studied by Albert et al.\ \cite{Albert+:04}
and by Chen et al.\ \cite{Chen+:07}.
In particular, work \cite{Chen+:07} gives an algorithm
that reports all window LIS (as opposed to just their lengths)
in time $O(\mathit{output})$.
In the same work, the algorithm is also generalised
for reporting all LIS in an arbitrary subset of $n$ substrings 
of the input permutation string, possibly of different sizes.
Deorowicz \cite{Deorowicz:12_CAI} considered the problem
of finding, for a given $w$, 
the maximum LIS score across all $w$-windows of the input string.
The resulting algorithm runs in time
$O\bigpa{n \log\log n + \min(n \lambda, 
n \bigceil{\frac{\lambda^3}{w}} \cdot \log \bigceil{\frac{w}{\lambda^2} + 1})}$,
where $\lambda$ is the output maximum LIS score.
In all the above versions of the problem, 
the length of each window LCS can be as high as $\Theta(w)$, 
and therefore the algorithms' running time can be as high as $\Theta(nw)$.

Using our techniques, the window LIS problem 
on a permutation string $a$ of size $n$ can be solved as follows.
Consider a subset of $2w$-windows in $a$, 
overlapping over prefixes and suffixes of length $w$:
$a\ang{0:2w}$, $a\ang{w:3w}$, $a\ang{2w:4w}$, \ldots
First, we run \algref{alg-perm} (seaweed doubling)
on each of these $2w$-windows $a'$
against the identity permutation string $\id$,
obtaining the semi-local seaweed matrix $\rP_{a',\id}$.
Then, we perform $2w-1$ string-substring LCS score queries
for every $w$-window of $a'$ against $\id$.
This can be done efficiently 
as a diagonal batch query via \thref{th-query-inc}.
Every $w$-window of $a$ is a substring 
in some string $a'$ among the considered subset of overlapping $2w$-windows,
hence we have obtained the full solution to the window LIS problem.
The overall running time is dominated 
by the calls to \algref{alg-perm},
which run in combined time $O(\frac{2n}{w} \cdot w \log^2 w) = O(n \log^2 w)$.

\index{problem!cyclic LCS!between permutations}%
\index{problem!cyclic LIS}%
The cyclic LCS problem has been defined in \secref{s-cyclic}.
The cyclic LCS problem on a pair of permutation strings 
is equivalent to the \emph{cyclic LIS problem},
which asks for the maximum LIS score
across all cyclic shifts of the input string.

The cyclic LIS problem has been considered by Albert et al.\ \cite{Albert+:07},
who gave a Monte Carlo randomised algorithm, 
running in time $O(n^{1.5} \log n)$ with small error probability.
Deorowicz \cite{Deorowicz:09} observed that \cite{Albert+:07}
also provides a deterministic algorithm for cyclic LIS,
running in time $O(n \lambda \log n)$,
and gave an improved algorithm 
running in time $O\bigpa{\min(n\lambda, n \log n + \lambda^3 \log n)}$.

We now give a simple algorithm for the cyclic LIS problem,
similar to the cyclic LCS algorithm given in \secref{s-cyclic},
but based on the semi-local LCS problem for permutation strings.
First, we run \algref{alg-perm} (semi-local LCS between permutation strings),
obtaining the semi-local seaweed matrix $\rP_{a,b}$.
Then, we run the algorithm of \thref{th-comp-mmult-comp} 
on matrix $\rP_{a,b}$ against itself,
obtaining the semi-local seaweed matrix $\rP_{aa,b}$.
Finally, we perform $n$ substring-string LCS queries
for every substring of $aa$ of length $n$ against string $b$.
The overall running time is dominated 
by the call to \algref{alg-perm},
which runs in time $O(n \log^2 n)$.

The resulting algorithm for the cyclic LIS problem 
improves on existing algorithms both in running time,
and by being deterministic.
In particular, our algorithm is faster 
than the algorithm of \cite{Deorowicz:09},
unless $l = o\bigpa{(n \log n)^{1/3}}$.

%%=-=-=-=-=-=-=-=-=-=-=-=-=-=-=-=-=-=-=-=-=-=-=-=-=-=-=-=-=-=-=-=-=-=-=-=-=-=%%
\mysection{Longest pattern-avoiding subsequences}
\label{s-permutation-lxs}

Two given permutation strings $a$, $b$ of equal length
(but generally over different alphabets) are called \emph{isomorphic}, 
if they have the same relative order of characters, 
i.e.\ $a(\hi) < a(\hj)$ iff $b(\hi) < b(\hj)$ for all $\hi$, $\hj$.
Given a target permutation string $t$ of length $n$ 
and a pattern permutation string $p$ of fixed length,
the \emph{longest $p$-isomorphic subsequence problem},
or simply the \emph{longest $p$-subsequence problem},
asks for the longest subsequence of $t$ that is isomorphic to $p$.
More generally, given a set of pattern permutation strings $X$,
the \emph{longest $X$-subsequence problem}
asks for the longest subsequence of $t$ 
that is isomorphic to any pattern string in $p$.
\begin{example}
The LIS problem can be interpreted as the longest $X$-subsequence problem,
where $X = \brc{\textsf{``1''}, \textsf{``12''}, \textsf{``123''}, \ldots, 
\textsf{``123\ldots n''}}$.
\end{example}

\index{problem!longest $Y$-avoiding subsequence}%
Given a set of \emph{antipattern} permutation strings $Y$,
the \emph{longest $Y$-avoiding subsequence problem}
asks for the longest subsequence of $t$
that \emph{does not} contain a subsequence isomorphic
to any string in $Y$.
\begin{example}
The LIS problem on a permutation string can be interpreted 
as the longest $\brc{\textsf{``21''}}$-avoiding subsequence problem.
\end{example}

For a detailed introduction into these problems and their connections,
see the work by Albert et al.\ \cite{Albert+:03} and references therein.

The LIS problem is the only nontrivial example 
of the longest $Y$-avoiding subsequence problem 
with antipatterns of length $2$.
Albert et al.\ \cite{Albert+:03} gave the full classification 
of the longest $Y$-avoiding subsequence problem
for all sets of antipatterns of length $3$.
There are 10 non-trivial sets of such antipatterns.
For each of these sets, 
the algorithms given in \cite{Albert+:03} run in polynomial time, 
ranging from $O(n \log n)$ to $O(n^5)$.
Two particular antipattern sets considered in \cite{Albert+:03} are
(in that work's original notation):
\begin{gather*}
C_3 = \brc{\textsf{``132'', ``213'', ``321''}}\\
C_4 = \brc{\textsf{``132'', ``213'', ``312''}}
\end{gather*}
For both these antipattern sets, algorithms given in \cite{Albert+:03} 
run in time $O(n^2 \log n)$.

We now give new algorithms 
for the longest $C_3$- and $C_4$-avoiding subsequence problems,
improving on the above algorithms in running time.

Permutation strings that are $C_3$-avoiding
are all cyclic rotations of an increasing permutation string.
The longest such subsequence in the target string can be found 
by the algorithm for the cyclic LCS problem between permutations
(\secref{s-permutation-cyclic}), running in time $O(n \log^2 n)$.

Permutation strings that are $C_4$-avoiding
are all obtained from an increasing permutation string 
by reversing some suffix.
The longest $C_4$-avoiding subsequence in the target string 
can be found as follows.
Let the target string $t$ be over the alphabet $\ang{0:n}$.
First, we call the standard LIS algorithm on $t$,
obtaining explicitly the prefix-prefix LCS scores
\begin{gather*}
\mathit{lcs}\bigpa{t\ang{0:\hi^+}, \id\ang{0:t(\hi)^+}} =
\mathit{lcs}\bigpa{t\ang{0:\hi^-}, \id\ang{0:t(\hi)^-}} + 1
\end{gather*}
for all $\hi \in \ang{0:n}$.
Independently, we call the seaweed doubling algorithm (\algref{alg-perm}) on $t$ 
against the reverse identity permutation $\ol{\id}$,
and use \thref{th-query} to process its output into a data structure
that allows efficient queries of all suffix-prefix LCS scores
$\mathit{lcs}\bigpa{t\ang{k:n}, \ol{\id}\ang{0:l}}$
for all $k,l \in \bra{0:n}$.
Finally, we obtain the solution 
to the longest $C_4$-avoiding subsequence problem as 
\begin{gather*}
\max_{\hi \in \ang{0:n}} \Bigpa{
\mathit{lcs}\bigpa{t\ang{0:\hi^+},\id\ang{0:t(\hi)^+}} +
\mathit{lcs}\bigpa{t\ang{\hi^+:n},\ol{\id}\ang{0:n-t(\hi)^+}}}
\end{gather*}
The overall running time is dominated by the call to \algref{alg-perm}, 
which runs in time $O(n \log^2 n)$.

%%=-=-=-=-=-=-=-=-=-=-=-=-=-=-=-=-=-=-=-=-=-=-=-=-=-=-=-=-=-=-=-=-=-=-=-=-=-=%%
\mysection{Longest piecewise monotone subsequences}
\label{s-permutation-piecewise}

The classical LIS problem asks for the longest increasing 
(or, equivalently, decreasing) subsequence in a permutation string.
A natural generalisation is to ask for the longest subsequence
that consists of a constant number of monotone pieces.
\index{problem!longest $k$-increasing subsequence}%
\index{problem!longest $k$-modal subsequence}%
In particular, given a permutation string $a$ of length $n$,
the \emph{longest $k$-increasing subsequence}
(respectively, \emph{longest $k$-modal subsequence}) problem 
asks for the longest subsequence in $a$
that is a concatenation of at most $k$ sequences,
all of which are increasing
(respectively, alternate between increasing and decreasing).
In the case of the longest $k$-modal subsequence problem,
we assume without loss of generality that $k$ is even.
Both problems can be solved as an instance of the LCS problem,
comparing the input permutation string $a$ against string $\id^k$,
i.e.\ the concatenation of $k$ copies of the identity permutation $\id$
(respectively, against string $(\id\,\ol{\id})^{k/2}$, 
i.e.\ the concatenation of $k$ alternating copies of $\id$ 
and its reverse $\ol{\id}$).
The resulting alignment dag is of size $n \times kn$,
and contains $kn$ match cells.
Using standard sparse LCS algorithms
\cite{Hunt_Szymanski:77,Apostolico_Guerra:87_Algorithmica},
such an instance of the LCS problem can be solved in time $O(nk \log n)$.
Demange et al.\ \cite{Demange+:09} gave a similar algorithm 
for the longest $k$-modal subsequence problem,
also running in time $O(nk \log n)$.

We now give new algorithms for the longest $k$-increasing subsequence
and the longest $k$-modal subsequence problems,
improving on the above algorithms in running time for sufficiently large $k$.

To solve the longest $k$-increasing subsequence problem,
we run \algref{alg-perm} (seaweed doubling),
obtaining the semi-local seaweed matrix $\rP_{\id,a}$,
from which we extract 
the string-substring seaweed matrix $\rP^\ssub_{\id,a}$.
We then run the algorithm of \thref{th-comp-mmult-comp} 
repeatedly $\log k$ times, 
obtaining the string-substring seaweed matrix $P_1 = \rP^\ssub_{\id^k,a}$.

For the the longest $k$-modal subsequence problem,
assume without loss of generality that $k$ is even.
To solve this problem, we first run \algref{alg-perm} (seaweed doubling) twice,
on strings $\id$ and $\ol{\id}$, respectively, against string $a$.
As as a result, we obtain the semi-local seaweed matrices 
$\rP_{\id,a}$ and $\rP_{\ol{\id},a}$.
We then obtain matrix $\rP_{\id\,\ol{\id},a}$ by \thref{th-comp-mmult-comp}.
From this matrix, we extract 
the string-substring seaweed matrix $\rP^\ssub_{\id\,\ol{\id},a}$.
We then run the algorithm of \thref{th-comp-mmult-comp} 
repeatedly $\log k-1$ times, 
obtaining the string-substring seaweed matrix
$P_2 = \rP^\ssub_{(\id\,\ol{\id})^{k/2},a}$.

The final step of the algorithm is identical for both problems:
we use the obtained string-substring seaweed matrix 
($P_1$ and $P_2$, respectively) to query the global LCS score.
Both described algorithms run in time 
$O(n \log^2 n) + \log k \cdot O(n \log n) = O(n \log^2 n)$.
This is faster than both the sparse LCS approach
and the algorithm of \cite{Demange+:09}, for all $k = \omega(\log n)$.

%%=-=-=-=-=-=-=-=-=-=-=-=-=-=-=-=-=-=-=-=-=-=-=-=-=-=-=-=-=-=-=-=-=-=-=-=-=-=%%
\mysection{Maximum clique in a circle graph}
\label{s-circle}

\begin{figure}[tb]
\centering
\subfloat[\label{f-circle-chord}The chord model]{%
\includegraphics{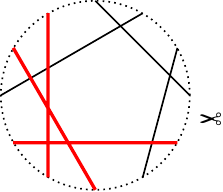}}
\qquad
\subfloat[\label{f-circle-int}An interval model]{%
\includegraphics{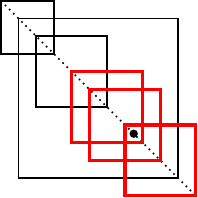}}
\caption{\label{f-circle} A circle graph and its maximum clique}
\end{figure}
\index{circle graph}
A \emph{circle graph} \cite{Even_Itai:71,Golumbic:04} is defined 
as the intersection graph of a set of chords in a circle,
i.e.\ the graph where each node represents a chord,
and two nodes are adjacent, whenever the corresponding chords intersect.
We consider the maximum clique problem on a circle graph.

\index{circle graph!interval model}
The \emph{interval model} of a circle graph is obtained
by cutting the circle at an arbitrary point 
and laying it out on a line, so that the chords become (closed) intervals.
The original circle graph is isomorphic 
to the overlap graph of its interval model,
i.e.\ the graph where each node represents an interval,
and two nodes are adjacent, whenever the corresponding intervals intersect 
but do not contain one another.

\begin{example}
\figref{f-circle} shows an instance of the maximum clique problem
on a six-node circle graph.
\sfigref{f-circle-chord} shows the set of chords defining a circle graph,
with one of the maximum cliques highlighted in bold red.
The cut point is shown by scissors.
\sfigref{f-circle-int} shows the corresponding interval model;
the dotted diagonal line contains the intervals,
each defined by the diagonal of a square.
The squares corresponding to the maximum clique 
are highlighted in bold red.
\end{example}

\index{problem!circle graph maximum clique}%
It has long been known that the maximum clique problem 
in a circle graph on $n$ nodes 
is solvable in polynomial time \cite{Gavril:73}.
A number of algorithms have been proposed for this problem
\cite{Rotem_Urrutia:81,Hsu:85,Masuda+:90,Apostolico+:92_DAM};
the problem has also been studied in the context 
of line arrangements in the hyperbolic plane 
\cite{Karzanov:79,Dress+:02}.
Given an interval model of a circle graph,
the running time of the above algorithms is $O(n^2)$ in the worst case,
i.e.\ when the input graph is dense.
In \cite{Tiskin:06_CPM,Tiskin:08_MCS},
we gave an algorithm running in time $O(n^{1.5})$.

We now give a new algorithm for the maximum clique problem in a circle graph,
improving on existing algorithms in running time.
The algorithm is based on 
the fast implicit distance multiplication procedure of \thref{th-mmult}.

Our algorithm takes as input
the interval model of a circle graph $G$ on $n$ nodes.
Without loss of generality, we may assume
that the set of interval endpoints is $\ang{0:2n}$. 
The interval model is represented by a permutation string $a$ of size $2n$,
where for each left (respectively, right) 
interval endpoint $\hi \in \ang{0:2n}$,
$a(\hi)$ is the corresponding right (respectively, left) endpoint.
Note that for all $\hi < \hj$,
an interval with left endpoint $\hi$ 
does not contain an interval with left endpoint $\hj$,
if and only if $a(\hi) < a(\hj)$.
Various alternative representations of interval models 
(e.g.\ the ones used in \cite{Rotem_Urrutia:81,Apostolico+:92_DAM})
can be converted to this representation in linear time.

In the interval model, a clique corresponds 
to a set of pairwise intersecting intervals,
none of which contains another interval from the set.
Recall that intervals in the line satisfy the \emph{Helly property}: 
if all intervals in a set intersect pairwise,
then they all intersect at a common point.
In our context, we only need to consider 
integer indices as intersection points.

Consider a clique in $G$.
Let $k \in \bra{1:2n-1}$ be a common intersection point 
of the intervals representing the clique, 
which is guaranteed to exist by the Helly property.
Since the intervals representing the clique cannot contain one another,
the sequence of their right endpoints is an increasing subsequence of $a$.
Let $\id$ be the identity permutation string of length $2n$.
From the observations above, it follows that 
the clique corresponds to a common subsequence 
of a prefix $a\ang{0:k}$ and a suffix $\id\ang{2n-k:2n}$.
Therefore, the maximum clique problem can be solved 
as an instance of the semi-local LCS problem.

\begin{algorithm}
\textbf{(Maximum clique in a circle graph)}%
\label{alg-circle}
\setlabelitbf
\nobreakitem[Input:]
interval model of circle graph $G$, 
represented by permutation string $a$ of size $2n$.
\item[Output:]
maximum-size clique of $G$,
represented by the corresponding set of intervals.
\item[Description.]
\setlabelit
\nobreakitem[First phase.]
We run \algref{alg-perm} on the input permutation string $a$
against the identity permutation string $\id$,
obtaining the seaweed matrix $\rP_{a,\id}$.
We then build the data structure of \thref{th-query}
for querying semi-local LCS scores of $a$ against $\id$.
\item[Second phase.]
The size of the maximum clique can now be obtained as
\begin{gather*}
\max_{k \in \bra{1:2n-1}} \lcs\bigpa{a\ang{0:k}, \id\ang{2n-k:2n}}
\end{gather*}
For each $k$, the prefix-suffix LCS score is queried 
from seaweed matrix $\rP_{a,\id}$ by \thref{th-query}.
The value $k^*$ for which the maximum score is attained
gives a common intersection point of the clique intervals.
\item[Third phase.]
Let $a'$ be a subsequence of the prefix $a\ang{0:k^*}$,
obtained by deleting all characters 
not belonging to the interval $\ang{2n-k^*:2n}$.
The intervals defining the maximum clique can now be obtained
by running a standard LIS algorithm on string $a'$,
and then tracing back the elements of the resulting LIS.
\setlabelitbf
\item[Cost analysis.]
\setlabelit
\item[First phase.]
The running time of \algref{alg-perm} is $O(n \log^2 n)$.
\item[Second phase.]
By \thref{th-query}, the combined running time 
of all the prefix-suffix queries is $O(n \log^2 n)$, 
if the queries are performed independently.
This time can be reduced to $O(n)$ by observing that
the queries can be performed as a single diagonal batch query
of the type described by \thref{th-query}.
\item[Third phase.]
The LIS algorithm on string $a'$ runs in time $O(n \log n)$.
\item[Total.]
The overall running time is $O(n \log^2 n)$.
\end{algorithm}

Like many algorithmic problems, 
the problem of finding a maximum clique in a circle graph
admits various parameterised versions.
Some relevant parameters are:
\begin{itemize}
\item the size $l$ of the maximum clique;
\item the \emph{thickness} $d$ of the interval model, i.e.\ 
  the maximum number of intervals containing a point,
  taken across all points in the line;
\item the number $e$ of graph edges.
\end{itemize}
For any interval model of a non-trivial circle graph,
we have $l \leq d \leq n \leq e \leq n^2$.
Notice that, given a permutation representing an interval model,
its thickness can be found in time $O(n \log^2 n)$
by building a range tree on the corresponding set of planar points,
and then performing $O(n)$ dominance counting queries.

Apostolico et al.\ \cite{Apostolico+:92_DAM} give algorithms
for the parameterised version 
of the maximum clique problem in a circle graph,
running in time $O(n \log n + e)$ and $O\bigpa{n \log n + nl \log(n/l)}$.
They also describe an algorithm for the maximum independent set problem,
parameterised by the interval model's thickness.

We now give a new algorithm for the maximum clique problem in a circle graph,
parameterised by the thickness of the input interval model.
Our algorithm improves on the parameterised algorithms 
of \cite{Apostolico+:92_DAM} for most values of the parameters.
The algorithm is an extended version of \algref{alg-circle}.

\begin{algorithm}
\textbf{(Maximum clique in a circle graph, parameterised by thickness)}%
\label{alg-circle-thickness} \setlabelitbf
\nobreakitem[Input:]
interval model of circle graph $G$, 
represented by string $a$ of size $2n$.
\item[Output:]
maximum-size clique of $G$,
represented by the corresponding set of intervals.
\item[Parameter:] 
thickness $d$, $d \leq n$, of the input interval model.
\nobreakitem[Description.]
\setlabelit
\nobreakitem[First phase.]
We run \algref{alg-perm} on string $a_r = a\ang{0:(r+1)d}$
against string $\id_r = \id\ang{rd:2n}$,
independently for all $r \in \bra{0:\frac{2n}{d}-1}$.
As will be shown in the algorithm's analysis,
in each run we obtain a seaweed matrix $\rP_{a_r,\id_r}$
with at most $4d$ non-trivial nonzeros.
For every $r$, 
we then build the data structure of \thref{th-query}
for querying semi-local LCS scores of $a_r$ against $\id_r$.
\item[Second phase.]
The size of the maximum clique can now be obtained as
\begin{gather*}
\max_{k \in \bra{1:2n-1}} \lcs\bigpa{a\ang{0:k}, \id\ang{2n-k:2n}} = {}\\
\max_{k \in \bra{1:2n-1}} 
  \lcs\bigpa{a_{\floor{k/d}}\ang{0:k}, \id_{\floor{k/d}}\ang{2n-k:2n}}
\end{gather*}
For each $k$, the prefix-suffix LCS score is queried 
from seaweed matrix $\rP_{a_{\floor{k/d}},\id_{\floor{k/d}}}$
by \thref{th-query}.
The value $k^*$ for which the maximum score is attained
gives a common intersection point of the clique intervals.
\item[Third phase.]
As in \algref{alg-circle}.
\setlabelitbf
\item[Cost analysis.]
\setlabelit
\nobreakitem[First phase.]
Consider the string decomposition
\begin{gather*}
a_r = a\ang{0:(r+1)d} = a\ang{0:rd} \; a\ang{rd:(r+1)d}
\end{gather*}
The alignment dag $\rG_{a_r,\id_r}$
is therefore the composite of alignment dags 
$\rG_{a\ang{0:rd},\id_r}$ and $\rG_{a\ang{rd:(r+1)d},\id_r}$.

The alignment dag $\rG_{a\ang{0:rd},\id_r}$
contains at most $d$ match cells, since every match 
corresponds to an interval containing point $rd$,
and there can be at most $d$ such intervals by the definition of thickness.
The alignment dag $\rG_{a\ang{rd:(r+1)d},\id_r}$
also contains at most $d$ match cells, 
since the length of the string $a\ang{rd:(r+1)d}$ is $d$.
Hence, the alignment dag $\rG_{a_r,\id_r}$
contains at most $d+d=2d$ matches.
Therefore, the time for each run of \algref{alg-perm} is $O(d \log^2 d)$,
and the overall running time of this phase 
is $O(n/d \cdot d \log^2 d) = O(n \log^2 d)$.
\item[Second phase.]
By \thref{th-query}, the combined running time 
of all the prefix-suffix queries is $O(n \log^2 d)$, 
if the queries are performed independently.
This time can be reduced to $O(n/d \cdot d) = O(n)$ by observing that
the queries can be performed as a set of $n/d$ diagonal batch queries
of the type described by \thref{th-query}.
\item[Third phase.]
String $a'$ contains at most $d$ characters, since every such character 
corresponds to an interval containing point $k^*$.
Therefore, the LIS algorithm on string $a'$ runs in time $O(d \log d)$.
\item[Total.]
The overall running time is $O(n \log^2 d)$.
\end{algorithm}
\algref{alg-circle-thickness} is faster than 
the $O(n \log n + e)$ algorithm of \cite{Apostolico+:92_DAM},
unless $e = o(n \log^2 d) = O(n \log^2 n)$. 
It is also faster than 
the $O\bigpa{n \log n + nl \log(n/l)}$ algorithm of \cite{Apostolico+:92_DAM},
unless $l = o\bigpa{\frac{\log^2 d}{\log n}} = O(\log n)$.

%%=-=-=-=-=-=-=-=-=-=-=-=-=-=-=-=-=-=-=-=-=-=-=-=-=-=-=-=-=-=-=-=-=-=-=-=-=-=%%
\mysection{Maximum common pattern between linear graphs}
\label{s-linear}

\index{linear graph}%
The concept of a \emph{linear graph},
introduced by Davydov and Batzoglou \cite{Davydov_Batzoglou:06},
is similar to an interval model of a circle graph defined in \secref{s-circle}.
The interval relations of disjointness, containment and overlapping
are denoted respectively by symbols $<$, $\sqsubset$ and $\between$.
A \emph{pattern} in a linear graph is defined as an ordered subset of intervals,
all of which satisfy pairwise a prescribed subset of relations. 

\index{problem!linear graph maximum common pattern}%
Fertin et al.\ \cite{Fertin+:10} considered 
the \emph{maximum common $S$-structured pattern ($S$-MCSP)} problem.
The problem asks for the maximum common pattern in a set of $n$ linear graphs,
each defined by at most $m$ intervals,
where the structure of the common pattern
is restricted by a prescribed subset of relations
$S \subseteq \brc{<, \sqsubset, \between}$.
In particular, the $\brc{\between}$-MCSP problem
asks for the maximum commonly-structured 
subset of pairwise overlapping intervals;
for $n=1$ this is equivalent to finding the maximum clique of a circle graph,
and for general $n$ is equivalent to finding the minimum-sized clique
among maximum cliques of the $n$ input circle graphs.
The $\brc{<,\sqsubset}$-MCSP problem
asks for the maximum commonly-structured subset of intervals,
no two of which are overlapping;
for $n=1$ this is equivalent to finding 
the maximum independent set of a circle graph;
however, for general $n$ the maximum commonly-structured independent set
of the $n$ input circle graphs
may be significantly different from (and smaller than)
each of the $n$ individual maximum independent sets.
The $\brc{<,\sqsubset,\between}$-MCSP problem
asks for the maximum commonly-structured subset of intervals
without any a priori restriction on its structure.

Extending and generalising a number of previous results,
paper \cite{Fertin+:10} considers the $S$-MCSP problem,
where $S$ runs over all seven nonempty subsets 
of $\brc{<,\sqsubset,\between}$.
For some of these seven variants, the algorithms use as a subroutine 
the algorithm of \cite{Tiskin:06_CPM,Tiskin:08_MCS}
for the maximum clique problem in a circle graph.
By plugging in the more efficient \algref{alg-circle},
we can obtain improved algorithms for those variants of the $S$-MCSP problem,
where finding the maximum clique in a circle graph is a bottleneck.

In particular, the $\brc{\between}$-MCSP problem 
is solved in \cite{Fertin+:10} by finding the maximum clique 
independently for $n$ circle graphs, 
each corresponding to one of the input linear graphs,
in overall time $O(nm^{1.5})$.
By plugging in \algref{alg-circle}, 
the running time is improved to $O(nm \log^2 m)$.

The $\brc{<,\between}$-MCSP problem 
is shown in \cite{Fertin+:10} to be NP-hard,
and to admit a polynomial-time $2h(k)$-approximation,
where $h(k)=\sum_{1 \leq i \leq k} 1/i = \ln n + O(1)$;
for the rest of this section, 
$k$ denotes the size of the optimal solution to the problem.
The approximation is obtained by $O(nm)$ calls
to the following subroutine:
given a circle graph of size $m$, and two integers $m_1$, $m_2$,
decide whether the graph contains $m_1$ disjoint cliques, each of size $m_2$.
This subroutine is performed in time $O(m^{2.5} \log m)$,
and therefore the overall running time is
$O(nm) \cdot O(m^{2.5} \log m) = O(nm^{3.5} \log m)$.
By a straightforward extension of \algref{alg-circle},
the running time of the subroutine is improved to $O(m \log^2 m)$,
and therefore the overall running time of the approximation algorithm 
is improved to 
$O(nm) \cdot O(m \log^2 m) = O(nm^2 \log^2 m)$.

The $\brc{\sqsubset,\between}$-MCSP problem 
is also shown in \cite{Fertin+:10} to be NP-hard,
and to admit a polynomial-time $k^{1/2}$-approximation.
The approximation is obtained by combining exact solutions
for the $\brc{\sqsubset}$-MCSP and $\brc{\between}$-MCSP problems
on the same input sets.
The exact solution for the $\brc{\between}$-MCSP is the bottleneck;
by plugging in the improved algorithm for this problem described above, 
the running time of the approximation algorithm 
for the $\brc{\sqsubset,\between}$-MCSP problem 
is improved from $O(nm^{1.5})$ to $O(nm \log^2 m)$.

Finally, paper \cite{Fertin+:10} 
argues that the $\brc{<,\sqsubset,\between}$-MCSP problem is NP-hard,
and gives several polynomial-time approximation algorithms.
In particular, it gives an $O(k^{2/3})$-approximation algorithm
running in time $O(nm^{1.5})$,
and an $O\bigpa{(k \log k)^{1/2}}$-approximation algorithm
running in time $O(nm^{3.5} \log m)$.
By using the techniques described above,
the running times of these approximation algorithms
are improved respectively to $O(nm \log^2 m)$ and $O(nm^2 \log^2 m)$.

%%===========================================================================%%

%%===========================================================================%%
\mychapter{Compressed string comparison}
\label{c-compressed}

In this chapter, we consider the semi-local comparison of compressed strings.

This chapter is organised as follows. 
In \secref{s-gc}, we introduce the grammar compression (GC) framework,
that generalises the classical LZ78 and LZW methods.
In \secref{s-gc-subrec-global}, we describe a folklore algorithm
for global subsequence recognition in a GC-string.
In \secref{s-gc-slcs}, we develop an efficient algorithm
for the extended substring-string problem 
between a plain pattern and a GC-string.
By application of this algorithm, in \secref{s-gc-lsubrec}
we obtain an algorithm for local subsequence recognition in a GC-string,
and in \secref{s-gc-tmatch}
an algorithm for threshold approximate matching in a GC-string;
both these algorithms improve on the existing ones in running time.

\mysection{Grammar-compressed strings}
\label{s-gc}

String compression is a classical paradigm, 
touching on many different areas of computer science.
From an algorithmic viewpoint, 
it is natural to ask whether compressed strings 
can be processed efficiently without decompression.
Early examples of such algorithms were given 
e.g.\ by Amir et al.\ \cite{Amir+:96} and by Rytter \cite{Rytter:99};
for a recent survey on the topic, see Lohrey \cite{Lohrey:12}.
Efficient algorithms for compressed strings
can also be applied to achieve speedup over ordinary string processing algorithms
for plain strings that are highly compressible.

We consider the following general model of compression.
\begin{definition}
\index{string!grammar-compressed (GC)}%
\index{straight-line program (SLP)}%
Let $t$ be a string of length $n$ (typically large).
String $t$ will be called a \emph{grammar-compressed string (GC-string)},
if it is generated by a context-free grammar,
also called a \emph{straight-line program (SLP)}.
\index{straight-line program (SLP)!statement}%
An SLP of length $\bar n$, $\bar n \leq n$,
is a sequence of $\bar n$ \emph{statements.}
A statement numbered $k$, $1 \leq k \leq \bar n$, 
has one of the following forms:
\begin{gather*}
t_k = \alpha\qquad t_k = uv
\end{gather*}
where $\alpha$ is an alphabet character,
and each of $u$, $v$ is either an alphabet character,
or symbol $t_i$ for some $i$, $1 \leq i < k$.
\end{definition}
We identify every symbol $t_r$ with the string it expands to;
in particular, we have $t = t_{\bar n}$.
In general, the plain string length $n$ 
can be exponential in the GC-string length $\bar n$.

\begin{example}
The \emph{Fibonacci string} \textsf{``ABAABABAABAAB''} of length 13
can be represented by the following SLP of length 6:
\begin{gather*}
t_1 = \textsf{A}\quad
t_2 = t_1 \textsf{B}\quad
t_3 = t_2 t_1\quad
t_4 = t_3 t_2\quad
t_5 = t_4 t_3\quad
t_6 = t_5 t_4
\end{gather*}
In general, a Fibonacci string of length $n$
can be represented by an SLP of length $\bar n$, where 
$n = F_{\bar n} = 
\bigpa{\frac{1}{\sqrt 5}-o(1)} \bigpa{\frac{1+\sqrt 5}{2}}^{\bar n}$
is the $\bar n$-th Fibonacci number.

This example is borrowed from Hermelin et al.\ \cite{Hermelin+:13}.
\end{example}

Kida et al.\ \cite{Kida+:03} introduced 
a more general compression model, called \emph{collage systems}.
Grammar compression is a equivalent 
to a subclass of collage systems called \emph{regular}.
As a special case, grammar compression includes 
the classical LZ78 and LZW compression schemes 
by Ziv, Lempel and Welch \cite{Ziv_Lempel:78,Welch:84}.
Both these schemes can be expressed by an SLP 
that consists of three sections:
\begin{itemize}
\item in the first section, all statements are of the form $t_k = \alpha$;
\item in the second section, all statements are of the form $t_k = t_i t_j$,
      where statement $j$ is from the first section;
\item in the third section, all statements are of the form $t_k = t_{k-1} t_j$,
      where statement $j$ is from the second section.
\end{itemize}
It should also be noted that certain other compression methods,
such as e.g.\ LZ77 \cite{Ziv_Lempel:77} and run-length compression,
do not fit directly into the grammar compression model.

\begin{comment}
Our goal is to design efficient algorithms on GC-strings.
While we do not allow full decompression
(since, in the worst case, this could be extremely inefficient),
we will assume that standard arithmetic operations on integers up to $n$
can be performed in constant time.
This simplifying assumption is based on the fact 
that most natural problems on GC-strings
operate on integer values (string lengths and indices)
that may be as high as $O(n)$.
The same assumption is made implicitly in previous works, 
e.g.\ by C\'egielski et al.\ \cite{Cegielski+:06}.
The assumption can be removed from all algorithms considered in this chapter 
by introducing \emph{symbolic indices} for strings,
which remap exponentially-sized index ranges to polynomially-sized ones,
while keeping enough information 
to restore the original index values in the algorithm's output.
\end{comment}

The algorithms in this section will take as input
a plain (uncompressed) \emph{pattern string} $p$ of length $m$,
and a grammar-compressed \emph{text string} $t$ of length $n$, 
generated by an SLP of length $\bar n$.
We aim at algorithms with running time 
that is a low-degree polynomial in $m$, $\bar n$,
but is independent of $n$ (which could be exponential in $\bar n$).

%%=-=-=-=-=-=-=-=-=-=-=-=-=-=-=-=-=-=-=-=-=-=-=-=-=-=-=-=-=-=-=-=-=-=-=-=-=-=%%
\mysection{Global subsequence recognition}
\label{s-gc-subrec-global}

The first problem that we consider on a compressed text
is the global subsequence recognition problem, introduced in \secref{s-lcs}.
Recall that on a plain text, this problem can be solved
in time $O(n)$ by a straightforward algorithm.
We now revisit this problem, assuming a plain pattern and a GC-text as inputs.
We also generalise the problem slightly,
looking for the length of the longest prefix of $p$ 
that is a subsequence of $t$.
The problem can be solved 
by a simple folklore algorithm as follows.

\begin{algorithm}[Global subsequence recognition]
\label{alg-gc-global}
\setlabelitbf
\nobreakitem[Input:]
plain pattern string $p$ of length $m$;
SLP of length $\bar n$, generating text string $t$ of length $n$.
\item[Output:]
an integer $k$, giving the length of the longest prefix of $p$
that is a subsequence of $t$.
String $t$ contains $p$ as a subsequence, if and only if $k=m$.
\item[Description.]
Recursion on the input SLP generating $t$.

\setlabelnormal
\item[Recursion base: $n=\bar n=1$.]
The output value $k \in \{0,1\}$ is determined 
by a single character comparison.

\item[Recursive step: $n \geq \bar n > 1$.]
Let $t=t' t''$ be the SLP statement defining string $t$.
Let $k'$ be the length of the longest prefix of $p$
that is a subsequence of $t'$.
Let $k''$ be the length of the longest prefix of $p \ldrop k'$ 
that is a subsequence of $t''$.
We call the algorithm recursively to obtain $k'$ and $k''$,
and then return $k=k'+k''$.

\nobreakitem[(End of recursive step)]
\setlabelitbf
\item[Cost analysis.]
The running time of the algorithm is $O(k \bar n)$.
The proof is by induction.
The running time of the recursive calls is respectively
$O(k' \bar n)$ and $O(k'' \bar n)$.
The overall running time of the algorithm
is $O(k' \bar n) + O(k'' \bar n) + O(1) = O(k \bar n)$.
In the worst case, this is $O(m \bar n)$.
\end{algorithm}

%%=-=-=-=-=-=-=-=-=-=-=-=-=-=-=-=-=-=-=-=-=-=-=-=-=-=-=-=-=-=-=-=-=-=-=-=-=-=%%
\mysection{Extended substring-string LCS}
\label{s-gc-slcs}

\index{problem!substring-string LCS!extended}%
We now consider the LCS problem.
Recall from \secrefs{s-lcs}, \ref{s-micro} that on a pair of plain strings,
it can be solved in time $O\bigpa{\frac{mn}{\log^2 n} + n}$, 
assuming that $m \leq n$.
The LCS problem on two GC-strings has been considered 
by Lifshits and Lohrey \cite{Lifshits_Lohrey:06,Lohrey:12},
and proven to be PP-hard (and therefore NP-hard).

In this section, we revisit the LCS problem,
now assuming a plain pattern $p$ and a GC-text $t$ as inputs.
Although, in principle, we would like to solve 
the more general semi-local LCS problem,
it would be impossible to do so 
while keeping the running time independent of $n$,
since the resulting semi-local seaweed matrix would require memory $O(m+n)$.
However, we are still able to solve the extended substring-string LCS problem
(i.e.\ string-substring, prefix-suffix and suffix-prefix LCS),
where the output only requires memory $O(m)$.

In the special case of LZ77 or LZW compression of the text,
the algorithm of Crochemore et al.\ \cite{Crochemore+:03_SIAM}
solves the LCS problem in time $O(m \bar n)$.
Thus, LZ77 or LZW compression of one of the input strings
only slows down the LCS computation by a constant factor
relative to the classical dynamic programming LCS algorithm,
or by a polylogarithmic factor relative to the best known LCS algorithms.

The general case of an arbitrary GC-text appears more difficult.
A GC-text is a special case of a context-free language,
which consists of a single string.
Therefore, the LCS problem between a GC-text and a plain pattern
can be regarded as a special case of the edit distance problem
between a context-free language given by a grammar of size $\bar n$,
and a pattern string of size $m$.
For this more general problem, Myers \cite{Myers:95}
gave an algorithm running in time 
$O(m^3 \bar n + m^2 \cdot \bar n \log \bar n)$.
In \cite{Tiskin:09_JMS}, we gave an algorithm 
for the three-way semi-local LCS problem 
between a GC-text and a plain pattern,
running in time $O(m^{1.5} \bar n)$.
Lifshits \cite{Lifshits:07} asked whether the LCS problem
in the same setting can be solved in time $O(m \bar n)$.

A new algorithm for the extended substring-string LCS problem 
can be obtained by an application of the techniques 
described in \chapref{c-semi}.
The resulting algorithm improves on existing algorithms in running time,
and approaches an answer to Lifshits' question within a logarithmic factor.

\begin{algorithm}[Extended substring-string LCS]
\label{alg-gc-slcs}
\setlabelitbf
\nobreakitem[Input:]
plain pattern $p$ of length $m$;
SLP of length $\bar n$, generating text $t$ of length $n$.
\item[Output:]
nonzeros of matrix $\rP^\subsX_{p,t}$.
\item[Description.]
First, we observe that, 
although the output matrix contains at most $m$ nonzeros,
its range is of size $m+n$, which may be exponentially larger.
To avoid an exponential growth of the indices, 
we will clean up the range by removing unused indices,
and deleting the corresponding zero row-column pairs from the matrix.
Formally, we describe this process as an order-preserving 
remapping of the index range.
\setlabelit
\item[First phase.]
Recursion on the input SLP generating $t$.

\setlabelnormal
\item[Recursion base: $n=\bar n=1$.]
The output can be computed by \algref{alg-seaweed} (seaweed combing)
on plain strings $p$ and $t$, of length $m$ and $1$ respectively.

\item[Recursive step: $n \geq \bar n > 1$.]
Let $t=t' t''$ be the SLP statement defining string $t$.
We call the algorithm recursively to obtain 
the nonzeros of matrices $\rP^\subsX_{p,t'}$, $\rP^\subsX_{p,t''}$.
The total number of nonzeros in each matrix is between $m$ and $2m$.
Conceptually, the ranges of these matrices are respectively
$\ang{-m:n' \mid 0:m+n'}$, and $\ang{-m:n'' \mid 0:m+n''}$.
However, the actual remapped range for each matrix
is $\ang{-m:2m \mid 0:3m}$ after the respective recursive call.

We now compute the matrix $\rP^\subsX_{p,t}$ 
from $\rP^\subsX_{p,t'}$, $\rP^\subsX_{p,t''}$
by \thref{th-comp-mmult-comp}.
The total number of nonzeros in this matrix 
is again between $m$ and $2m$.
Conceptually, the range of this matrix is $\ang{-m:n \mid 0:m+n}$.
However, the actual remapped index range 
after the application of \thref{th-comp-mmult-comp}
is $\ang{-m:4m \mid 0:5m}$.
Therefore, there are at least $2m$ indices $\hi \in \ang{0:4m}$,
such that the row $\rP^\subsX_{p,t}(\hi,*)$
and the column $\rP^\subsX_{p,t}(*,\hi)$ both contain only zeros.
We now delete exactly $2m$ such rows and columns from the matrix,
and remap the index range to $\ang{-m:2m \mid 0:3m}$,
while preserving the linear order of the indices.
\nobreakitem[(End of recursive step.)]

\setlabelit
\item[Second phase.]
We now have the nonzeros of the output matrix $\rP^\subsX_{p,t}$,
remapped to the range $\ang{-m:2m \mid 0:3m}$.
This is already sufficient to query the global LCS score,
or extended substring-string LCS scores 
for pattern $p$ against text $t$.
However, if explicit indices of the nonzeros 
in the output matrix are required,
the index range can be remapped back to $\ang{-m:n \mid 0:m+n}$
by reversing every remapping step in the first phase.

\setlabelitbf
\item[Cost analysis.]
\setlabelit
\nobreakitem[First phase.]
The cost of a recursive step is dominated 
by the application of \thref{th-comp-mmult-comp},
which runs in time $O(m \log m)$.
There are $\bar n$ recursive steps in total,
therefore the first phase runs in time $O(m \log m \cdot \bar n)$.
\item[Second phase.]
For each nonzero, the inverse remapping can be performed recursively
in time $O(\bar n)$.
There are $m$ nonzeros in total, 
therefore the second phase runs in time $O(m \bar n)$.
\item[Total.]
The overall running time is $O(m \log m \cdot \bar n)$.
\end{algorithm}

\algref{alg-gc-slcs} provides, as a special case,
an algorithm for the LCS problem
between a plain string and a GC-string,
running in time $O(m \log m \cdot \bar n)$;
the LCS score can easily be queried 
from the algorithm's output matrix by \thref{th-ps}.
This running time should be contrasted 
with standard LCS algorithms on plain strings,
running in time $O\bigpa{\frac{mn}{\log^2 n} + n}$
\cite{Masek_Paterson:80,Crochemore+:03_SIAM},
and with the PP-hardness of the LCS problem on two GC-strings
\cite{Lifshits_Lohrey:06,Lohrey:12}.

Hermelin et al.\ \cite{Hermelin+:13} and Gawrychowski \cite{Gawrychowski:12}
refined the application of our techniques as follows.
They consider the rational-weighted alignment problem
(equivalently, the LCS or Levenshtein distance problems)
on a pair of GC-strings $a$, $b$
of total compressed length $\bar r = \bar m + \bar n$,
parameterised by the strings' total plain length $r = m + n$.
The algorithm of \cite{Hermelin+:13} 
runs in time $O\bigpa{r \log (r/\bar r) \cdot \bar r}$,
which is improved in \cite{Gawrychowski:12}
to $O\bigpa{r \log^{1/2} (r/\bar r) \cdot \bar r}$.
In both cases, our algorithm of \thref{th-mmult} is used as a subroutine.

%% fully-compressed LCS: grammar, LZ78

%%=-=-=-=-=-=-=-=-=-=-=-=-=-=-=-=-=-=-=-=-=-=-=-=-=-=-=-=-=-=-=-=-=-=-=-=-=-=%%
\mysection{Local subsequence recognition}
\label{s-gc-lsubrec}

\index{problem!subsequence recognition!local}%
The local subsequence recognition problem was introduced in \secref{s-slcs} 
as a special case of the semi-local LCS problem,
and defined in \secref{s-amatch} (\defref{def-lsubrec})
as a variant of the approximate matching problem.
In the context of local subsequence recognition,
a substring of text $t$ is called a \emph{matching substring},
if it contains the pattern $p$ as a subsequence.
A matching substring will be called \emph{minimally matching},
if it is inclusion-minimal,
i.e.\ it has no proper matching substring.

We recall that, depending on the output filtering,
local subsequence recognition can take the following forms:
the \emph{minimum-window subsequence recognition problem},
which asks for the locations of all substrings of $t$ 
that are minimally matching,
and the \emph{fixed-window subsequence recognition problem},
which asks for the locations 
of all the matching substrings of a fixed length $w$.
A combination of these two problems
is the \emph{bounded minimal-window subsequence recognition problem}, 
which asks for the locations 
of all the minimally matching substrings below a fixed length $w$.

Clearly, the output size for the described 
\emph{reporting versions} of these problems
may be exponential in $\bar n$;
therefore, we have to parameterise the running time by the output size,
which we denote by $\mathit{output}$.
We will also consider the \emph{counting version} for each of the above problems,
which, instead of locations of all the matching substrings,
only asks for their overall number.
The running time of the counting algorithms described in this section 
will be the running time of the corresponding reporting algorithm 
with $\mathit{output} = O(1)$.

The minimal-window, fixed-window and bounded minimal-window 
subsequence recognition problems for a GC-text against a plain pattern
have been considered by C\'egielski et al.\ \cite{Cegielski+:06}.
For each problem, they gave an algorithm running in time 
$O(m^2 \log m \cdot \bar n + \mathit{output})$.
In \cite{Tiskin:09_JMS}, we gave an algorithm 
improving the running time for these problems
to $O(m^{1.5} \bar n + \mathit{output})$,
and then in \cite{Tiskin:11_CSR} 
improved it to $O(m \log m \cdot \bar n + \mathit{output})$
by an extended version of \algref{alg-gc-slcs}.
Yamamoto et al.\ \cite{Yamamoto+:11}, using elementary techniques,
gave an even faster algorithm, 
running in time $O(m \cdot \bar n + \mathit{output})$.

We now describe the algorithm of \cite{Tiskin:11_CSR} 
for local subsequence recognition.
Although inferior, both in running time and simplicity, 
to the algorithm of \cite{Yamamoto+:11},
it will serve as a warm-up for the material presented in the next section.

We extend \algref{alg-gc-slcs} (extended string-substring LCS) as follows.
In addition to the extended substring-string matrix $\rP^\subsX_{p,t}$,
we now also make use of the cross-semi-local matrix $\rP_{p;t',t''}$.
This matrix is used for reporting the minimally matching substrings
that are cross-substrings in the current recursive step.
\begin{algorithm}[Local subsequence recognition]
\label{alg-gc-lsubrec}
\setlabelitbf
\nobreakitem[Input:]
plain pattern $p$ of length $m$;
SLP of length $\bar n$, generating text $t$ of length $n$.
\item[Output:]
locations (or count) of minimally matching substrings in $t$.
\item[Description.]
Similarly to \algref{alg-gc-slcs},
index remapping has to be performed in the background
in order to avoid an exponential growth of the indices.
To simplify the exposition, we now assume constant-time index arithmetic,
keeping the index remapping implicit.
\setlabelit
\nobreakitem[First phase.]
Recursion on the input SLP generating $t$.

\setlabelnormal
\item[Recursion base: $n=\bar n=1$.]
As in \algref{alg-gc-slcs}, 
the extended substring-string matrix $\rP^\subsX_{p,t}$
can be computed by \algref{alg-seaweed} (seaweed combing)
on plain strings $p$ and $t$, of length $m$ and $1$ respectively.
String $t$ is matching, if and only if $m=1$ and $t=p$;
in this case, $t$ is also minimally matching.

\item[Recursive step: $n \geq \bar n > 1$.]
Let $t=t' t''$ be the SLP statement defining string $t$.
We run a recursive step of \algref{alg-gc-slcs},
obtaining the extended substring-string seaweed matrix $\rP^\subsX_{p,t}$.
In addition, we obtain the cross-semi-local seaweed matrix $\rP_{p;t',t''}$
by \thref{th-comp-mmult-comp}.

As in \algref{alg-gc-lsubrec}, we run a recursive step of \algref{alg-gc-slcs},
obtaining the extended substring-string seaweed matrix $\rP^\subsX_{\Tp,\Tt}$.
In addition, we obtain the cross-semi-local seaweed matrix $\rP_{\Tp;\Tt',\Tt''}$
by \thref{th-comp-mmult-comp}.
These two subpermutation matrices are typically very sparse:
their index range is of size $O(m+n)$, where $n$ is typically much higher than $m$,
whereas the number of nonzeros in either matrix is at most $2m$
(for matrix $\rP_{\Tp;\Tt',\Tt''}$, it is exactly $m$).

By \thref{th-ps}, a substring $t\ang{i:j}$ is matching, 
if and only if $\rP_{p,t}^{T \Sigma T}(i,j) = 0$,
i.e.\ the point $(i,j)$ in the score matrix $\rH_{p,t}$ 
is not $\gtrless$-dominated by any nonzeros in the seaweed matrix $\rP_{p,t}$.
Recall that a substring $t\ang{i:j}$ is a cross-substring, 
if $i \in \bra{0:n'-1}$, $j \in \bra{n'+1:n}$;
in other words, a cross-substring consists of 
a non-empty suffix of $t'$ and a non-empty prefix of $t''$.
A point $(i,j)$ corresponding to a cross-substring
lies within the cross-semi-local score matrix $\rH_{p;t',t''}$,
and can only be $\gtrless$-dominated by one of the $m$ nonzeros 
in the cross-semi-local seaweed matrix $\rP_{p;t',t''}$.

Let
\begin{gather*}
L = \Bigbrc{
\bigpa{\hi_{0^+},\hj_{0^+}} \ll \bigpa{\hi_{1^+},\hj_{1^+}} 
\ll \dots \ll \bigpa{\hi_{s^-},\hj_{s^-}}}
\end{gather*}
be the $\ll$-chain of all $\gtrless$-maximal nonzeros in $\rP_{p;t',t''}$,
where $s = \abs{L} \leq m$.
If a point is $\gtrless$-dominated 
by any nonzeros in $\rP_{p;t',t''}$,
then it is dominated by some $\gtrless$-maximal nonzero, 
i.e.\ by a point in $L$.
Therefore, a cross-substring $t\ang{i:j}$ is matching, 
if and only if point $(i,j)$ is not $\gtrless$-dominated by any point in $L$.

Consider the set of all points in $\rH_{p;t',t''}$
that are not $\gtrless$-dominated by any point in $L$.
We are interested in the $\gtrless$-minimal points in this set.
Such points form a $\ll$-chain, interleaved with $L$.
Its two endpoints are the degenerate boundary points 
$\bigpa{\hi_{0^+}^-,n'}$, $\bigpa{n',\hj_{s^-}^+}$;
the remaining points form a $\ll$-chain of size $s-1$:
\begin{gather*}
M = \Big\{
\bigpa{\hi_{1^+}^-,\hj_{1^-}^+} \ll
\bigpa{\hi_{2^+}^-,\hj_{2^-}^+} \ll \dots \ll
\bigpa{\hi_{(s-1)^+}^-,\hj_{(s-1)^-}^+}
\Big\}
\end{gather*}
Let $M^\ssub = M \cap \bra{0:n' \mid n':n}$
be the subset of points in $M$ within
the string-cross-substring score matrix $\rH^\ssub_{p;t',t''}$.
A non-degenerate cross-substring $t\ang{i:j}$ is minimally matching, 
if and only if $(i,j) \in M^\ssub$.
The number of such cross-substrings is 
$\abs{M^\ssub} \leq \abs{M} = s-1 \leq m-1$.
\nobreakitem[(End of recursive step)]

\setlabelit
\item[Second phase.]
For every SLP symbol, we now have the relative locations 
of its minimally matching cross-substrings.
Furthermore, every non-trivial substring of $t$ 
corresponds to a cross-substring for some SLP symbol,
under an appropriate transformation of indices.
By another recursion on the structure of the SLP,
it is now straightforward to obtain either the absolute locations, 
or the count of all the minimally matching substrings in $t$.

\setlabelitbf
\item[Cost analysis.]
\setlabelit
\nobreakitem[First phase.]
As in \algref{alg-gc-slcs},
each seaweed matrix multiplication runs in time $O(m \log m)$.
The $\ll$-chains $L$ and $M$
can be obtained in time $O(m)$.
Hence, the running time of a recursive step is $O(m \log m)$.
There are $\bar n$ recursive steps in total,
therefore the whole recursion runs in time $O(m \log m \cdot \bar n)$.

\item[Second phase.]
For every SLP symbol, there are at most $m-1$
minimally matching cross-substrings.
Given the output of the first phase, 
the absolute locations of all minimally matching substrings in $t$ 
can be reported in time $O(m \bar n + \mathit{output})$,
and their count can be obtained in time $O(m \bar n)$.

\item[Total.]
The overall running time 
is $O(m \log m \cdot \bar n + \mathit{output})$ for reporting,
and $O(m \log m \cdot \bar n)$ for counting.
\end{algorithm}

\begin{figure}[tb]
\centering

\subfloat[\label{f-gc-lsubrec-plain}%
Cross-semi-local matrices $\rH_{p;t',t''}$, $\rP_{p;t',t''}$; 
chains $M$, $L$]{%
\makebox[\textwidth]{\includegraphics{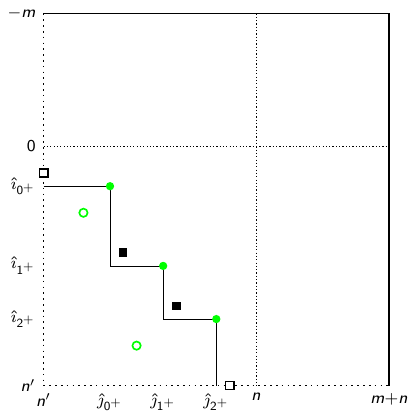}}}

\subfloat[\label{f-gc-lsubrec-seaweed}%
Corresponding seaweed braid]{%
\makebox[\textwidth]{\includegraphics{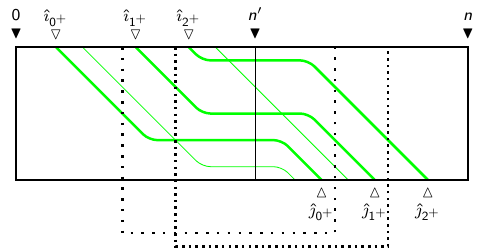}}}
\caption{\label{f-gc-lsubrec}%
A snapshot of \algref{alg-gc-lsubrec} (local subsequence recognition)}
\end{figure}
\begin{example}
\figref{f-gc-lsubrec} shows a snapshot of a recursive step
in the first phase of \algref{alg-gc-lsubrec}.
\sfigref{f-gc-lsubrec-plain} shows 
the cross-semi-local matrices $\rH_{p;t',t''}$ and $\rP_{p;t',t''}$; 
in this particular example, all the nonzeros of the latter
lie within the string-cross-substring matrix $\rP^\ssub_{p,t}$.
\sfigref{f-gc-lsubrec-seaweed} shows the corresponding seaweed braid.
Matrix $\rP_{p;t',t''}$ contains $m=5$ nonzeros, 
shown by green circles in \sfigref{f-gc-lsubrec-plain},
and by green seaweeds in \sfigref{f-gc-lsubrec-seaweed}.
Out of these five nonzeros, three are $\gtrless$-maximal.
These nonzeros form the $\ll$-chain $L$;
they are shown by filled circles 
connected by a thin zig-zag line in \sfigref{f-gc-lsubrec-plain},
and by thicker seaweeds in \sfigref{f-gc-lsubrec-seaweed}.
The remaining two nonzeros are shown by hollow circles
(respectively, by thinner seaweeds).
All points in $\rH_{p;t',t''}$ that are not $\gtrless$-dominated 
by any nonzeros in $\rP_{p;t',t''}$
are located above and to the right of the thin zig-zag line.
The four $\gtrless$-minimal such points are shown 
in \sfigref{f-gc-lsubrec-plain} by black squares;
among these, two are degenerate boundary points (hollow squares),
and the other two form the $\ll$-chain $M$ (filled squares).
Both points of chain $M$ lie within the range $\bra{0:n' \mid n':n}$,
therefore $M^\ssub = M$, and each of the two points in $M$ 
corresponds to a minimally matching non-degenerate cross-substring in $t$.
These two substrings are shown in \sfigref{f-gc-lsubrec-seaweed}
by dotted brackets.

By \thref{th-ps}, a substring in $t$ is matching,
if and only if the corresponding rectangle in the alignment dag
is not pierced by a seaweed 
entering at its left-hand boundary and leaving at its right-hand boundary.
Notice that the bracketed substrings of $t$ in \figref{f-gc-lsubrec-seaweed}
are exactly the two inclusion-minimal non-degenerate cross-substrings 
satisfying this property.
\end{example}

An algorithm for the fixed-window subsequence recognition problem
can be obtained from \algref{alg-gc-lsubrec} as follows.
Substrings $t\ang{i:j}$ of length $w$ 
correspond to points $(i,j)$ lying on the diagonal $j-i=w$
in the semi-local score matrix $\rH_{p,t}$.
Consider the set of all points on this diagonal,
$\gtrless$-dominated by any point in the $\ll$-chain $L$, 
introduced in \algref{alg-gc-lsubrec}.
This set consists of a (not necessarily disjoint) union 
of $s$ diagonal intervals
\begin{gather*}
\label{eq-u}
U = \bigcup_{\Hu \in \ang{0:s}} \Bigbrc{\text{$(i,i+w)$ such that 
$i \in \bigbra{\hi_{\Hu}^+:\hj_{\Hu}^- -w}$}}
\end{gather*}
where any interval of negative length is by convention considered empty.
In every recursive step, the interval endpoints in the set $U$
can be computed in time $O(m)$.

A cross-substring $t\ang{i:i+w}$, $i \in \bra{n'-w+1:n'-1}$, is matching,
if and only if $(i,i+w) \not\in U$.
Therefore, each point corresponding to a cross-substring of $t$
can be reported in constant time.

An algorithm for the bounded minimal-window subsequence recognition problem
can be obtained from \algref{alg-gc-lsubrec} 
by discarding in every recursive step 
the minimally matching cross-substrings of length exceeding $w$.

The overall running time 
of both the above modifications of \algref{alg-gc-lsubrec} is still 
$O(m \log m \cdot \bar n + \mathit{output})$.

%%=-=-=-=-=-=-=-=-=-=-=-=-=-=-=-=-=-=-=-=-=-=-=-=-=-=-=-=-=-=-=-=-=-=-=-=-=-=%%
\section{Edit distance matching}
\label{s-gc-tmatch}

\index{problem!edit distance matching}%
The edit distance matching problem was introduced 
in \secref{s-amatch} (\defref{def-edmatch}).
In the context of edit distance matching,
a substring of text $t$ will be called a \emph{matching substring},
if it has alignment score at least $h$ against pattern $p$
(alternatively, edit distance at most $k$),
where $h$ (respectively, $k$) is a fixed threshold.

\begin{comment}
An alternative, more natural method of filtering the output
is only to count or report the following special type of matching substrings.
We call a matching substring $t[i:j]$ \emph{(strictly) locally optimal},
if it is (strictly) closer to $p$ (in terms of the edit distance)
than the substrings $t[i \pm 1:j]$ and $t[i:j \pm 1]$.
The complete set of matching substrings
can be computed efficiently from the locally optimal ones,
by considering the neighborhoods of their starting and ending positions,
and testing the edit distance of the corresponding substrings from $p$.
\end{comment}

Approximate pattern matching on compressed text 
has been studied by K\"arkk\"ainen et al.\ \cite{Karkkainen+:03}.
For a GC-text of length $\bar n$,
an uncompressed pattern of length $m$, 
and an edit distance threshold $k > 0$,
the (suitably generalised) algorithm of \cite{Karkkainen+:03} 
solves the edit distance matching problem 
in time $O(m \bar n k^2 + \mathit{output})$.
In the special case of LZ78 or LZW compression, 
the running time is reduced to $O(m \bar n k + \mathit{output})$.
Bille et al.\ \cite{Bille+:11} gave an efficient general scheme for adapting 
an arbitrary edit distance matching algorithm to work on a GC-text.
The running time of the resulting algorithms
is parameterised by the text's plain length $n$;
note that $\log n \leq \bar n$.
In particular, when the algorithms 
by Landau and Vishkin \cite{Landau_Vishkin:89}
and by Cole and Hariharan \cite{Cole_Hariharan:02}
are each plugged into this scheme,
the resulting algorithm runs respectively in time
\begin{gather*}
O\bigpa{\bar nmk + \bar n \log_2 n + \mathit{output}}\\
O\bigpa{\bar n(m+k^4) + \bar n \log_2 n + \mathit{output}}
\end{gather*}
In the special case of LZ78 or LZW compression,
Bille et al.\ \cite{Bille+:09} show that it is possible 
to remove the term $\bar n \log_2 n$,
reducing the running time respectively to 
$O\pa{\bar nmk + \mathit{output}}$ and
$O\bigpa{\bar n(m+k^4) + \mathit{output}}$.

Using the techniques of the previous sections,
we now show how the edit distance matching problem on a GC-text
can be solved more efficiently,
for a sufficiently high value of the threshold $k$.
The algorithm extends \algrefs{alg-gc-slcs} (extended string-substring LCS)
and \ref{alg-gc-lsubrec} (local subsequence recognition),
and assumes an edit distance with arbitrary rational weights.
As in \algref{alg-gc-lsubrec}, we assume for simplicity 
the constant-time index arithmetic,
keeping the index remapping implicit.
In the algorithm's description, it will be convenient 
to extend the integer and half-integer interval notation 
to non-integer endpoints as follows:
\begin{gather*}
\bra{u:v} = \bigbra{\ceil{u}:\floor{v}}\qquad
\ang{u:v} = \bigang{\ceil{u}:\floor{v}}
\end{gather*}
for any real $u$, $v$.

\begin{algorithm}[Edit distance matching]
\label{alg-gc-tmatch}
\setlabelitbf
\nobreakitem[Parameters:]
character alignment weights for a mismatch $\wmismatch$ and for a gap $\wgap$, 
where $2\wgap \leq \wmismatch < 0$.
These weights are assumed to be constant rationals.
The weight for a match is fixed as $\wmatch=0$. 
\item[Input:]
plain pattern string $p$ of length $m$;
SLP of length $\bar n$, generating text string $t$ of length $n$;
score threshold $h = -k < 0$.
\item[Output:]
unique starting positions of matching substrings in $t$.
\item[Description.]
\setlabelit
\nobreakitem[First phase.]
Recursion on the input SLP generating $t$.

To reduce the problem to an unweighted LCS score,
we apply the normalisation and the blow-up techniques 
described in \secref{s-weighted}.
Following \defref{def-normalised},
we introduce the normalised weights
$\wmatch^* = 1$,
$\wmismatch^* = \frac{\wmismatch-2\wgap}{-2\wgap}$,
$\wgap^* = 0$.
Let $\wmismatch^* = \tfrac{\mu}{\nu} < 1$,
where $\mu$, $\nu$ are positive natural numbers.
We transform strings $p$, $t$ into
the corresponding blown-up strings $\Tp$, $\Tt$
of length $\Tm = \nu m$, $\Tn = \nu n$, respectively.

\setlabelnormal
\item[Recursion base: $n=\bar n=1$, $\Tn = \nu$.]
The extended substring-string seaweed matrix $\rP^\subsX_{\Tp,\Tt}$
can be computed by \algref{alg-seaweed} (seaweed combing)
on plain strings $\Tp$ and $\Tt$, of length $\nu m$ and $\nu$ respectively.
This matrix can be used to query 
the LCS score $\rH_{\Tp,\Tt}(0,\nu)$ between $\Tp$ and $\Tt$.
String $t$ is matching, if and only if the corresponding 
weighted alignment score is above the threshold:
$\cH_{p,t}(0,1) \geq h$.

\item[Recursive step: $n \geq \bar n > 1$, $\Tn = \nu n$.]
Let $t=t' t''$ be the SLP statement defining string $t$.
We have $\Tt=\Tt' \Tt''$ for the corresponding blown-up strings.

As in \algref{alg-gc-lsubrec} (local subsequence recognition), 
we obtain recursively
the extended substring-string seaweed matrix $\rP^\subsX_{\Tp,\Tt}$
and the cross-semi-local seaweed matrix $\rP_{\Tp;\Tt',\Tt''}$
by \thref{th-comp-mmult-comp}.
These two subpermutation matrices are typically very sparse:
$\rP^\subsX_{\Tp,\Tt}$ contains at most $2\Tm = 2\nu m$ nonzeros,
and $\rP_{\Tp;\Tt',\Tt''}$ exactly $\Tm = \nu m$ nonzeros.

Similarly to \algref{alg-gc-lsubrec}, the matching substrings in $t$
can now be determined by the nonzeros 
of the cross-semi-local seaweed matrix $\rP_{\Tp;\Tt',\Tt''}$.
However, this time it is no longer sufficient
to consider just its $\gtrless$-maximal nonzeros;
we now have to consider all its $\Tm$ nonzeros.
Let us denote the indices of these nonzeros,
in increasing order independently for each dimension, by
\begin{gather}
\label{eq-nz-order}
\hi_{0^+} < \hi_{1^+} < \ldots < \hi_{\Tm^-}\qquad
\hj_{0^+} < \hj_{1^+} < \ldots < \hj_{\Tm^-}
\end{gather}
These two index sequences define an $\Tm\times\Tm$ 
non-contiguous permutation submatrix of nonzeros in $\rP_{\Tp;\Tt',\Tt''}$:
\begin{gather}
\label{eq-amatch-l}
P(\Hs,\Ht) = \rP_{\Tp;\Tt',\Tt''}(\hi_{\Hs},\hj_{\Ht})
\end{gather}
for all $\Hs,\Ht \in \ang{0:\Tm}$.

Index sequence $\hi_{\Hs}$ (respectively, $\hj_{\Ht}$)
partitions the range $\bra{-\Tm:\Tn'}$  (respectively, $\bra{\Tn':\Tm+\Tn}$)
into $\Tm+1$ disjoint non-empty intervals of varying lengths:
\begin{alignat*}{2}
{}
&\bra{-\Tm:\Tn'} &&= \bra{\hi_{0^-}:\hi_{0^+}} \uplus 
  \bra{\hi_{0^+}:\hi_{1^+}} \uplus\cdots\uplus \bra{\hi_{\Tm^-}:\hi_{\Tm^+}}\\
&\bra{\Tn':\Tm+\Tn} &&= \bra{\hj_{0^-}:\hj_{0^+}} \uplus 
  \bra{\hj_{0^+}:\hj_{1^+}} \uplus\cdots\uplus \bra{\hj_{\Tm^-}:\hj_{\Tm^+}}
\end{alignat*}
where the boundary indices are defined as
\begin{gather*}
\textstyle
\hi_{0^-}   = (-\Tm)^- \qquad
\hi_{\Tm^+} = (\Tn')^+ \qquad 
\hj_{0^-}   = (\Tn')^- \qquad
\hj_{\Tm^+} = (\Tm + \Tn)^+
\end{gather*}
(note that we are making use of the interval notation with non-integer endpoints).
Therefore, we have a partitioning 
of the cross-semi-local score matrix $\rH_{\Tp;\Tt',\Tt''}$
into $(\Tm+1)^2$ disjoint non-empty rectangular 
\emph{$H$-blocks} of varying dimensions.
Consider an arbitrary $H$-block
\begin{gather}
\label{eq-amatch-block}
\textstyle
\rH_{\Tp;\Tt',\Tt''}\bigbra{\hi_{u^-}:\hi_{u^+} \mid \hj_{v^-}:\hj_{v^+}}
\end{gather}
where $u, v \in \bra{0:\Tm}$.
By definition of the index sequences \eqref{eq-nz-order},
a nonzero of matrix $\rP_{\Tp;\Tt',\Tt''}$ can only occur
at a meeting point of four different $H$-blocks.
Therefore, matrix $\rP_{\Tp;\Tt',\Tt''}$ 
is identically zero inside every $H$-block:
we have $\rP_{\Tp;\Tt',\Tt''}(\hi,\hj) = 0$ for all 
$\hi \in \bigang{\hi_{u^-}:\hi_{u^+}}$,
$\hj \in \bigang{\hj_{v^-}:\hj_{v^+}}$,
given any fixed index pair $u, v \in \bra{0:\Tm}$
(again using interval notation with non-integer endpoints).
Therefore, given a fixed $H$-block,
all its points are $\gtrless$-dominated 
by the same set of nonzeros in matrix $\rP_{\Tp;\Tt',\Tt''}$
(and hence also in $\rP_{\Tp,\Tt}$).
The number of nonzeros in this set is
\begin{gather}
\label{eq-amatch-d}
d = \rP_{\Tp;\Tt',\Tt''}^{T\Sigma T}(\hi_{u^+}^-,\hj_{v^-}^+) =
\rP_{\Tp,\Tt}^{T\Sigma T}(\hi_{u^+}^-,\hj_{v^-}^+)
\end{gather}
where the reference point within an $H$-block is chosen arbitrarily 
as its bottom-left ($\gtrless$-minimal) point $(\hi_{u^+}^-,\hj_{v^-}^+)$.
Since the value of $d$ is constant across the $H$-block,
by \thref{th-ps} all its entries have identical value: 
\begin{gather}
\label{eq-amatch-md}
\rH_{\Tp;\Tt',\Tt''}(i,j)=\Tm-d
\end{gather}
for all $i \in \bra{\hi_{u^-}:\hi_{u^+}}$, $j \in \bra{\hj_{v^-}:\hj_{v^+}}$.

We now switch our focus from the blown-up strings $\Tp$, $\Tt'$, $\Tt''$
back to the original strings $p$, $t'$, $t''$.
The partitioning of the LCS score matrix $\rH_{\Tp;\Tt',\Tt''}$ into $H$-blocks
induces a partitioning of the alignment score matrix $\cH_{p;t',t''}$
into $(\tilde m+1)^2$ disjoint rectangular 
\emph{$\cH$-blocks} of varying dimensions.
The $\cH$-block corresponding to $H$-block \eqref{eq-amatch-block} is
\begin{gather}
\label{eq-amatch-block-ind}
\textstyle
\cH_{p;t',t''} \Bigbra{
\frac{\hi_{u^-}}{\nu}:\frac{\hi_{u^+}}{\nu} \mid
\frac{\hj_{v^-}}{\nu}:\frac{\hj_{v^+}}{\nu}}
\end{gather}
Note that, 
although an $H$-block \eqref{eq-amatch-block} is by definition non-empty,
the corresponding $\cH$-block \eqref{eq-amatch-block-ind} may be empty.
This happens whenever either of the intervals
$\Bigbra{\frac{\hi_{u^-}}{\nu}:\frac{\hi_{u^+}}{\nu}}$ or
$\Bigbra{\frac{\hj_{v^-}}{\nu}:\frac{\hj_{v^+}}{\nu}}$
contains no integer points
(i.e.\ either of the intervals
$\bigbra{\hi_{u^-}:\hi_{u^+}}$ or $\bigbra{\hj_{v^-}:\hj_{v^+}}$
contains no multiples of $\nu$).
However, for ease of notation, we will assume
that all the $\cH$-blocks are non-empty.

Although all entries within an $H$-block \eqref{eq-amatch-block} are constant,
the entries within the corresponding $\cH$-block \eqref{eq-amatch-block-ind}
will typically vary.
By \eqref{eq-amatch-md} and \defref{def-normalised}, we have
\begin{gather}
\label{eq-amatch-dd}
\textstyle
\cH_{p;t',t''}(i,j) = 
\frac{\Tm-d}{\nu} \cdot (-2\wgap) + (m+j-i) \cdot \wgap
\end{gather}
where 
$i \in \Bigbra{\frac{\hi_{u^-}}{\nu}:\frac{\hi_{u^+}}{\nu}}$,
$j \in \Bigbra{\frac{\hj_{v^-}}{\nu}:\frac{\hj_{v^+}}{\nu}}$.
Recall that $\wgap < 0$. 
Therefore, the score within an $\cH$-block 
is maximised when $j-i$ is minimised, so the maximum score is attained 
by the block's bottom-left (i.e.\ $\gtrless$-minimal) entry 
$\cH_{p;t',t''}\Bigpa{
  \bigfloor{\frac{\hi_{u^+}}{\nu}},\bigceil{\frac{\hj_{v^-}}{\nu}}}$.

We are interested in the bottom-left entries of all the $\cH$-blocks,
since that is where block maxima are attained.
The leftmost column and the bottom row of these entries
(respectively 
$\cH_{p;t',t''}\Bigpa{\bigfloor{\frac{\hi_{u^+}}{\nu}},n'}$ and
$\cH_{p;t',t''}\Bigpa{n',\bigceil{\frac{\hj_{v^-}}{\nu}}}$ for all $u$, $v$),
lie on the bottom-left boundary of matrix $\cH_{p;t',t''}$;
all such boundary points are degenerate.
The remaining block maxima form an $\Tm \times \Tm$ non-contiguous submatrix
\begin{gather*}
\textstyle H(u,v) = 
\cH_{p;t',t''}\Bigpa{\bigfloor{\frac{\hi_{u^+}}{\nu}},\bigceil{\frac{\hj_{v^-}}{\nu}}}
\end{gather*}
where $u \in \bra{0:\Tm-1}$, $v \in \bra{1:\Tm}$.

Let
\begin{gather*}
\textstyle H^\ssub = 
\bigpa{H(u,v) \text{ such that } 
\Bigpa{\bigfloor{\frac{\hi_{u^+}}{\nu}},\bigceil{\frac{\hj_{v^-}}{\nu}}} 
\in \bra{0:n' \mid n':n}}
\end{gather*}
be the submatrix of $H$ within
the string-cross-substring score matrix $\rH^\ssub_{p;t',t''}$.
Since matrix $\cH_{p;t',t''}$ is anti-Monge,
its submatrices $H$ and $H^\ssub$ are also anti-Monge.

We now need to obtain the row maxima of matrix $H^\ssub$.
Let
\begin{gather*}
\textstyle N(u,v) = \frac{\nu}{2\wgap} H^\ssub(u,v) = {}
\why{by \eqref{eq-amatch-dd}, \eqref{eq-amatch-d}, \eqref{eq-amatch-l}}\\ 
\textstyle P^{T\Sigma T}(u,v) - 
\Tm + \frac{\nu \Bigpa{m + 
\bigceil{\frac{\hj_{v^-}}{\nu}} - \bigfloor{\frac{\hi_{u^+}}{\nu}}} \wgap}{2\wgap}
\end{gather*}
Since $\wgap < 0$, the problem of finding row maxima of $H^\ssub$
is equivalent to finding row minima of matrix $N$,
or, equivalently, column minima of the transpose matrix $N^T$.
This matrix (and therefore $N$ itself) is subunit-Monge: we have
\begin{gather*}
N^T(v,u) = N(u,v) = P^{T\Sigma}(v,u) + b(u) + c(v)
\end{gather*}
where
\begin{gather*}
b(u) = - \frac{\nu \bigfloor{\frac{\hi_{u^+}}{\nu}} \cdot \wgap}{2\wgap-\wmatch}\qquad
c(v) = - \Tm + \frac{\nu \Bigpa{m + \bigceil{\frac{\hj_{v^-}}{\nu}}} \wgap}{2\wgap-\wmatch}
\end{gather*}
Therefore, the column minima of $N^T$ 
can be found by either \lmref{lm-rowmin-loglog} or \lmref{lm-rowmin} 
(replacing row minima with column minima by symmetry).

The unique starting positions of matching non-degenerate cross-substrings in $t$
can now be obtained as indices of row maxima in $H^\ssub$
scoring above the threshold $h$.
\nobreakitem[(End of recursive step)]

\setlabelit
\item[Second phase.]
As in \algref{alg-gc-lsubrec}, 
substituting ``matching'' for ``minimally matching''.

\setlabelitbf
\item[Cost analysis.]
\setlabelit
\nobreakitem[First phase.]
As in \algref{alg-gc-slcs}, each seaweed matrix multiplication 
runs in time $O(\Tm \log \Tm) = O(m \log m)$.
The algorithms of \lmref{lm-rowmin-loglog} and \lmref{lm-rowmin} 
run in time $O(\Tm \log\log \Tm) = O(m \log\log m)$ and $O(\Tm) = O(m)$,
respectively.
Hence, the running time of a recursive step is $O(m \log m)$.
There are $\bar n$ recursive steps in total,
therefore the whole recursion runs in time $O(m \log m \cdot \bar n)$.

\item[Second phase.]
As in \algref{alg-gc-lsubrec},
the absolute unique matching positions of all matching substrings in $t$ 
can be found in time $O(m \bar n)$.

\item[Total.]
The overall running time 
is $O(m \log m \cdot \bar n)$.
\end{algorithm}
\algref{alg-gc-tmatch} improves on the algorithm of \cite{Karkkainen+:03}
for $k = \omega\bigpa{(\log m)^{1/2}}$ in the case of general GC-compression,
and $k = \omega(\log m)$ in the case of LZ78 or LZW compression.
\algref{alg-gc-tmatch} also improves 
on the algorithms of \cite{Bille+:09,Bille+:11}
for $k = \omega\bigpa{(m \log m)^{1/4}}$,
in the case of both general GC-compression and LZ78 or LZW compression.

\begin{figure}[tb]
\centering

\subfloat[\label{f-gc-tmatch-plain}%
Cross-semi-local matrices $\rH_{\Tp;\Tt',\Tt''}$, $\rP_{\Tp;\Tt',\Tt''}$;
submatrices $H$, $P$]{%
\makebox[\textwidth]{\includegraphics{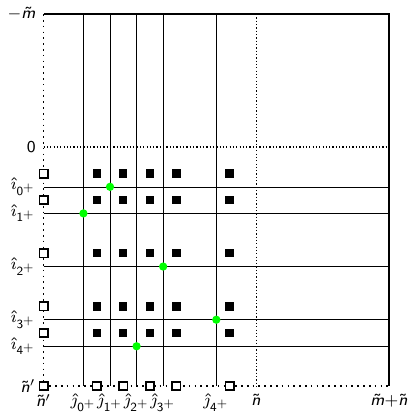}}}

\subfloat[\label{f-gc-tmatch-seaweed}Corresponding seaweed braid]{%
\makebox[\textwidth]{\includegraphics{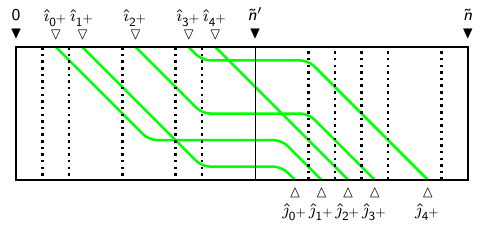}}}
\caption{\label{f-gc-tmatch}%
A snapshot of \algref{alg-gc-tmatch} (threshold approximate matching)}
\end{figure}
\begin{example}
\figref{f-gc-tmatch} shows a snapshot of a recursive step
in the first phase of \algref{alg-gc-tmatch},
which is assumed to run on the same input as in \figref{f-gc-lsubrec}.
\sfigref{f-gc-tmatch-plain} shows 
the cross-semi-local matrices $\rH_{\Tp;\Tt',\Tt''}$ and $\rP_{\Tp;\Tt',\Tt''}$; 
\sfigref{f-gc-tmatch-seaweed} shows the corresponding seaweed braid.
The cross-semi-local matrix and the seaweed braid
are identical to the ones in \figref{f-gc-lsubrec}.
In \sfigref{f-gc-tmatch-plain}, submatrix $P$
and the partitioning of $\rH_{\Tp;\Tt',\Tt''}$ 
into $H$-blocks is shown by a grid of thin solid lines.
The bottom-left entry in each corresponding $\cH$-block 
(note that it need not be bottom-left in the $H$-block)
is marked by a black square;
among these, some are degenerate boundary points (hollow squares),
and the others form submatrix $H$.
In \sfigref{f-gc-tmatch-seaweed}, the corresponding substring boundaries
are shown by dotted lines.
All elements of $H$ lie within the range $\bra{0:\Tn' \mid \Tn':\Tn''}$,
therefore $H^\ssub = H$,
and its row maxima scoring above the threshold $h$
correspond to unique starting positions 
of matching non-degenerate cross-substrings in $t$.
\end{example}

%%===========================================================================%%

%%===========================================================================%%
\mychapter{The transposition network method}
\label{c-network}

From \chapref{c-semi}, we already know
that the structure of semi-local string comparison
can be expressed in two equivalent forms:
the distance multiplication monoid of simple unit-Monge matrices,
and the monoid of seaweed braids.
In this chapter, we show that this structure can be seen
in yet another alternative form,
based on the classical concept of comparison networks.

%%=-=-=-=-=-=-=-=-=-=-=-=-=-=-=-=-=-=-=-=-=-=-=-=-=-=-=-=-=-=-=-=-=-=-=-=-=-=%%
\mysection{Seaweed combing as a transposition network}
\label{s-network}

Comparison networks were first considered as a computation model 
by Batcher \cite{Batcher:68} (see also \cite{Cormen+:01,AlHaj_Batcher:11}).
\begin{definition}
\label{def-network-comp}
\index{network!comparison}%
A \emph{circuit} represents a computation as a dag (directed acyclic graph).
The internal nodes of a circuit are labeled by elementary operations
on values, which are passed along the edges;
source and sink nodes represent the inputs and outputs, respectively.
A \emph{comparator node} (or simply \emph{comparator}) 
is a node of indegree and outdegree 2,
which sorts its two operands in increasing order.
In other words, a comparator node 
compares the operands on the incoming edges, 
and returns each of the minimum and the maximum operand 
on a prescribed outgoing edge.
A \emph{comparison network} is a circuit 
where all internal nodes are comparator nodes.
\end{definition}

The most well-studied types of comparison networks 
are the ones that either sort their inputs, 
or merge two disjoint subsets of inputs.
In particular, Batcher \cite{Batcher:68} 
gave classical merging networks with $O(n \log n)$ comparators, 
and sorting networks with $O(n \log^2 n)$ comparators.
Ajtai et al.\ \cite{Ajtai+:83_STOC,Ajtai+:83_Combinatorica} 
gave an asymptotically optimal sorting network with $O(n \log n)$ comparators;
their construction was subsequently simplified 
by Paterson \cite{Paterson:90} and by Seiferas \cite{Seiferas:09}.

Comparison networks are usually visualised by \emph{wire diagrams}
(also known as \emph{Knuth diagrams}),
where the values propagate across the network
along a set of parallel \emph{wires}.
Every comparator is represented by a directed line segment, 
drawn orthogonally between two (not necessarily adjacent) wires.
The order in which a comparator returns the minimum and the maximum output 
is consistent across all the comparators in the network.
The most common convention on wire diagrams 
(adopted e.g.\ by Knuth \cite{Knuth:98_3}) 
is to draw the wires horizontally, directed from left to right;
sometimes (e.g.\ in \cite{Paterson:90}), 
they are drawn vertically, directed from top to bottom.
In our setting, it will be convenient to draw the wires diagonally,
directed from top-left to bottom-right.
Comparator segments will be directed so that 
the minimum output is returned on the bottom-left, 
and the maximum on the top-right.

We will be dealing exclusively 
with the following restricted type of comparison network.
\begin{definition}
\label{def-network-trans}
\index{network!transposition}%
A comparison network is called a \emph{transposition network},
if in its wire diagram, all the comparisons are between adjacent wires.
\end{definition}

Every grid-diagonal dag (as in \defref{def-gd-dag})
can be associated with a transposition network as follows.
\begin{definition}
\label{def-gd-network}
Let $G$ be a grid-diagonal dag. 
Its corresponding transposition network $\mathcal{N}(G)$ 
has a diagram of $m+n$ wires, laid over dag $G$ diagonally,
so that every horizontal and every vertical edge in $G$ 
is crossed by exactly one wire.
Hence, every cell is crossed diagonally by exactly two adjacent wires.
The cell contains a comparator between these two wires,
if and only if the cell's diagonal edge has weight $0$.
\end{definition}

\begin{figure}[tb]
\centering
\subfloat[\label{f-align1-seaweeds}Alignment dag with a seaweed braid]{%
\includegraphics{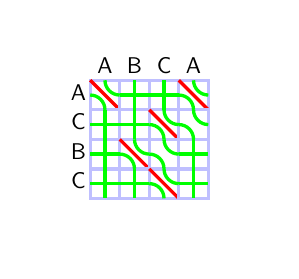}}
\qquad
\subfloat[\label{f-align1-network}Corresponding transposition network]{%
\includegraphics{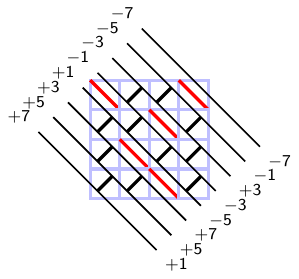}}
\caption{\label{f-align1}
Seaweed combing by a transposition network}
\end{figure}

\begin{example}
\figref{f-align1} illustrates \defref{def-gd-network}
on an alignment dag for strings $a=\textsf{``ACBC''}$, $b=\textsf{``ABCA''}$.
\sfigref{f-align1-seaweeds} shows the alignment dag $\rG_{a,b}$.
Following our usual colour conventions, 
the diagonal edges of weight $1$ are shown in red,
and the diagonal edges of weight $0$ are omitted.
\sfigref{f-align1-seaweeds} also shows the reduced seaweed braid
laid over the alignment dag $\rG_{a,b}$.
\sfigref{f-align1-network} shows in black
the corresponding transposition network, laid diagonally over $\rG_{a,b}$.
By \defref{def-gd-network}, a cell contains a comparator, 
if and only if it does not contain a red diagonal edge.
\end{example}

\begin{comment}
\begin{figure}[tb]
\centering
%
\includegraphics{f-network-diamond}
%
\caption{\label{f-network-diamond} 
The diamond network}
\end{figure}

A full-mismatch alignment dag corresponds to the \emph{diamond network},
shown in \figref{f-network-diamond}.
A full-match alignment dag corresponds to the trivial transposition network
without any comparators.
\end{comment}

Let us denote the input and output arrays
of a transposition network $\mathcal{N}(\rG_{a,b})$
by $x\ang{-m:n}$ and $y\ang{0:m+n}$, respectively.
Assuming all input values of the network are distinct,
each value traces a well-defined path through the network.
We write $\pi(\hi)=\hj$, if the input $x(\hi)$ ends up as the output $y(\hj)$.

Although a transposition network $\mathcal{N}(\rG_{a,b})$ 
is in general neither merging nor sorting,
it is still useful to consider its operation on a certain type of input.
Let the input array $x$ be anti-sorted (i.e.\ sorted in decreasing order).
It turns out that, given such an input, 
the operation of the network $\mathcal{N}(\rG_{a,b})$ emulates seaweed combing,
and the resulting bijection $\pi$ coincides with the bijection 
defined by the semi-local seaweed matrix $\rP_{a,b}$.
\begin{theorem}
\label{th-network}
Consider an alignment dag $\rG_{a,b}$ 
and the corresponding transposition network $\mathcal{N}(\rG_{a,b})$.
For all $\hi \in \ang{-m:n}$, $\hj \in \ang{0:m+n}$,
we have $\pi(\hi)=\hj$, if and only if $\rP_{a,b}(\hi,\hj)=1$.
\end{theorem}
\begin{proof}
Consider any pair of input values in $\mathcal{N}(\rG_{a,b})$.
Initially, these two values are anti-sorted.
Assume that the two values in question meet at some comparator.
If by that point their paths have not yet crossed,
then they arrive at the comparator anti-sorted, and leave it sorted,
so the two paths cross for the first time.
Otherwise, the two paths have previously crossed once,
therefore the values arrive and leave the comparator sorted,
and their paths never cross again.

Since the comparators are located in mismatch cells of $\rG_{a,b}$,
the described operation of each comparator 
is equivalent to seaweed combing (\algref{alg-seaweed}).
Therefore, the paths of the values correspond 
to the seaweeds in the resulting reduced seaweed braid,
and the output of $\mathcal{N}(\rG_{a,b})$ 
is described in terms of the seaweed matrix $\rP_{a,b}$ as claimed.
\end{proof}

\begin{example}
\figref{f-align1} illustrates the connection 
between seaweed braids and transposition networks,
as described by \thref{th-network}.
\sfigref{f-align1-seaweeds} shows the reduced seaweed braid
obtained by seaweed combing (\algref{alg-seaweed}),
laid over the alignment dag $\rG_{a,b}$.
\sfigref{f-align1-network} shows 
an anti-sorted input array of distinct values,
and the corresponding output array.
Each value in \sfigref{f-align1-network} 
traces a path through the network; 
the paths are not shown explicitly, 
but can be reconstructed by running the network ``by hand''.
Every path turns out to be identical 
to the layout of the corresponding seaweed in \sfigref{f-align1-seaweeds}.
\end{example}

Extending \thref{th-network}, it is natural to consider the situation 
where the input values to a transposition network $\mathcal{N}(\rG_{a,b})$
are anti-sorted, but not necessarily distinct.
An extreme case of that is an anti-sorted \emph{binary input}:
a sequence of ones, followed by a sequence of zeros.
While in this case, the information on semi-local LCS scores is lost,
it turns out that the output still contains sufficient information
to obtain the ordinary (global) LCS score. 
\begin{theorem}
\label{th-network-bin}
Consider an alignment dag $\rG_{a,b}$ 
and the corresponding semi-local seaweed matrix $\rP_{a,b}$.
Let the transposition network $\mathcal{N}(\rG_{a,b})$ 
operate on an anti-sorted input array $x\ang{-m:n}$,
consisting of $m$ $1$-values and $n$ $0$-values:
\begin{gather*}
x(\hi) =
\begin{cases}
1 & \text{if $\hi \in \ang{-m:0}$} \\
0 & \text{if $\hi \in \ang{0:n}$} 
\end{cases}
\end{gather*}
Let $y\ang{0:m+n}$ be the output array of $\mathcal{N}(\rG_{a,b})$.
Then we have
\begin{gather*}
\textstyle
\lcs(a,b) = \sum_{\hj \in \ang{0:n}} y(\hj) = m - \sum_{\hj \in \ang{n:m+n}} y(\hj)
\end{gather*}
\end{theorem}
\begin{proof}
We have
\begin{gather*}
\lcs(a,b) = n - \rP_{a,b}^\Sigma(0,n) = {}
\why{\thref{th-ps}; definition of $\Sigma$} \\
\textstyle
n - \sum_{(\hi,\hj) \in \ang{0:n \mid 0:n}} \rP_{a,b}(\hi,\hj) = {}
\why{$\rP_{a,b}$ is a permutation matrix} \\
\textstyle
n - \bigpa{n - \sum_{(\hi,\hj) \in \ang{-m:0 \mid 0:n}} \rP_{a,b}(\hi,\hj)} = {}
\why{cancellation} \\
\textstyle
\sum_{(\hi,\hj) \in \ang{-m:0 \mid 0:n}} \rP_{a,b}(\hi,\hj) = {}
\why{$\rP_{a,b}$ is a permutation matrix} \\
\textstyle
\bigabs{\bigbrc{(\hi,\hj) \in \ang{-m:0 \mid 0:n}: \rP_{a,b}(\hi,\hj) = 1}} = {}
\why{\thref{th-network}} \\
\textstyle
\bigabs{\bigbrc{(\hi,\hj) \in \ang{-m:0 \mid 0:n}: \pi(\hi) = \hj}} = {}
\why{definition of $x$} \\
\textstyle
\bigabs{\bigbrc{\hj \in \ang{0:n} : x(\pi^{-1}(\hj)) = 1}} = {}
\why{definition of $\pi$} \\
%\why{$x(\hi) = y (\pi(\hi))$} \\
%
\textstyle
\bigabs{\bigbrc{\hj \in \ang{0:n} : y(\hj) = 1 }} = {}
\why{$y$ binary} \\
\textstyle
\sum_{\hj \in \ang{0:n}} y(\hj) = {}
\why{$\sum_{\hj \in \ang{0:m+n}} y(\hj) = \sum_{\hi \in \ang{-m:n}} x(\hi) = m$} \\
\textstyle
m - \sum_{\hj \in \ang{n:m+n}} y(\hj)
\end{gather*}
\end{proof}
\begin{figure}[tb]
\centering
\subfloat[\label{f-align1-seaweeds-bin}Alignment dag with a partitioned seaweed braid]{%
\includegraphics{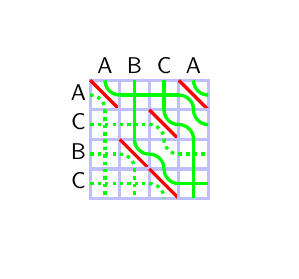}}
\qquad
\subfloat[\label{f-align1-network-bin}Corresponding transposition network with binary input]{%
\includegraphics{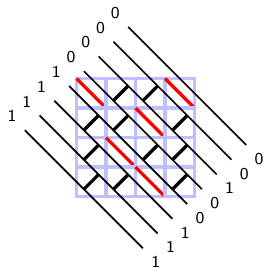}}
\caption{\label{f-align1-bin}
Binary seaweed combing by a transposition network}
\end{figure}

\begin{example}
\figref{f-align1-bin} illustrates \thref{th-network-bin}
on the same pair of strings as \figref{f-align1}.
\sfigref{f-align1-seaweeds-bin} shows 
the seaweed braid of \sfigref{f-align1-seaweeds},
laid over the alignment dag $\rG_{a,b}$.
The seaweeds originating at the top 
(respectively, the left-hand side) of the dag
are shown by solid (respectively, dotted) lines.
\sfigref{f-align1-network-bin} shows
the transposition network of \sfigref{f-align1-network},
laid over the alignment dag $\rG_{a,b}$.
The network is given an anti-sorted binary input.
The path of each $0$-value (respectively, $1$-value)
corresponds to a solid (respectively, dotted) seaweed.
We have 
\begin{gather*}
\textstyle
\lcs(a,b) = \sum_{\hj \in \ang{0:4}} y(\hj) = 3\\
\textstyle
\lcs(a,b) = 4 - \sum_{\hj \in \ang{4:8}} y(\hj) = 4-1 = 3
\end{gather*}
as claimed by \thref{th-network-bin}.
\end{example}

\begin{comment}
%%=-=-=-=-=-=-=-=-=-=-=-=-=-=-=-=-=-=-=-=-=-=-=-=-=-=-=-=-=-=-=-=-=-=-=-=-=-=%%
\mysection{Sparse LCS and semi-local LCS}
\label{s-sparse}

Standard LIS algorithm
(see e.g.\ \cite{Dijkstra:80,Bespamyatnikh_Segal:00}). $O(n \log n)$.

Sparse LCS. Equivalent to maximum non-crossing matching [Malucelli+:93].

Sparse semi-local LCS. Grid partitioning.

\end{comment}

%%=-=-=-=-=-=-=-=-=-=-=-=-=-=-=-=-=-=-=-=-=-=-=-=-=-=-=-=-=-=-=-=-=-=-=-=-=-=%%
\mysection{Parameterised LCS}
\label{s-parameterised}

An algorithm's complexity is most commonly defined 
to be a function of a single argument: the input size.
However, in the pursuit of efficiency, algorithms may also be designed 
to be sensitive to various other parameters of the input.
In the context of string comparison, the two most relevant parameters are:
\begin{itemize}
\item the input strings' alignment score; 
we consider primarily the LCS score $\lambda = \mathit{lcs}(a,b)$;
\item the input strings' edit distance; 
we consider primarily the LCS distance
$\kappa = \mathit{dist}_\mathit{LCS}(a,b) = m + n - 2\lambda$.
\end{itemize}
In this section, we study algorithms for the LCS problem
that are sensitive to these parameters.
We consider \emph{low-similarity} and \emph{high-similarity} LCS algorithms.
The running times of such algorithms 
are parameterised respectively by $\lambda$ and $\kappa$,
so that advantage can be taken of the low value of a parameter.
More generally, one can also use weighted alignment scores 
or weighted edit distances (e.g.\ the Levenshtein distance) as parameters.

As we aim for algorithms that, for a low value of the parameter,
run substantially faster than $\Theta(mn)$,
we cannot afford to perform 
all the $mn$ pairwise comparisons of characters from each string.
As in \secref{s-micro}, we assume a character comparison model 
that allows comparison outcomes ``equal'', ``less than'' and ``greater than'',
so that the ``missing'' comparisons can be obtained by transitivity,
and algorithms with running time $o(mn)$ become possible.

For simplicity, we ignore the trivial cases 
$\lambda = 0$ (the two input strings have no characters in common)
and $\kappa = 0$ (the two input strings are identical).
As usual, we denote the alphabet size by $\sigma$.
Without loss of generality, we assume $m \leq n$.
% CHECK AND CORRECT WHERE NECESSARY

\paragraph{Low-similarity comparison (sensitive to low $\lambda$).}
In this type of comparison, 
the overall number of matching character pairs will be low.
To locate these matching pairs effectively, the input strings $a$ and $b$
are preprocessed into a data structure that allows efficient queries 
defined by the \emph{string identification problem}
\cite{Aho+:76_JACM,Hsu_Du_84_JCSS,Rick:00_IPL}.
%
% EXPLAIN
%
The preprocessing builds a binary search tree on each input string,
and returns a partitioning of both $a$ and $b$ into character equality classes.
%
% CHECK
%
This preprocesing procedure runs in time $m \log\sigma$ 
(respectively, $n \log\sigma$).

% sort and cross-match characters between $a$ and $b$
% time $n \log\sigma$

After the preprocessing, the low-similarity LCS problem can be solved 
by one of the algorithms due to Hirschberg \cite{Hirschberg:77_JACM}, 
Hsu and Du \cite{Hsu_Du_84_JCSS} (see also Apostolico \cite{Apostolico:87_IPL}),
Apostolico and Guerra \cite{Apostolico_Guerra:87_Algorithmica}.
All these algorithms run in time $O(n\lambda)$.
Apostolico, Browne and Guerra \cite{Apostolico+:92_TCS} proposed another algorithm
that requires no preprocessing, 
and runs in time $O(n \lambda \log \sigma)$ and linear space.

\paragraph{High-similarity comparison (sensitive to low $\kappa$).}
In this type of comparison, 
the overall number of matching character pairs may be as high as $mn$.
However, an efficient high-similarity LCS algorithm
may not need to look at all these matches;
speaking informally, a good algorithm 
``will know where to look for relevant matches''.
In fact, it is sufficient to consider
$O(n \cdot \kappa)$ character matches overall.
In contrast to low-similarity comparison, 
there is no need for preprocessing the input strings.

Efficient high-similarity LCS algorithms have been given by
Ukkonen \cite{Ukkonen:85_IC},
Myers \cite{Myers:86_Algorithmica},
Wu et al.\ \cite{Wu+:90_IPL}.
All these algorithms run in time $O(n \cdot \kappa)$.
Apostolico, Browne and Guerra \cite{Apostolico+:92_TCS} proposed another algorithm
that runs in time $O(n \cdot \kappa)$ and linear space.

% (in main alg: from each block must leave unsplit at most $O(1)$ cells on average
% hence find splitting match in time $O(1)$ by linear search)

\paragraph{Flexible comparison (sensitive to both low $\lambda$ and low $\kappa$).}
This type of comparison can be achieved
by preprocessing the input string as for low-similarity comparison,
and then running both comparison types alongside each other.
However, dedicated flexible comparison algorithms have also been proposed.
In particular, the flexible LCS algorithm by Hirschberg \cite{Hirschberg:77_JACM}
runs in time $O(\lambda\kappa \log n)$,
and one by Rick \cite{Rick:00_IPL}
in time $O(\lambda\kappa)$ and linear space. 

We now describe an algorithm based on the comparison network method.
The algorithm is sensitive to both parameters $\lambda$ and $\kappa$,
providing flexible LCS computation 
efficient in both the low- and the high-similarity case.
Our algorithm matches existing flexible-LCS algorithms in running time.
We call it the \emph{waterfall algorithm}.

We preprocess input strings to build a data structure 
for efficient querying of the matches, and also \emph{match successor queries}:
given $\hl$, $\hi$, find lowest $\hj \geq \hi$ such that $a(\hl) = b(\hj)$.

% (in main alg: from each block must split off at most $O(1)$ cells on average
% find splitting match in time $O(1)$ by binary search tree built at preprocessing)

\index{algorithm!waterfall}%
\begin{algorithm}
\textbf{(Parameterised LCS: The waterfall algorithm)}
\label{alg-waterfall}
\setlabelitbf
\nobreakitem[Input:]
strings $a$, $b$ of length $m$, $n$, respectively.
\item[Output:]
the LCS score $\lcs(a,b)$.
\item[Description.]
The algorithm is based on 
transposition network $\mathcal N(\rG_{a,b})$ with binary input, 
as described by \thref{th-network-bin}.
Such a network can be evaluated efficiently as follows.
The dag $\rG_{a,b}$ is processed row-by-row.
Instead of performing binary value comparisons
within individual mismatch cells of a given row of cells,
we partition the row into contiguous horizontal blocks,
and combine the comparisons within each block
into a single constant-time operation.
A block may contain both match and mismatch cells;
as we move vertically from one dag row to the next,
the blocks may split or merge.
We keep track of each block's endpoints in a parameter-sensitive way,
achieving an overall speedup 
whenever either of the parameters $\lambda$, $\kappa$ is low.

Let us now fill in the details.
Consider a row of nodes in the alignment dag $\rG_{a,b}$,
corresponding to a fixed index $l \in \bra{0:m}$.
The nodes in this row are connected by $n$ horizontal edges,
each of which is crossed by a single wire 
of the transposition network $\mathcal N(\rG_{a,b})$.
Hence, when running the network on a binary input as in \thref{th-network-bin},
each row of nodes corresponds to a sequence of $n$ binary values.
We regard this sequence as partitioned
into contiguous runs of $0$-values, called \emph{blocks},
alternating with contiguous runs of $1$-values, called \emph{gaps}.
We only need to deal explicitly with the blocks,
leaving the processing of gaps implicit (hence the block/gap terminology).
If the input strings happen to be highly dissimilar or highly similar,
i.e.\ the respective parameter $\lambda$ or $\kappa$ is low,
then every dag row will have only a small number of blocks (and gaps). 

Consider the operation of the transposition network row by row,
moving downwards across the alignment dag.
In the top row of nodes $l=0$, 
we have a single block of $0$-values coming as input to the network,
with index set spanning the full row $\ang{0:n}$.

Now, given a fixed row of cells $\hl \in \ang{0:m}$,
consider a transition from its top boundary $\hl^-$ 
to its bottom boundary $\hl^+$.
Given the endpoints of all blocks in row of nodes $\hl^-$,
we need to find the endpoints for the blocks in row of nodes $\hl^+$.
This transformation of endpoints can be achieved in two phases: 
\emph{block splitting} and \emph{block merging}.

\setlabelit
\item[Block splitting.] 
Consider a block in row of nodes $\hl^-$, 
spanning the index set $\ang{i_0:i_1}$.
Let $\hi \in \ang{i_0:i_1}$ be the index 
of the leftmost match cell (if one exists) immediately below the given block:
$\hi = \min \brc{\hk \in \ang{i_0:i_1} \mid a(\hl) = b(\hk)}$.
If index $\hi$ exists, then block $\ang{i_0:i_1}$ 
is split at this index into a pair of subblocks:
\begin{itemize}
\item the left subblock $\ang{i_0:\hi^-}$, which is empty if $i = i_0^+$,
and consists of the whole block $\ang{i_0:i_1}$ if index $\hi$ does not exist;
this subblock is kept unchanged in row of nodes $\hl^+$;
\item the right subblock $\ang{\hi^-:i_1}$;
this subblock is shifted by one unit to the right,
forming a new block $\ang{\hi^+:i_1+1}$ in row of nodes $\hl^+$
(unless $i_1=n$, in which case the rightmost $0$-value 
is considered to be ``shifted out'' off the right boundary of the dag,
and the resulting new block is $\ang{\hi^+:i_1}$).
\end{itemize}
A unit-length gap $\ang{\hi^-:\hi^+}$ is created 
between the two subblocks in row $\hl^+$.
The gap at the right of the right subblock 
(which exists in row $\hl^-$, unless $i_1=n$)
is reduced by one unit in row $\hl^+$.

\item[Block merging.]
This occurs whenever a gap between two blocks 
has been closed (reduced to length $0$) 
as a result of shifting a subblock in the previous splitting phase.
Suppose that subblock shifting
has resulted in a pair of touching blocks 
$\ang{i_0:i_1}$ and $\ang{i_1:i_2}$ in row $\hl^+$.
Then, we merge them into a single block $\ang{i_0:i_2}$.
Note that the left touching block $\ang{i_0:i_1}$
must have been split off the right-hand side 
of a larger block in the preceding splitting phase.
Therefore, blocks can only merge in pairs: 
it is impossible for three or more blocks 
to merge together in the same merging phase.
Upon the completion of all block merging, 
we have the correct block endpoints in row of nodes $\hl^+$.

After all the $m$ rows of cells in the alignment dag have been processed,
the values $y(\hj)$, $\hj \in \ang{0:m+n}$ are obtained as follows:
\begin{itemize}
\item for $\hj \in \ang{0:n}$, the values $y(\hj)$
correspond to the sequence of blocks and gaps 
resulting from the final iteration of the split/merge procedure;
\item for $\hj \in \ang{n:m+n}$, the values $y(\hj)$
correspond, in reverse order, 
to the sequence of $0$-values (and, implicitly, $1$-values)
``shifted out'' off the right boundary of the dag
in each of the $m$ iterations of the split/merge procedure.
\end{itemize}
The LCS score of the input strings $a$, $b$
can now be obtained by \thref{th-network-bin}.
\setlabelitbf

We now need to show that the described algorithm emulates correctly
the operation of the transposition network $\mathcal{N}(\rG_{a,b})$
on binary input, as defined by \thref{th-network-bin}.
Consider the operation of the algorithm in row of cells $\hl$.
The operation within an individual cell 
depends on whether the top edge of the cell belongs to a left subblock 
(which can be the whole block in case it does not get split), 
a right subblock, or a gap.
We also need to consider separately 
the leftmost cell under a right subblock 
(the ``subblock splitting cell''), 
and the leftmost cell under a gap following a right subblock 
(the ``gap filling cell'').
We thus have the following cases for a cell's operation
(where L = ``left'', R = ``right'', T = ``top'', B = ``bottom'').
\begin{center}
\begin{tabular}{|l|ll|ll|l|}
\hline
Cell's top edge & \multicolumn{2}{|c|}{Input} & 
\multicolumn{2}{|c|}{Output} & Match? \\
& L & T & B & R & \\ \hline
In L subblock            & 1 & 0 & 0 & 1 & mismatch \\ \hline
Leftmost in R subblock  & 1 & 0 & 1 & 0 & match \\
Other in R subblock     & 0 & 0 & 0 & 0 & either \\ \hline
Leftmost in gap after R subblock & 0 & 1 & 0 & 1 & either \\
Other in gap                & 1 & 1 & 1 & 1 & either \\ \hline
\end{tabular}
\end{center}
It is now straightforward to check that 
the input-output relationship defined by the above table
is consistent with the operation of all the individual comparators
in the network $\mathcal{N}(\rG_{a,b})$,
%
\begin{comment}
\item[Cost analysis.]
%
Preprocessing.

sort characters in string $b$, time $O(n \log\sigma)$

scan chars in sequence

for each character build binary search tree in string $b$.
Total time $O(n)$.

Overall time $O(n \log\sigma)$

Iteration: scan match list.

Alternative 1: 
For each match: find containing block by binary search.
If leftmost in its containing block, split.
Time $O(r \log \lambda)$ or $O(r \log\log \lambda)$ via van Emde Boas.
Generalises Crochemore/Porat
(in the same way as Hunt/Szymanski generalises standard LIS).

Alternative 2:

Need to answer queries: given a char and position, nearest match to the right

Char freq $O(1)$: 
no preproc, scan all matches.
Query time $O(1)$.

Char freq $O(n)$, at most $O(1)$ such chars: 
trivial preproc in $O(n)$, all queries explicitly. 
Query time $O(1)$.

Char freq $O(n^{1/2})$, at most $O(n^{1/2})$ such chars:
partition into $n^{1/2}$ blocks

Process each block.
Walk around tree to find leftmost covered match cell.

Assuming constant $\sigma$: trivial preprocessing for each char in $O(n)$, overall $O(n \sigma)$

$m$ iterations.
In each iteration, for each block find leftmost covered match cell.
Constant $\sigma$: precompute for each char, $O(1)$.
\end{comment}
%
and that the algorithm's overall running time is $O(\lambda \cdot \kappa)$.
\end{algorithm}

The name ``waterfall algorithm'' for \algref{alg-waterfall}
is justified by the following interpretation.
Let us think of the $0$-values as a non-compressible water
that flows through the alignment dag under gravity.
The blocks of adjacent $0$-values correspond to water jets
that may split or merge while flowing through the dag.
Initially, there is just a single jet of width $n$,
falling vertically down through the top boundary of the dag.
The diagonal match edges are rocks that form barriers for the water:
whenever a jet hits a rock in its path, it gets displaced 
by one unit to the right, following the rock's inclination.
This displacement may result in a jet splitting
at the top-left tip of the rock;
the displaced water may also fill the gap to its right
and merge with a jet falling vertically just beyond that gap.
The amount of water that emerges from the dag 
at its right-hand side (respectively, at its bottom)  
equals the LCS score $\lambda$ (respectively, the LCS distance $\mu$)
of the input strings. 

\begin{figure}[tb]
\centering
\subfloat[\label{f-waterfall-losim}Low-similarity]{%
\includegraphics{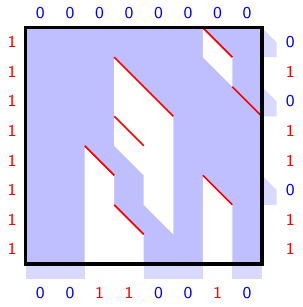}}
\qquad
\subfloat[\label{f-waterfall-hisim}High-similarity]{%
\includegraphics{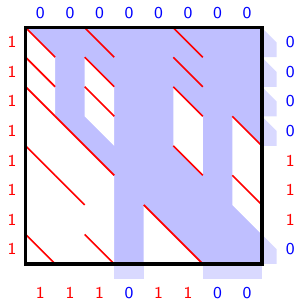}}
\caption{\label{f-waterfall}
The waterfall algorithm}
\end{figure}
\begin{example}
\figref{f-waterfall} illustrates two separate runs 
of the waterfall algorithm (\algref{alg-waterfall}):
\begin{itemize}
\item \sfigref{f-waterfall-losim} shows the low-similarity case,
with input strings $a = \textsf{``ABCBDABE''}$, $b = \textsf{``FFDBCFAC''}$;
\item \sfigref{f-waterfall-hisim} shows the high-similarity case,
with input strings $a = \textsf{``AAABABCA''}$, $b = \textsf{``ABADCADB''}$.
\end{itemize}
Following our usual convention,
the diagonal edges in match cells are shown in red.
The rest of the alignment dag, as well as the input strings,
are kept implicit.
The anti-sorted input to the transposition network
is represented by the sequences of red $1$-values and blue $0$-values 
along the left (respectively, the top) boundary of the alignment dag.

The iterative procedure of block splitting and merging
is shown by filling in the interior of the alignment dag
with blue and white pattern as follows.
Every horizontal edge in the alignment dag
corresponds to a blue (respectively, white) streak in the pattern,
whenever that edge is crossed by an (implicit) wire
carrying a $0$- (respectively, $1$-) value.
Therefore, for each row index $l$ in the dag, 
the corresponding sequence of blocks and gaps
is represented by a sequence of continuous 
blue and white streaks in a horizontal line at level $l$.
The transition of block sequences 
between every pair of successive rows is shown by connecting 
the corresponding pairs of blue streaks with a blue strip.
For the vertical (respectively, diagonal) transition of a single $0$-value,
the connecting strip 
has the shape of a unit square (respectively, unit-width parallelogram).

The output of the transposition network
is represented by the mixed sequence of red $1$-values and blue $0$-values 
along the bottom and right boundaries of the alignment dag.
Each output $0$-value is also shown by a blue tab.
By \thref{th-network-bin}, we have
\begin{gather*}
\textstyle
\lcs(a,b) = \sum_{\hj \in \ang{0:8}} y(\hj) = 3 =
8 - \sum_{\hj \in \ang{4:8}} y(\hj) = 8-5 = 3\\
\textstyle
\lcs(a,b) = \sum_{\hj \in \ang{0:8}} y(\hj) = 5 =
8 - \sum_{\hj \in \ang{4:8}} y(\hj) = 8-3 = 5
\end{gather*}
in \sfigrefs{f-waterfall-losim} and \ref{f-waterfall-hisim},
respectively.

Just as the splitting/merging procedure in \algref{alg-waterfall}
is not symmetric with respect to blocks and gaps,
so the pattern in \figref{f-waterfall} is not symmetric 
either with respect to horizontal and vertical directions,
or with respect to $0$- and $1$-values.
For this reason, while the blocks are represented by the usual blue colour,
the gaps are left uncoloured 
(i.e.\ are represented by the ``background'' white).
The red colour, which we would normally use to represent $1$-values,
is not present in the pattern.
\end{example}

Parameterised LCS algorithms are closely related 
to \emph{threshold LCS} algorithms.
Here, a threshold value for the parameter $\lambda$ or $\kappa$ is given,
and the algorithm is required to output the LCS score of the input strings,
as long as this value is below the threshold,
or to report ``excess'',
if the LCS score (respectively, LCS distance) exceeds the threshold;
in the latter case, there is no requirement 
for the score to be output explicitly.
In general, a threshold algorithm 
can be obtained from a parameterised algorithm
by setting an appropriate time-out threshold 
as a function of the parameter threshold, and reporting ``excess'', 
if the algorithm's running time exceeds the time-out threshold.
Alternatively, a threshold LCS algorithm can be obtained 
from the waterfall algorithm (\algref{alg-waterfall}) by reporting ``excess'', 
if the current number of blocks 
exceeds the threshold value for $\lambda$ (respectively, $\kappa$).

%Parameterised incremental LCS: \cite{Landau+:07}

%%=-=-=-=-=-=-=-=-=-=-=-=-=-=-=-=-=-=-=-=-=-=-=-=-=-=-=-=-=-=-=-=-=-=-=-=-=-=%%
\begin{comment}

\mysection{Dynamic LCS}
\label{s-dynamic}

\begin{definition}
%
\emph{block dag}. Layered dag, nodes correspond to blocks.
Edges represent split and merge operations on blocks.
%
\end{definition}
%
\begin{theorem}
%
Update block dag in time $O(m)$.
%
\end{theorem}

\end{comment}
%%=-=-=-=-=-=-=-=-=-=-=-=-=-=-=-=-=-=-=-=-=-=-=-=-=-=-=-=-=-=-=-=-=-=-=-=-=-=%%
\mysection{Bit-parallel LCS}
\label{s-bitpar}

The most efficient practical method 
for computing the (global) LCS score for a pair of strings
is by \emph{bit-parallel} algorithms.
These algorithms take advantage of bitwise Boolean operations on bit vectors
available in modern processors, often in combination 
with arithmetic operations on the same vectors as integers.
We denote by $w$ the \emph{word} (standard bit vector) length;
in modern general-purpose processors, word length is often $w = 64$.

Early bit-parallel string comparison algorithms
were given by Allison and Dix \cite{Allison_Dix:86_IPL}
and by Myers \cite{Myers:99_JACM}.
Crochemore at al.\ \cite{Crochemore+:01_IPL} proposed 
an efficient bit-parallel LCS algorithm, running in time $O(mn/w)$.
For every $w$ cells of the alignment dag $\rG_{a,b}$, the algorithm only performs
five elementary operations (one arithmetic and four Boolean).
Hyyr\"o \cite{Hyyro:04_AWOCA} improved this to four operations
(two arithmetic and two Boolean).

\begin{figure}[tb]
\centering
\subfloat[\label{f-bitpar-1-c}Binary transposition network cell]{%
\vtop{\vskip0pt\hbox{\includegraphics{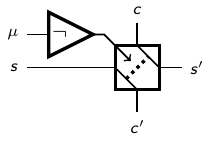}}}}
\qquad
\subfloat[\label{f-bitpar-1-tt}Corresponding truth table]{%
\mbox{\scriptsize
\begin{tabular}[t]{|l|cccccccc|}
\hline
$s$   & 0 & 1 & 0 & 1 & 0 & 1 & 0 & 1 \\
$c$   & 0 & 0 & 1 & 1 & 0 & 0 & 1 & 1 \\
$\mu$ & 0 & 0 & 0 & 0 & 1 & 1 & 1 & 1 \\ \hline
$s'$  & 0 & 1 & 1 & \colorbox{gray!40}{\textbf{1}} & 0 & 0 & 1 & 1 \\
$c'$  & 0 & 0 & 0 & 1 & 0 & 1 & 0 & 1 \\ \hline
\end{tabular}}}

\subfloat[\label{f-bitpar-2-c}Boolean circuit for 
  $2c' + s' \becomes s + (s \land \mu) + c$]{%
\vtop{\vskip0pt\hbox{\includegraphics{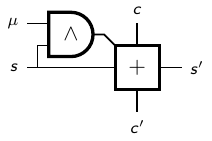}}}}
\qquad
\subfloat[\label{f-bitpar-2-tt}Corresponding truth table]{%
\scriptsize
\begin{tabular}[t]{|l|cccccccc|}
\hline
$s$   & 0 & 1 & 0 & 1 & 0 & 1 & 0 & 1 \\
$c$   & 0 & 0 & 1 & 1 & 0 & 0 & 1 & 1 \\
$\mu$ & 0 & 0 & 0 & 0 & 1 & 1 & 1 & 1 \\ \hline
$s'$  & 0 & 1 & 1 & \colorbox{gray!40}{\textbf{0}} & 0 & 0 & 1 & 1 \\
$c'$  & 0 & 0 & 0 & 1 & 0 & 1 & 0 & 1 \\ \hline
\end{tabular}}
\caption{\label{f-bitpar} 
Binary transposition network via a Boolean circuit}
\end{figure}

Both algorithms \cite{Crochemore+:01_IPL,Hyyro:04_AWOCA} 
can be viewed as an implementation of a binary transposition network, 
described in \secref{s-network},
by standard bit-parallel processor instructions.
\figref{f-bitpar} shows the main idea of such an implementation
for the algorithm of \cite{Crochemore+:01_IPL}.
Consider a cell in a binary transposition network,
and let us denote its input bits by $s$, $c$,
and its output bits by $s'$, $c'$, as shown in \sfigref{f-bitpar-1-c}.
Let us denote by $\mu$ an extra input bit, which takes value $1$ 
if and only if the current dag cell is a match cell.
The operation of a network cell is fully described 
by the truth table in \sfigref{f-bitpar-1-tt}.

Now consider a Boolean circuit shown in \sfigref{f-bitpar-2-c}.
The circuit consists of an $\land$-gate and a \emph{full adder} element,
which adds its three input bits arithmetically,
and returns the sum as two separate output bits.
Let us again denote the input bits by $s$, $c$, $\mu$,
and the output bits by $s'$, $c'$.
Then, the circuit computes a Boolean-arithmetic expression
\begin{equation*}
2c' + s' \becomes s + (s \land \mu) + c
\end{equation*}
The operation of such a circuit is fully described 
by the truth table in \sfigref{f-bitpar-2-tt}.

Note that the truth tables in \sfigrefs{f-bitpar-1-tt} and \ref{f-bitpar-2-tt}
differ in just the one highlighted bit.
This difference can be corrected by two extra Boolean operations,
resulting in a Boolean-arithmetic expression that fully implements
the operation of a single transposition network cell:
\begin{equation}
\label{eq-cipr-bit}
2c' + s' \becomes \bigpa{s + (s \land \mu) + c} \lor (s \land \lnot \mu)
\end{equation}

Now consider a row of $n$ cells in the alignment dag $\rG_{a,b}$, 
assuming for the moment $n \leq w$.
Let $S$ denote a word variable 
that will hold the input $s$, and then the output $s'$, for each cell, 
least significant bit first.
Likewise, let $M$ denote a word constant 
that holds the match parameter $\mu$ for each cell, 
least significant bit first.
(The input $c$ and output $c'$ will not be represented explicitly,
but will instead correspond to a propagating carry bit 
in word integer addition, 
from the least significant bit all the way to the most significant bit.)
The operation of the transposition network $\mathcal{N}(\rG_{a,b})$
in the given row of cells corresponds to evaluating an expression
\begin{equation}
\label{eq-cipr}
S \becomes \bigpa{S + (S \land M)} \lor (S \land \lnot M)
\end{equation}
which is obtained from \eqref{eq-cipr-bit} by identifying 
the output $c'$ of each cell with the input $c$ of the next cell in the row.
Here, the Boolean operations are bitwise,
and the addition is on integers represented by the words.
Note that for an exact correspondence with \eqref{eq-cipr},
the roles of 0-values and 1-values 
in the waterfall algorithm (\algref{alg-waterfall}) must be exchanged 
(or, alternatively, the cells must be composed into words
by columns, rather than by rows).

If $n > w$, then each row of the alignment dag
is partitioned into $\ceil{n/w}$ words of length $w$.
In this case, expression \eqref{eq-cipr} needs to be modified
to allow carry propagation from each word to the next word in its row.

The described bit-parallel approach can be extended 
to provide even more efficient string comparison
in the case of highly similar strings.
We consider the threshold version of the high-similarity LCS problem,
as described in \secref{s-parameterised}.
Let $\kappa$ denote a threshold on $\mathit{dist}_\mathit{LCS}(a,b)$,
i.e.\ the LCS distance between the input strings $a$, $b$.
We assume for simplicity that the value $\kappa$ is odd, and that $m=n$. 
Under these assumptions, a highest-scoring path in the alignment dag $\rG_{a,b}$
must lie strictly within a symmetric diagonal band of width $\kappa+1$, 
unless $\mathit{dist}_\mathit{LCS}(a,b)$ exceeds $\kappa$.
Hence, the waterfall algorithm (\algref{alg-waterfall}) 
can be modified as follows.
First, we ignore all character matches outside the band,
as well as on the band's lower-left and upper-right boundaries.
We also ignore all the (explicit) input 0-values
and all the (implicit) input 1-values outside the band;
therefore, we only have 
$\Half[k+1]$ (explicit) input 0-values at the top of the band,
and $\Half[k+1]$ (implicit) input 1-values at the left-hand side of the band.
Finally, we create artificial \emph{separator matches}
along the bottom-left boundary of the band.
The LCS score $\mathit{lcs}(a,b)$ can be obtained
by counting the output 0-values at the bottom of the band
(or, symmetrically, the output 1-values at the right-hand side of the band).
The algorithm reports ``exceed'', if all the input 0-values
are output at the bottom of the band.
The algorithm runs in time $O(m\kappa/w)$.

\begin{figure}[tb]
\centering
\includegraphics{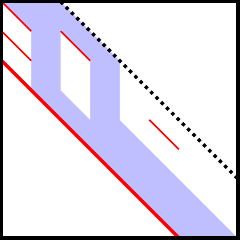}
\caption{\label{f-waterfall-hisim-bitpar}
  High-similarity bit-parallel waterfall algorithm}
\end{figure}

\figref{f-waterfall-hisim-bitpar} shows a run of the resulting 
high-similarity bit-parallel waterfall algorithm, for $\kappa=3$.
The visual conventions are similar to those 
in \figrefs{f-waterfall-losim}, \ref{f-waterfall-hisim}.
Note that the band is of width $\kappa+1 = 4$,
and that there are $\Half[\kappa+1] = 2$ input 0-values at the band's top.
Both these values end up as output 0-values at the band's bottom,
hence the algorithm returns ``exceed'' in the given run.

The described algorithm is particularly easy to implement 
when the bandwidth $m+1 \leq w-1$.
In such a case, every row of the band fits into a single word.
The bit-parallel five- (respectively, four-) instruction sequence
of either of \cite{Crochemore+:01_IPL,Hyyro:04_AWOCA} can be used;
the only modification required is an extra shift instruction in each row,
to account for the band right-shifting by $1$ when moving to the next row.

Still further optimisation is possible 
in the case of \emph{multi-string} comparison.
This type of comparison has been considered 
e.g.\ by by Hyyr\"o et al.\ \cite{Hyyro+:04_WEA}.
Here, we asked to compute the LCS score for string $a$ 
against each of the $r$ strings $b_0$, \ldots, $b_{r-1}$, all of length $n$.
We assume that we are given a single threshold $\kappa$ 
on all $\mathit{dist}_\mathit{LCS}(a,b_s)$, $0 \leq s < r$.
As before, we assume for simplicity 
that the value $\kappa$ is odd, and that $m=n$.
The problem can be solved by $r$ independent runs 
of the high-similarity bit-parallel waterfall algorithm.
However, it is possible to combine these runs 
into a single bit-parallel computation, where in each step, 
we evaluate a single row from every one of the $r$ bands.
The bands are packed together in a single super-band of width $r(\kappa+1)$;
individual bands within the super-band 
are separated by diagonals of separator matches.

\begin{figure}[tb]
\centering
\includegraphics{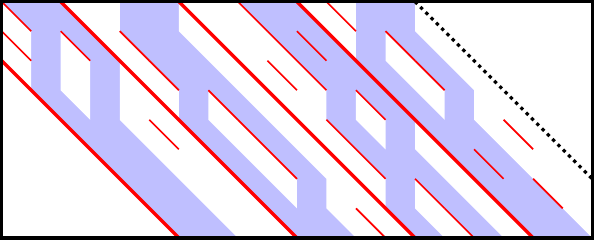}
\caption{\label{f-waterfall-hisim-bitpar-multi}
  High-similarity bit-parallel multi-string waterfall algorithm}
\end{figure}

\figref{f-waterfall-hisim-bitpar-multi} shows a run of the resulting 
high-similarity bit-parallel multi-string waterfall algorithm, 
for $\kappa=3$ and $r=4$.
The leftmost band (band $0$) is identical 
to the band in \figref{f-waterfall-hisim-bitpar}.
The algorithm returns ``exceed'' for bands $0$, $1$, $3$.
For band $2$, we have a single 0-value output at the bottom of the band,
hence $\mathit{lcs}(a,b_2) = n-1$ by \thref{th-network-bin}.

Finally, note that all the described techniques 
for bit-parallel string comparison
are compatible with integer-weighted alignment, 
by application of the blow-up technique described in \secref{s-weighted}.
For example, we can obtain bit-parallel algorithms
for Levenshtein alignment score (or, symmetrically, Levenshtein edit distance)
by blowing up the alignment dag by a factor of $2$ in each dimension.

%%=-=-=-=-=-=-=-=-=-=-=-=-=-=-=-=-=-=-=-=-=-=-=-=-=-=-=-=-=-=-=-=-=-=-=-=-=-=%%
\begin{comment}
\mysection{Subword-parallel LCS}
\label{s-swpar}

%Subword-parallel, including MP-RAM.

\end{comment}
%%===========================================================================%%

%%===========================================================================%%
\mychapter{Beyond semi-locality}
\label{c-beyond}

In this chapter, our aim is to extend the approach of the previous chapters
beyond semi-local string comparison,
with the ultimate goal of efficient fully-local string comparison.
We study several problems, where for each of the two input strings,
their various substrings are selected for pairwise comparison.
This type of string comparison is arguably the most important,
in particular for biological applications. 
However, it is also the most computationally expensive,
therefore there is particular need for efficient algorithms.

% In particular,
% in \secref{s-dynamic} we introduce
% the dynamic LCS problem,
% and describe an efficient data structure for this problem,
% based on seaweed braids.
This chapter is organised as follows.
In \secref{s-window}, we introduce 
the window-substring and window-window LCS problems,
and give an efficient algorithm for these problems.
This algorithm provides a refinement for the standard dot plot method, 
by allowing window-window string comparison
based on the LCS score, rather than the less sensitive Hamming score.
In \secref{s-quasi}, we introduce the quasi-local LCS problem,
which generalises the semi-local, 
window-substring and window-window LCS problems,
and give an efficient algorithm for this problem.
By application of this algorithm, in \secref{s-spliced} 
we obtain an algorithm for sparse spliced alignment
under an arbitrary rational edit distance metric,
which improves on existing algorithms for this problem.

\mysection{Window-local LCS and alignment plots}
\label{s-window}

\index{algorithm!Smith--Waterman--Gotoh}%
Numerous methods of local string comparison 
have been proposed in the past.
The \emph{Smith--Waterman--Gotoh algorithm} \cite{Smith_Waterman:81,Gotoh:82}
allows one to obtain the highest-scoring pair
across all substring pairs in the two input strings.
It can also be generalised to report all substring pairs 
scoring above a certain threshold.
A significant drawback of the Smith--Waterman--Gotoh algorithm
is that it generally favours long, less precise substring alignments 
over short, more precise ones
(as noted e.g.\ by Arslan et al.\ \cite{Arslan+:01}).
The quality of the alignment is also dependent on the scoring scheme:
for example, for the simplest LCS score,
the algorithm only provides the trivial global comparison,
so the method is generally only useful
for weighted alignment scores with sufficiently high penalties
(negative score weights) for gaps.

The drawbacks of the Smith--Waterman--Gotoh approach 
have been addressed by the techniques 
of normalised alignment by Arslan et al.\ \cite{Arslan+:01}, 
and length-constrained alignment 
by Arslan and E{\u g}ecio{\u g}lu \cite{Arslan_Egecioglu:02}.
Such methods typically have very high computational cost. 
To reduce this cost, various approximation algorithms for the same problems
have been suggested (see \cite{Arslan_Egecioglu:02}).

An alternative approach,  
is to obtain local comparison of two strings by computing an \emph{alignment plot,}
i.e.\ an exhaustive set of alignments for pairs of fixed-length windows in both sequences.
A similar method is used by Rasmussen et al.\ \cite{Rasmussen+:06},
where an algorithm for local alignment in fixed-length windows in described.
The difference of our approach from that of \cite{Rasmussen+:06}
is that the latter is designed to be very efficient
when the similarity threshold is high,
but becomes prohibitively expensive when looking 
for relatively low-similarity matches (e.g.\ 70\% similarity).
In contrast, our method works equally efficiently 
for any similarity threshold.

\index{substring!window}%
\index{dot plot!Hamming-filtered}%
\index{Hamming score}%
Recall from \secref{s-cyclic} that, 
given a fixed parameter $w$, we call a substring of length $w$ 
a \emph{$w$-window} in the corresponding string.
String comparison in windows has a long history.
One of its early instances is \emph{dot plots}
(also known as \emph{diagonal plots} or \emph{dot matrices}), 
introduced by Gibbs and McIntyre \cite{Gibbs_McIntyre:70}
and by Maizel and Lenk \cite{Maizel_Lenk:81}
(see also Crochemore et al. \cite{Crochemore+:07}).
In addition to numerical scores, 
dot plots provide a convenient visualisation of string comparison.
In the context of dot plots, processing a pair of windows
is usually referred to as \emph{filtering}.
The standard filtering method
compares every $w$-window of string $a$ against every $w$-window of string $b$
in terms of their \emph{Hamming score},
i.e.\ the count of matching character pairs under a rigid alignment model,
where every character must be aligned 
(resulting in either a match or a mismatch),
and no gaps are allowed.
A Hamming-filtered dot plot can be computed in time $O(mn)$ 
by the algorithm of \cite{Maizel_Lenk:81,Mueller+:06}.
This algorithm has been implemented in several software packages
(see e.g.\ \cite{Sonnhammer_Durbin:95,Rice+:00,CLC-protein:06}).
A faster suffix-tree based algorithm has been proposed and implemented 
by Krumsiek et al.\ \cite{Krumsiek+:07}.
Enhancement of the dot plot approach have been proposed
by Huang and Zhang \cite{Huang_Zhang:04}
and by Putonti et al.\ \cite{Putonti+:06}.

In contrast with the Smith--Waterman--Gotoh algorithm,
the dot plot method gives the user more flexibility
to select the biologically significant substring alignments,
by providing all the local scores
between fixed-size windows of the input strings.
However, the Hamming scoring scheme used by this method within each window pair
is less sensitive than even the LCS score,
and especially than the weighted alignment score used by Smith--Waterman--Gotoh.
This tradeoff motivates us to combine the best features of the two approaches
in the following definition.

%Another popular approach to local string comparison 
%is by using heuristics, such as the \emph{BLAST algorithm} \cite{Altschul+:90}.
%Heuristic methods can be very fast, 
%but they are not as sensitive or detailed
%as the exact methods considered in this paper.

\begin{definition}
\label{def-window-local}
\index{problem!window-window LCS}%
\index{problem!window-substring LCS}%
Given strings $a$, $b$, and a window length $w$, the \emph{window-window}
(respectively, \emph{window-substring}) \emph{LCS problem} 
asks for the LCS score of for every $w$-window in $a$ 
against every $w$-window (respectively, every substring) in $b$.
\end{definition}%
The window-window LCS problem can be seen as a refinement of the dot plot method
and a complement to the Smith--Waterman--Gotoh method.
As in the dot plot method,
we compute all window-window comparison scores between the input strings.
However, instead of the Hamming score, our method is based on the LCS score,
and is therefore potentially more sensitive.
The method can be further extended to use weighted alignment scores.
\index{dot plot!alignment-filtered}%
By analogy with Hamming-filtered dot plots, we call the resulting matrix
of window-window alignment scores an \emph{alignment-filtered dot plot},
or simply an \emph{alignment plot}.
Recently, the alignment plot method has been applied
to the detection of alignment-conserved regions in DNA 
by Picot et al.\ \cite{Picot+:10}.
Computation of an alignment plot is also a key subroutine
in the approximate repeat searching method 
by Federico et al.\ \cite{Federico+:11}.

A solution to the window-substring LCS problem 
can be represented in space $O(mn \log n)$ 
by the data structure of \thref{th-query},
built on the string-substring seaweed matrix
for each window of $a$ against $b$.
An individual window-substring LCS score query 
can be performed on this data structure by \thref{th-query} 
in time $O(\log^2 n)$.
The same data structure can be used to obtain explicitly 
a solution to the window-window LCS problem in time $O(mn)$,
by performing a diagonal batch query directly 
on each of the string-substring seaweed matrices.
Thus, string-substring seaweed matrices provide a unified solution 
for both the window-substring and the window-window LCS problems.

A straightforward algorithm for the window-substring 
(and therefore also window-window) LCS problem
can be obtained by solving the semi-local LCS problem
independently to each window of string $a$ against whole string $b$.
Using \algref{alg-seaweed} (seaweed combing) as a subroutine,
the resulting algorithm runs in time $O(mnw)$.
If window length $w$ is sufficiently large,
the running time can be improved slightly 
by the micro-block speedup (\algref{alg-seaweed-micro}).

We now give an algorithm for the window-substring (and window-window) LCS problem
that improves substantially on the above approach,
and matches the asymptotic running time
of both the Hamming-scored dot plot and the Smith--Waterman--Gotoh methods.

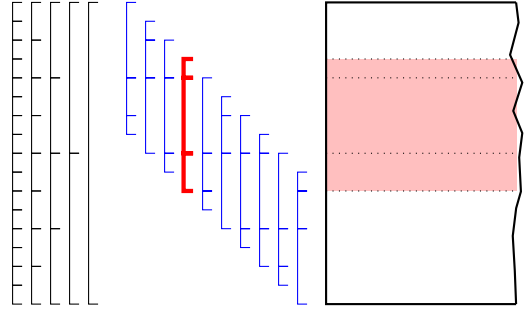
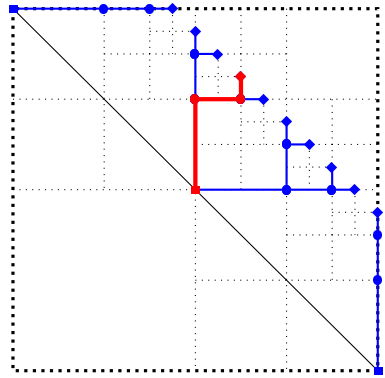
\begin{figure}[tb]
\centering
\subfloat[\label{f-window-1}%
Canonical substrings and windows]{%
\beginpgfgraphicnamed{f-window-1}%
\begin{tikzpicture}[x=0.25cm, y=-0.25cm]
\tikzstyle{every to}=[to path = {-- +(-0.5,0) |- (\tikztotarget)}]
\path (0,0) coordinate(p);
\foreach \k in {0,1,...,15}{%
  \draw (p) to ++(0,1) coordinate(p);}
\path (1,0) coordinate(p);
\foreach \k in {0,2,...,15}{%
  \draw (p) to ++(0,2) coordinate(p);}
\path (2,0) coordinate(p);
\foreach \k in {0,4,...,15}{%
  \draw (p) to ++(0,4) coordinate(p);}
\path (3,0) coordinate(p);
\foreach \k in {0,8,...,15}{%
  \draw (p) to ++(0,8) coordinate(p);}
\path (4,0) coordinate(p);
\draw (p) to ++(0,16);
\path (6,0) coordinate(p);
\foreach \k in {0,4,...,9}{%
  \draw[blue] (p) to ++(0,4) to ++(0,2) to ++(0,1) 
  (p) +(4,4) coordinate(p);}
\path (6,0) +(2,2) coordinate(p);
\foreach \k in {2,6,...,9}{%
  \draw[blue] (p) to ++(0,2) to ++(0,4) to ++(0,1) 
  (p) +(4,4) coordinate(p);}
\path (6,0) +(1,1) coordinate(p);
\foreach \k in {1,5,...,9}{%
  \draw[blue] (p) to ++(0,1) to ++(0,2) to ++(0,4) 
  (p) +(4,4) coordinate(p);}
\path (6,0) +(3,3) coordinate(p);
\foreach \k in {3,7,...,9}{%
\ifthenelse{\equal{\k}{3}}{
  \draw[red,ultra thick] (p) to ++(0,1) to ++(0,4) to ++(0,2)}{
  \draw[blue] (p) to ++(0,1) to ++(0,4) to ++(0,2)}
  (p) +(4,4) coordinate(p);}
\path (16,0) coordinate(p);
\fill[red!25] (p) ++(0,3) rectangle +(10,7);
\draw[dotted] (p) ++(0,3) -- +(10,0) ++(0,1) -- +(10,0) 
  ++(0,4) -- +(10,0) ++(0,2) -- +(10,0);
\draw[thick] (p) +(10,0) -- +(0,0) -- +(0,16) -- +(10,16);
\draw[thick,decorate,decoration=random steps] (p) +(10,0) -- +(10,16);
\end{tikzpicture}
\endpgfgraphicnamed}

\subfloat[\label{f-window-2}%
Window decomposition forest]{%
\beginpgfgraphicnamed{f-window-2}%
\begin{tikzpicture}[x=0.3cm, y=-0.3cm]
\tikzstyle{every node}=[inner sep=0pt,outer sep=0pt,minimum size=0pt]
\draw[very thick,dotted] (0,0) rectangle (16,16);
\draw (0,0) -- (16,16);
\foreach \i/\j in {0/0}
  \draw[dotted] (\i,\j) ++(8,0) -- ++(0,16) (\i,\j) ++(0,8) -- ++(16,0);
\foreach \i/\j in {0/0,8/0,8/8}
  \draw[dotted] (\i,\j) ++(4,0) -- ++(0,8) (\i,\j) ++(0,4) -- ++(8,0);
\foreach \i/\j in {4/0,8/0,8/4,12/4,12/8}
  \draw[dotted] (\i,\j) ++(2,0) -- ++(0,4) (\i,\j) ++(0,2) -- ++(4,0);
\foreach \i/\j in {6/0,8/2,10/4,12/6,14/8}
  \draw[dotted] (\i,\j) ++(1,0) -- ++(0,2) (\i,\j) ++(0,1) -- ++(2,0);
\foreach \j/\i in {7/0,15/8}
  \draw[blue,thick] (\j,\i) node{\bulletd} -- ++(-1,0) node{\bulletc} -- 
    ++(-2,0) node{\bulletc} -- ++(-4,0) node{\bullets};
\foreach \j/\i in {11/4}
  \draw[blue,thick] (\j,\i) node{\bulletd} -- ++(-1,0) node{\bulletc} -- 
    ++(-2,0) node{\bulletc} -- ++(0,4) node{\bullets};
\foreach \j/\i in {9/2,13/6}
  \draw[blue,thick] (\j,\i) node{\bulletd} -- ++(-1,0) node{\bulletc} -- 
    ++(0,2) node{\bulletc};
\foreach \j/\i in {8/1,10/3,12/5,14/7}
  \draw[blue,thick] (\j,\i) node{\bulletd} -- ++(0,1) node{\bulletc};
\foreach \j/\i in {16/9}
  \draw[blue,thick] (\j,\i) node{\bulletd} --
     ++(0,1) node{\bulletc} -- ++(0,2) node{\bulletc} -- ++(0,4) node{\bullets};
\draw[red,ultra thick] (10,3) node{\bulletd} -- ++(0,1) node{\bulletc} -- 
  ++(-2,0) node{\bulletc} -- ++(0,4) node{\bullets};
\end{tikzpicture}
\endpgfgraphicnamed}
\caption{\label{f-window} An execution of \algref{alg-window-lcs}
(window-substring LCS)}
\end{figure}
\begin{algorithm}[Window-substring LCS]
\label{alg-window-lcs}
\setlabelitbf
\nobreakitem[Input:]
strings $a$, $b$ of length $m$, $n$, respectively;
window length $w$.
\item[Output:]
nonzeros of the string-substring seaweed matrix $\rP^\ssub_{a\ang{i:j},b}$
for every $w$-window $a\ang{i:j}$, $j-i=w$, against full string $b$.
\item[Description.]
\index{substring!canonical}%
We call a substring $a\ang{i:j}$ \emph{canonical}, if $j-i=s$, 
where $s$ is an arbitrary power of $2$, $1 \leq s \leq m$,
and $i$, $j$ are both multiples of $s$.
In particular, every individual character of $a$ 
is a canonical substring with $s=1$.
Every substring of $a$ of length $t$ can be decomposed
into a concatenation of at most $\log t$ canonical substrings.

In the following, by processing a substring $a\ang{i:j}$,
we mean computing the string-substring seaweed matrix $\rP^\ssub_{a\ang{i:j},b}$.

\setlabelit
\item[First phase.]
We process all canonical substrings of $a$ 
by the following recursive procedure.

\setlabelnormal
\item[Recursion base: $s=1$.]
For every canonical substring $a\ang{i:j}$ of length $j-i=1$,
matrix $\rP^\ssub_{a\ang{i:j},b}$ 
is computed by \algref{alg-seaweed} (seaweed combing)
on strings $a\ang{i:j}$ and $b$, of length $1$ and $n$ respectively.

\item[Recursive step: $s>1$.]
For every canonical substring $a\ang{i:j}$ of length $j-i=s$,
we have $a\ang{i:j} = a\bigang{i:\Half[i+j]}\: a\bigang{\Half[i+j]:j}$,
where $a\bigang{i:\Half[i+j]}$, $a\bigang{\Half[i+j]:j}$
are both canonical substrings of length $s/2$.
By \thref{th-comp-mmult}, we have
\begin{gather*}
\rP^\ssub_{a\ang{i:j},b} = 
\rP^\ssub_{a\bigang{i:\Half[i+j]},b} \boxdot \rP^\ssub_{a\bigang{\Half[i+j]:j},b}
\end{gather*}
The implicit matrix distance product can be computed by \thref{th-comp-mmult-comp}.
\item[(End of recursive step)]

\setlabelit
\item[Second phase.]
Let $s$ be an arbitrary power of $2$, $1 \leq s \leq m$.
We define \emph{$(s,t)$-window} to be a substring $a\ang{i:j}$, 
such that $j-i=st$, and $i$ and $j$ are both multiples of $s$.
Intuitively, $s$ is the amount of shift between two successive $(s,t)$-windows,
and $t$ is the number of different $(s,t)$-windows 
that cover any single character in the string $a$ 
(except possibly near its boundaries).
Note that a $(1,t)$-window is the same as a $t$-window,
and an $(s,1)$-window is the same as a canonical substring of length $s$.
Given parameters $s$, $t$, we solve the problem 
of processing all $(s,t)$-windows by the following recursive procedure.

\setlabelnormal
\item[Recursion base: $t=1$.]
For any $s$, all $(s,1)$-windows are canonical substrings of length $s$,
therefore they have already been processed in the first phase.

\item[Recursive step: $t>1$.] \ %
\begin{trivlist}
\item{\itshape Case of $t$ even.}
First, we call the second phase recursively
to process all $(s,t-1)$-windows.
Now consider an $(s,t)$-window $a\ang{i:j}$.
We have the decomposition
\begin{gather}
\label{eq-window-decomp}
a\ang{i:j} = a\ang{i:j-s}\: a\ang{j-s:j} = a\ang{i:i+s}\: a\ang{i+s:j}
\end{gather}
Observe that the prefix $a\ang{i:j-s}$ and the suffix $a\ang{i+s:j}$
are both $(s,t-1)$-windows;
also, the suffix $a\ang{j-s:j}$ and the prefix $a\ang{i:i+s}$
are both $(s,1)$-windows.
Therefore, we can now process all $(s,t)$-windows $a\ang{i:j}$ as follows:
\begin{gather}
\label{eq-window-even}
\rP^\ssub_{a\ang{i:j},b} = 
\begin{cases}
\rP^\ssub_{a\ang{i:j-s},b} \boxdot \rP^\ssub_{a\ang{j-s:j},b}\\
\rP^\ssub_{a\ang{i:i+s},b} \boxdot \rP^\ssub_{a\ang{i+s:j},b}
\end{cases}
\end{gather}
where the choice between the two alternatives can be made arbitrarily.

\item{\itshape Case of $t$ odd.}
We first call the second phase recursively
to process all $(2s,\Half[t-1])$-windows
(in other words, ``every other $(s,t-1)$-window'').
Now consider an $(s,t)$-window $a\ang{i:j}$.
As before, we have the decomposition \eqref{eq-window-decomp},
where the prefix $a\ang{i:j-s}$ and the suffix $a\ang{i+s:j}$
are both $(s,t-1)$-windows.
Furthermore, if $\frac{i}{s}$ is even and and $\frac{j}{s}$ is odd,
then the prefix $a\ang{i:j-s}$ is a $(2s,\Half[t-1])$-window.
Likewise, if $\frac{i}{s}$ is odd and and $\frac{j}{s}$ is even,
then the suffix $a\ang{i+s:j}$ is a $(2s,\Half[t-1])$-window.
The suffix $a\ang{j-s:j}$ and the prefix $a\ang{i:i+s}$
are both $(s,1)$-windows.
Therefore, we can now process all $(s,t)$-windows $a\ang{i:j}$ as follows:
\begin{gather}
\label{eq-window-odd}
\rP^\ssub_{a\ang{i:j},b} = 
\begin{cases}
\rP^\ssub_{a\ang{i:j-s},b} \boxdot \rP^\ssub_{a\ang{j-s:j},b} & 
\text{if $\frac{i}{s}$ even and $\frac{j}{s}$ odd} \\
\rP^\ssub_{a\ang{i:i+s},b} \boxdot \rP^\ssub_{a\ang{i+s:j},b} &
\text{if $\frac{i}{s}$ odd and $\frac{j}{s}$ even}
\end{cases}
\end{gather}
Note that the equations in \eqref{eq-window-odd} 
are identical to the ones in \eqref{eq-window-even},
but the choice between the alternatives is now determined 
by the parity of the values $\frac{i}{s}$, $\frac{j}{s}$.
The two alternatives in \eqref{eq-window-odd} are exhaustive:
$\frac{j}{s} - \frac{i}{s} = \frac{j-i}{s} = \frac{t}{s}$ must be odd,
since $t$ is odd and $s$ is a power of $2$.
\end{trivlist}

The implicit matrix distance products 
in \eqref{eq-window-even}, \eqref{eq-window-odd}
are computed by \thref{th-comp-mmult-comp}.
In each case, the product is between two string-substring seaweed matrices:
one for substring $a\ang{i:j-s}$ or $a\ang{i+s:j}$,
already processed by the recursive call,
and the other for a canonical substring $a\ang{j-s:j}$ or $a\ang{i:i+s}$.
\item[(End of recursive step)]

\setlabelitbf
\item[Cost analysis.]

\setlabelit
\nobreakitem[First phase.]
Starting at the top recursion level,
the number of seaweed matrix multiplications 
doubles in every subsequent level down the recursion.
The running time for each matrix multiplication is $O(n \log s)$,
where the parameter $s$ halves in every subsequent level of recursion.
Therefore, the running time of the whole phase 
is dominated by the bottom level of the recursion,
where we have $O(m)$ matrix multiplications, each running in time $O(n)$.
The running time of the whole phase is $O(m) \cdot O(n) = O(mn)$.

\item[Second phase.]
Starting at the top recursion level,
the number of seaweed matrix multiplications 
halves in every recursion level where the parameter $t$ is odd,
which is at least a half of all the recursion levels.
In a recursion level where $t$ is even, 
the number of matrix multiplications does not change.
The running time for each matrix multiplication is $O(n \log s)$,
where the parameter $s$ doubles 
in every recursion level where the parameter $t$ is odd,
and does not change in every recursion level where the parameter $t$ is even.
Therefore, the running time of the whole phase 
is dominated by the top level of the recursion,
where we have $O(m)$ matrix multiplications, each running in time $O(n)$.
The running time of the whole phase is $O(m) \cdot O(n) = O(mn)$.

\item[Total.]
The running time for both the first and the second phase,
and therefore the overall running time, is $O(mn)$.
\end{algorithm}
%
\begin{comment}
\item[Memory.]
Storing the implicit highest-score matrices
for all $O(m)$ canonical substrings requires memory $O(mn)$.
However, only the matrices for canonical substrings of length at most $w$
are actually required,
and even these matrices need not be all stored simultaneously.
By processing the canonical substrings incrementally 
within a sliding window of length $w$,
and discarding the data after the window has passed,
%at any given instant we only need to store implicit highest-score matrices
%for $O(w)$ canonical substrings that fit in the current window.
%Hence, 
the memory cost can be reduced to $O(wn)$.
%The details will be given in the full version of the paper.
\end{comment}
%
Note that the asymptotic running time of \algref{alg-window-lcs}
is independent of the window length $w$.

\begin{example}
\figref{f-window} shows an execution of \algref{alg-window-lcs}
on string $a$ of length $16$ with window size $7$,
against string $b$ of arbitrary length.
\sfigref{f-window-1} shows the canonical substrings of $a$
of lengths $1$, $2$, $4$, $8$, $16$ in black,
and windows of length $7$ in blue.
For each window, the figure shows its decomposition into canonical substrings.
For one of the windows, highlighed in thick red,
the corresponding area is outlined in the alignment dag.
\sfigref{f-window-2} represents substrings of $a$ by points in the plane,
and the decompositions into canonical substrings by a forest of trees.
Each window $a\ang{i:j}$ corresponds to a leaf of a decomposition tree.
The canonical substrings in a decomposition of the window
correspond to the edges on the path from this leaf to the root of the tree; 
for the window shown in red, its corresponding path is also highlighted in red. 
The leaves, internal nodes and roots of decomposition trees
are shown respectively by diamonds, circles and squares.
\end{example}

%%=-=-=-=-=-=-=-=-=-=-=-=-=-=-=-=-=-=-=-=-=-=-=-=-=-=-=-=-=-=-=-=-=-=-=-=-=-=%%
\mysection{Quasi-local LCS}
\label{s-quasi}

\index{substring!prescribed}%
We now generalise window-local string comparison as follows.
Suppose that we are given an arbitrary set of \emph{prescribed} substrings 
of various lengths in string $a$.
We assume that all the prescribed substrings are non-empty,
and denote their number by $k$, $m \leq k \leq \binom{m}{2}$.

\begin{definition}
\label{def-quasi-local}
\index{problem!quasi-local LCS}%
Given strings $a$, $b$, the \emph{quasi-local LCS problem} 
asks for the LCS score of every prescribed substring in $a$ 
against every substring in $b$.
\end{definition}
The quasi-local LCS problem includes as special cases
the semi-local, window-window, window-substring and fully-local LCS problems,
as well as length-constrained local alignment
considered by Arslan and E{\u g}ecio{\u g}lu \cite{Arslan_Egecioglu:02}.
The solution of the quasi-local LCS problem
can be represented in space $O(kn \log n)$
by the data structure of \thref{th-query},
built on the string-substring seaweed matrix 
for each prescribed substring of $a$ against $b$.
An individual quasi-local LCS score query 
can be performed on this data structure in time $O(\log^2 n)$.

A straightforward algorithm for the quasi-local LCS problem
can be obtained by applying \algref{alg-seaweed} (seaweed combing)
independently to each prescribed substring $a$ against whole string $b$.
The resulting algorithm runs in time $O(kmn)$.
If all the prescribed substrings are sufficiently long,
then the running time can be improved slightly
by the micro-block speedup (\algref{alg-seaweed-micro}).

We now give an algorithm for the quasi-local LCS problem
that improves substantially on the above approach.
\begin{figure}[tb]
\centering
\subfloat[\label{f-quasi-1}%
Canonical and prescribed substrings]{%
\beginpgfgraphicnamed{f-quasi-1}%
\begin{tikzpicture}[x=0.25cm, y=-0.25cm]
\tikzstyle{every to}=[to path = {-- +(-0.5,0) |- (\tikztotarget)}]
\path (0,0) coordinate(p);
\foreach \k in {0,1,...,15}{%
  \draw (p) to ++(0,1) coordinate(p);}
\path (1,0) coordinate(p);
\foreach \k in {0,2,...,15}{%
  \draw (p) to ++(0,2) coordinate(p);}
\path (2,0) coordinate(p);
\foreach \k in {0,4,...,15}{%
  \draw (p) to ++(0,4) coordinate(p);}
\path (3,0) coordinate(p);
\foreach \k in {0,8,...,15}{%
  \draw (p) to ++(0,8) coordinate(p);}
\path (4,0) coordinate(p);
\draw (p) to ++(0,16);
\path (6,0) coordinate(p);
\draw[blue] (p) ++(0,4) to ++(0,2);
\path (7,0) coordinate(p);
\draw[blue] (p) ++(0,10) to ++(0,2) to ++(0,1);
\path (8,0) coordinate(p);
\draw[blue] (p) ++(0,7) to ++(0,1) to ++(0,4);
\path (9,0) coordinate(p);
\draw[blue] (p) ++(0,7) to ++(0,1) to ++(0,4) to ++(0,1);
\path (10,0) coordinate(p);
\draw[blue] (p) ++(0,0) to ++(0,4) to ++(0,2) to ++(0,1);
\draw[blue] (p) ++(1,1) to ++(0,1) to ++(0,2) to ++(0,4);
\draw[blue] (p) ++(2,2) to ++(0,2) to ++(0,4) to ++(0,1);
\draw[blue] (p) ++(3,6) to ++(0,2) to ++(0,4) to ++(0,1);
\draw[blue] (p) ++(4,7) to ++(0,1) to ++(0,4) to ++(0,2);
\draw[blue] (p) ++(5,9) to ++(0,1) to ++(0,2) to ++(0,4);
\path (16,0) coordinate(p);
\draw[blue] (p) ++(0,2) to ++(0,2) to ++(0,4) to ++(0,2);
\draw[red,ultra thick] (p) ++(1,3) to ++(0,1) to ++(0,4) to ++(0,2) to ++(0,1);
\draw[blue] (p) ++(2,8) to ++(0,8);
\path (19,0) coordinate(p);
\draw[blue] (p) ++(0,2) to ++(0,2) to ++(0,4) to ++(0,4);
\path (20,0) coordinate(p);
\draw[blue] (p) ++(0,2) to ++(0,2) to ++(0,4) to ++(0,8);
\path (21,0) coordinate(p);
\draw[blue] (p) ++(0,0) to ++(0,8) to ++(0,4) to ++(0,2) to ++(0,1);
\path (22,0) coordinate(p);
\fill[red!25] (p) ++(0,3) rectangle +(10,8);
\draw[dotted] (p) ++(0,3) -- +(10,0) ++(0,1) -- +(10,0) 
  ++(0,4) -- +(10,0) ++(0,2) -- +(10,0) ++(0,1) -- +(10,0);
\draw[thick] (p) +(10,0) -- +(0,0) -- +(0,16) -- +(10,16);
\draw[thick,decorate,decoration=random steps] (p) +(10,0) -- +(10,16);
\end{tikzpicture}
\endpgfgraphicnamed}

\subfloat[\label{f-quasi-2}%
Prescribed substring decomposition forest]{%
\hspace{3cm}
\beginpgfgraphicnamed{f-quasi-2}%
\begin{tikzpicture}[x=0.3cm, y=-0.3cm]
\tikzstyle{every node}=[inner sep=0pt,outer sep=0pt,minimum size=0pt]
\draw[very thick,dotted] (0,0) rectangle (16,16);
\draw (0,0) -- (16,16);
\foreach \j/\i in {0/0}
  \draw[dotted] (\j,\i) ++(8,0) -- ++(0,16) (\j,\i) ++(0,8) -- ++(16,0);
\foreach \j/\i in {0/0,8/0,8/8}
  \draw[dotted] (\j,\i) ++(4,0) -- ++(0,8) (\j,\i) ++(0,4) -- ++(8,0);
\foreach \j/\i in {4/0,8/0,8/4,8/8,4/4,12/0,12/4,12/8,12/12}
  \draw[dotted] (\j,\i) ++(2,0) -- ++(0,4) (\j,\i) ++(0,2) -- ++(4,0);
\foreach \j/\i in {6/0,8/2,10/2,12/6,12/10,14/0,14/8}
  \draw[dotted] (\j,\i) ++(1,0) -- ++(0,2) (\j,\i) ++(0,1) -- ++(2,0);
\draw[blue,thick]
  (7,0) node{\bulletd} -- ++(-1,0) node{\bulletc} -- 
    ++(-2,0) node{\bulletc} -- ++(-4,0) node{\bullets}
  (15,0) node{\bulletd} -- ++(-1,0) node{\bulletc} -- 
    ++(-2,0) node{\bulletc} -- ++(-4,0) node{\bulletc} to[bend left] 
    ++(-8,0) node{\bullets}
  (8,1) node{\bulletd} -- ++(0,1) node{\bulletc} -- 
    ++(0,2) node{\bulletc} -- ++(0,4) node{\bullets}
  (9,2) node{\bulletd} -- ++(-1,0) node{\bulletc}
  (10,2) node{\bulletd} to[bend right] ++(-2,0) node{\bulletc}
  (12,2) node{\bulletd} to[bend right] ++(-4,0) node{\bulletc}
  (16,2) node{\bulletd} to[bend right] ++(-8,0) node{\bulletc}
  (11,3) node{\bulletd} -- ++(0,1) node{\bulletc} -- 
    ++(-1,0) node{\bulletc} -- ++(-2,0) node{\bulletc} -- 
    ++(0,4) node{\bullets}
  (6,4) node{\bulletd} -- ++(-2,0) node{\bullets}
  (13,6) node{\bulletd} -- ++(-1,0) node{\bulletc} to[bend right] ++(0,2)
    node{\bulletc} -- ++(-4,0) node{\bullets}
  (12,7) node{\bulletd} -- ++(0,1) node{\bulletc}
  (13,7) node{\bulletd} -- ++(0,1) node{\bulletc} -- ++(-1,0) node{\bulletc} 
  (14,7) node{\bulletd} -- ++(0,1) node{\bulletc} to[bend left] 
    ++(-2,0) node{\bulletc} 
  (16,8) node{\bulletd} to[bend right] ++(0,8) node{\bullets}
  (16,9) node{\bulletd} -- ++(0,1) node{\bulletc} -- 
    ++(0,2) node{\bulletc} -- ++(0,4) node{\bullets}
  (13,10) node{\bulletd} -- ++(-1,0) node{\bulletc} -- ++(0,2) node{\bullets};
\draw[red,ultra thick]
  (11,3) node{\bulletd} -- ++(0,1) node{\bulletc} -- 
    ++(-1,0) node{\bulletc} -- ++(-2,0) node{\bulletc} -- 
    ++(0,4) node{\bullets};
\end{tikzpicture}
\endpgfgraphicnamed
\hspace{3cm}}
\caption{\label{f-quasi} An execution of \algref{alg-quasi-lcs}
(quasi-local LCS)}
\end{figure}
\begin{algorithm}[Quasi-local LCS]
\label{alg-quasi-lcs}
\setlabelitbf
\nobreakitem[Input:]
strings $a$, $b$ of length $m$, $n$, respectively;
a set of $k$ endpoint index pairs for the prescribed substrings in $a$.
\item[Output:]
nonzeros of the string-substring seaweed matrix
for every prescribed substring of string $a$ against full string $b$.
\item[Description.]
The algorithm structure is similar to the one of \algref{alg-window-lcs}.

\setlabelit
\item[First phase.]
As in \algref{alg-window-lcs}.

\setlabelit
\item[Second phase.]
First, we remove from consideration all canonical prescribed substrings,
since they have already been processed in the first phase.
We then process all the remaining prescribed substrings of $a$ 
by the following recursive procedure.
Let $s$ be a parameter, assumed to be a power of $2$. 
Initially, we set $s=1$.
At every level of recursion, the endpoint indices
of the prescribed substrings are multiples of $s$.

\setlabelnormal
\item[Recursion base: the set of prescribed substrings is empty.]
In this case, the problem is trivial.

\item[Recursive step: the set of prescribed substrings is nonempty.]
We call the second phase procedure recursively with the parameter $2s$, 
and the following set of prescribed substrings.
For each currently prescribed non-canonical substring $a\ang{i,j}$,
the corresponding new prescribed substring in the recursive call
is $a\bigang{2s\bigceil{\frac{i}{2s}} : 2s\bigfloor{\frac{j}{2s}}}$,
unless this substring is empty.
Informally, we round $i$ and $j$ to a multiple of $2s$;
index $i$ is rounded up and index $j$ down.
Note that different currently prescribed substrings
may correspond to the same new prescribed substring in the recursive call.

The recursive call results in the processing 
of all the prescribed substrings $a\ang{i,j}$ 
where $i$, $j$ are multiples of $2s$;
in other words, where $\frac{i}{s}$ and $\frac{j}{s}$ are both even.
We then process all remaining prescribed substrings $a\ang{i:j}$ as follows:
\begin{gather*}
\label{eq-quasi}
\rP^\ssub_{a\ang{i:j},b} = 
\begin{cases}
\rP^\ssub_{a\ang{i:j-s},b} \boxdot \rP^\ssub_{a\ang{j-s:j},b}&
\text{$\frac{i}{s}$ even, $\frac{j}{s}$ odd}\\
\rP^\ssub_{a\ang{i:i+s},b} \boxdot \rP^\ssub_{a\ang{i+s:j},b}&
\text{$\frac{i}{s}$ odd, $\frac{j}{s}$ even}\\
\rP^\ssub_{a\ang{i:i+s},b} \boxdot \rP^\ssub_{a\ang{i+s:j-s},b} \boxdot 
\rP^\ssub_{a\ang{j-s:j},b}&
\text{$\frac{i}{s}$ odd, $\frac{j}{s}$ odd}
\end{cases}
\end{gather*}
The seaweed matrix products are computed by \thref{th-comp-mmult-comp}.
In each case, the product 
is between two or three string-substring seaweed matrices:
one matrix for substring $a\ang{i:j-s}$, $a\ang{i+s:j}$ or $a\ang{i+s:j-s}$,
already processed by the recursive call;
and the other one or two matrices for canonical substrings 
$a\ang{j-s:j}$ and/or $a\ang{i:i+s}$.
\item[(End of recursive step)]
\setlabelitbf
\nobreakitem[Cost analysis.]
\setlabelit
\item[First phase.]
As in \algref{alg-window-lcs},
the total running time of this phase is $O(mn)$.

\item[Second phase.]
In every level of the recursion, 
the number of matrix multiplications is at most $O(k)$.
The running time for each matrix multiplication is at most $O(n \log m)$.
The recursion has $\log m$ levels.
Therefore, the running time is dominated 
by the top level of the recursion,
where we have $O(m)$ matrix multiplications, each running in time $O(n)$.
Therefore, the running time of the whole phase 
is $O(k \log m \cdot n \log m) = O(kn \log^2 m)$.

For values of $k$ close to the fully-local case
$k = \binom{m}{2} = O(m^2)$,
a sharper analysis is possible.
In this case, the running time of the whole phase is $O(m^2n)$.

\item[Total.]
The overall running time is dominated by the second phase,
and is therefore $O(kn \log^2 m)$.
For values of $k$ close to $\binom{m}{2}$, the running time is $O(m^2n)$.
\end{algorithm}

Note that in the fully-local case, i.e.\ the case 
where all $\binom{m}{2}$ nonempty substrings of $a$ are prescribed,
the same asymptotic time can be obtained 
by solving the prefix-substring LCS problem
independently for each of the $m$ nonempty suffixes of string $a$.
Each of these prefix-substring LCS instances 
can be solved by a separate run of \algref{alg-seaweed} (seaweed combing).

\begin{example}
\figref{f-quasi} shows an execution of \algref{alg-window-lcs}
on string $a$ of length $16$, 
with $16$ prescribed substrings of various sizes,
against string $b$ of arbitrary length.
Conventions are the same as in \figref{f-window}.
\end{example}

%%=-=-=-=-=-=-=-=-=-=-=-=-=-=-=-=-=-=-=-=-=-=-=-=-=-=-=-=-=-=-=-=-=-=-=-=-=-=%%
\mysection{Sparse spliced alignment}
\label{s-spliced}

Assembling a gene from candidate exons 
is an important problem in computational biology.
Several alternative approaches to this problem 
have been developed over time.
\index{spliced alignment}%
One of such approaches is \emph{spliced alignment}
by Gelfand et al.\ \cite{Gelfand+:96} (see also \cite{Gusfield:97}),
which scores different candidate exon chains within a DNA sequence
by comparing them to a known related gene sequence.
In this method, the two sequences are modelled 
respectively by strings $a$, $b$ of lengths $m$, $n$ respectively.
A subset of substrings in string $a$ are marked as candidate exons.
The comparison between sequences is made by string alignment.
The algorithm for spliced alignment given in \cite{Gelfand+:96}
runs in time $O(m^2 n)$.

In general, the number of candidate exons $k$ 
may be as high as $\binom{m}{2} = O(m^2)$.
\index{spliced alignment!sparse}%
The method of \emph{sparse spliced alignment}
makes a realistic assumption that, prior to the assembly, 
the set of candidate exons undergoes some filtering,
after which only a small fraction of candidate exons remains.
Kent et al.\ \cite{Kent+:06} give an algorithm for sparse spliced alignment
that, in the special case $k=O(m)$, runs in time $O(m^{1.5}n)$.
By a direct application of the quasi-local LCS problem (\secref{s-quasi}),
the running time can be reduced to $O(mn \log^2 m)$.
Sakai \cite{Sakai:11} gave an improved algorithm,
running in time $O(mn \log n)$.

For higher values of $k$, all the described algorithms provide 
a smooth transition in running time to the dense case $k=\binom{m}{2}$.
In this case, the algorithms' running time $O(m^2 n)$ 
is asymptotically equal to the algorithm of \cite{Gelfand+:96}.

We now describe an algorithm for sparse spliced alignment,
based on the approach of \cite{Sakai:11}.
We keep the notation and terminology of the previous sections;
in particular, candidate exons are represented 
by prescribed substrings of string $a$.
We say that substring $a\ang{i':j'}$ 
\emph{precedes} substring $a\ang{i'':j''}$, if $j' \leq i''$.
A \emph{precedence chain} of substrings is a chain
in the partial order of substring precedence.
We identify every precedence chain with the string obtained
by concatenating all its component substrings in the order of precedence.

The algorithm uses a generalisation of the standard network alignment method,
equivalent to the one used by \cite{Kent+:06}.
For simplicity, we describe the algorithm under LCS score;
using the blow-up technique of \secref{s-weighted}, 
the algorithm can be generalised 
to an arbitrary alignment score with rational weights.

\begin{algorithm}[Sparse spliced alignment]
\label{alg-spliced}
\setlabelitbf
\nobreakitem[Input:]
strings $a$, $b$ of length $m$, $n$, respectively;
a set of $k$ endpoint index pairs for the prescribed substrings in $a$.
\item[Output:]
the precedence chain of prescribed substrings in $a$,
giving the highest LCS score against string $b$.
\item[Description.]
The algorithm runs in two phases.
\setlabelit
\item[First phase.]
As in \algref{alg-window-lcs}.

\item[Second phase.]
The problem is now solved by dynamic programming as follows.
Let $u_j(s)$ denote the highest LCS score 
across all precedence chains of prescribed substrings 
in prefix string $a \ltake j$, taken against prefix string $b \ltake s$.
With each $j \in \bra{0:n}$, we associate the integer vector $u_j = u_j(*)$.
We initialise $u_0$ as the zero vector.
We then compute the vectors $u_j$, $j \in \bra{1:n}$, 
in the order of increasing $j$.
Let $a_0=a\ang{i_0:j}$, $a_1=a\ang{i_1:j}$, \ldots, $a_t=\ang{i_t:j}$
be all the prescribed substrings of $a$ terminating at index $j$.
We have
\begin{gather}
\label{eq-spliced}
u_j = u_{j-1} \oplus 
\bigpa{u_{i_0} \odot \rH^\ssub_{a_0,b}} \oplus
\bigpa{u_{i_1} \odot \rH^\ssub_{a_1,b}} \oplus \ldots \oplus
\bigpa{u_{i_t} \odot \rH^\ssub_{a_t,b}}
\end{gather}
for all $j \in \bra{1:n}$.

The matrices $\rH^\ssub_{a_0,b}$, $\rH^\ssub_{a_1,b}$, \ldots, $\rH^\ssub_{a_t,b}$
do not need to be computed explicitly.
Instead, it is straightforward to compute 
each of the vector-matrix distance products in \eqref{eq-spliced}
by up to $\log m$ instances of implicit vector-matrix distance product,
using the decomposition of each of $a_0$, $a_1$, \ldots, $a_t$
into up to $\log m$ canonical substrings, 
along with the corresponding seaweed matrices obtained in the first phase.

The solution score is now given by the value $u_m(n)$.
The solution precedence chain of prescribed substrings 
can be obtained by tracing the dynamic programming sequence backwards
from vector $u_m$ to vector $u_0$. 
\setlabelitbf
\item[Cost analysis.]
\setlabelit
\item[First phase.]
As in \algref{alg-window-lcs},
the total running time of this phase is $O(mn)$.

\item[Second phase.]
For each of $k$ prescribed substrings of $a$,
we execute up to $\log m$ instances 
of implicit matrix-vector distance multiplication.
Every such instance runs in time $O(n)$ by \thref{th-mvmult}.
Therefore, the total running time of this phase is $O(kn \log m)$.

\item[Total.]
The overall running time of the algorithm is dominated by the second phase, 
and is therefore $O(kn \log m)$.
\end{algorithm}

Similarly to \algref{alg-quasi-lcs}, 
a sharper analysis for values of $k$ close to $\binom{m}{2}$ 
leads to a smooth transition to the running time $O(m^2 n)$ 
in the dense case $k = \binom{m}{2}$,
which is asymptotically equal to the running time 
of the dense spliced algorithm of \cite{Gelfand+:96}.

%Also matches alignment with non-overlapping inversions in $O(n^3 \log n)$
%\cite{Alves+:05_1}, in the case of unit costs.

\begin{comment}
%%=-=-=-=-=-=-=-=-=-=-=-=-=-=-=-=-=-=-=-=-=-=-=-=-=-=-=-=-=-=-=-=-=-=-=-=-=-=%%
\mysection{Threshold local LCS}

Non-overlapping repeats:
see Benson \cite{Benson:95}, Kannan and Myers \cite{Kannan_Myers:96}.

%%=-=-=-=-=-=-=-=-=-=-=-=-=-=-=-=-=-=-=-=-=-=-=-=-=-=-=-=-=-=-=-=-=-=-=-=-=-=%%
\mysection{Two-round local LCS}

Preprocessing $\tilde O(mn)$.

Query substring of $a$: extra $\tilde O(n)$.

Query substring of $b$: extra $\tilde O(1)$.

Local LCS on permutations. Generalises planar chains.

\end{comment}

%%===========================================================================%%

%%===========================================================================%%
\mychapter{Conclusions}

We have presented a number of existing and new 
algorithmic techniques and applications
related to semi-local string comparison.
Our approach unifies a substantial number of previously unrelated
problems and techniques, 
and in many cases allows us to match or improve existing algorithms.

A number of questions related to the semi-local string comparison framework
remain open.
In particular, it is not yet clear whether the framework 
can be extended to arbitrary real costs,
or to sequence alignment with non-linear gap penalties.

In summary, semi-local string comparison 
turns out to be a powerful algorithmic technique, 
which unifies, and often improves on, a number of previous approaches 
to various substring- and subsequence-related problems.
It is likely that further development of this approach
will give it even more scope and power.

%%===========================================================================%%
\chapter{Acknowledgement}

This work was conceived in a discussion with Gad Landau in Haifa.
The imaginative term ``seaweeds'' was coined by Yuri Matiyasevich 
during a presentation by the author in St.~Petersburg.
I thank El\.zbieta Babij, Philip Bille, Pawe\l{} Gawrychowski, Tim Griffin,
Dima Grigoriev, Peter Krusche, Gad Landau, Victor Levandovsky, Sergei Nechaev, 
Dima Pasechnik, Lu\'is Russo, Andrei Sobolevski, 
Nikolai Vavilov, Oren Weimann, and Michal Ziv-Ukelson 
for fruitful discussions,
and many anonymous referees for their comments that helped to improve this work.

%%===========================================================================%%

\bibliographystyle{plain}
\bibliography{auto,books,algebra,string,graph,sort,scan}
\printindex

%%===========================================================================%%
\end{document}